\documentclass[12pt,thesis,letter]{book}
\usepackage{graphicx}

\usepackage{amsmath}
\usepackage{amssymb}
\usepackage{amsfonts}
\usepackage{amsthm}

\newcommand{\Z}{\mathbb{Z}}
\newcommand{\Zp}{\Z_{+}}
\newcommand{\Real}{\mathbb{R}}
\newcommand{\Comp}{\mathbb{C}}

\newcommand{\Nset}{\mathbb{N}}

\newcommand{\pow}{\mathop{\rm pow}}
\newcommand{\lift}{\mathbf{L}_h}

\newcommand{\dlift}[1]{\mathbb{L}_{#1}}
\newcommand{\idlift}[1]{\mathbb{L}_{#1}^{-1}}
\newcommand{\dliftsys}[2]{\mathcal{L}_{#1}(#2)}
\newcommand{\hold}[1]{\mathcal{H}_{#1}}
\newcommand{\ghold}[1]{\widetilde{\mathcal{H}}_{#1}}
\newcommand{\holdfunc}{\bf{H}}
\newcommand{\samp}[1]{\mathcal{S}_{#1}}
\newcommand{\gsamp}[1]{\widetilde{\mathcal{S}}_{#1}}
\newcommand{\norm}[1]{\|#1\|}

\newcommand{\twonorm}[1]{\|#1\|}
\newcommand{\us}[1]{\uparrow\!#1}
\newcommand{\ds}[1]{\downarrow\!#1}

\newcommand{\lft}{{\mathcal F}_{l}}
\newcommand{\sinc}{\mathrm{sinc}}

\newcommand{\sspc}[4]{\left[ \begin{array}{c|c}#1&#2\\ \hline #3&#4\end{array}\right]}
\newcommand{\gss}[9]{
  \left[ \begin{array}{c|cc}
      #1&#2&#3\\ 
      \hline
      #4&#6&#7\\
      #5&#8&#9\\
  \end{array}\right]
}

\newcommand{\B}{\mathbf{B}}
\newcommand{\C}{\mathbf{C}}
\newcommand{\D}{\mathbf{D}}
\newcommand{\T}{{\mathcal T}}

\newcommand{\W}{{\mathcal W}}
\newcommand{\G}{{\mathbf G}}

\newcommand{\dd}{\|d\|_\infty\leq\frac{\Delta}{2}}
\newcommand{\vect}[1]{\boldsymbol{#1}}

\newtheorem{theorem}{Theorem}[chapter]
\newtheorem{deff}{Definition}[chapter]
\newtheorem{problem}{Problem}[chapter]
\newtheorem{lemma}{Lemma}[chapter]
\newtheorem{prop}{Proposition}[chapter]

\begin{document}
\pagenumbering{roman}
\begin{titlepage}
  \begin{center}
    \begin{huge}
      \begin{bf}
         Multirate Digital Signal Processing
	 via Sampled-Data $H^\infty$ Optimization \vspace{2.7cm}\\
      \end{bf} 
    \end{huge}
    \begin{Large}
      \begin{bf}
        Dissertation \vspace{0.7cm}\\ 
        Submitted in partial fulfillment of\\ 
        the requirements for the degree of\\ 
        Doctor of Informatics \vspace{1.5cm}\\
        Masaaki Nagahara \vspace{1.4cm}\\
%        Department of Applied Analysis\\
%	and Complex Dynamical Systems \vspace{0.7cm}\\ 
	Graduate School of Informatics \vspace{0.7cm}\\
        Kyoto University \vspace{1.4cm}\\
        February 2003
      \end{bf}
    \end{Large}
  \end{center}
\end{titlepage}
\pagestyle{empty}
\newpage
\section*{Abstract}
In this thesis, we present a new method for designing multirate signal processing
and digital communication systems via sampled-data $H^\infty$ control theory.
The difference between our method and conventional ones is in the signal spaces.
Conventional designs are executed in the discrete-time domain,
while our design takes account of both the discrete-time and the continuous-time signals.
Namely, our method can take account of the characteristic of the original analog signal
and the influence of the A/D and D/A conversion.
While the conventional method often indicates that an ideal digital low-pass filter is preferred,
we show that the optimal solution need not be an ideal low-pass when the original analog signal
is not completely band-limited.
This fact can not be recognized only in the discrete-time domain.
Moreover, we consider quantization effects.
We discuss the stability and the performance of quantized sampled-data control systems.
We justify $H^\infty$ control to reduce distortion caused by the quantizer.
Then we apply it to differential pulse code modulation.
While the conventional $\Delta$ modulator is not optimal and besides not stable,
our modulator is stable and optimal with respect to the $H^\infty$-norm.
We also give an LMI (Linear Matrix Inequality) solution to the optimal $H^{\infty}$ 
approximation of IIR (Infinite Impulse Response) filters via FIR (Finite Impulse Response) filters.  
A comparison with the Nehari shuffle is made with
a numerical example, and it is observed that 
the LMI solution generally performs better.  
Another numerical study also indicates that 
there is a trade-off between the pass-band and
stop-band approximation characteristics.
\newpage
\section*{Acknowledgments}
I would like to express my sincere gratitude to everyone who
helped me and contributed in various ways toward completion of 
this work.

First of all, I am most grateful to my supervisor Professor
Yutaka Yamamoto.  Since I entered the Graduate School of Kyoto
University, he has constantly guided me in every respect of
science and technology.  Without this knowledge he has endowed
me, it is unthinkable that I could proceed this far toward 
the completion of this thesis.  He also taught me fundamental
knowledge on sampled-data control and signal processing on
which this thesis is based.  His strong leadership and ideas
have been a constant source of encouragement and have had
a definite influence on this work.  Without his supervision,
this thesis would not exist today.  

I would also like to thank Dr.~Hisaya Fujioka
for his helpful discussions and support in my research,
in particular, those on sampled-data control.
His warm encouragement was also great help to completing 
this thesis.

My gratitude should be extended to Dr.~Yuji Wakasa
for his valuable advice.  In particular, his advice on convex optimization
was very valuable to my research.

I would thank Mr.~Kenji Kashima, 
%Ms.~Yoko Koyanagi, 
Mr.~Shinichiro Ashida,
and current and past members of the intelligent and control systems laboratory
for their academic stimulation and dairy friendship.

Special thanks are addressed toward my parents for their deep understanding
and encouragement.

This research was partially supported by Japan Society for the Promotion 
of Science (JSPS).
\newpage
\pagestyle{myheadings}
\markboth{\sl Contents}{\sl Contents}
\tableofcontents
\newpage
\pagestyle{myheadings}
\markboth{\sl Notation}{\sl Notation}
\section*{Notation}

\begin{description}
  \item [$\Zp$]\hspace{-2mm}: non-negative integers. 
  \item [$\Nset$]\hspace{-2mm}: natural numbers.
  \item [$\Real$]\hspace{-2mm}: real numbers.
  \item [$\Real_+$]\hspace{-2mm}: non-negative real numbers.
  \item [$\Real^n$]\hspace{-2mm}: $n$-dimensional vector space over $\Real$.
  \item [$\Comp$]\hspace{-2mm}: complex numbers.
  \item [$l^2$]\hspace{-2mm}: real-valued square summable sequences.
  \item [$l^\infty$]\hspace{-2mm}: real-valued bounded sequences.
  \item [$L^2[0,\infty)$ and $L^2[0,h)$]\hspace{-2mm}:
        Lebesgue spaces consisting of square integrable real 
        functions on $[0,\infty)$ and $[0,h)$, respectively.
	$L^2[0,\infty)$ may be abbreviated to $L^2$.
  \item [$l^2_{L^2[0,h)}$]\hspace{-2mm}:
        square summable sequences whose values are in $L^2[0,h)$.
  \item [$\sspc{A}{B}{C}{D}$]\hspace{-2mm}: transfer function whose realization is $\{A,B,C,D\}$,
	that is, $\sspc{A}{B}{C}{D}(\lambda):=C(\lambda I-A)^{-1}B+D$, where $\lambda:=s$ in continuous-time
	and $\lambda:=z$ in discrete-time.
  \item [$\lft(P,K)$]\hspace{-2mm}: linear fractional transformation of $P$ and $K$, that is,
	if $P=\begin{bmatrix}P_{11}&P_{12}\\P_{21}&P_{22}\end{bmatrix}$ then $\lft(P,K):=P_{11}+P_{12}K(I-P_{22}K)^{-1}P_{21}$.
  \item [$\samp{h}$]\hspace{-2mm}: ideal sampler with sampling period $h$.
  \item [$\hold{h}$]\hspace{-2mm}: zero order hold with sampling period $h$.
  \item [$\lift$]\hspace{-2mm}: lifting for continuous-time signals with sampling period $h$.
  \item [$\dlift{N}$]\hspace{-2mm}: lifting for discrete-time signals by factor $N$ (discrete-time lifting).
  \item [$\dliftsys{N}{P}$]\hspace{-2mm}: fast-discretizing and discrete-time lifting of continuous system $P$, that is,
		$\dliftsys{N}{P}:=\dlift{N}\samp{h/N}P\hold{h/N}\idlift{N}$.
  \item [$\us{M}$]\hspace{-2mm}: upsampler with upsampling factor $M$.
  \item [$\ds{M}$]\hspace{-2mm}: downsampler with downsampling factor $M$.
  \item [$r(A)$]\hspace{-2mm}: maximum absolute value of eigenvalues of matrix $A$.
  \item [$A^T$]\hspace{-2mm}: transpose of a matrix (or a vector) $A$.
\end{description}
\newpage
\setcounter{page}{1}
\pagenumbering{arabic}
\pagestyle{headings}

\chapter{Introduction}
\label{ch:introduction}
This thesis presents a new method in
multirate digital signal processing and digital communication systems.

When we execute these design procedures,
we must first discretize
%\footnote{Generally, discretization means sampling and quantizing. 
%However, in this chapter we are concerned only with the sampling.}
the original analog signal (e.g., speech, audio or visual image),
then the discretized signal is processed in the discrete-time domain
(e.g., filtering, compressing, transmitting etc.),
and finally we reconstruct an analog signal from the discrete signal.

Conventionally, the design is performed mostly in the discrete-time domain
by assuming that the original analog signal is fully band-limited up to the Nyquist frequency.
Under this assumption, the sampling theorem gives a method for reconstructing
an analog signal from a sampled signal \cite{Jer77,Zay}.
\begin{theorem}[Shannon et al. ]
\label{th:shannon}
Let $f(t)$ be a continuous-time signal ideally band-limited to the range 
$(-\pi/h, \pi/h)$, that is, its Fourier transform $\hat{f}(\omega)$ is zero outside this interval.
Then $f(t)$ can uniquely be recovered from its sampled values $f(nh),\; n=0, \pm 1, \pm 2, \ldots$ 
via the formula
\begin{equation}
f(t) = \sum_{n=-\infty}^{\infty} f(nh) \frac{\sin \pi(t/h-n)}{\pi(t/h-n)} = \sum_{n=-\infty}^{\infty} f(nh) \sinc (t/h-n).
\label{eq:shannon}
\end{equation}
\end{theorem}
Although most of the conventional studies of digital signal processing are based on Theorem \ref{th:shannon}
\cite{Fli,Pro,ProSal,Vai,Zay,Zel}, 
we encounter two questions in the implementation: the question of D/A and A/D conversions.

We consider the first question.
The process (\ref{eq:shannon}) is a kind of D/A conversion;
we convert sampled values $\{f(nh)\}$ to modulated impulses $\{f(nh)\delta(nh)\}$,
then filter them by the ideal low-pass filter\footnote{
The frequency response of this filter is equal to 1 up to the Nyquist frequency $\omega_N$ 
and zero beyond $\omega_N$.} \cite{Zel}.
In practice, this conversion is physically impossible,
and in reality a zero-order hold followed by a sharp low-pass filter is often used.
We should notice that the effect of such a real situation of
D/A conversion can never be taken into account only in the discrete-time domain.

The question that we must consider next is the assumption of full band-limitation.
The assumption that Theorem \ref{th:shannon} requires is in reality impossible
because no real analog signal is fully band-limited.
In order to achieve the assumption, a sharp low-pass filter is attached before sampling.
However, such a sharp low-pass characteristic will deteriorate the quality of the analog signal,
and moreover even such a filter does not satisfy the assumption.
Therefore we have to consider the effect of sampling (i.e., {\em aliasing}),
which can never be captured in the discrete-time domain.

To answer these questions, we must take account of not only discrete-time signals but also continuous-time signals.
Therefore, we propose to consider an analog performance to design digital systems
by using the modern {\em sampled-data control theory}.
Sampled-data control theory deals with control systems that consist
of continuous-time plants to be controlled, discrete-time controllers controlling them,
and ideal A/D and D/A converters.
The modern theory can take the intersample behavior of 
sampled-data control systems into account.
The key idea is {\em lifting} \cite{Yam90,BamPeaFraTan91,Yam94}.
Although sampled-data systems are not time-invariant,
the lifting method leads to time-invariant discrete-time models.
These models are infinite-dimensional, 
which can be reduced to equivalent discrete-time finite-dimensional ones
\cite{BamPea92,CheFra,KabHar93}.

There is another method for sampled-data systems: {\em fast-sampling/fast-hold} (FSFH) method \cite{KelAnd92,YamMadAnd99}.
The idea is to approximate the continuous-time inputs by step functions of sufficiently small step size
and also approximate the continuous-time outputs by taking their samples by sufficiently fast ideal sampler.
The approximated system will be a periodically time-varying discrete-time system (finite-dimensional),
which can be transformed a time-invariant discrete-time system by using the {\em discrete-time lifting}
\cite{Mey90,CheFra}

Based on these studies, the $H^\infty$ optimal sampled-data control
has been studied \cite{Toi92,BamPea92,KabHar93,HayHarYam94}.
The $H^\infty$ optimality criterion is suitable for the frequency characteristic,
which is intuitive for one who designs the system.
%Furthermore, the $H^\infty$ norm of a sampled-data system is equal to its $L^2$ induced norm \cite{NagFoi,CheFra}.

In the last few years, several articles have studied digital signal processing
via sampled-data control theory.
Chen and Francis \cite{CheFra95} studied a multirate filter bank design problem
with an $H^\infty$ criterion.
Although this design was done in the discrete-time domain, they brought the modern $H^\infty$ control theory
to the digital signal processing, and thus they threw new light on the subject.
The first study of digital signal processing in the sampled-data setting was made by Khargonekar and Yamamoto \cite{KhaYam96}.
They formulated a single-rate signal reconstruction problem by using sampled-data theory.
After their study, this method was developed to multirate signal processing \cite{YamFujKha97,IsiYam98,IsiYamFra99,YamNagFuj00,NagYam00},
digital communications \cite{NagYam01b, NagYam03}
and quantizer design \cite{NagAshYam02}.

The purpose of this thesis is to answer the questions 
mentioned above by sampled-data $H^\infty$ optimal control theory,
and to show that this method is effective in designing digital systems.

The organization of this thesis is as follows.
\begin{itemize}
\item Chapter \ref{ch:SDC} surveys sampled-data control theory.
We describe the fundamental difficulty in sampled-data systems, 
and introduce the lifting method, 
which resolved this difficulty:
by using it 
we can define the frequency response and the $H^\infty$ optimal control.
The fast-sampling/fast-hold method for multirate sampled-data control systems is
discussed in detail, because we mainly use this method in this thesis.
\item Chapter \ref{ch:multirate} deals with multirate signal processing,
in particular, interpolation, decimation and sampling rate conversion.
We present a new method for designing these systems via the sampled-data $H^\infty$
optimization.
Design examples shows that our method is superior to the conventional one.
\item Chapter \ref{ch:comm} presents a new design of digital communication systems.
Under signal compression and channel distortion, we design an optimal transmitting/receiving
filter by using the sampled-data $H^\infty$ optimization. 
We show also design examples to indicate that our method is effective in
digital communication.
\item Chapter \ref{ch:qqq} investigates issues of quantization.
Since quantization is a nonlinear operation, we introduce a linearized model
(i.e., additive noise model).
By using this model, we discuss stability and performance of a quantized sampled-data
control system.
Then we apply it to differential pulse code modulation (DPCM) systems.
Design examples are shown and our design is superior to the conventional
$\Delta$ modulation.
\item Chapter \ref{ch:firapp} presents a new method to approximate an IIR
filter by an FIR filter, which directly yields an
optimal approximation with respect to the $H^{\infty}$ error norm.
We show that this design
problem can be reduced to an LMI (Linear Matrix Inequality).
We will also make a comparison via a numerical example 
with an existing method, known as the Nehari shuffle.
\item Chapter \ref{ch:conclusion} concludes this thesis with a summary of the results presented
and future perspectives.
\end{itemize}
\chapter{Sampled-Data Control Theory}
\label{ch:SDC}
\section{Sampled-data control systems}
A sampled-data control system is a system 
in which a continuous-time plant is to be controlled by
a discrete-time controller.
Consider the unity-feedback sampled-data control system shown in Figure \ref{fig:SDC_sdfeedback}.
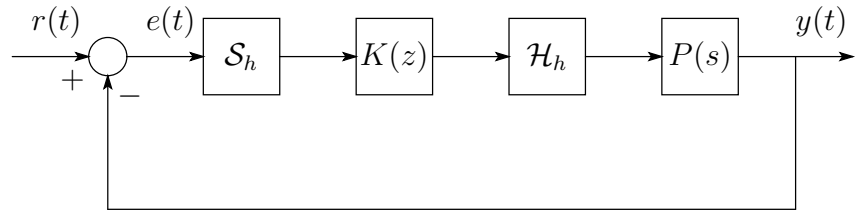
\begin{figure}[t]
\begin{center}
%\input{SDC_sdfeedback}
%WinTpicVersion2.15
\unitlength 0.1in
\begin{picture}(44.00,10.70)(4.00,-12.00)
% VECTOR 2 0 3 0
% 2 400 800 800 800
% 
\special{pn 8}%
\special{pa 400 400}%
\special{pa 800 400}%
\special{fp}%
\special{sh 1}%
\special{pa 800 400}%
\special{pa 733 380}%
\special{pa 747 400}%
\special{pa 733 420}%
\special{pa 800 400}%
\special{fp}%
% CIRCLE 2 0 3 0
% 4 900 800 800 800 800 800 800 800
% 
\special{pn 8}%
\special{ar 900 400 100 100  0.0000000 6.2831853}%
% VECTOR 2 0 3 0
% 2 1000 800 1400 800
% 
\special{pn 8}%
\special{pa 1000 400}%
\special{pa 1400 400}%
\special{fp}%
\special{sh 1}%
\special{pa 1400 400}%
\special{pa 1333 380}%
\special{pa 1347 400}%
\special{pa 1333 420}%
\special{pa 1400 400}%
\special{fp}%
% BOX 2 0 3 0
% 2 1400 600 1800 1000
% 
\special{pn 8}%
\special{pa 1400 200}%
\special{pa 1800 200}%
\special{pa 1800 600}%
\special{pa 1400 600}%
\special{pa 1400 200}%
\special{fp}%
% VECTOR 2 0 3 0
% 2 1800 800 2200 800
% 
\special{pn 8}%
\special{pa 1800 400}%
\special{pa 2200 400}%
\special{fp}%
\special{sh 1}%
\special{pa 2200 400}%
\special{pa 2133 380}%
\special{pa 2147 400}%
\special{pa 2133 420}%
\special{pa 2200 400}%
\special{fp}%
% BOX 2 0 3 0
% 2 2200 600 2600 1000
% 
\special{pn 8}%
\special{pa 2200 200}%
\special{pa 2600 200}%
\special{pa 2600 600}%
\special{pa 2200 600}%
\special{pa 2200 200}%
\special{fp}%
% VECTOR 2 0 3 0
% 2 2600 800 3000 800
% 
\special{pn 8}%
\special{pa 2600 400}%
\special{pa 3000 400}%
\special{fp}%
\special{sh 1}%
\special{pa 3000 400}%
\special{pa 2933 380}%
\special{pa 2947 400}%
\special{pa 2933 420}%
\special{pa 3000 400}%
\special{fp}%
% BOX 2 0 3 0
% 2 3000 600 3400 1000
% 
\special{pn 8}%
\special{pa 3000 200}%
\special{pa 3400 200}%
\special{pa 3400 600}%
\special{pa 3000 600}%
\special{pa 3000 200}%
\special{fp}%
% VECTOR 2 0 3 0
% 2 3400 800 3800 800
% 
\special{pn 8}%
\special{pa 3400 400}%
\special{pa 3800 400}%
\special{fp}%
\special{sh 1}%
\special{pa 3800 400}%
\special{pa 3733 380}%
\special{pa 3747 400}%
\special{pa 3733 420}%
\special{pa 3800 400}%
\special{fp}%
% BOX 2 0 3 0
% 2 3800 600 4200 1000
% 
\special{pn 8}%
\special{pa 3800 200}%
\special{pa 4200 200}%
\special{pa 4200 600}%
\special{pa 3800 600}%
\special{pa 3800 200}%
\special{fp}%
% VECTOR 2 0 3 0
% 2 4200 800 4800 800
% 
\special{pn 8}%
\special{pa 4200 400}%
\special{pa 4800 400}%
\special{fp}%
\special{sh 1}%
\special{pa 4800 400}%
\special{pa 4733 380}%
\special{pa 4747 400}%
\special{pa 4733 420}%
\special{pa 4800 400}%
\special{fp}%
% LINE 2 0 3 0
% 2 4500 800 4500 1600
% 
\special{pn 8}%
\special{pa 4500 400}%
\special{pa 4500 1200}%
\special{fp}%
% VECTOR 2 0 3 0
% 2 900 1600 900 900
% 
\special{pn 8}%
\special{pa 900 1200}%
\special{pa 900 500}%
\special{fp}%
\special{sh 1}%
\special{pa 900 500}%
\special{pa 880 567}%
\special{pa 900 553}%
\special{pa 920 567}%
\special{pa 900 500}%
\special{fp}%
% STR 2 0 3 0
% 3 500 600 500 700 2 0
% $r(t)$
\put(5.0000,-3.0000){\makebox(0,0)[lb]{$r(t)$}}%
% STR 2 0 3 0
% 3 1100 600 1100 700 2 0
% $e(t)$
\put(11.0000,-3.0000){\makebox(0,0)[lb]{$e(t)$}}%
% STR 2 0 3 0
% 3 4500 600 4500 700 2 0
% $y(t)$
\put(45.0000,-3.0000){\makebox(0,0)[lb]{$y(t)$}}%
% STR 2 0 3 0
% 3 780 770 780 870 4 0
% $+$
\put(7.8000,-4.7000){\makebox(0,0)[rt]{$+$}}%
% STR 2 0 3 0
% 3 950 850 950 950 1 0
% $-$
\put(9.5000,-5.5000){\makebox(0,0)[lt]{$-$}}%
% STR 2 0 3 0
% 3 1600 700 1600 800 5 0
% $\samp{h}$
\put(16.0000,-4.0000){\makebox(0,0){$\samp{h}$}}%
% STR 2 0 3 0
% 3 2400 700 2400 800 5 0
% $K(z)$
\put(24.0000,-4.0000){\makebox(0,0){$K(z)$}}%
% STR 2 0 3 0
% 3 3200 700 3200 800 5 0
% $\hold{h}$
\put(32.0000,-4.0000){\makebox(0,0){$\hold{h}$}}%
% STR 2 0 3 0
% 3 4000 700 4000 800 5 0
% $P(s)$
\put(40.0000,-4.0000){\makebox(0,0){$P(s)$}}%
% LINE 2 0 3 0
% 2 4500 1600 900 1600
% 
\special{pn 8}%
\special{pa 4500 1200}%
\special{pa 900 1200}%
\special{fp}%
\end{picture}%
\end{center}
\caption{Sampled-data control system}
\label{fig:SDC_sdfeedback}
\end{figure}
In this figure, $P(s)$ is a continuous-time plant and $K(z)$ is a discrete-time controller.
In order to include $K(z)$ in this control system,
we need an interface.
Therefore we introduce
the ideal sampler $\samp{h}$ and the zero-order hold $\hold{h}$ with sampling time $h$.
\begin{deff}%ideal sampler and zero-order hold
    The ideal sampler $\samp{h}$ and the zero-order hold $\hold{h}$ are defined
    as follows:
    \begin{align}
	\samp{h}
	&:  L^2[0,\infty) \ni u \longmapsto v \in l^2,\quad
	v[k] := u(kh), \nonumber\\
	\hold{h} 
	&: l^2 \ni v \longmapsto u \in L^2[0,\infty),\quad
	u(kh+\theta) := H(\theta)v[k], \nonumber\\
	&\qquad k=0,1,2,\ldots, \nonumber\\
	\intertext{where $H(\cdot)$ is the hold function defined as follows:}
	H(\theta) &:= \begin{cases}
			 1, &\theta \in [0,h),\\
			 0, &\text{\rm otherwise}.
		     \end{cases}
	\label{eq:SDC_hold_func}
    \end{align}
\end{deff}
In practice, a quantization error occurs in the A/D conversion.
We however omit this quantization error here. 
The quantization effects are discussed in Chapter \ref{ch:qqq}.

The system contains both continuous-time and discrete-time signals and is theoretically
regarded as a periodically time-varying system \cite{CheFra}.
Since this system is not time-invariant, it is difficult to analyze or design by using
such conventional machinery as transfer functions or frequency responses.

\section{Lifting}
Conventionally there are two ways for designing sampled-data systems.
One is the following: first design a continuous-time controller in the continuous-time domain,
and then discretize the controller.
A typical discretization method is the Tustin (or bilinear) transformation \cite{CheFra,YamEnc,DCDS}.
If the sampling period is sufficiently small, the designed system may perform well.
However, if the sampling period is not small enough,
the performance not only deteriorates, but also the closed-loop system may become unstable.

The other is to approximate the continuous-time signal by a discrete-time one
by considering only the signals at sampling points.
As a result the sampled-data system becomes a discrete-time time-invariant system.
This method preserves the closed-loop stability.
However, the output of the system may sometimes induce very large intersample ripples
despite a very small sampling period.
The reason is that the method ignores the intersample behavior.

Recently, Yamamoto \cite{Yam90, Yam94} studied the problem of sampled-data control systems.
He developed what is now called the {\em lifting} method, which takes the intersample behavior into account
and gives an exact, not approximated, time-invariant discrete-time system for a sampled-data control system.

We begin by defining the lifting operator.
\begin{deff}%Lifting
    Define the lifting operator $\lift$ by
    \begin{gather*}
	\lift : L^2[0,\infty) \longrightarrow l^2_{L^2[0,h)}:\quad
	 \{f(t)\}_{t\in\mathbb{R}_+} \longmapsto \{\widetilde{f}[k](\theta)\}_{k=0}^\infty,
	\quad \theta\in [0,h),\\
	\widetilde{f}[k](\theta) := f(kh+\theta) \in L^2[0,h).
    \end{gather*}
\end{deff}
By lifting, continuous-time signals in $L^2[0,\infty)$ will become
discrete-time signals whose values are in $L^2[0,h)$,
hence the sampled-data system can be rewritten as a time-invariant discrete-time system
with infinite-dimensional signal spaces.
As a result, we can introduce the concept of transfer functions or the frequency response
for sampled-data systems, and hence we can treat sampled-data systems
as time-invariant systems without approximation.

%\section{Lifting model for sampled-data control system}
Consider the standard sampled-data control system in Figure \ref{fig:SDC_gplant},
\begin{figure}[t]
    \begin{center}
\setlength{\unitlength}{0.4mm}
\begin{picture}( 169, 126)
\put( 151.8,  25.3){\line(-1, 0){   8.4}}
\put( 151.8,  67.5){\line( 0,-1){  42.2}}
\put(  16.9,  67.5){\line( 1, 0){  42.2}}
\put(  16.9,  25.3){\line( 0, 1){  42.2}}
\put(  16.9,  25.3){\vector( 1, 0){   8.4}}
\put( 151.8,  67.5){\vector(-1, 0){  42.2}}
\put( 101.2,  25.3){\vector( 1, 0){  25.3}}
\put(  42.2,  25.3){\vector( 1, 0){  25.3}}
\put( 160.2, 101.2){\vector(-1, 0){  50.6}}
\put(  59.0, 101.2){\vector(-1, 0){  50.6}}
\put( 126.5,  16.9){\framebox(  16.9,  16.9){$\hold{h}$}}
\put(  67.5,   8.4){\framebox(  33.7,  33.7){$K_d$}}
\put(  25.3,  16.9){\framebox(  16.9,  16.9){$\samp{h}$}}
\put(  59.0,  59.0){\framebox(  50.6,  50.6){$G$}}
\put(  17.9, 110.1){$z$}
\put( 144.3, 110.1){$w$}
\put( 144.3,  76.4){$u$}
\put(  17.9,  76.4){$y$}
\put(  51.6,  34.2){$y_d$}
\put( 110.6,  34.2){$u_d$}
\end{picture}
    \end{center}
    \caption{Sampled-data control system}
    \label{fig:SDC_gplant}
\end{figure}
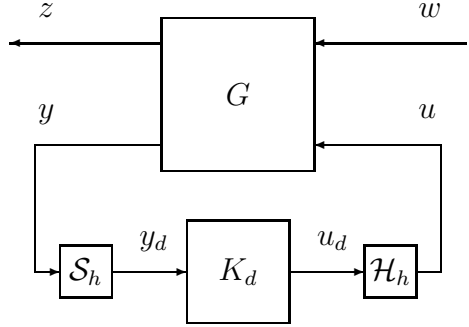
where $G$ is a continuous-time generalized plant, whose
state-space equation is given as follows:
\begin{align*}
       \dot{x} &= Ax + \left[\begin{array}{cc}B_1&B_2\end{array}\right]
	  \left[\begin{array}{c}w\\u\end{array}\right],\\
    \left[
     \begin{array}{c}
	 z\\
	 y
     \end{array}
    \right]
    &=
    \left[
     \begin{array}{ccc}
	 C_1\\
	 C_2\\
     \end{array}
    \right]x
    +
    \left[
     \begin{array}{ccc}
	 D_{11}&D_{12}\\
	 0&0
     \end{array}
    \right]
    \left[
     \begin{array}{c}
	 w\\
	 u
     \end{array}
    \right].
\end{align*}

The signal $w$ is the exogenous input consisting of reference commands,
disturbance or sensor noise, 
while $z$ is the signal to be controlled to have 
a desirable performance. 
Note that both of these signals are continuous-time.
The system $K_d$ is a digital controller.

The design problem of the sampled-data control system in Figure \ref{fig:SDC_gplant}
is to obtain the controller $K_d$ 
which stabilizes the closed-loop system and
makes the performance from $w$ to $z$ desirable.

This system is a sampled-data control system, 
and hence it is a periodically time varying system.
However, by lifting the continuous-time signals $z$ and $w$, and taking
$\widetilde{z}:=\lift z$ and $\widetilde{w}:=\lift w$,
the sampled-data control system in Figure \ref{fig:SDC_gplant} can be 
converted to a time-invariant discrete-time system shown in Figure \ref{fig:SDC_gplant_lift}.
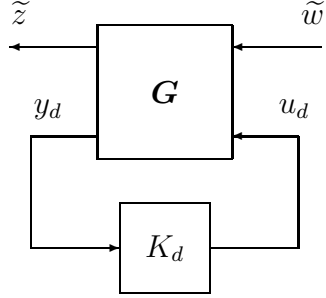
\begin{figure}[t]
   \begin{center}
\begin{picture}( 143, 126)
\put(  42.2,  59.0){\framebox(  50.6,  50.6){$\boldsymbol{G}$}}
\put(  42.2, 101.2){\vector(-1, 0){  33.7}}
\put( 126.5, 101.2){\vector(-1, 0){  33.7}}
\put(  42.2,  67.5){\line(-1, 0){  25.3}}
\put(  16.9,  67.5){\line( 0,-1){  42.2}}
\put( 118.0,  67.5){\vector(-1, 0){  25.3}}
\put(  84.3,  25.3){\line( 1, 0){  33.7}}
\put( 118.0,  25.3){\line( 0, 1){  42.2}}
\put(  16.9,  25.3){\vector( 1, 0){  33.7}}
\put(  50.6,   8.4){\framebox(  33.7,  33.7){$K_d$}}
\put(   9.4, 110.1){$\widetilde{z}$}
\put( 119.0, 110.1){$\widetilde{w}$}
\put( 110.6,  76.4){$u_d$}
\put(  17.9,  76.4){$y_d$}
\end{picture}
   \end{center}
   \caption{Lifted sampled-data control system}
   \label{fig:SDC_gplant_lift}
\end{figure}
Namely, the state-space equation of the lifted system $\boldsymbol{G}$ is obtained
as follows:
\begin{gather*}
\label{eq:gplant1}
%\boldsymbol{G}:
\left[
 \begin{array}{c}
     x[k+1]\\
     \widetilde{z}[k]\\
     y_d[k]
 \end{array}
\right]
    =
    \left[
    \begin{array}{ccc}
     A_d&\B_1&B_{d2}\\
	     \C_1&\D_{11}&\D_{12}\\
	     C_{d2}&0&0
	 \end{array}
    \right]
	\left[
 \begin{array}{c}
     x[k]\\
     \widetilde{w}[k]\\
     u_d[k]
 \end{array}
\right],\\
k=0,1,\ldots,
\end{gather*}
\begin{equation*}
    \label{eq:gplant2}
    \begin{split}
	A_d     &:= e^{Ah},\qquad B_{d2}  := \int_0^h e^{A(h-\tau)}B_2d\tau,\qquad C_{d2}:= C_2,\\
	\B_1 &: L^2[0,h)\longrightarrow \Real^n : w \mapsto \int_0^h e^{A(h-\tau)}B_1w(\tau)d\tau,\\
	\C_1 &: \Real^n \longrightarrow L^2[0,h): x \mapsto C_1e^{A\theta}x,\\
	\D_{11} &: L^2[0,h) \longrightarrow L^2[0,h)\\
	  &: w  \mapsto \int_0^\theta C_1e^{A(\theta-\tau)}B_1w(\tau)d\tau
	    + D_{11}w(\theta),\\
	\D_{12}  &:\Real^m \longrightarrow L^2[0,h)\\
	  &: u_d \mapsto \int_0^\theta C_1e^{A(\theta-\tau)}B_2H(\tau)d\tau u_d
	    + D_{12}H(\theta)u_d,\\
	&\theta \in [0,h).
    \end{split}
\end{equation*}
In this equation, $n$ and $m$ are the dimensions of $x$ and $u_d$, respectively, and 
$H(\cdot)$ is the hold function defined by (\ref{eq:SDC_hold_func}).
Note that $\B_1$, $\C_1$, $\D_{11}$ and $\D_{12}$ are operators in infinite-dimensional spaces,
while $A_d$, $B_{d2}$ and $C_{d2}$ are matrices.
Therefore the lifted system becomes a discrete-time time-invariant system with
infinite-dimensional operators.

\section{Frequency response and $H^\infty$ optimization}
In the previous section, we have shown that sampled-data systems can be represented as time-invariant
discrete-time systems. We can then define their transfer function or frequency response.
The concept of frequency response for sampled-data systems is introduced by Yamamoto and Khargonekar
\cite{YamKha96} and its computation is developed, for example, in \cite{HFKY95}.

Let $\T$ denote the system from $w$ to $z$ in Figure \ref{fig:SDC_gplant}, and
$\widetilde{\T}$ the lifted system of $\T$.
Define the state-space equation for $\widetilde{\T}$ as follows:
\begin{equation}
    \label{eq:SDC_SD_sys}
    \begin{split}
	x_s[k+1] &= Ax_s[k] + \B\widetilde{w}[k],\\
	\widetilde{z}[k] &= \C x_s[k] + \D\widetilde{w}[k],
	\qquad k=0,1,2,\ldots.
    \end{split}
\end{equation}
Note that $A$ is a matrix, while $\B$,
$\C$ and $\D$ are infinite-dimensional operators.
We assume that $A$ is a power stable matrix,
that is, $A^n \rightarrow 0$ as $n\rightarrow\infty$.

For the lifted signal
$\{\widetilde{f}[k]\}_{k=0}^{\infty},$
define its $z$ {\em transform} as
\begin{equation*}
    \mathcal{Z}[\widetilde{f}](z) := \sum_k^\infty
	\widetilde{f}[k]z^{-k}.
\end{equation*}
It follows that the {\em transfer function} $\widetilde{\T}(z)$ of the sampled-data system
can be defined as follows:
    \begin{equation*}
	\widetilde{\T}(z) :=
	    \D+\C (zI-A)^{-1}\B,
	\qquad z\in \Comp.
    \end{equation*}
The $z$ transform $\widetilde{\T}(z)$ is a linear operator on $L^2[0,h)$ with
a complex variable $z$. 
By substituting $e^{j \omega h}$ for $z$ in $\widetilde{\T}(z)$, we can define
the {\em frequency response operator}.
\begin{deff} %Frequency response operator
Define the frequency response operator of sampled-data system $\T$ by
\begin{equation*}
	\widetilde{\T}(e^{j\omega h})=\D+\C(e^{j\omega h}I-A)^{-1}\B:
	L^2_{[0,h)}\longrightarrow L^2[0,h), \quad \omega \in \Real.
\end{equation*}
\end{deff}
The norm of the frequency response operator $\widetilde{\T}(e^{j \omega h})$
\begin{equation*}
    \|\widetilde{\T}(e^{j\omega h})\| 
	:= \sup_{\begin{subarray}{c}v\in L^2[0,h)\\ v\neq 0\end{subarray}}
	\frac{\|\widetilde{T}(e^{j\omega h})v\|_{L^2[0,h)}}{\|v\|_{L^2[0,h)}},
\end{equation*}
is called the {\em gain} at $\omega$.
The $H^\infty$-{\em norm} of the sampled-data system is then given by
\begin{equation*}
    \|\widetilde{\T}\|_\infty := \sup_{\omega \in [0,2\pi/h)}
	\|\widetilde{T}(e^{j\omega h})\|.
\end{equation*}
%It follows that the $H^\infty$-norm is the supremum of the gain 
%on the frequency range $[0,2\pi/h)$.
The $H^\infty$-norm  $\|\widetilde{\T}\|_\infty$ is equivalent to the $L^2$ induced norm of 
the sampled-data system (\ref{eq:SDC_SD_sys}), that is, 
% \cite{BamPea92,CheFra}, that is, 
\begin{equation*}
\|\widetilde{\T}\|_\infty 
	= \|\T\| 
	:= \sup_{\begin{subarray}{c}w\in L^2[0,\infty)\\ w\neq 0\end{subarray}}
	\frac{\|\T w\|_{L^2[0,\infty)}}{\|w\|_{L^2[0,\infty)}}.
\end{equation*}

Sampled-data $H^\infty$ control problem is to find a discrete-time controller $K_d$
which stabilizes the closed-loop system and makes the $H^\infty$-norm of the system small.
The $H^\infty$ control has the following advantages:
\begin{itemize}
\item Since the $H^\infty$-norm is the $L^2$ induced norm of the sampled-data control system,
we can formulate the worst case design.
\item Many robustness requirements for the design against the system uncertainty can be 
described by $H^\infty$-norm constraint.
\item We can shape the frequency characteristic with frequency weights,
which is intuitive to designers.
\end{itemize}

Although $\T$ is infinite-dimensional, the $H^\infty$ optimization can be
equivalently transformed to that for a finite-dimensional discrete-time system \cite{Toi92,BamPea92,KabHar93,HayHarYam94,CheFra}.
Note that the obtained finite-dimensional discrete-time optimization problem takes intersample behavior
into account.

On the other hand, there is another method for obtaining a finite-dimensional discrete-time system:
{\em fast-sampling/fast-hold approximation} \cite{KelAnd92,CheFra}.
This method provides not an equivalent but an approximated model,
however its computation is easier than that of the method giving equivalent models.

Moreover, with regard to signal reconstruction problem (main theme of this thesis),
the method giving equivalent models yields some conservative results.
The reason is as follows:
we often allow signal reconstruction to have time delays since the system does not have any feedback loop.
In other words, a sampled value at a certain time can be reconstructed by using future sampled values,
that is, the filter is allowed to be non-causal.
However, the method giving equivalent method does not readily apply to this situation,
while the fast-sampling/fast-hold method does.

For the reasons mentioned above, we adopt the approximation method in this thesis to consider the problem.
In the following section, we discuss this method in detail.

\section{Fast discretization of multirate sampled-data systems}
\label{sec:SDC_fsfh}
In this section, we deal with multirate sampled-data control systems
by using the fast-sampling/fast-hold method.
The fast-sampling/fast-hold technique is a method for approximating
the performance of sampled-data systems.
The procedure is as follows:
\begin{itemize}
\item discretize the continuous-time input by a hold with sampling period $h/N$,
\item discretize the continuous-time output by a sampler with sampling period $h/N$.
\end{itemize}
With large $N\in\Nset$, the discretized signals may be  good approximation of the continuous signals.

Figure \ref{fig:SDC_fsfh} illustrates the procedure.
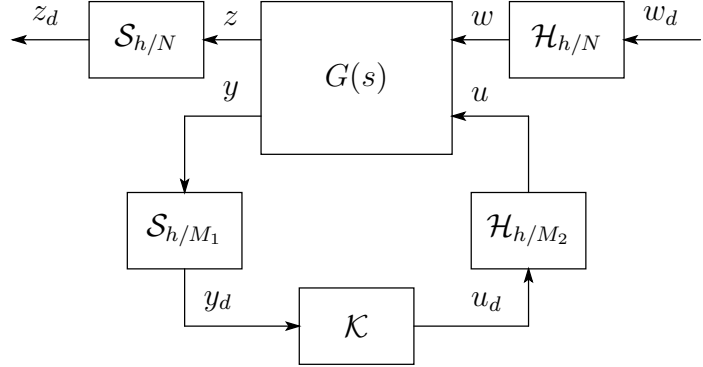
\begin{figure}[t]
    \begin{center}
%	\input{SDC_mr_fsfh}
%WinTpicVersion2.15
\unitlength 0.1in
\begin{picture}(36.00,19.50)(1.00,-21.00)
% BOX 2 0 3 0
% 2 500 600 1100 1000
% 
\special{pn 8}%
\special{pa 500 200}%
\special{pa 1100 200}%
\special{pa 1100 600}%
\special{pa 500 600}%
\special{pa 500 200}%
\special{fp}%
% BOX 2 0 3 0
% 2 1400 600 2400 1400
% 
\special{pn 8}%
\special{pa 1400 200}%
\special{pa 2400 200}%
\special{pa 2400 1000}%
\special{pa 1400 1000}%
\special{pa 1400 200}%
\special{fp}%
% BOX 2 0 3 0
% 2 2700 600 3300 1000
% 
\special{pn 8}%
\special{pa 2700 200}%
\special{pa 3300 200}%
\special{pa 3300 600}%
\special{pa 2700 600}%
\special{pa 2700 200}%
\special{fp}%
% BOX 2 0 3 0
% 2 2500 1600 3100 2000
% 
\special{pn 8}%
\special{pa 2500 1200}%
\special{pa 3100 1200}%
\special{pa 3100 1600}%
\special{pa 2500 1600}%
\special{pa 2500 1200}%
\special{fp}%
% BOX 2 0 3 0
% 2 2200 2100 1600 2500
% 
\special{pn 8}%
\special{pa 2200 1700}%
\special{pa 1600 1700}%
\special{pa 1600 2100}%
\special{pa 2200 2100}%
\special{pa 2200 1700}%
\special{fp}%
% BOX 2 0 3 0
% 2 1300 1600 700 2000
% 
\special{pn 8}%
\special{pa 1300 1200}%
\special{pa 700 1200}%
\special{pa 700 1600}%
\special{pa 1300 1600}%
\special{pa 1300 1200}%
\special{fp}%
% STR 2 0 3 0
% 3 1000 1700 1000 1800 5 0
% $\samp{h/M_1}$
\put(10.0000,-14.0000){\makebox(0,0){$\samp{h/M_1}$}}%
% STR 2 0 3 0
% 3 2800 1700 2800 1800 5 0
% $\hold{h/M_2}$
\put(28.0000,-14.0000){\makebox(0,0){$\hold{h/M_2}$}}%
% STR 2 0 3 0
% 3 1900 2200 1900 2300 5 0
% $K$
\put(19.0000,-19.0000){\makebox(0,0){$\mathcal{K}$}}%
% STR 2 0 3 0
% 3 1900 900 1900 1000 5 0
% $G(s)$
\put(19.0000,-6.0000){\makebox(0,0){$G(s)$}}%
% STR 2 0 3 0
% 3 3000 700 3000 800 5 0
% $\hold{h/N}$
\put(30.0000,-4.0000){\makebox(0,0){$\hold{h/N}$}}%
% STR 2 0 3 0
% 3 800 700 800 800 5 0
% $\samp{h/N}$
\put(8.0000,-4.0000){\makebox(0,0){$\samp{h/N}$}}%
% STR 2 0 3 0
% 3 1200 620 1200 720 2 0
% $z$
\put(12.0000,-3.2000){\makebox(0,0)[lb]{$z$}}%
% STR 2 0 3 0
% 3 2500 620 2500 720 2 0
% $w$
\put(25.0000,-3.2000){\makebox(0,0)[lb]{$w$}}%
% STR 2 0 3 0
% 3 2500 1020 2500 1120 2 0
% $u$
\put(25.0000,-7.2000){\makebox(0,0)[lb]{$u$}}%
% STR 2 0 3 0
% 3 1200 1020 1200 1120 2 0
% $y$
\put(12.0000,-7.2000){\makebox(0,0)[lb]{$y$}}%
% STR 2 0 3 0
% 3 3400 620 3400 720 2 0
% $w_d$
\put(34.0000,-3.2000){\makebox(0,0)[lb]{$w_d$}}%
% STR 2 0 3 0
% 3 200 620 200 720 2 0
% $z_d$
\put(2.0000,-3.2000){\makebox(0,0)[lb]{$z_d$}}%
% STR 2 0 3 0
% 3 1100 2120 1100 2220 2 0
% $y_d$
\put(11.0000,-18.2000){\makebox(0,0)[lb]{$y_d$}}%
% STR 2 0 3 0
% 3 2500 2120 2500 2220 2 0
% $u_d$
\put(25.0000,-18.2000){\makebox(0,0)[lb]{$u_d$}}%
% VECTOR 2 0 3 0
% 2 3700 800 3300 800
% 
\special{pn 8}%
\special{pa 3700 400}%
\special{pa 3300 400}%
\special{fp}%
\special{sh 1}%
\special{pa 3300 400}%
\special{pa 3367 420}%
\special{pa 3353 400}%
\special{pa 3367 380}%
\special{pa 3300 400}%
\special{fp}%
% VECTOR 2 0 3 0
% 2 2700 800 2400 800
% 
\special{pn 8}%
\special{pa 2700 400}%
\special{pa 2400 400}%
\special{fp}%
\special{sh 1}%
\special{pa 2400 400}%
\special{pa 2467 420}%
\special{pa 2453 400}%
\special{pa 2467 380}%
\special{pa 2400 400}%
\special{fp}%
% VECTOR 2 0 3 0
% 2 1400 800 1100 800
% 
\special{pn 8}%
\special{pa 1400 400}%
\special{pa 1100 400}%
\special{fp}%
\special{sh 1}%
\special{pa 1100 400}%
\special{pa 1167 420}%
\special{pa 1153 400}%
\special{pa 1167 380}%
\special{pa 1100 400}%
\special{fp}%
% VECTOR 2 0 3 0
% 2 500 800 100 800
% 
\special{pn 8}%
\special{pa 500 400}%
\special{pa 100 400}%
\special{fp}%
\special{sh 1}%
\special{pa 100 400}%
\special{pa 167 420}%
\special{pa 153 400}%
\special{pa 167 380}%
\special{pa 100 400}%
\special{fp}%
% LINE 2 0 3 0
% 2 1400 1200 1000 1200
% 
\special{pn 8}%
\special{pa 1400 800}%
\special{pa 1000 800}%
\special{fp}%
% VECTOR 2 0 3 0
% 2 1000 1200 1000 1600
% 
\special{pn 8}%
\special{pa 1000 800}%
\special{pa 1000 1200}%
\special{fp}%
\special{sh 1}%
\special{pa 1000 1200}%
\special{pa 1020 1133}%
\special{pa 1000 1147}%
\special{pa 980 1133}%
\special{pa 1000 1200}%
\special{fp}%
% LINE 2 0 3 0
% 2 1000 2000 1000 2300
% 
\special{pn 8}%
\special{pa 1000 1600}%
\special{pa 1000 1900}%
\special{fp}%
% VECTOR 2 0 3 0
% 2 1000 2300 1600 2300
% 
\special{pn 8}%
\special{pa 1000 1900}%
\special{pa 1600 1900}%
\special{fp}%
\special{sh 1}%
\special{pa 1600 1900}%
\special{pa 1533 1880}%
\special{pa 1547 1900}%
\special{pa 1533 1920}%
\special{pa 1600 1900}%
\special{fp}%
% LINE 2 0 3 0
% 2 2200 2300 2800 2300
% 
\special{pn 8}%
\special{pa 2200 1900}%
\special{pa 2800 1900}%
\special{fp}%
% VECTOR 2 0 3 0
% 2 2800 2300 2800 2000
% 
\special{pn 8}%
\special{pa 2800 1900}%
\special{pa 2800 1600}%
\special{fp}%
\special{sh 1}%
\special{pa 2800 1600}%
\special{pa 2780 1667}%
\special{pa 2800 1653}%
\special{pa 2820 1667}%
\special{pa 2800 1600}%
\special{fp}%
% LINE 2 0 3 0
% 2 2800 1600 2800 1200
% 
\special{pn 8}%
\special{pa 2800 1200}%
\special{pa 2800 800}%
\special{fp}%
% VECTOR 2 0 3 0
% 2 2800 1200 2400 1200
% 
\special{pn 8}%
\special{pa 2800 800}%
\special{pa 2400 800}%
\special{fp}%
\special{sh 1}%
\special{pa 2400 800}%
\special{pa 2467 820}%
\special{pa 2453 800}%
\special{pa 2467 780}%
\special{pa 2400 800}%
\special{fp}%
\end{picture}%
    \end{center}
    \caption{Fast-sampling/fast-hold discretization}
    \label{fig:SDC_fsfh}
\end{figure}
Assume that the controller $\mathcal{K}$ is $(M_1,M_2)$-periodic ($M_1, M_2 \in \Nset$) \cite{Mey90}, 
that is, $z^{M_2}\mathcal{K}z^{-M_1}=\mathcal{K}$
where $z$ is the unit advance and $z^{-1}$ is the unit time delay.

\subsection{Discrete-time lifting}
By attaching a fast-sampler and a fast-hold as shown in Figure \ref{fig:SDC_fsfh},
the multirate sampled-data system will be converted to a finite-dimensional discrete-time system
(we will show this in the next section).
However, this system has three sampling periods: $h/M_1$, $h/M_2$ and $h/N$,
and hence the system will be time-varying (to put it precisely, periodically time-varying).
In order to equivalently convert a multirate system to a single-rate one,
the {\em discrete-time lifting} \cite{KhaPooTan85,Mey90,CheFra} is useful.
The definition is as follows:
\begin{deff}%Discrete-time lifting
Define the discrete-time lifting $\dlift{N}$ and its inverse $\idlift{N}$ by
\begin{align*}
    \dlift{N}&: l^2\longrightarrow l^2:
\left\{v[0],v[1],\ldots \right\}\mapsto
       \left\{
          \left[
            \begin{array}{c}
		v[0]\\
		v[1]\\
		\vdots\\
		v[N-1]
	    \end{array}
          \right],
          \left[
            \begin{array}{c}
		v[N]\\
		v[N+1]\\
		\vdots\\
		v[2N-1]
            \end{array}
          \right],\ldots
    \right\},\\
    \idlift{N}&: l^2\longrightarrow l^2:\\
    &    \left\{
          \left[
            \begin{array}{c}
		v_0[0]\\
		v_1[0]\\
		\vdots\\
		v_{N-1}[0]
	    \end{array}
          \right],
          \left[
            \begin{array}{c}
		v_0[1]\\
		v_1[1]\\
		\vdots\\
		v_{N-1}[1]
            \end{array}
          \right],\ldots
    \right\}\mapsto
     \left\{v_0[0],v_1[0],\ldots,v_{N-1}[0],v_0[1],v_1[1],\ldots \right\}.
\end{align*}
\end{deff}
The discrete-time lifting $\dlift{N}$ converts a 1-dimensional signal
into an $N$-dimensional signal
and the sampling rate becomes $N$ times slower.
This operation makes it possible to equivalently convert multirate systems into
single-rate systems, and hence its analysis and design become easier.
Note that the discrete-time lifting $\dlift{N}$ is norm-preserving,
namely, $\|\dlift{N}v\|=\|v\|$, $v\in l^2$ and so is $\idlift{N}$.

\subsection{Approximating multirate sampled-data systems}
By using the discrete-time lifting, the multirate system shown in Figure \ref{fig:SDC_fsfh}
is converted to a single-rate discrete-time system.
Take
\begin{equation*}
G(s) = \left[ \begin{array}{cc} G_{11}(s) & G_{12}(s)\\ G_{21}(s) & G_{22}(s) \end{array}\right],
\end{equation*}
where $G_{11}$, $G_{12}$ and $G_{22}$ are strictly proper and $G_{21}$ is proper.
Let their state-space realization be
\begin{equation*}
G_{ij}(s) = \sspc{A}{B_{j}}{C_{i}}{D_{ij}}(s),\quad i, j=1,2.
\end{equation*}

Let the system from $w_d$ to $z_d$ in Figure \ref{fig:SDC_fsfh} be $T_{dN}$.
Then $T_{dN}$ can be rewritten as follows:
\begin{equation*}
\begin{split}
T_{dN}&=\lft(G_{dN},\mathcal{K}),\\
T_{dN}&:=
	\left[\begin{array}{cc}\samp{h/N} &0 \\ 0 &I\end{array}\right]
	\left[ \begin{array}{cc} G_{11} & G_{12}\\ G_{21} & G_{22} \end{array}\right]
	\left[\begin{array}{cc}\hold{h/N} &0 \\ 0 &I\end{array}\right].
\end{split}
\end{equation*}
Then we apply the discrete-time lifting to $w_d$ and $z_d$ as shown in Figure \ref{fig:SDC_fsfh_lifted}.
\begin{figure}[t]
\begin{center}
%\input{SDC_mr_fsfh_lifted}
%WinTpicVersion2.15
\unitlength 0.1in
\begin{picture}(50.00,19.50)(0.00,-21.00)
% BOX 2 0 3 0
% 2 1100 600 1700 1000
% 
\special{pn 8}%
\special{pa 1100 200}%
\special{pa 1700 200}%
\special{pa 1700 600}%
\special{pa 1100 600}%
\special{pa 1100 200}%
\special{fp}%
% BOX 2 0 3 0
% 2 2000 600 3000 1400
% 
\special{pn 8}%
\special{pa 2000 200}%
\special{pa 3000 200}%
\special{pa 3000 1000}%
\special{pa 2000 1000}%
\special{pa 2000 200}%
\special{fp}%
% BOX 2 0 3 0
% 2 3300 600 3900 1000
% 
\special{pn 8}%
\special{pa 3300 200}%
\special{pa 3900 200}%
\special{pa 3900 600}%
\special{pa 3300 600}%
\special{pa 3300 200}%
\special{fp}%
% BOX 2 0 3 0
% 2 2800 2100 2200 2500
% 
\special{pn 8}%
\special{pa 2800 1700}%
\special{pa 2200 1700}%
\special{pa 2200 2100}%
\special{pa 2800 2100}%
\special{pa 2800 1700}%
\special{fp}%
% BOX 2 0 3 0
% 2 1900 1600 1300 2000
% 
\special{pn 8}%
\special{pa 1900 1200}%
\special{pa 1300 1200}%
\special{pa 1300 1600}%
\special{pa 1900 1600}%
\special{pa 1900 1200}%
\special{fp}%
% STR 2 0 3 0
% 3 1600 1700 1600 1800 5 0
% $\samp{h/M_1}$
\put(16.0000,-14.0000){\makebox(0,0){$\samp{h/M_1}$}}%
% STR 2 0 3 0
% 3 3400 1700 3400 1800 5 0
% $\hold{h/M_2}$
\put(34.0000,-14.0000){\makebox(0,0){$\hold{h/M_2}$}}%
% STR 2 0 3 0
% 3 2500 2200 2500 2300 5 0
% $K$
\put(25.0000,-19.0000){\makebox(0,0){$\mathcal{K}$}}%
% STR 2 0 3 0
% 3 2500 900 2500 1000 5 0
% $G(s)$
\put(25.0000,-6.0000){\makebox(0,0){$G(s)$}}%
% STR 2 0 3 0
% 3 3600 700 3600 800 5 0
% $\hold{h/N}$
\put(36.0000,-4.0000){\makebox(0,0){$\hold{h/N}$}}%
% STR 2 0 3 0
% 3 1400 700 1400 800 5 0
% $\samp{h/N}$
\put(14.0000,-4.0000){\makebox(0,0){$\samp{h/N}$}}%
% STR 2 0 3 0
% 3 1800 620 1800 720 2 0
% $z$
\put(18.0000,-3.2000){\makebox(0,0)[lb]{$z$}}%
% STR 2 0 3 0
% 3 3100 620 3100 720 2 0
% $w$
\put(31.0000,-3.2000){\makebox(0,0)[lb]{$w$}}%
% STR 2 0 3 0
% 3 3100 1020 3100 1120 2 0
% $u$
\put(31.0000,-7.2000){\makebox(0,0)[lb]{$u$}}%
% STR 2 0 3 0
% 3 1800 1020 1800 1120 2 0
% $y$
\put(18.0000,-7.2000){\makebox(0,0)[lb]{$y$}}%
% STR 2 0 3 0
% 3 4000 620 4000 720 2 0
% $w_d$
\put(40.0000,-3.2000){\makebox(0,0)[lb]{$w_d$}}%
% STR 2 0 3 0
% 3 800 620 800 720 2 0
% $z_d$
\put(8.0000,-3.2000){\makebox(0,0)[lb]{$z_d$}}%
% STR 2 0 3 0
% 3 1700 2120 1700 2220 2 0
% $y_d$
\put(17.0000,-18.2000){\makebox(0,0)[lb]{$y_d$}}%
% STR 2 0 3 0
% 3 3100 2120 3100 2220 2 0
% $u_d$
\put(31.0000,-18.2000){\makebox(0,0)[lb]{$u_d$}}%
% BOX 2 0 3 0
% 2 800 1000 400 600
% 
\special{pn 8}%
\special{pa 800 600}%
\special{pa 400 600}%
\special{pa 400 200}%
\special{pa 800 200}%
\special{pa 800 600}%
\special{fp}%
% BOX 2 0 3 0
% 2 4200 600 4600 1000
% 
\special{pn 8}%
\special{pa 4200 200}%
\special{pa 4600 200}%
\special{pa 4600 600}%
\special{pa 4200 600}%
\special{pa 4200 200}%
\special{fp}%
% STR 2 0 3 0
% 3 600 700 600 800 5 0
% $\dlift{N}$
\put(6.0000,-4.0000){\makebox(0,0){$\dlift{N}$}}%
% STR 2 0 3 0
% 3 4400 700 4400 800 5 0
% $\idlift{N}$
\put(44.0000,-4.0000){\makebox(0,0){$\idlift{N}$}}%
% STR 2 0 3 0
% 3 120 620 120 720 2 0
% $zz$
\put(1.2000,-3.2000){\makebox(0,0)[lb]{$\widetilde{z}_d$}}%
% STR 2 0 3 0
% 3 4720 620 4720 720 2 0
% $ww$
\put(47.2000,-3.2000){\makebox(0,0)[lb]{$\widetilde{w}_d$}}%
% VECTOR 2 0 3 0
% 2 5000 800 4600 800
% 
\special{pn 8}%
\special{pa 5000 400}%
\special{pa 4600 400}%
\special{fp}%
\special{sh 1}%
\special{pa 4600 400}%
\special{pa 4667 420}%
\special{pa 4653 400}%
\special{pa 4667 380}%
\special{pa 4600 400}%
\special{fp}%
% VECTOR 2 0 3 0
% 2 4200 800 3900 800
% 
\special{pn 8}%
\special{pa 4200 400}%
\special{pa 3900 400}%
\special{fp}%
\special{sh 1}%
\special{pa 3900 400}%
\special{pa 3967 420}%
\special{pa 3953 400}%
\special{pa 3967 380}%
\special{pa 3900 400}%
\special{fp}%
% VECTOR 2 0 3 0
% 2 3300 800 3000 800
% 
\special{pn 8}%
\special{pa 3300 400}%
\special{pa 3000 400}%
\special{fp}%
\special{sh 1}%
\special{pa 3000 400}%
\special{pa 3067 420}%
\special{pa 3053 400}%
\special{pa 3067 380}%
\special{pa 3000 400}%
\special{fp}%
% VECTOR 2 0 3 0
% 2 2000 800 1700 800
% 
\special{pn 8}%
\special{pa 2000 400}%
\special{pa 1700 400}%
\special{fp}%
\special{sh 1}%
\special{pa 1700 400}%
\special{pa 1767 420}%
\special{pa 1753 400}%
\special{pa 1767 380}%
\special{pa 1700 400}%
\special{fp}%
% VECTOR 2 0 3 0
% 2 1100 800 800 800
% 
\special{pn 8}%
\special{pa 1100 400}%
\special{pa 800 400}%
\special{fp}%
\special{sh 1}%
\special{pa 800 400}%
\special{pa 867 420}%
\special{pa 853 400}%
\special{pa 867 380}%
\special{pa 800 400}%
\special{fp}%
% VECTOR 2 0 3 0
% 2 400 800 0 800
% 
\special{pn 8}%
\special{pa 400 400}%
\special{pa 0 400}%
\special{fp}%
\special{sh 1}%
\special{pa 0 400}%
\special{pa 67 420}%
\special{pa 53 400}%
\special{pa 67 380}%
\special{pa 0 400}%
\special{fp}%
% LINE 2 0 3 0
% 2 2000 1200 1600 1200
% 
\special{pn 8}%
\special{pa 2000 800}%
\special{pa 1600 800}%
\special{fp}%
% VECTOR 2 0 3 0
% 2 1600 1200 1600 1600
% 
\special{pn 8}%
\special{pa 1600 800}%
\special{pa 1600 1200}%
\special{fp}%
\special{sh 1}%
\special{pa 1600 1200}%
\special{pa 1620 1133}%
\special{pa 1600 1147}%
\special{pa 1580 1133}%
\special{pa 1600 1200}%
\special{fp}%
% LINE 2 0 3 0
% 2 1600 2000 1600 2300
% 
\special{pn 8}%
\special{pa 1600 1600}%
\special{pa 1600 1900}%
\special{fp}%
% VECTOR 2 0 3 0
% 2 1600 2300 2200 2300
% 
\special{pn 8}%
\special{pa 1600 1900}%
\special{pa 2200 1900}%
\special{fp}%
\special{sh 1}%
\special{pa 2200 1900}%
\special{pa 2133 1880}%
\special{pa 2147 1900}%
\special{pa 2133 1920}%
\special{pa 2200 1900}%
\special{fp}%
% LINE 2 0 3 0
% 2 2800 2300 3400 2300
% 
\special{pn 8}%
\special{pa 2800 1900}%
\special{pa 3400 1900}%
\special{fp}%
% VECTOR 2 0 3 0
% 2 3400 2300 3400 2000
% 
\special{pn 8}%
\special{pa 3400 1900}%
\special{pa 3400 1600}%
\special{fp}%
\special{sh 1}%
\special{pa 3400 1600}%
\special{pa 3380 1667}%
\special{pa 3400 1653}%
\special{pa 3420 1667}%
\special{pa 3400 1600}%
\special{fp}%
% BOX 2 0 3 0
% 2 3100 2000 3700 1600
% 
\special{pn 8}%
\special{pa 3100 1600}%
\special{pa 3700 1600}%
\special{pa 3700 1200}%
\special{pa 3100 1200}%
\special{pa 3100 1600}%
\special{fp}%
% LINE 2 0 3 0
% 2 3400 1600 3400 1200
% 
\special{pn 8}%
\special{pa 3400 1200}%
\special{pa 3400 800}%
\special{fp}%
% VECTOR 2 0 3 0
% 2 3400 1200 3000 1200
% 
\special{pn 8}%
\special{pa 3400 800}%
\special{pa 3000 800}%
\special{fp}%
\special{sh 1}%
\special{pa 3000 800}%
\special{pa 3067 820}%
\special{pa 3053 800}%
\special{pa 3067 780}%
\special{pa 3000 800}%
\special{fp}%
\end{picture}%
\end{center}
\caption{Lifted and fast-discretized sampled-data system}
\label{fig:SDC_fsfh_lifted}
\end{figure}
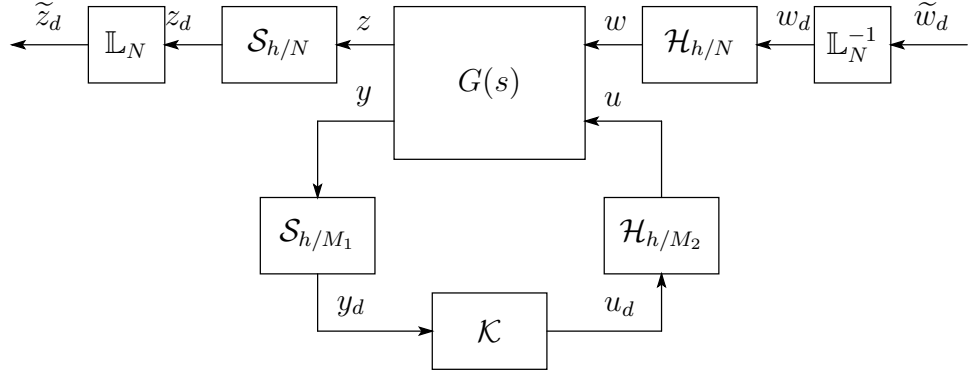
Let the system from $\widetilde{w}_d$ to $\widetilde{z}_d$ in Figure \ref{fig:SDC_fsfh_lifted}
be $\widetilde{T}_{dN}$. 
Note that since $\dlift{N}$ and $\idlift{N}$ are norm-preserving mappings,
we have $\|T_{dN}\|=\|\widetilde{T}_{dN}\|$.
Then we rewrite $\widetilde{T}_{dN}$ as follows:
\begin{equation*}
\begin{split}
\widetilde{T}_{dN} &= \lft(\widehat{G}_{dN}, \mathcal{K}),\\
\widehat{G}_{dN} &:=
	\left[\begin{array}{cc}\dlift{N} &0 \\ 0 &I\end{array}\right]
	\left[\begin{array}{cc}\samp{h/N} &0 \\ 0 &I\end{array}\right]
	\left[ \begin{array}{cc} G_{11} & G_{12}\\ G_{21} & G_{22} \end{array}\right]
	\left[\begin{array}{cc}\hold{h/N} &0 \\ 0 &I\end{array}\right]
	\left[\begin{array}{cc}\idlift{N} &0 \\ 0 &I\end{array}\right].
\end{split}
\end{equation*}

Let us turn to the controller $\mathcal{K}$.
By the assumption that $\mathcal{K}$ is $(M_1, M_2)$-periodic,
$K_d:=\dlift{M_2}K\idlift{M_1}$ is time-invariant \cite{Mey90}.
By using this property, we rearrange $\widetilde{G}_{dN}$ as follows:
\begin{equation*}
\begin{split}
\widetilde{T}_{dN} &= \lft(\widetilde{G}_{dN}, K_d),\\
\widetilde{G}_{dN} &:=
	\left[\begin{array}{cc}I &0 \\ 0 &\dlift{M_1}\samp{h/M_1}\end{array}\right]
	\widehat{G}_{dN}
	\left[\begin{array}{cc}I &0 \\ 0 &\hold{h/M_2}\idlift{M_2}\end{array}\right]\\
&=
	\left[\begin{array}{cc}\dlift{N}\samp{h/N} &0 \\ 0 &\dlift{M_1}\samp{h/M1}\end{array}\right]
	\left[ \begin{array}{cc} G_{11} & G_{12}\\ G_{21} & G_{22} \end{array}\right]
	\left[\begin{array}{cc}\hold{h/N}\idlift{N} &0 \\ 0 &\hold{M_2}\idlift{M_2}\end{array}\right].
\end{split}
\end{equation*}

Finally, we convert $\widetilde{G}_{dN}$ to a simple time-invariant system
by the following proposition.
\begin{prop}
\label{prop:dlift}
Assume $N=kM_1M_2$, $k\in\Nset$. Then we have the following identities:
\begin{equation}
\dlift{M_1}\samp{h/M_1} = S\dlift{N}\samp{h/N},\quad
\hold{h/M_2}\idlift{M_2} = \hold{h/N}\idlift{N}H,
\end{equation}
where
\begin{equation}
\begin{split}
S &:= \left.\left[\begin{array}{ccc}p& & \\ &\ddots & \\ & &p
	\end{array}\right]\right\}M_1,\quad 
p := [1,\underbrace{0,\ldots,0}_{kM_2-1}],\\
H &:= \underbrace{\left[\begin{array}{ccc}q& & \\ &\ddots & \\ & &q
	\end{array}\right]}_{M_2},\quad
q := [\underbrace{1,\ldots,1}_{kM_1}]^T.
\end{split}
\end{equation}
\end{prop}
This proposition gives us the following equation:
\begin{equation*}
\begin{split}
\widetilde{G}_{dN} &= 
	\left[\begin{array}{cc}\dlift{N}\samp{h/N} &0 \\ 0 &S\dlift{N}\samp{h/N}\end{array}\right]
	\left[ \begin{array}{cc} G_{11} & G_{12}\\ G_{21} & G_{22} \end{array}\right]
	\left[\begin{array}{cc}\hold{h/N}\idlift{N} &0 \\ 0 &\hold{h/N}\idlift{N}H\end{array}\right]\\
&=
	\left[\begin{array}{cc}I&0 \\ 0 &S\end{array}\right]
	\left[ \begin{array}{cc} \dliftsys{N}{G_{11}} & \dliftsys{N}{G_{12}}\\ \dliftsys{N}{G_{21}} & \dliftsys{N}{G_{22}} \end{array}\right]
	\left[\begin{array}{cc}I &0 \\ 0 &H\end{array}\right],
\end{split}
\end{equation*}
where $\dliftsys{N}{G}:=\dlift{N}\samp{h/N}G\hold{h/N}\idlift{N}$.
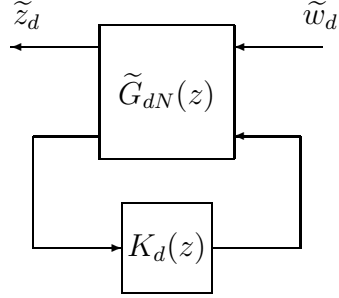
\begin{figure}[t]
    \begin{center}
\begin{picture}( 143, 126)
\put(  42.2,  59.0){\framebox(  50.6,  50.6){$\widetilde{G}_{dN}(z)$}}
\put(  42.2, 101.2){\vector(-1, 0){  33.7}}
\put( 126.5, 101.2){\vector(-1, 0){  33.7}}
\put(  42.2,  67.5){\line(-1, 0){  25.3}}
\put(  16.9,  67.5){\line( 0,-1){  42.2}}
\put( 118.0,  67.5){\vector(-1, 0){  25.3}}
\put(  84.3,  25.3){\line( 1, 0){  33.7}}
\put( 118.0,  25.3){\line( 0, 1){  42.2}}
\put(  16.9,  25.3){\vector( 1, 0){  33.7}}
\put(  50.6,   8.4){\framebox(  33.7,  33.7){$K_d(z)$}}
\put(   9.4, 110.1){$\widetilde{z}_d$}
\put( 119.0, 110.1){$\widetilde{w}_d$}
%\put( 110.6,  76.4){$u_d$}
%\put(  17.9,  76.4){$y_d$}
\end{picture}
    \end{center}
    \caption{FSFH discretized system}
    \label{fig:SDC_fsfh_d}
\end{figure}
Note that $\dliftsys{N}{G_{ij}}=:G_{d,ij}$ ($i,j=1,2$) is a time-invariant discrete-time system \cite{Mey90,CheFra}.
In fact, a state-space realization of $G_{d,ij}$ is given as follows:
\begin{equation*}
G_{d,ij}=
	\left[
	\begin{array}{c|cccc}
		A_d^N & A_d^{N-1}B_{d,j} & A_d^{N-2}B_{d,j} &\ldots &B_{d,j}\\\hline
		C_i & D_{ij} & 0 & \ldots & 0\\
		C_iA_d & C_iB_{d,j} & D_{ij} & \ldots & 0\\
		\vdots & \vdots & \vdots & &\vdots\\
		C_iA_d^{N-1} & C_iA_d^{N-2}B_{d,j} & C_iA_d^{N-3}B_{d,j} &\ldots & D_{ij}
	\end{array}
	\right], \quad i,j=1,2,
\end{equation*}
where
\begin{equation*}
A_d := e^{Ah/N},\quad B_{d,j}:= \int_0^{h/N}e^{At}B_jdt.
\end{equation*}
We show the obtained time-invariant discrete-time system in Figure \ref{fig:SDC_fsfh_d}.

When $N$ becomes larger,
the performance of the discrete-time system $T_{dN}:=\lft(\widetilde{G}_{dN},K_d)$ converges to that of
the original sampled-data system $\T:=\lft(G,\hold{h/M_2}\mathcal{K}\samp{h/M_1})$ \cite{YamMadAnd99}.
Therefore, if we take sufficiently large $N$, the error between $\T$ and $T_{dN}$ will small.
The error estimate for the fast-sampling factor $N$ is discussed in \cite{YamAndNag02}.

We summarize the above discussion as a theorem:
\begin{theorem}
For the multirate sampled-data control system $\T:=\lft(G,\hold{h/M_2}\mathcal{K}\samp{h/M_1})$,
there exists a time-invariant discrete-time system $T_{dN}:=\lft(\widetilde{G}_{dN},K_d)$ such that
\begin{equation*}
\lim_{N\rightarrow\infty}\|T_{dN}\|=\|\T\|.
\end{equation*}
\end{theorem}
%\begin{prop}
%\begin{align}
%\dlift{N}\samp{h/N}\hold{h} &= \left[\begin{array}{c}1\\\vdots\\1\end{array}\right],\\
%\samp{h}\hold{h/N}\idlift{N} &= \left[\begin{array}{cccc}1&0&\ldots&0\end{array}\right].
%\end{align}
%\end{prop}
%\begin{cor}
%\begin{align}
%\hold{h} &= \hold{h/N}\idlift{N}\left[\begin{array}{c}1\\\vdots\\1\end{array}\right],\\
%\samp{h} & =\left[\begin{array}{cccc}1&0&\ldots&0\end{array}\right]\dlift{N}\samp{h/N}.
%\end{align}
%\end{cor}
%\subsection{Polyphase representation and discrete-time lifting}
%\section{Polyphase representation and discrete-time lifting}
%\begin{prop}
%
%\end{prop}
%
%\input{multirate}
\chapter{Multirate Signal Processing}
\label{ch:multirate}
\section{Introduction}
Multirate techniques are now very popular in digital signal
processing.  They are particularly effective in subband coding,
and various techniques for economical information saving have
been developed \cite{Fli,Vai,Zel}.  

One example is  signal decoding in audio/speech processing.  
For example, in the commercial CD format, the
sampling frequency is 44.1 kHz, but one hardly employs the same
sampling period in decoding.  A popular technique is 
{\em interpolation}.
The process is as follows:
first {\em upsample} the encoded digital signal
(i.e., inserting zeros between two consecutive samples), 
remove the
parasitic {\em imaging} spectra via a digital low-pass filter, and
then convert it back to an analog signal with a hold device and
an analog low-pass filter.  
Imaging is a phenomenon due to zeros inserted by upsampling,
and yields high frequency noise.

The chief advantage here is that
one can employ a fast hold device, and do not have to use a very sharp
analog filter (thereby avoiding much phase distortion induced by
a sharp analog filter).  

Another example is signal compression.
Due to the limitation of the bandwidth of communication channels
or of the size of storage devices,
one often needs to compress the signals.
In signal compression, multirate processing plays a major role.
The fundamental operation for signal compression is {\em decimation}.
Decimation is to reduce the sampling rate of a signal;
the process is to first filter out the {\em aliasing} components by using a digital 
low-pass filter, and then {\em downsample} the filtered signal.
Downsampling is an operation to keep every $M$-th sample ($M$ is a natural number) 
and remove in-between samples.
The aliasing caused by upsampling 
is comparative to aliasing in A/D conversion, that is,
removing samples causes frequency overlapping.

%The popular compression method JPEG or MPEG uses
%the subband coding \cite{Fli,Vai,Zel}, and in the coding system decimation
%must be designed in order to effectively compress.

By combining interpolators and decimators, we can obtain a sampling rate converter.
In commercial applications, there are
many different sampling rates employed:  for example 
48kHz for DAT and 44.1kHz for audio CD.  
The conversion from one sampling rate to the other
becomes necessary.  In such a process, it is clearly
required that the information loss be as little as possible.

The conventional way of doing this is as follows:  Suppose
we want to convert a signal $v$ with sampling frequency $f_1$ Hz
to another signal $u$ with sampling frequency $f_2$ Hz.  
Suppose also that there exist (coprime) integers $L_1$ and $L_2$ such that 
$f_{1}L_1 = f_{2}L_2$.  We first upsample $v$ by factor $L_1$,
to make the sampling frequency $f_{1}L_1$.  Suppose that the original
signal is perfectly band-limited in the range  
$|\omega| < f_1/2$.  We then introduce a digital filter 
$H(z)$ to filter out the undesirable imaging component.
After this, the obtained signal is downsampled by factor $L_2$ to
become a signal with sampling frequency $f_{2}=f_1L_1/L_2$ Hz.

In the existing literature, it is a commonly accepted principle
that one inserts a very sharp digital low-pass filter after
the upsampler or before the downsampler to eliminate the
effect of imaging or aliasing components \cite{Fli,Vai,Zel}.
This is based on the following
reasoning: Suppose that the original signal is fully band-limited.
Then the imaging (aliasing) components induced by upsampling (downsampler) 
is not relevant to the original analog signal 
and hence they must be removed by a low-pass filter.  
If the original signal is band-limited, the closer is this filter to an ideal filter, the better.  

In practice, however, no signals are fully band-limited in
a practical range of a passband, and they obey only an approximate
frequency characteristic.  The argument above is thus valid only
in an approximate sense.  One may rephrase this as a problem of
robustness:  namely, when the original signals are not fully
band-limited but obey only a certain frequency characteristic,
how close should the digital filter be to the ideal low-pass filter?

This type of question has been seldom addressed in the signal processing
literature until very recently.  However, this can be properly
placed in the framework of sampled-data control theory, and there are now
several investigations that apply the sampled-data control methodology
to digital signal processing.

%Among them, Chen and Francis \cite{CheFra95} solves the design of multirate filter
%banks in the discrete-time $H^{\infty}$ setting; Khargonekar
%and Yamamoto \cite{KhaYam96} formulates and solves a single-rate
%signal reconstruction problem with optimal analog $H^{\infty}$ 
%performance.  This has been generalized in \cite{YamFujKha97,YamKha98} to
%a multirate context.  A multirate D/A conversion has been studied
%in \cite{IsiYam98}

We formulate a multirate digital signal reconstruction problem
under the assumption that the original analog signal is subject to a certain frequency
characteristic, but not fully band-limited.
Under the assumption we will optimize the analog performance with
an $H^\infty$ optimality criterion.

This may also be regarded as an optimal D/A converter design.
We will show that performance improvement is possible 
over  conventional low-pass filters.  It is also seen that
the presented method can be used as a new design method for a
low-pass filter.

In this chapter, we first introduce the fundamentals of multirate digital signal processing.
The conventional idea of its design is also discussed and we point out that 
there are some drawbacks in its idea.
We then  present an alternative method of designing multirate systems,
that is, sampled-data $H^\infty$ design.
The last section provides design examples and we show
the advantages of the present method.

\section{Interpolators and decimators}
In this section, we introduce mathematical definitions of interpolators and decimators.
An interpolator or a decimator is implemented with upsamplers $\us{M}$ and downsamplers $\ds{M}$
respectively.
We begin by defining the upsampler and the downsampler.
\begin{deff}%upsamplers and down samplers.
For discrete-time signal $\{x[k]\}_{k=0}^\infty$, define the upsampler $\us{M}$ and 
the downsampler $\ds{M}$ by
    \begin{equation*}
	\begin{split}
	    \us{M} &: \{x[k]\}_{k=0}^\infty 
	    \longmapsto \{x[0],\underbrace{0,0,\ldots,0}_{\text{$M\!-\!1$}},x[1],0,\ldots \},\\
	    \ds{M} &: \{x[k]\}_{k=0}^\infty \longmapsto \{x[0],x[M],x[2M],\ldots\}.
	\end{split}
    \end{equation*}
\end{deff}

The upsampling operation is implemented by inserting $M-1$ equidistant zero-valued
samples between two consecutive samples
of $x[k]$ before the sampling rate is multiplied by the factor $M$.
Figure \ref{fig:multirate_intro_us} indicates the upsampling operation.
\begin{figure}[t]
   \begin{center}
      \includegraphics[width=6cm]{multirate_intro_us.eps}
   \end{center}
   \caption{Upsampling operation $y=(\us{2})x$}
   \label{fig:multirate_intro_us}
\end{figure}

On the other hand, the downsampling operation is implemented by keeping every $M$-th sample of $x[k]$ and
removing in-between samples to generate $y[k]$, then the sampling rate 
becomes multiplied by $1/M$. This procedure is illustrated in Figure \ref{fig:multirate_intro_ds}.
\begin{figure}[t]
   \begin{center}
      \includegraphics[width=6cm]{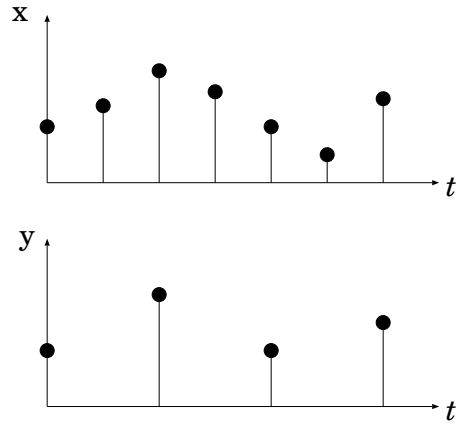}
   \end{center}
   \caption{Downsampling operation $y=(\ds{2})x$}
   \label{fig:multirate_intro_ds}
\end{figure}

Note that the upsampler is left-invertible (and the downsampler is right-invertible), 
that is, $(\ds{M})(\us{M})=I$,
however $(\us{M})(\ds{M})$ is not the identity. In fact, by $(\us{M})(\ds{M})$ the discrete-time signal
$x=\{x[0],x[1],\ldots\}$ is converted into
\begin{equation*}
(\us{M})(\ds{M})x = \{x[0], \underbrace{0,\ldots,0}_{M-1}, x[M], 0,\ldots\}\neq x.
\end{equation*}
To put it differently, downsampling is a lossy data compression, that is, 
the original signal cannot be perfectly reconstructed from the downsampled signal.
On the other hand, we have the duality relation:
$(\us{M})^* = (\ds{M}),\quad (\ds{M})^* = (\us{M})$,
that is, for any signal $x\in l^2$ and $y\in l^2$, we have
$\langle (\us{M})x, y \rangle = \langle x, (\ds{M})y \rangle$,
$\langle (\ds{M})x, y \rangle = \langle x, (\us{M})y \rangle$,
where $\langle \cdot, \cdot \rangle$ denotes the inner product on $l^2$, that is, 
$\langle x, y \rangle:=\sum_{k=0}^{\infty}y[k]x[k]$.

In addition, an upsampler and a decimator are
represented with the discrete-time lifting and its inverse (see Section \ref{sec:SDC_fsfh}):
\begin{equation}
\label{eq:dlift_prop2}
(\us{M})=\idlift{M}\left[\begin{array}{cccc}1&0&\ldots&0\end{array}\right]^T,\quad
(\ds{M})=\left[\begin{array}{cccc}1&0&\ldots&0\end{array}\right]\dlift{M},
\end{equation}
and vice versa:
\begin{equation}
\label{eq:dlift_prop1}
	\dlift{M} := (\ds{M})
	\left[
	\begin{array}{cccc}
	    1 & z & \ldots & z^{M-1}
	\end{array}
	\right]^T,\quad
	\idlift{M} :=
 	\left[
	\begin{array}{cccc}
	    1 & z^{-1} & \cdots & z^{-M+1}
	\end{array}
	\right](\us{M}).
\end{equation}
These relations are used to design interpolators or decimators
in Section \ref{sec:design_interp} and \ref{sec:design_decim}.

Having defined the upsampler and the downsampler, 
we can now explain the interpolator and the decimator.

An interpolator consists of two parts: an upsampler and a digital filter.
Figure \ref{fig:multirate_intro_interp} shows the block diagram of an interpolator.
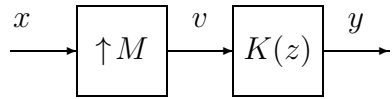
\begin{figure}[t]
    \begin{center}
\begin{picture}( 152,  51)
\put(   0.0,  16.9){\vector( 1, 0){  25.3}}
\put(  25.3,   0.0){\framebox(  33.7,  33.7){$\us{M}$}}
\put(  59.0,  16.9){\vector( 1, 0){  25.3}}
\put(  84.3,   0.0){\framebox(  33.7,  33.7){$K(z)$}}
\put( 118.0,  16.9){\vector( 1, 0){  25.3}}
\put(   1.0,  25.8){$x$}
\put(  68.5,  25.8){$v$}
\put( 127.5,  25.8){$y$}
\end{picture}
    \end{center}
    \caption{Interpolator}
    \label{fig:multirate_intro_interp}
\end{figure}
First, the upsampler inserts zeros and increases the sampling rate.
Then the filter $K(z)$, called {\em interpolation filter},
operates on the $M-1$ zero-valued samples inserted by the upsampler $\us{M}$
to yield nozero values between the original samples, as illustrated in 
Figure \ref{fig:multirate_intro_interp2}.
\begin{figure}[t]
    \begin{center}
	\includegraphics[width=7cm]{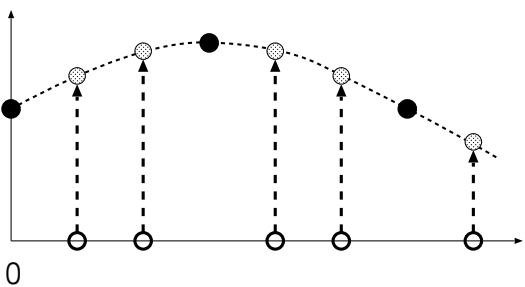}
    \end{center}
    \caption{Signal interpolation by interpolator}
    \label{fig:multirate_intro_interp2}
\end{figure}

Let us now consider the interpolation in the frequency domain.
Assume the sampling period of signal $x$ in Figure \ref{fig:multirate_intro_interp} is 1, that is, 
the Nyquist frequency is $\pi$, and its Fourier transform $X(\omega)$ has characteristic
as shown above in Figure \ref{fig:multirate_intro_interp_freq}.
\begin{figure}[t]
    \begin{center}
	\includegraphics[width=7cm]{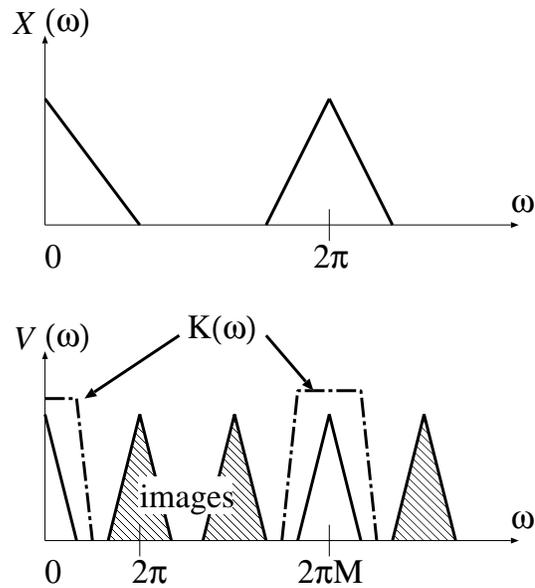}
    \end{center}
    \caption{Imaging components caused by upsampler}
    \label{fig:multirate_intro_interp_freq}
\end{figure}
By upsampling, the Nyquist frequency of the upsampled signal $v$ becomes
$2\pi M$, and there occur unwanted frequency imaging components 
(the shaded portions below in Figure \ref{fig:multirate_intro_interp_freq}).
In order to remove these imaging components, the frequency response $K(\omega)$ of the interpolation filter 
must be of low-pass characteristic with cut-off frequency $\pi$ as shown in Figure \ref{fig:multirate_intro_interp_freq}. 
Therefore, a very sharp low-pass filter close to the ideal one
is often used.

On the other hand, a decimator is  constructed by using a downsampler and a digital filter,
whose block diagram is shown in Figure \ref{fig:multirate_intro_decim}.
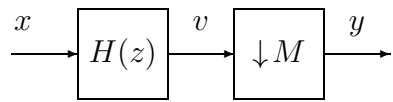
\begin{figure}[t]
    \begin{center}
\begin{picture}( 152,  51)
\put(   0.0,  16.9){\vector( 1, 0){  25.3}}
\put(  25.3,   0.0){\framebox(  33.7,  33.7){$H(z)$}}
\put(  59.0,  16.9){\vector( 1, 0){  25.3}}
\put(  84.3,   0.0){\framebox(  33.7,  33.7){$\ds{M}$}}
\put( 118.0,  16.9){\vector( 1, 0){  25.3}}
\put(   1.0,  25.8){$x$}
\put(  68.5,  25.8){$v$}
\put( 127.5,  25.8){$y$}
\end{picture}
    \end{center}
    \caption{Decimator}
    \label{fig:multirate_intro_decim}
\end{figure}
To see the role of the filter $H(z)$, 
let us look at the frequency domain.
Assume the sampling period of signal $v$ in Figure \ref{fig:multirate_intro_decim} is 1
(the Nyquist Frequency is $\pi$) and its Fourier transform $V(\omega)$ has characteristic
as indicated above in Figure \ref{fig:multirate_intro_decim_freq}.
\begin{figure}[t]
    \begin{center}
	\includegraphics[width=7cm]{multirate_intro_decim_freq.eps}
    \end{center}
    \caption{Aliasing}
    \label{fig:multirate_intro_decim_freq}
\end{figure}
If $V(\omega)$ is not band-limited to $|\omega|\leq \pi/M$,
the spectrum $Y(\omega)$ obtained after downsampling will overlap
as shown below in Figure \ref{fig:multirate_intro_decim_freq}.
This overlapping (the shaded portions below in Figure \ref{fig:multirate_intro_decim_freq})
is called {\em aliasing}.

The filter $H(z)$, called {\em decimation filter} or {\em anti-aliasing filter},
is connected before the downsampler to avoid aliases in advance.
To eliminate the aliasing completely, the filter $H(z)$ must be the ideal
low-pass  filter with cut-off frequency $\pi/M$.
Therefore, similarly to the case of interpolation,
a very sharp low-pass filter close to the ideal one
is often employed.

As mentioned above, the interpolation filter $K(z)$ or the decimation filter $H(z)$ 
ideally has the characteristic shown in Figure \ref{fig:multirate_intro_filter}.
\begin{figure}[t]
    \begin{center}
\begin{picture}( 295, 101)
\put( 178.1,  88.8){$|H(\omega)|$}
\put(   9.4,  88.8){$|K(\omega)|$}
\put( 279.2,  17.4){$\omega$}
\put( 110.6,  17.4){$\omega$}
\put( 224.7,   4.9){$\pi/M$}
\put(  56.0,   4.9){$\pi/M$}
\put( 186.5,  51.1){$1$}
\put(  7.9,  56.5){$M$}
\thicklines
\put( 227.7,  50.6){\line( 0,-1){  33.7}}
\put( 193.9,  50.6){\line( 1, 0){  33.7}}
\put(  59.0,  59.0){\line( 0,-1){  42.2}}
\put(  25.3,  59.0){\line( 1, 0){  33.7}}
\thinlines
\put(  25.3,  16.9){\vector( 0, 1){  67.5}}
\put(  25.3,  16.9){\vector( 1, 0){  84.3}}
\put(  25.3,  16.9){\vector( 0, 1){  67.5}}
\put(  25.3,  16.9){\vector( 1, 0){  84.3}}
\put( 193.9,  16.9){\vector( 0, 1){  67.5}}
\put( 193.9,  16.9){\vector( 1, 0){  84.3}}
\end{picture}
    \end{center}
    \caption{Ideal characteristic of interpolation filter $K(z)$ and decimation filter $H(z)$}
    \label{fig:multirate_intro_filter}
\end{figure}
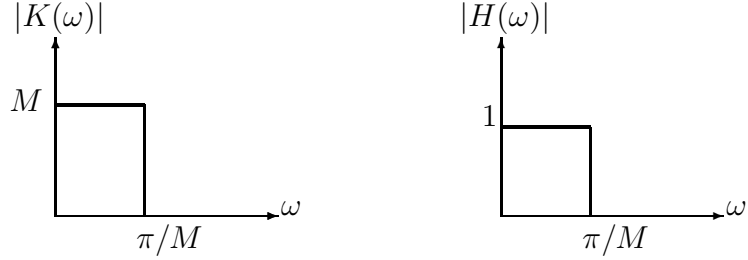
In practice, we cannot realize a filter with such a frequency characteristic.
Therefore the filter is conventionally designed such that the frequency response approximates
that of the ideal filter. 
%Figure \ref{fig:multirate_intro_filter} shows the frequency characteristic 
%of the ideal interpolation filter $K(z)$ or the ideal decimation filter $H(z)$.

As far as the discrete-time signals are concerned,
that is, the spectra of the original analog signals are completely band-limited 
by the Nyquist frequency, the design may be correct.
However, the real analog signals have the spectra beyond the Nyquist frequency,
and hence the ideal characteristic in Figure \ref{fig:multirate_intro_filter}
will not be necessarily optimal.
In the following sections,
we propose an alternative method that takes the analog performance,
in particular, the frequency component over the Nyquist frequency, into account.

\section{Design of interpolators}
\label{sec:design_interp}
\subsection{Problem formulation}
We start by formulating a design problem for (sub)optimal interpolators.
Consider the block diagram shown in Figure~\ref{fig:oversampleDAC}.
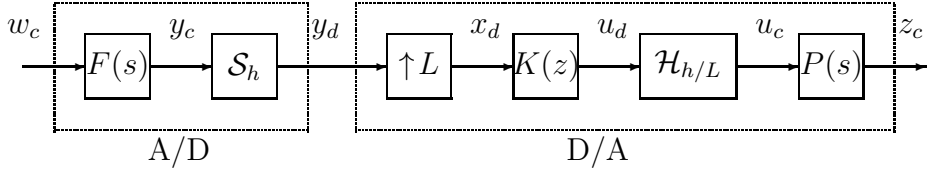
\begin{figure}[t]
\begin{center}
\setlength{\unitlength}{0.5mm}
\begin{picture}( 253,  55)
\put( 105.4,  21.1){\framebox(  16.9,  16.9){$\us{L}$}}
\put( 139.1,  21.1){\framebox(  16.9,  16.9){$K(z)$}}
\put( 172.8,  21.1){\framebox(  25.3,  16.9){$\hold{h/L}$}}
\put( 215.0,  21.1){\framebox(  16.9,  16.9){$P(s)$}}
\put(  59.0,  21.1){\framebox(  16.9,  16.9){$\samp{h}$}}
\put(  25.3,  21.1){\framebox(  16.9,  16.9){$F(s)$}}
\put(  16.9,  12.6){\dashbox{0.5}(  67.5,  33.7){}}
\put(  97.0,  12.6){\dashbox{0.5}( 143.3,  33.7){}}
\put( 231.9,  29.5){\vector( 1, 0){  16.9}}
\put(  42.2,  29.5){\vector( 1, 0){  16.9}}
\put(   8.4,  29.5){\vector( 1, 0){  16.9}}
\put( 122.3,  29.5){\vector( 1, 0){  16.9}}
\put( 156.0,  29.5){\vector( 1, 0){  16.9}}
\put( 198.1,  29.5){\vector( 1, 0){  16.9}}
\put(  75.9,  29.5){\vector( 1, 0){  29.5}}
\put(  41.7,   4.7){A/D}
\put( 153.1,   4.7){D/A}
\put(   4.4,  38.4){$w_c$}
\put(  47.4,  38.4){$y_c$}
\put(  85.3,  38.4){$y_d$}
\put( 127.5,  38.4){$x_d$}
\put( 161.2,  38.4){$u_d$}
\put( 203.4,  38.4){$u_c$}
\put( 241.3,  38.4){$z_c$}
\end{picture}
\end{center}
\caption{Signal reconstruction by interpolator}
\label{fig:oversampleDAC}
\end{figure}
The incoming signal $w_c$ first goes through an
anti-aliasing filter $F(s)$ and the filtered signal
$y_{c}$ becomes nearly (but not entirely) band-limited.  
The filter $F(s)$ governs the frequency-domain characteristic of the
analog signal $y_{c}$.  This signal is then 
sampled by the sampler $\samp{h}$ to become a discrete-time signal $y_d$ 
with sampling period $h$.  
%This signal is usually stored or 
%transmitted with some media (e.g., CD) or a channel.

To restore $y_{c}$ we usually let it pass through a digital
filter, a hold device and then an analog filter.  The present
setup however places yet one more step: 
The discrete-time signal $y_{d}$ is first upsampled by $\us{L}$
, and becomes another discrete-time signal $x_{d}$ with
sampling period $h/L$.  The discrete-time signal $x_d$ is then
processed by a digital filter $K(z)$, 
becomes a continuous-time signal $u_{c}$
by going through the zero-order hold $\hold{h/L}$ (that works
in sampling period $h/L$), and then becomes the final signal
by passing through an analog filter $P(s)$.  An advantage here
is that one can use a fast hold device $\hold{h/L}$ thereby making
more precise signal restoration possible.  
The objective here is to design the digital filter
$K(z)$ for given $F(s)$, $L$ and $P(s)$.

Figure~\ref{fig:multirate_design_interp_error} shows the block diagram of the error
system for the design.  The delay in the upper portion of the
diagram corresponds to the fact that we allow a certain amount
of time delay for signal reconstruction.
\begin{figure}[t]
\begin{center}
\setlength{\unitlength}{0.5mm}
\begin{picture}( 240,  55)
\footnotesize
\put( 224.4,  13.1){$-$}
\put( 211.8,  34.2){$+$}
\put( 232.9,  30.0){$e_c$}
\put(   5.2,  30.0){$w_c$}
\put(  43.2,  25.8){$y_c$}
\put( 223.4,  25.3){\vector( 1, 0){  12.6}}
\put( 219.2,   8.4){\vector( 0, 1){  12.6}}
\put( 219.2,  42.2){\vector( 0,-1){  12.6}}
\put( 219.2,  25.3){\circle{ 8}}
\put( 126.5,  42.2){\line( 1, 0){  92.7}}
\put( 210.8,   8.4){\line( 1, 0){   8.4}}
\put( 177.1,   8.4){\vector( 1, 0){  16.9}}
\put( 134.9,   8.4){\vector( 1, 0){  16.9}}
\put( 101.2,   8.4){\vector( 1, 0){  16.9}}
\put(  67.5,   8.4){\vector( 1, 0){  16.9}}
\put(  42.2,   8.4){\vector( 1, 0){   8.4}}
\put(  42.2,  42.2){\vector( 1, 0){  59.0}}
\put(  33.7,  25.3){\line( 1, 0){   8.4}}
\put(  42.2,  42.2){\line( 0,-1){  33.7}}
\put(   0.0,  25.3){\vector( 1, 0){  16.9}}
\put( 101.2,  33.7){\framebox(  25.3,  16.9){$e^{-mhs}$}}
\put(  16.9,  16.9){\framebox(  16.9,  16.9){$F(s)$}}
\put(  50.6,   0.0){\framebox(  16.9,  16.9){$\samp{h}$}}
\put(  84.3,   0.0){\framebox(  16.9,  16.9){$\us{L}$}}
\put( 118.0,   0.0){\framebox(  16.9,  16.9){$K(z)$}}
\put( 151.8,   0.0){\framebox(  25.3,  16.9){$\hold{h/L}$}}
\put( 193.9,   0.0){\framebox(  16.9,  16.9){$P(s)$}}
\end{picture}
\end{center}
\caption{Signal reconstruction error system}
\label{fig:multirate_design_interp_error}
\end{figure}
Let $\T_{I}$ denote the input/output operator
from $w_c$ to $e_c:=z_c(t)-u_c(t-mh)$.  
Our design problem is as follows:
\begin{problem}
\label{prob:Hinf_interp}
Given a stable, strictly proper $F(s)$, stable, proper $P(s)$, upsampling factor $L\in \Nset$,
delay step $m\in \Nset$, sampling period $h>0$ and  
an attenuation level
$\gamma>0$, find a digital filter $K(z)$
such that
\begin{equation}
	\norm{\T_{I}}:= \sup_{\substack{w_c \in L^2[0,\infty)\\w_c\neq 0}}
	    \frac{\twonorm{\T_{I}w_c}_{L^2[0,\infty)}}{\twonorm{w_c}_{L^2[0,\infty)}} < \gamma .
	    \label{eq:Hinf_interp}
\end{equation}
\end{problem}
The norm in (\ref{eq:Hinf_interp}) is an $L^2$ induced norm, 
which is equal to the $H^\infty$-norm of the system $\T_{I}$,
% \cite{BamPea92,CheFra},
so Problem \ref{prob:Hinf_interp} is the $H^\infty$ optimization problem.

\subsection{Reduction to a finite-dimensional problem}
A difficulty in Problem \ref{prob:Hinf_interp}
is that it involves
a continuous-time delay, and hence it is an infinite-dimensional
problem.  Another difficulty is that it contains
the upsampler $\us{L}$
so that it makes the overall system time-varying (to be precise, periodically time-varying).

However we can reduce the problem to a finite-dimensional single-rate one. 
Let $\{A_F,B_F,C_F,0\}$ be a realization of $F(s)$, that is,
the state equation of $F(s)$:
    \begin{equation*}
	\dot{x}_F = A_Fx_F + B_Fw_c,\quad y_c = C_Fx_F.
    \end{equation*}
%Define the following operator:
%    \begin{equation}
%	\begin{split}
%	    \D_{11} &: L^2[0,h) \rightarrow L^2[0,h)\\
%	    &: w_k \mapsto \int_0^\theta C_Fe^{A_F(\theta-\tau)}B_Fw_k(\tau)d\tau
%	\end{split}
%	\label{eq:sekkei:D11}
%   \end{equation}
\begin{theorem}
For the error system $\T_{I}$, there exist (finite-dimensional) discrete-time systems
$\{T_{I,N}: N=L,2L,\ldots\}$ such that
\begin{equation} \label{eq:theorem_interp}
  \lim_{N\rightarrow\infty}\|T_{I,N}\| = \|\T_{I}\|.
\end{equation}
\end{theorem}
\begin{proof}
We first reduce the problem to a single-rate one.
Recall the property (\ref{eq:dlift_prop1}) of the discrete-time lifting $\dlift{L}$ and its inverse
$\idlift{L}$:
\begin{equation*}
    \dlift{L} := (\ds{L})
	\left[\begin{array}{cccc}
	  1 & z & \ldots & z^{L-1}
	\end{array}\right]^T,\quad
    \idlift{L} :=
 	\left[\begin{array}{cccc}
	  1 & z^{-1} & \cdots & z^{-L+1}
	\end{array}\right](\us{L}).
\end{equation*}
Then $K(z)(\us{L})$ can be rewritten as
\begin{equation*}
    K(z)(\us{L})
	    = \idlift{L}\widetilde{K}(z),\quad
    \widetilde{K}(z)
	    := \dlift{L}K(z)\idlift{L}
		  \left[\begin{array}{cccc}
			1&0&\ldots&0
		  \end{array}\right]^T.
\end{equation*}
The lifted system $\widetilde{K}(z)$ is an LTI, single-input/$L$-output system
that satisfies
\begin{displaymath}
    K(z) = \left[\begin{array}{cccc}
		 1&z^{-1}&\cdots &z^{-L+1}
	   \end{array}\right] \widetilde{K}(z^L).
\end{displaymath}
Using the generalized hold $\ghold{h}$ defined by
\begin{equation*}
    \ghold{h} : l^2 \ni
    \vect{v} \mapsto u\in L^2,\quad
    u(kh+\theta) = {\bf H}(\theta) \vect{v}[k],\quad
    \theta \in [0,h),\quad k=0,1,2,\ldots,
\end{equation*}
where $\holdfunc(\cdot)$ is the hold function:
\begin{equation*}
    \holdfunc(\theta)
	:= \begin{cases}
		      \left[\begin{array}{ccccc}1&0&0&\ldots&0\end{array}\right],
		      \theta \in [0,h/L),\\
		      \left[\begin{array}{ccccc}0&1&0&\ldots&0\end{array}\right],
		      \theta \in [h/L,2h/L),\\
		      \qquad\qquad\cdots\\
		      \left[\begin{array}{ccccc}0&0&\ldots&0&1\end{array}\right],
		      \theta \in [(L-1)h/L, h),
	    \end{cases}
\end{equation*}
we obtain the identity
\[
    \hold{h/L}\idlift{L}=\ghold{h}.
\]
This yields
\[
    \hold{h/L}K(z)(\us{L})\samp{h}=\ghold{h}\widetilde{K}(z)\samp{h}.
\]
Hence Figure~\ref{fig:multirate_design_interp_error} is equivalent to 
Figure~\ref{fig:multirate_design_interp_error2}.
\begin{figure}[t]
    \begin{center}
\setlength{\unitlength}{0.5mm}
\begin{picture}( 219,  55)
\put( 190.7,  13.1){$-$}
\put( 126.5,   0.0){\framebox(  16.9,  16.9){$\ghold{h}$}}
\put( 199.1,  30.0){$e_c$}
\put(   4.0,  30.0){$w_c$}
\put( 189.7,  25.3){\vector( 1, 0){  12.6}}
\put( 185.5,   8.4){\vector( 0, 1){  12.6}}
\put(  84.3,   0.0){\framebox(  25.3,  16.9){$\widetilde{K}(z)$}}
\put( 185.5,  25.3){\circle{ 8}}
\put(  50.6,   0.0){\framebox(  16.9,  16.9){$\samp{h}$}}
\put( 177.1,   8.4){\line( 1, 0){   8.4}}
\put( 143.3,   8.4){\vector( 1, 0){  16.9}}
\put(  42.2,  42.2){\line( 0,-1){  33.7}}
\put(  16.9,  16.9){\framebox(  16.9,  16.9){$F(s)$}}
\put(   0.0,  25.3){\vector( 1, 0){  16.9}}
\put(  42.2,   8.4){\vector( 1, 0){   8.4}}
\put(  42.2,  42.2){\vector( 1, 0){  59.0}}
\put(  33.7,  25.3){\line( 1, 0){   8.4}}
\put( 160.2,   0.0){\framebox(  16.9,  16.9){$P(s)$}}
\put( 101.2,  33.7){\framebox(  25.3,  16.9){$e^{-mhs}$}}
\put( 126.5,  42.2){\line( 1, 0){  59.0}}
\put( 185.5,  42.2){\vector( 0,-1){  12.6}}
\put(  67.5,   8.4){\vector( 1, 0){  16.9}}
\put( 109.6,   8.4){\vector( 1, 0){  16.9}}
\end{picture}
    \end{center}
    \caption{Reduced single-rate problem}
    \label{fig:multirate_design_interp_error2}
\end{figure}

We modify the diagram in Figure \ref{fig:multirate_design_interp_error2} into the
diagram in Figure \ref{fig:multirate_design_interp_gplant_SD},
\begin{figure}[t]
    \begin{center}
\setlength{\unitlength}{0.55mm}
\begin{picture}( 207, 126)
%\put(  85.3, 118.5){${\cal G}_s$}
%\put(  16.9,  46.4){\framebox( 151.8,  67.5){}}
\put( 173.0,  68.0){$u_d$}
\put(   2.7,  68.0){$y_d$}
\put(   1.0, 101.7){$e_c$}
\put( 173.8, 101.7){$w_c$}
\put( 177.1,  12.6){\line( 0, 1){  50.6}}
\put( 139.1,  63.2){\vector(-1, 0){   8.4}}
\put(   8.4,  63.2){\line( 0,-1){  50.6}}
\put(  21.1,  63.2){\line(-1, 0){  12.6}}
\put(  54.8,  97.0){\vector(-1, 0){  54.8}}
\put( 181.3,  97.0){\vector(-1, 0){  50.6}}
\put( 177.1,  63.2){\vector(-1, 0){  12.6}}
\put(  54.8,  63.2){\vector(-1, 0){   8.4}}
\put(   8.4,  12.6){\vector( 1, 0){  67.5}}
\put( 109.6,  12.6){\line( 1, 0){  67.5}}
\put(  75.9,   0.0){\framebox(  33.7,  25.3){$\widetilde{K}(z)$}}
\put( 139.1,  50.6){\framebox(  25.3,  25.3){$\ghold{h}$}}
\put(  21.1,  50.6){\framebox(  25.3,  25.3){$\samp{h}$}}
\put(  54.8,  50.6){\framebox(  75.9,  59.0){$\left[\begin{array}{cc}e^{-mhs}F(s),&\!\!-P(s)\\F(s),&\!\!0\end{array}\right]$}}
\end{picture}
    \end{center}
    \caption{Sampled-data system $\T_{I}$}
    \label{fig:multirate_design_interp_gplant_SD}
\end{figure}
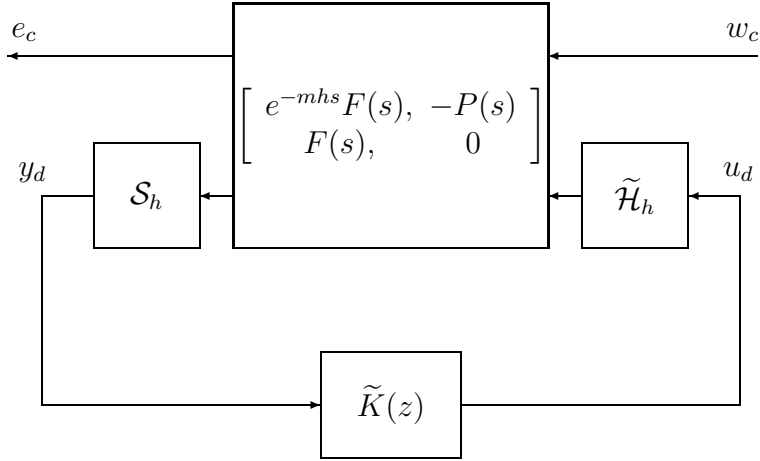
and we introduce the fast-sampling/fast-hold approximation \cite{KelAnd92,YamMadAnd99}
in order to obtain a finite-dimensional discrete-time system approximately.
Figure \ref{fig:multirate_design_interp_fsfh} illustrates the procedure.
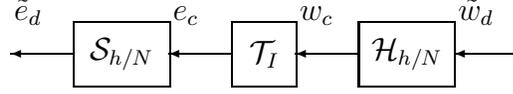
\begin{figure}[t]
    \begin{center}
\setlength{\unitlength}{0.5mm}
\begin{picture}( 143,  25)
\put(  16.9,   0.0){\framebox(  25.3,  16.9){$\samp{h/N}$}}
\put( 134.9,   8.4){\vector(-1, 0){  16.9}}
\put(  92.7,   8.4){\vector(-1, 0){  16.9}}
\put(  59.0,   8.4){\vector(-1, 0){  16.9}}
\put(  16.9,   8.4){\vector(-1, 0){  16.9}}
\put(  92.7,   0.0){\framebox(  25.3,  16.9){$\hold{h/N}$}}
\put(  59.0,   0.0){\framebox(  16.9,  16.9){$\T_{I}$}}
\put(  43.2,  17.4){$e_c$}
\put(   1.0,  17.4){$\tilde{e}_{d}$}
\put(  76.9,  17.4){$w_c$}
\put( 119.0,  17.4){$\tilde{w}_{d}$}
\end{picture}
    \end{center}
    \caption{Fast-sampling/fast-hold discretization}
    \label{fig:multirate_design_interp_fsfh}
\end{figure}
By the fast-sampling/fast-hold approximation,
we obtain the approximated discrete-time system $T_{I,N}$ ($N:=Ll,\; l\in \Nset$),
\begin{equation*}
    T_{I,N}(z) = z^{-m}G_{I11,N}(z)+ G_{I12,N}(z)\widetilde{K}(z)G_{I21,N}(z),
\end{equation*}
where
\begin{gather*}
    G_{I11,N}:=z^{-m}\dliftsys{N}{F},\quad
    G_{I12,N}:=-\dliftsys{N}{P}H,\quad
    G_{I21,N}:=S\dliftsys{N}{F},\\
    S := [1,\underbrace{0,\ldots,0}_{N-1}],\quad
    H := \underbrace{\left[\begin{array}{ccc}q& & \\ &\ddots & \\ & &q
	\end{array}\right]}_{L},\quad
    q := [\underbrace{1,\ldots,1}_{l}]^T.
\end{gather*}
Figure \ref{fig:multirate_interp_gplant_fsfh} shows the obtained discrete-time
system,
\begin{figure}[t]
    \begin{center}
\setlength{\unitlength}{0.4mm}
\begin{picture}( 200, 114)
%\begin{picture}( 400, 500)
\put(  25.3,  63.2){\line(-1, 0){  16.9}}
\put( 109.6,  12.6){\line( 1, 0){  33.7}}
\put( 173.0,  68.0){$u_d$}
\put(   2.7,  68.0){$y_d$}
\put(   1.0, 101.7){$\tilde{e}_d$}
\put( 173.8, 101.7){$\tilde{w}_d$}
\put( 177.1,  12.6){\line( 0, 1){  50.6}}
\put( 139.1,  12.6){\line( 1, 0){  37.9}}
\put( 185.5,  97.0){\vector(-1, 0){  25.3}}
\put(   8.4,  63.2){\line( 0,-1){  50.6}}
\put(  75.9,   0.0){\framebox(  33.7,  25.3){$\widetilde{K}(z)$}}
\put(   8.4,  12.6){\vector( 1, 0){  67.5}}
\put(  25.3,  97.0){\vector(-1, 0){  25.3}}
\put(  25.3,  50.6){\framebox( 134.9,  59.0){
    $\left[
    \begin{array}{cc}
    z^{-m}G_{I11,N} & G_{I12,N}\\
     G_{I21,N}       & 0
    \end{array}
    \right]
    $}}
\put( 177.1,  63.2){\vector(-1, 0){  16.9}}
\end{picture}
    \end{center}
    \caption{Discrete-time system $T_{I,N}$}
    \label{fig:multirate_interp_gplant_fsfh}
\end{figure}
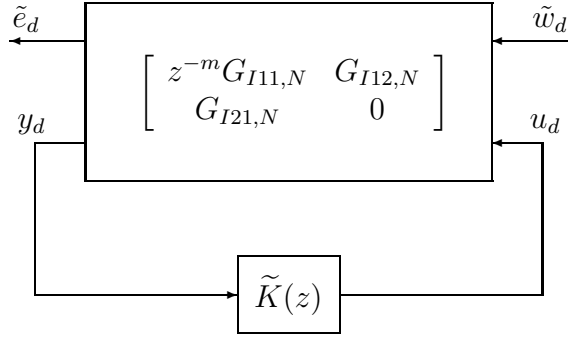
where $\widetilde{K}(z)$ is an LTI, single-input/$L$-output system
that satisfies
\begin{displaymath}
    K(z) = \left[\begin{array}{cccc}
		 1&z^{-1}&\cdots &z^{-L+1}
	   \end{array}\right] \widetilde{K}(z^L).
\end{displaymath}
The convergence of (\ref{eq:theorem_interp}) is guaranteed in \cite{YamMadAnd99}.
\end{proof}

\section{Design of decimators}
\label{sec:design_decim}
\subsection{Problem formulation}
We now formulate a design problem for optimal decimators.
While this can be considered dually with interpolators, it is
less studied in the literature.  Downsampling occurs usually in
the filter bank design, and its independent design has received
less attention.

Consider the block diagram Figure~\ref{fig:multirate_design_decim}.
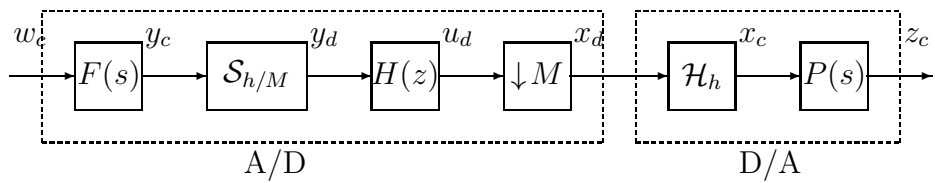
\begin{figure}[t]
\begin{center}
\setlength{\unitlength}{0.52mm}
\begin{picture}( 245,  51)
\put( 186.5,   0.5){D/A}
\put(  60.0,   0.5){A/D}
\put( 228.7,  34.2){$z_c$}
\put( 186.5,  34.2){$x_c$}
\put( 144.3,  34.2){$x_d$}
\put( 110.6,  34.2){$u_d$}
\put(  76.9,  34.2){$y_d$}
\put(  34.7,  34.2){$y_c$}
\put(   1.0,  34.2){$w_c$}
\put( 160.2,   8.4){\dashbox(  67.5,  33.7){}}
\put(   8.4,   8.4){\dashbox( 143.3,  33.7){}}
\put( 219.2,  25.3){\vector( 1, 0){  16.9}}
\put( 185.5,  25.3){\vector( 1, 0){  16.9}}
\put( 143.3,  25.3){\vector( 1, 0){  25.3}}
\put( 109.6,  25.3){\vector( 1, 0){  16.9}}
\put(  75.9,  25.3){\vector( 1, 0){  16.9}}
\put(  33.7,  25.3){\vector( 1, 0){  16.9}}
\put(   0.0,  25.3){\vector( 1, 0){  16.9}}
\put( 202.4,  16.9){\framebox(  16.9,  16.9){$P(s)$}}
\put( 168.6,  16.9){\framebox(  16.9,  16.9){$\hold{h}$}}
\put( 126.5,  16.9){\framebox(  16.9,  16.9){$\ds{M}$}}
\put(  92.7,  16.9){\framebox(  16.9,  16.9){$H(z)$}}
\put(  50.6,  16.9){\framebox(  25.3,  16.9){$\samp{h/M}$}}
\put(  16.9,  16.9){\framebox(  16.9,  16.9){$F(s)$}}
\end{picture}
\end{center}
\caption{Signal reconstruction with decimator}
\label{fig:multirate_design_decim}
\end{figure}
The incoming signal $w_c$ first goes through an
anti-aliasing filter $F(s)$ and the filtered signal
$y_{c}$ becomes nearly (but not entirely) band-limited.  
%$F(s)$ governs the frequency-domain characteristic of the
%analog signal $y_{c}$.  
This signal is then 
sampled by $\samp{h/M}$ to become a discrete-time signal $y_d$ 
with sampling period $h/M$.  
%This signal is usually stored or 
%transmitted with some media (e.g., CD) or a channel.

%To restore $y_{c}$ we usually let it pass through a digital
%filter, a hold device and then an analog filter.  The present
%setup however places yet one more step: 
The discrete-time signal $y_{d}$ is first processed by a digital filter
$H(z)$. Then the filtered signal $x_{d}$ is downsampled by $\ds{M}$,
and becomes another discrete-time signal $u_{d}$ with
sampling period $h$.  The discrete-time signal $u_d$ then
becomes a continuous-time signal $u_{c}$
by going through the zero-order hold $\hold{h}$,
%(that works
%in sampling period $h/M$), 
and then becomes the final signal
by going through an analog filter $P(s)$.  
%An advantage here
%is that one can use a fast hold device $\hold{h/M}$ thereby making
%more precise signal restoration possible.  
The objective here is to design the digital filter
$H(z)$ for given $F(s)$, $M$ and $P(s)$.

Figure~\ref{fig:multirate_design_decim_error} shows the block diagram of the error
system for the design.  The delay in the upper portion of the
diagram corresponds to the fact that we allow a certain amount
of time delay for signal reconstruction.
\begin{figure}[t]
\begin{center}
\setlength{\unitlength}{0.5mm}
\begin{picture}( 245,  59)
\footnotesize    
\put(   0.0,  25.3){\vector( 1, 0){  16.9}}
\put(  42.2,  42.2){\line( 0,-1){  33.7}}
\put(  33.7,  25.3){\line( 1, 0){   8.4}}
\put(  42.2,  42.2){\vector( 1, 0){  59.0}}
\put( 224.4,  13.1){$-$}
\put( 211.8,  34.2){$+$}
\put(  43.2,  25.8){$y_c$}
\put(   5.2,  30.0){$w_c$}
\put( 223.4,  25.3){\vector( 1, 0){  12.6}}
\put( 219.2,   8.4){\vector( 0, 1){  12.6}}
\put( 219.2,  42.2){\vector( 0,-1){  12.6}}
\put( 219.2,  25.3){\circle{ 8}}
\put( 126.5,  42.2){\line( 1, 0){  92.7}}
\put( 210.8,   8.4){\line( 1, 0){   8.4}}
\put( 101.2,  33.7){\framebox(  25.3,  16.9){$e^{-mhs}$}}
\put(  16.9,  16.9){\framebox(  16.9,  16.9){$F(s)$}}
\put(  42.2,   8.4){\vector( 1, 0){   8.4}}
\put(  50.6,   0.0){\framebox(  25.3,  16.9){$\samp{h/M}$}}
\put(  92.7,   0.0){\framebox(  16.9,  16.9){$H(z)$}}
\put( 126.5,   0.0){\framebox(  16.9,  16.9){$\ds{M}$}}
\put( 160.2,   0.0){\framebox(  16.9,  16.9){$\hold{h}$}}
\put( 193.9,   0.0){\framebox(  16.9,  16.9){$P(s)$}}
\put(  75.9,   8.4){\vector( 1, 0){  16.9}}
\put( 109.6,   8.4){\vector( 1, 0){  16.9}}
\put( 143.3,   8.4){\vector( 1, 0){  16.9}}
\put( 177.1,   8.4){\vector( 1, 0){  16.9}}
\put( 228.7,  34.2){$e_c$}
\end{picture}
\end{center}
\caption{Signal reconstruction error system}
\label{fig:multirate_design_decim_error}
\end{figure}
Let $\T_{D}$ denote the input/output operator from $w_c$ to $e_c$ in Figure \ref{fig:multirate_design_decim_error}.
Our design problem is then as follows:
\begin{problem}
    \label{prob:Hinf_decim}
Given a stable, strictly proper $F(s)$, stable, proper $P(s)$, 
downsampling factor $M\in \Nset$, delay step $m\in \Nset$, sampling period $h>0$ and
an attenuation level $\gamma>0$, find a digital filter $H(z)$
such that
    \begin{equation}
	\norm{\T_{D}}:= \sup_{\substack{w_c \in L^2[0,\infty)\\ w_c\neq 0}}
	    \frac{\twonorm{\T_{D}w_c}_{L^2[0,\infty)}}{\twonorm{w_c}_{L^2[0,\infty)}} < \gamma.
	    \label{eq:Hinf_decim}
    \end{equation}
\end{problem}

\subsection{Reduction to a finite-dimensional problem}
\begin{theorem}
For the error system $\T_{D}$, there exist (finite-dimensional) discrete-time systems
$\{T_{D,N}: N=M,2M,\ldots\}$ such that
\begin{equation} \label{eq:theorem_decim}
  \lim_{N\rightarrow\infty}\|T_{D,N}\| = \|\T_{D}\|.
\end{equation}

%    '±'±'Å $\widetilde{H}(z)$ 'Í $H(z)$ 'Ì—£ŽUŽžŠÔƒŠƒtƒeƒBƒ"ƒO'Å' 'èCŽŸ
%    Ž®'Å—^'¦'ç'ê'éD
%    \begin{equation*}
%	\widetilde{H}(z)
%	    := \left[\begin{array}{cccc}1&0& \cdots &0\end{array}\right]\dlift{M}H(z)\idlift{M}
%   \end{equation*}
\end{theorem}	
\begin{proof}
Using the discrete-time lifting $\dlift{M}$ we rewrite $(\ds{M})H(z)$
as
\begin{equation*}
	(\ds{M})H(z)
	  = \widetilde{H}(z)\dlift{M},\quad
	\widetilde{H}(z)
	    := \left[\begin{array}{cccc}1&0&\cdots &0\end{array}\right]\dlift{M}H(z)\idlift{M},
\end{equation*}
where the system $\widetilde{H}(z)$ is an LTI, $M$-input/single-output system that 
satisfies
\begin{equation*}
    H(z) = \widetilde{H}(z^M)
	\left[
	 \begin{array}{cccc}
	     1&
	     z&
	     \ldots&
	     z^{M-1}
	 \end{array}
       \right]^T.
\end{equation*}
Using the generalized sampler $\gsamp{h}$ defined by
\begin{equation*}
	\gsamp{h}
	:  L^2 \ni u \longmapsto \vect{v} \in l^2,\quad
	\vect{v}[k] :=
	\left[
	\begin{array}{c}
	    u\bigl(kh\bigr)\\
	    u\bigl(kh+h/M\bigr)\\
	    \vdots\\
	    u\bigl(kh+(M\!-\!1)h/M\bigr)\\
	\end{array}
	\right],\quad
	k=0,1,2,\ldots,
\end{equation*}
we obtain the identity
\begin{equation*}
    \dlift{M}\samp{h/M} = \gsamp{h}.
\end{equation*}
Hence Figure~\ref{fig:multirate_design_decim_error} is equivalent to 
Figure~\ref{fig:multirate_design_decim_error2}.
\begin{figure}[t]
    \begin{center}
\setlength{\unitlength}{0.5mm}
\begin{picture}( 219,  55)
\put( 190.7,  13.1){$-$}
\put( 199.1,  30.0){$e_c$}
\put(   4.0,  30.0){$w_c$}
\put( 189.7,  25.3){\vector( 1, 0){  12.6}}
\put( 185.5,   8.4){\vector( 0, 1){  12.6}}
\put( 185.5,  25.3){\circle{ 8}}
\put( 177.1,   8.4){\line( 1, 0){   8.4}}
\put( 143.3,   8.4){\vector( 1, 0){  16.9}}
\put(  42.2,  42.2){\line( 0,-1){  33.7}}
\put(   0.0,  25.3){\vector( 1, 0){  16.9}}
\put(  42.2,   8.4){\vector( 1, 0){   8.4}}
\put(  42.2,  42.2){\vector( 1, 0){  59.0}}
\put(  33.7,  25.3){\line( 1, 0){   8.4}}
\put( 126.5,  42.2){\line( 1, 0){  59.0}}
\put( 185.5,  42.2){\vector( 0,-1){  12.6}}
\put(  67.5,   8.4){\vector( 1, 0){  16.9}}
\put( 109.6,   8.4){\vector( 1, 0){  16.9}}
\put(  16.9,  16.9){\framebox(  16.9,  16.9){$F(s)$}}
\put(  50.6,   0.0){\framebox(  16.9,  16.9){$\gsamp{h}$}}
\put( 126.5,   0.0){\framebox(  16.9,  16.9){$\hold{h}$}}
\put(  84.3,   0.0){\framebox(  25.3,  16.9){$\widetilde{H}(z)$}}
\put( 160.2,   0.0){\framebox(  16.9,  16.9){$P(s)$}}
\put( 101.2,  33.7){\framebox(  25.3,  16.9){$e^{-mhs}$}}
\end{picture}
    \end{center}
    \caption{Reduced single-rate problem}
    \label{fig:multirate_design_decim_error2}
\end{figure}
As has been mentioned above, this can be reduced to a finite-dimensional
discrete-time system.

We modify the block diagram in Figure \ref{fig:multirate_design_decim_error2} into the 
block diagram in Figure \ref{fig:multirate_design_decim_gplant_SD}.
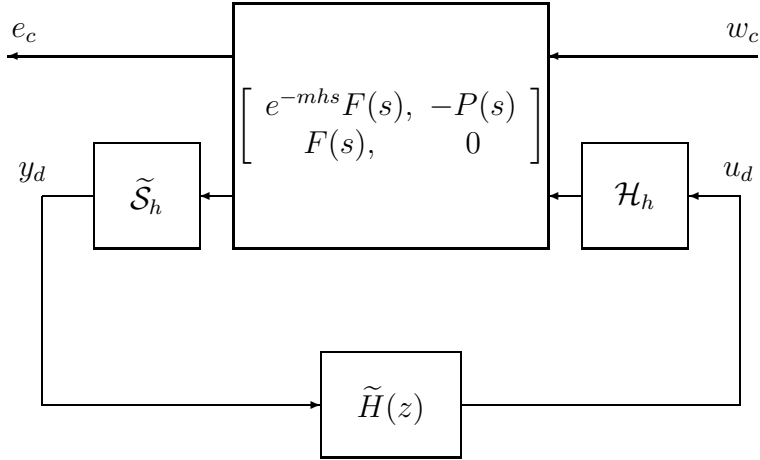
\begin{figure}[t]
    \begin{center}
\setlength{\unitlength}{0.55mm}
\begin{picture}( 207, 126)
%\put(  85.3, 118.5){${\cal G}_s$}
%\put(  16.9,  46.4){\framebox( 151.8,  67.5){}}
\put( 173.0,  68.0){$u_d$}
\put(   2.7,  68.0){$y_d$}
\put(   1.0, 101.7){$e_c$}
\put( 173.8, 101.7){$w_c$}
\put( 177.1,  12.6){\line( 0, 1){  50.6}}
\put( 139.1,  63.2){\vector(-1, 0){   8.4}}
\put(   8.4,  63.2){\line( 0,-1){  50.6}}
\put(  21.1,  63.2){\line(-1, 0){  12.6}}
\put(  54.8,  97.0){\vector(-1, 0){  54.8}}
\put( 181.3,  97.0){\vector(-1, 0){  50.6}}
\put( 177.1,  63.2){\vector(-1, 0){  12.6}}
\put(  54.8,  63.2){\vector(-1, 0){   8.4}}
\put(   8.4,  12.6){\vector( 1, 0){  67.5}}
\put( 109.6,  12.6){\line( 1, 0){  67.5}}
\put(  75.9,   0.0){\framebox(  33.7,  25.3){$\widetilde{H}(z)$}}
\put( 139.1,  50.6){\framebox(  25.3,  25.3){$\hold{h}$}}
\put(  21.1,  50.6){\framebox(  25.3,  25.3){$\gsamp{h}$}}
\put(  54.8,  50.6){\framebox(  75.9,  59.0){$\left[\begin{array}{cc}e^{-mhs}F(s),&\!\!-P(s)\\F(s),&\!\!0\end{array}\right]$}}
\end{picture}
    \end{center}
    \caption{Sampled-data system $\T_{D}$}
    \label{fig:multirate_design_decim_gplant_SD}
\end{figure}
Then by using the fast-sampling/fast-hold method, the sampled-data system in Figure
\ref{fig:multirate_design_decim_gplant_SD} is approximated to the following discrete-time system
$T_{D,N}$ ($N:=Ml,\; l\in \Nset$),
\begin{equation*}
    T_{D,N}(z) = z^{-m}G_{D11,N}(z)+ G_{D12,N}(z)\widetilde{K}(z)G_{D21,N}(z),
\end{equation*}
where
\begin{gather*}
    G_{D11,N}:=z^{-m}\dliftsys{N}{F},\quad
    G_{D12,N}:=-\dliftsys{N}{P}H,\quad
    G_{D21,N}:=S\dliftsys{N}{F},\\
S := \left.\left[\begin{array}{ccc}p& & \\ &\ddots & \\ & &p
	\end{array}\right]\right\}M,\quad 
p := [1,\underbrace{0,\ldots,0}_{l-1}],\quad
H := [\underbrace{1,\ldots,1}_{N}]^T.
\end{gather*}
Figure \ref{fig:multirate_design_decim_gplant_fsfh} shows the obtained 
finite-dimensional discrete-time
system.
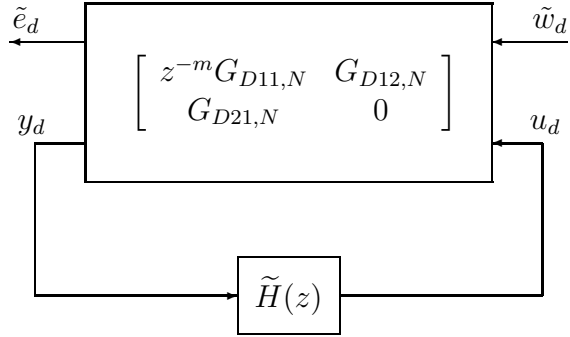
\begin{figure}[t]
    \begin{center}
\setlength{\unitlength}{0.4mm}
\begin{picture}( 200, 114)
%\begin{picture}( 400, 500)
\put(  25.3,  63.2){\line(-1, 0){  16.9}}
\put( 109.6,  12.6){\line( 1, 0){  33.7}}
\put( 173.0,  68.0){$u_d$}
\put(   2.7,  68.0){$y_d$}
\put(   1.0, 101.7){$\tilde{e}_d$}
\put( 173.8, 101.7){$\tilde{w}_d$}
\put( 177.1,  12.6){\line( 0, 1){  50.6}}
\put( 139.1,  12.6){\line( 1, 0){  37.9}}
\put( 185.5,  97.0){\vector(-1, 0){  25.3}}
\put(   8.4,  63.2){\line( 0,-1){  50.6}}
\put(  75.9,   0.0){\framebox(  33.7,  25.3){$\widetilde{H}(z)$}}
\put(   8.4,  12.6){\vector( 1, 0){  67.5}}
\put(  25.3,  97.0){\vector(-1, 0){  25.3}}
\put(  25.3,  50.6){\framebox( 134.9,  59.0){
    $\left[
    \begin{array}{cc}
    z^{-m}G_{D11,N} & G_{D12,N}\\
     G_{D21,N}       & 0
    \end{array}
    \right]
    $}}
\put( 177.1,  63.2){\vector(-1, 0){  16.9}}
\end{picture}
    \end{center}
    \caption{Discrete-time system $T_{D,N}$}
    \label{fig:multirate_design_decim_gplant_fsfh}
\end{figure}

The convergence of (\ref{eq:theorem_decim}) is guaranteed in \cite{YamMadAnd99}.
\end{proof}
Note that the decimation filter 
\begin{equation*}
    H(z) = \widetilde{H}(z^M)
	\left[
	 \begin{array}{cccc}
	     1&
	     z&
	     \ldots&
	     z^{M-1}
	 \end{array}
       \right]^T,
\end{equation*}
may not be causal, thus we adopt
the following filter:
\begin{equation*}
    H(z) = z^{-M}\widetilde{H}(z^M)
	\left[
	 \begin{array}{cccc}
	     1&
	     z&
	     \ldots&
	     z^{M-1}
	 \end{array}
       \right]^T.
\end{equation*}

\section{Design of sampling rate converters}
\label{sec:design_src}
By combining interpolators and decimators, we can construct
a sampling rate converter.
Figure \ref{fig:multirate_design_src} shows a sampling rate converter,
where an interpolation with factor $M_1$ is followed by a decimation with factor $M_2$. 
By this converter, the sampling rate of the input signal is changed by the factor $M_1/M_2$.
\begin{figure}[t]
    \begin{center}
\begin{picture}( 219,  34)
\put( 193.9,  16.9){\vector( 1, 0){  25.3}}
\put( 134.9,  16.9){\vector( 1, 0){  25.3}}
\put(  59.0,  16.9){\vector( 1, 0){  25.3}}
\put(   0.0,  16.9){\vector( 1, 0){  25.3}}
\put( 160.2,   0.0){\framebox(  33.7,  33.7){$\ds{M_2}$}}
\put(  84.3,   0.0){\framebox(  50.6,  33.7){$L(z)$}}
\put(  25.3,   0.0){\framebox(  33.7,  33.7){$\us{M_1}$}}
\end{picture}
    \end{center}
    \caption{Sampling rate converter ($M_1 : M_2$)}
    \label{fig:multirate_design_src}
\end{figure}
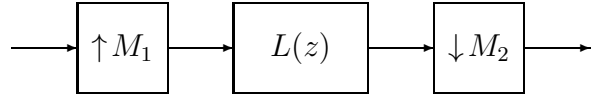
In the application of digital audio, the conversion from CD signals at 44.1kHz
to DAT signals at 48kHz is realized by a converter with the factors
$M_1 = 3\times 7^2$ and $M_2 = 2^5 \times 5$.
Conventionally, the digital filter
$L(z)$ is designed to be a low-pass filter with
the cut-off frequency $\omega = \pi/M,\quad M:=\max(M_1,M_2)$ \cite{Fli,Vai}.
In this section, we design the filter $L(z)$ that combines
the interpolation filter $H(z)$ designed by the method discussed in Section \ref{sec:design_interp}
and the decimation filter $H(z)$ in Section \ref{sec:design_decim}. 
The designed sampling rate converter will be in the form illustrated in Figure \ref{fig:multirate_design_src2}
\begin{figure}[t]
    \begin{center}
\begin{picture}( 320,  51)
\put(   8.4,  25.3){\vector( 1, 0){  33.7}}
\put(  42.2,   8.4){\framebox(  33.7,  33.7){$\us{M_1}$}}
\put(  75.9,  25.3){\vector( 1, 0){  33.7}}
\put( 109.6,   8.4){\framebox(  33.7,  33.7){$K(z)$}}
\put( 177.1,   8.4){\framebox(  33.7,  33.7){$H(z)$}}
\put( 244.5,   8.4){\framebox(  33.7,  33.7){$\ds{M_2}$}}
\put( 143.3,  25.3){\vector( 1, 0){  33.7}}
\put( 210.8,  25.3){\vector( 1, 0){  33.7}}
\put( 278.2,  25.3){\vector( 1, 0){  33.7}}
\end{picture}
    \end{center}
    \caption{Sampling rate converter with interpolator and decimator}
    \label{fig:multirate_design_src2}
\end{figure}
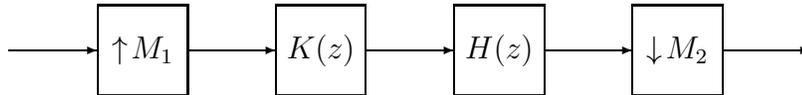
The block diagrams of the error system for sampled-data $H^\infty$ design of 
an interpolation filter $K(z)$ and a decimation filter $H(z)$ are
shown in Figure \ref{fig:multirate_design_src_interp} and in Figure \ref{fig:multirate_design_src_decim},
respectively.
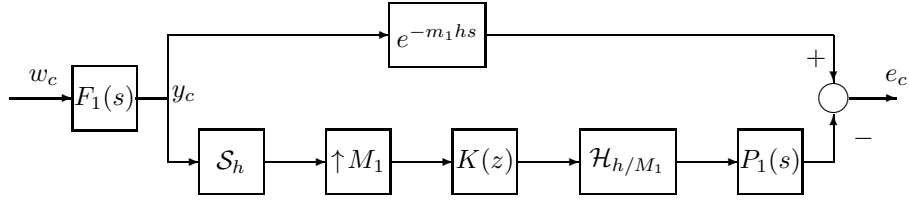
\begin{figure}[t]
    \begin{center}
\setlength{\unitlength}{0.5mm}
\begin{picture}( 240,  55)
\footnotesize
\put( 224.4,  13.1){$-$}
\put( 211.8,  34.2){$+$}
\put( 232.9,  30.0){$e_c$}
\put(   5.2,  30.0){$w_c$}
\put(  43.2,  25.8){$y_c$}
\put( 223.4,  25.3){\vector( 1, 0){  12.6}}
\put( 219.2,   8.4){\vector( 0, 1){  12.6}}
\put( 219.2,  42.2){\vector( 0,-1){  12.6}}
\put( 219.2,  25.3){\circle{ 8}}
\put( 126.5,  42.2){\line( 1, 0){  92.7}}
\put( 210.8,   8.4){\line( 1, 0){   8.4}}
\put( 177.1,   8.4){\vector( 1, 0){  16.9}}
\put( 134.9,   8.4){\vector( 1, 0){  16.9}}
\put( 101.2,   8.4){\vector( 1, 0){  16.9}}
\put(  67.5,   8.4){\vector( 1, 0){  16.9}}
\put(  42.2,   8.4){\vector( 1, 0){   8.4}}
\put(  42.2,  42.2){\vector( 1, 0){  59.0}}
\put(  33.7,  25.3){\line( 1, 0){   8.4}}
\put(  42.2,  42.2){\line( 0,-1){  33.7}}
\put(   0.0,  25.3){\vector( 1, 0){  16.9}}
\put( 101.2,  33.7){\framebox(  25.3,  16.9){$e^{-m_1hs}$}}
\put(  16.9,  16.9){\framebox(  16.9,  16.9){$F_1(s)$}}
\put(  50.6,   0.0){\framebox(  16.9,  16.9){$\samp{h}$}}
\put(  84.3,   0.0){\framebox(  16.9,  16.9){$\us{M_1}$}}
\put( 118.0,   0.0){\framebox(  16.9,  16.9){$K(z)$}}
\put( 151.8,   0.0){\framebox(  25.3,  16.9){$\hold{h/M_1}$}}
\put( 193.9,   0.0){\framebox(  16.9,  16.9){$P_1(s)$}}
\end{picture}
    \end{center}
    \caption{Error system for designing interpolator}
    \label{fig:multirate_design_src_interp}
\end{figure}
\begin{figure}[t]
    \begin{center}
\setlength{\unitlength}{0.5mm}
\begin{picture}( 245,  59)
\footnotesize    
\put(   0.0,  25.3){\vector( 1, 0){  16.9}}
\put(  42.2,  42.2){\line( 0,-1){  33.7}}
\put(  33.7,  25.3){\line( 1, 0){   8.4}}
\put(  42.2,  42.2){\vector( 1, 0){  59.0}}
\put( 224.4,  13.1){$-$}
\put( 211.8,  34.2){$+$}
\put(  43.2,  25.8){$y_c$}
\put(   5.2,  30.0){$w_c$}
\put( 223.4,  25.3){\vector( 1, 0){  12.6}}
\put( 219.2,   8.4){\vector( 0, 1){  12.6}}
\put( 219.2,  42.2){\vector( 0,-1){  12.6}}
\put( 219.2,  25.3){\circle{ 8}}
\put( 126.5,  42.2){\line( 1, 0){  92.7}}
\put( 210.8,   8.4){\line( 1, 0){   8.4}}
\put( 101.2,  33.7){\framebox(  25.3,  16.9){$e^{-m_2h_2s}$}}
\put(  16.9,  16.9){\framebox(  16.9,  16.9){$F_2(s)$}}
\put(  42.2,   8.4){\vector( 1, 0){   8.4}}
\put(  50.6,   0.0){\framebox(  25.3,  16.9){$\samp{h_2/{M_2}}$}}
\put(  92.7,   0.0){\framebox(  16.9,  16.9){$H(z)$}}
\put( 126.5,   0.0){\framebox(  16.9,  16.9){$\ds{M_2}$}}
\put( 160.2,   0.0){\framebox(  16.9,  16.9){$\hold{h_2}$}}
\put( 193.9,   0.0){\framebox(  16.9,  16.9){$P_2(s)$}}
\put(  75.9,   8.4){\vector( 1, 0){  16.9}}
\put( 109.6,   8.4){\vector( 1, 0){  16.9}}
\put( 143.3,   8.4){\vector( 1, 0){  16.9}}
\put( 177.1,   8.4){\vector( 1, 0){  16.9}}
\put( 228.7,  34.2){$e_c$}
    \end{picture}
    \end{center}
    \caption{Error system for designing decimator}
    \label{fig:multirate_design_src_decim}
\end{figure}
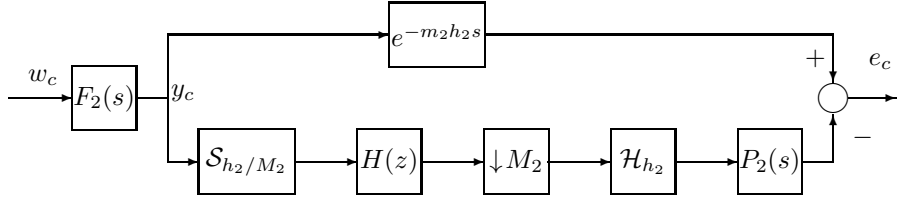
In these diagrams, $h_2 = h\cdot\frac{M_2}{M_1}$ and the analog low-pass filters
$F_1(s)$ and $F_2(s)$ have characteristics illustrated in Figure \ref{fig:multirate_design_src_filter}.
\begin{figure}[t]
    \begin{center}
        \includegraphics[width=0.7\linewidth]{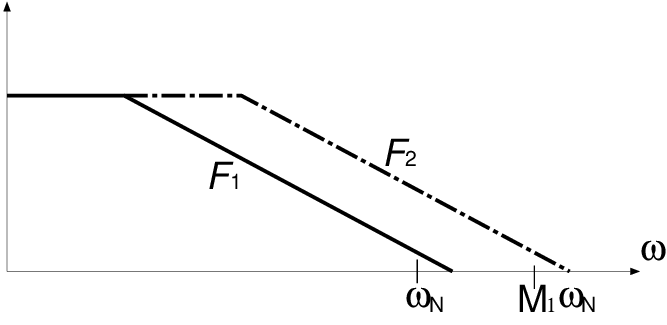}
    \end{center}
    \caption{Characteristic of $F_1$ and $F_2$}
    \label{fig:multirate_design_src_filter}
\end{figure}
The filters $F_1$ and $F_2$ take account of the characteristic
of the analog input signal, and hence we can design filters $K(z)$ and $H(z)$ 
that optimize the analog performance using the sampled-data system design method.

The advantage of the design mentioned above is that
we can design a converter with large $M_1$ or $M_2$.
For example, to design a converter for CD/DAT (i.e., $M_1=3\times 7^2$ and $M_2=2^5\times 5$), 
we first design interpolators with $M=3,7$ and decimators with $M=2,5$
as shown in Figure \ref{fig:multirate_design_src_ids},
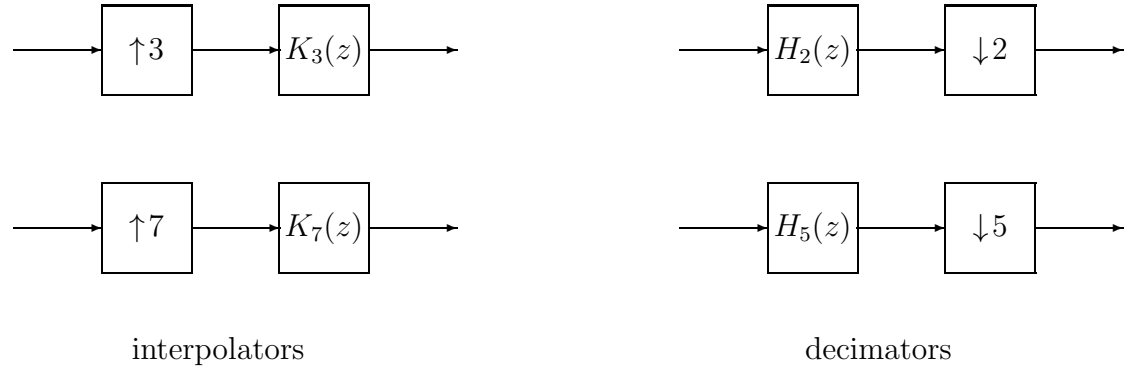
\begin{figure}[t]
    \begin{center}
\setlength{\unitlength}{0.7mm}
\begin{picture}( 228,  84)
\put( 158.8,   8.9){decimators}
\put(  30.9,   8.9){interpolators}
\put( 151.8,  59.0){\framebox(  16.9,  16.9){$H_2(z)$}}
\put( 151.8,  25.3){\framebox(  16.9,  16.9){$H_5(z)$}}
\put( 185.5,  59.0){\framebox(  16.9,  16.9){$\ds{2}$}}
\put( 185.5,  25.3){\framebox(  16.9,  16.9){$\ds{5}$}}
\put( 134.9,  67.5){\vector( 1, 0){  16.9}}
\put( 168.6,  67.5){\vector( 1, 0){  16.9}}
\put(  75.9,  33.7){\vector( 1, 0){  16.9}}
\put(  42.2,  33.7){\vector( 1, 0){  16.9}}
\put(   8.4,  33.7){\vector( 1, 0){  16.9}}
\put(  75.9,  67.5){\vector( 1, 0){  16.9}}
\put(  42.2,  67.5){\vector( 1, 0){  16.9}}
\put(   8.4,  67.5){\vector( 1, 0){  16.9}}
\put( 202.4,  67.5){\vector( 1, 0){  16.9}}
\put( 134.9,  33.7){\vector( 1, 0){  16.9}}
\put( 168.6,  33.7){\vector( 1, 0){  16.9}}
\put( 202.4,  33.7){\vector( 1, 0){  16.9}}
\put(  59.0,  25.3){\framebox(  16.9,  16.9){$K_7(z)$}}
\put(  25.3,  25.3){\framebox(  16.9,  16.9){$\us{7}$}}
\put(  59.0,  59.0){\framebox(  16.9,  16.9){$K_3(z)$}}
\put(  25.3,  59.0){\framebox(  16.9,  16.9){$\us{3}$}}
\end{picture}
    \end{center}
    \caption{Interpolators and decimators for CD/DAT sampling rate conversion}
    \label{fig:multirate_design_src_ids}
\end{figure}
and then by combining the interpolation filters $K_3(z)$, $K_7(z)$ and 
the decimation filters $H_2(z)$, $H_5(z)$,
we obtain the sampling rate converter as illustrated in Figure \ref{fig:multirate_design_src_ids_series}.
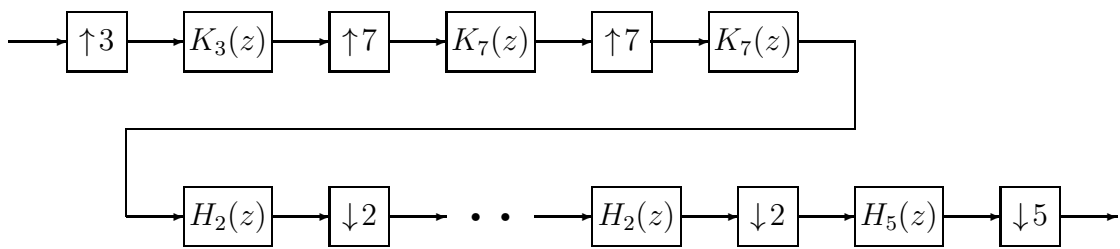
\begin{figure}[t]
    \begin{center}
\setlength{\unitlength}{0.46mm}
\begin{picture}( 337,  84)
\put( 151.8,  16.9){\circle*{ 2}}
\put( 143.3,  16.9){\circle*{ 2}}
\put( 252.9,  50.6){\line( 0,-1){   8.4}}
\put(  42.2,  42.2){\line( 0,-1){   8.4}}
\put( 118.0,  16.9){\vector( 1, 0){  16.9}}
\put( 295.1,   8.4){\framebox(  16.9,  16.9){$\ds{5}$}}
\put( 252.9,   8.4){\framebox(  25.3,  16.9){$H_5(z)$}}
\put( 160.2,  16.9){\vector( 1, 0){  16.9}}
\put( 219.2,   8.4){\framebox(  16.9,  16.9){$\ds{2}$}}
\put( 177.1,   8.4){\framebox(  25.3,  16.9){$H_2(z)$}}
\put( 101.2,   8.4){\framebox(  16.9,  16.9){$\ds{2}$}}
\put(  59.0,   8.4){\framebox(  25.3,  16.9){$H_2(z)$}}
\put( 278.2,  16.9){\vector( 1, 0){  16.9}}
\put( 312.0,  16.9){\vector( 1, 0){  16.9}}
\put(  84.3,  16.9){\vector( 1, 0){  16.9}}
\put( 252.9,  67.5){\line( 0,-1){  16.9}}
\put( 236.1,  67.5){\line( 1, 0){  16.9}}
\put( 193.9,  67.5){\vector( 1, 0){  16.9}}
\put( 160.2,  67.5){\vector( 1, 0){  16.9}}
\put( 118.0,  67.5){\vector( 1, 0){  16.9}}
\put(  84.3,  67.5){\vector( 1, 0){  16.9}}
\put(  42.2,  67.5){\vector( 1, 0){  16.9}}
\put(   8.4,  67.5){\vector( 1, 0){  16.9}}
\put( 202.4,  16.9){\vector( 1, 0){  16.9}}
\put( 236.1,  16.9){\vector( 1, 0){  16.9}}
\put(  42.2,  16.9){\vector( 1, 0){  16.9}}
\put(  42.2,  33.7){\line( 0,-1){  16.9}}
\put(  42.2,  42.2){\line( 1, 0){ 210.8}}
\put( 210.8,  59.0){\framebox(  25.3,  16.9){$K_7(z)$}}
\put( 177.1,  59.0){\framebox(  16.9,  16.9){$\us{7}$}}
\put( 134.9,  59.0){\framebox(  25.3,  16.9){$K_7(z)$}}
\put( 101.2,  59.0){\framebox(  16.9,  16.9){$\us{7}$}}
\put(  59.0,  59.0){\framebox(  25.3,  16.9){$K_3(z)$}}
\put(  25.3,  59.0){\framebox(  16.9,  16.9){$\us{3}$}}
\end{picture}
    \end{center}
    \caption{Sampling rate converter for CD/DAT}
    \label{fig:multirate_design_src_ids_series}
\end{figure}

On the other hand, there is a design method for sampling rate converters
by periodically time-varying systems \cite{IsiYamFra99}.
However, the order of the design in the method will be $M_1\times M_2$,
and hence the design of a converter with very large number of $M_1$ or $M_2$ 
such as CD/DAT ($M_1\times M_2=23520$) has much difficulty in numerical computation.

\section{Design examples}
\label{sec:multirate_examples}
\subsection{Design of interpolators}
\label{sec:multirate_examples_interp}
In this section, we compare the interpolator designed by the method discussed in Section \ref{sec:design_interp}
with that designed by Johnston's method \cite{Joh80}.
The parameters are as follows: interpolation ratio $M=2$, sampling period $h=1$ and
delay step $m=2$.
The analog filters $F(s)$ and $P(s)$ are
\begin{equation*}
    F(s) = \frac{1}{(Ts + 1)(0.1Ts + 1)},\quad T:=22.05/\pi\approx 7.0187,
	\quad P(s) = 1.
\end{equation*}
Note that the low-pass filter $F(s)$ has first order attenuation in the frequency range 
$\omega$ $\in$ $[0.14248, 1.4248]$ [rad/sec],
and second order attenuation in the range $\omega > 1.4248$ [rad/sec].
The filter simulates the frequency energy distribution
of a typical orchestral music.
The Johnston filter is taken to be of order $31$.

The frequency responses of the obtained filters are shown in Figure
\ref{fig:multirate_examples_interp_filter}.
The sampled-data design filter is of order $7$, 
which is lower than the Johnston filter.
\begin{figure}[t]
    \begin{center}
	\includegraphics[width=0.7\linewidth]{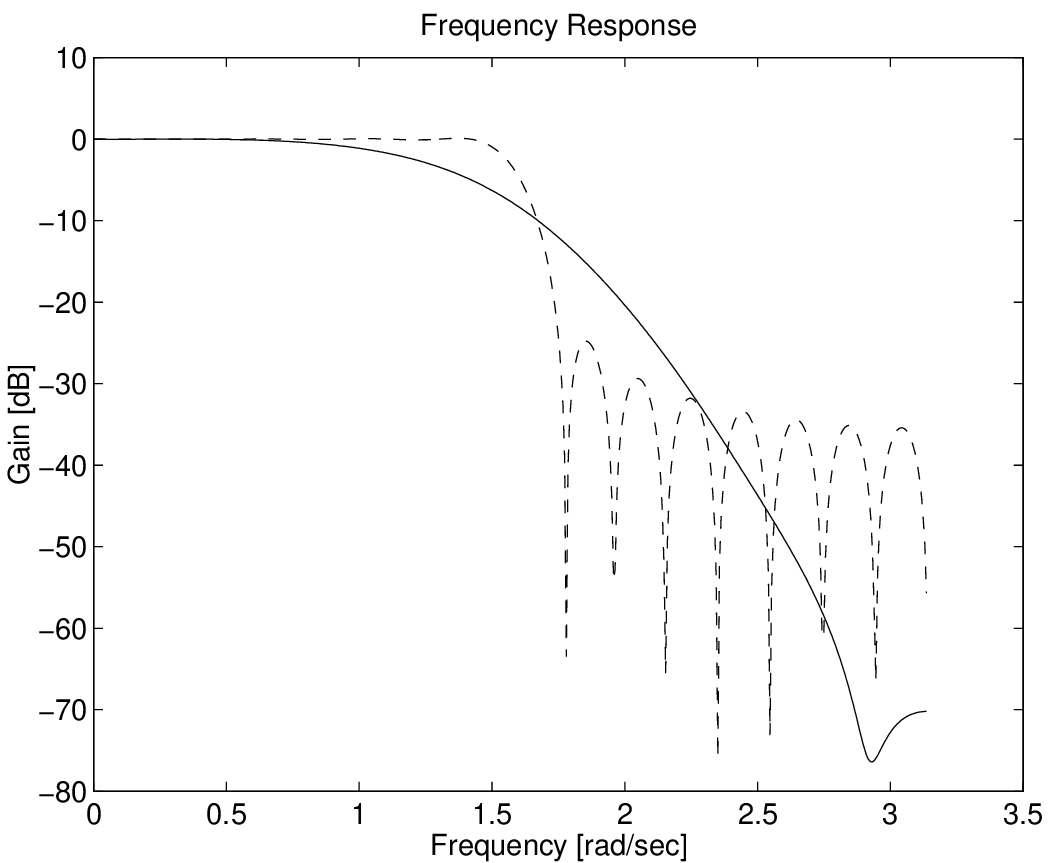}
    \end{center}
    \caption{Frequency responses of interpolation filters:
    sampled-data design (solid) and Johnston filter (dash)}
    \label{fig:multirate_examples_interp_filter}
\end{figure}
The Johnston filter shows the sharper decay beyond the 
cut-off frequency $\omega = \pi/2$,
while the 
filter obtained by the sampled-data design shows a rather slow decay.

On the other hand, the reconstruction error (see Figure \ref{fig:multirate_design_interp_error})
characteristic
in Figure \ref{fig:multirate_examples_interp_Tew} exhibits quite
an admirable performance in spite of the 
low-order of the sampled-data design filter.
It is almost comparable with 31st order Johnston filter.
\begin{figure}[t]
    \begin{center}
	\includegraphics[width=0.7\linewidth]{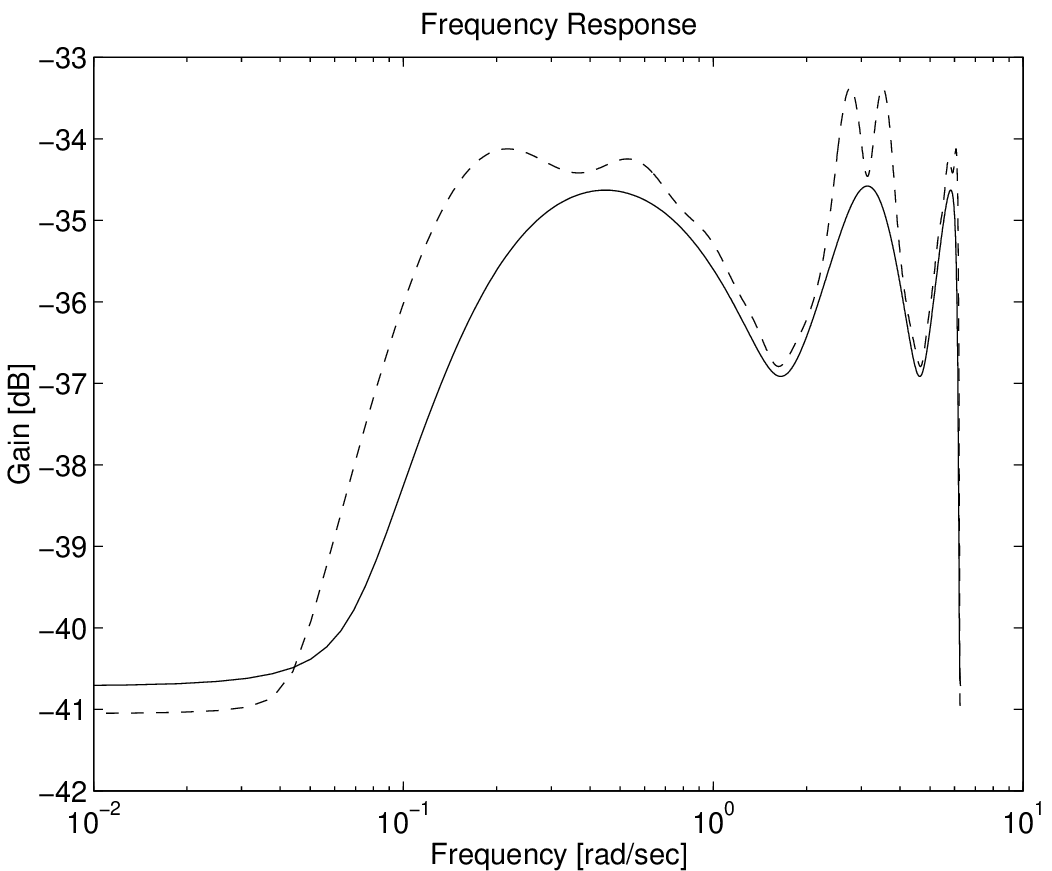}
    \end{center}
    \caption{Frequency responses of error system: sampled-data design (solid) and
     Johnston filter design (dash)}
    \label{fig:multirate_examples_interp_Tew}
\end{figure}

While for those frequencies much below the cut-off frequency the gain characteristic of 
the sampled-data design is not as good as the Johnston filter, 
the sampled-data designed filter need not be inferior.
To see this, let us see the time responses against rectangular waves
in Figure \ref{fig:multirate_examples_interp_t_y} (sampled-data designed)
\begin{figure}[t]
    \begin{center}
	\includegraphics[width=0.7\linewidth]{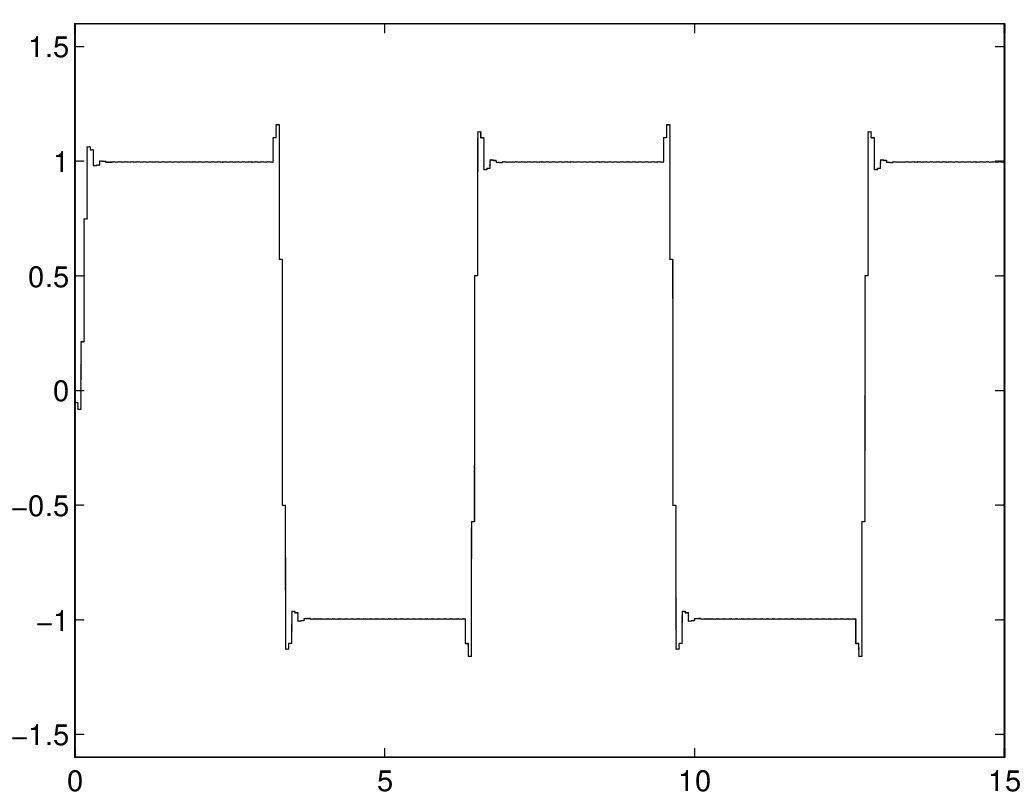}
    \end{center}
    \caption{Time response (sampled-data design)}
    \label{fig:multirate_examples_interp_t_y}
\end{figure}
and in Figure \ref{fig:multirate_examples_interp_t_j} (Johnston filter).
\begin{figure}[t]
    \begin{center}
	\includegraphics[width=0.7\linewidth]{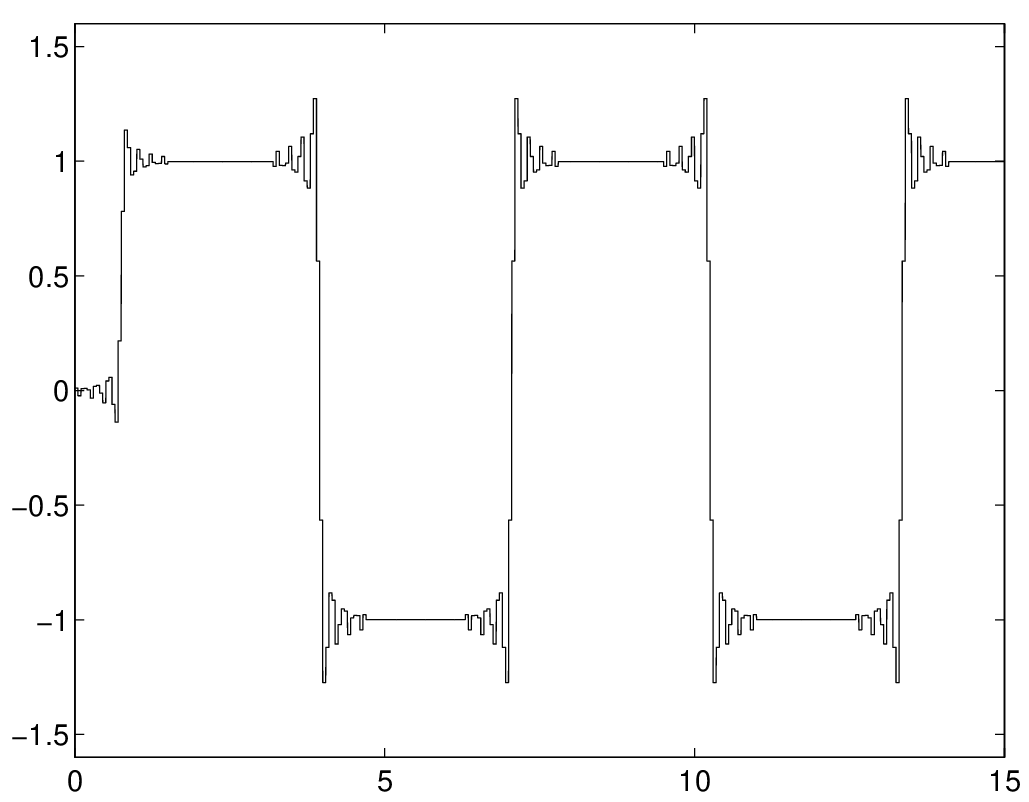}
    \end{center}
    \caption{Time response (Johnston filter)}
    \label{fig:multirate_examples_interp_t_j}
\end{figure}

The Johnston filter exhibits a very typical Gibbs phenomenon 
(i.e., we can see ringing caused by the sharp characteristic of the filter),
whereas the one by
the sampled-data design shows a much smaller peak near the edge.
We also note that the sampled-data designed filter has nearly linear phase as shown in 
Figure \ref{fig:multirate_example_interp_phase}.
\begin{figure}[t]
    \begin{center}
	\includegraphics[width=0.7\linewidth]{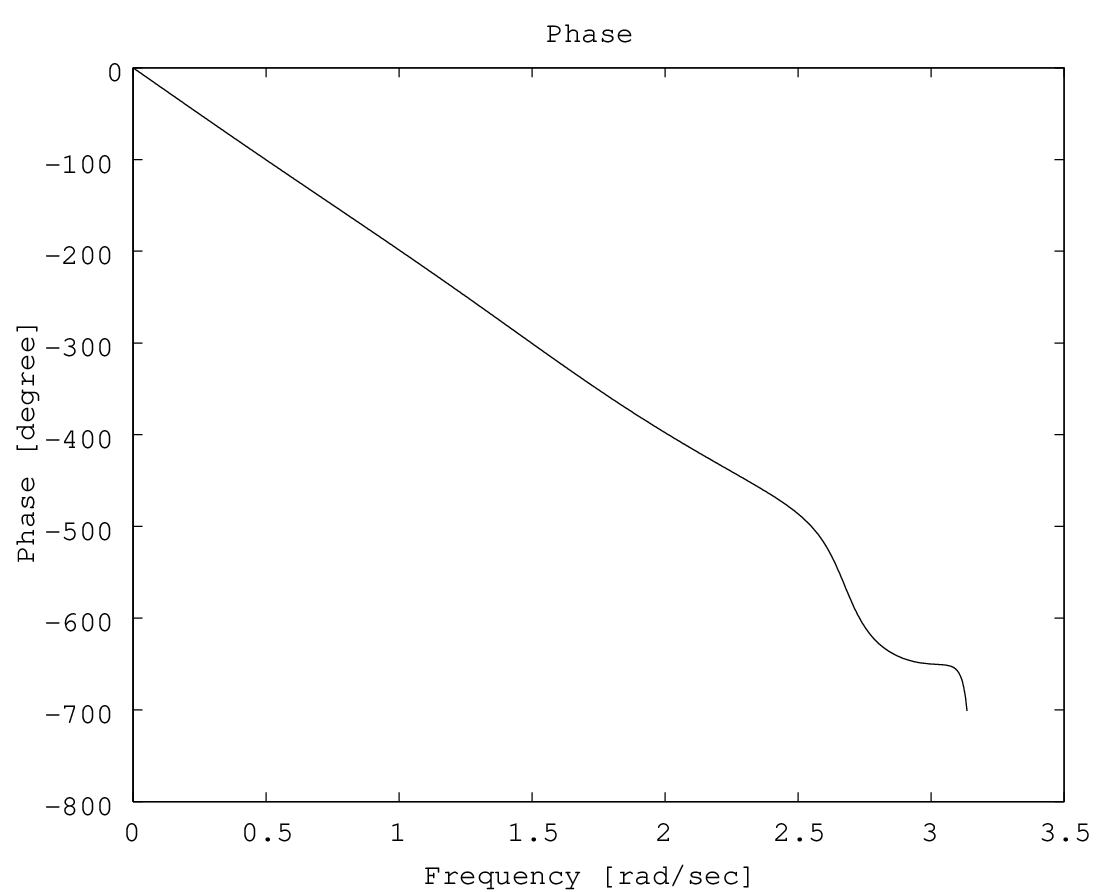}
    \end{center}
    \caption{Phase response (Sampled-data design)}
    \label{fig:multirate_example_interp_phase}
\end{figure}

\subsection{Design of decimators}
We now present an example of the $H^\infty$ design of decimators discussed in Section \ref{sec:design_decim}.
For comparison, we take the Johnston filter of order 31.

Let the decimation ratio $M=2$ and the other parameters is the same
as the interpolator design in the previous section.

Figure \ref{fig:multirate_examples_decim_filter} shows the frequency response
of the decimation filters.
Note that the filter designed by the sampled-data method is of order $6$.
\begin{figure}[t]
    \begin{center}
	\includegraphics[width=0.7\linewidth]{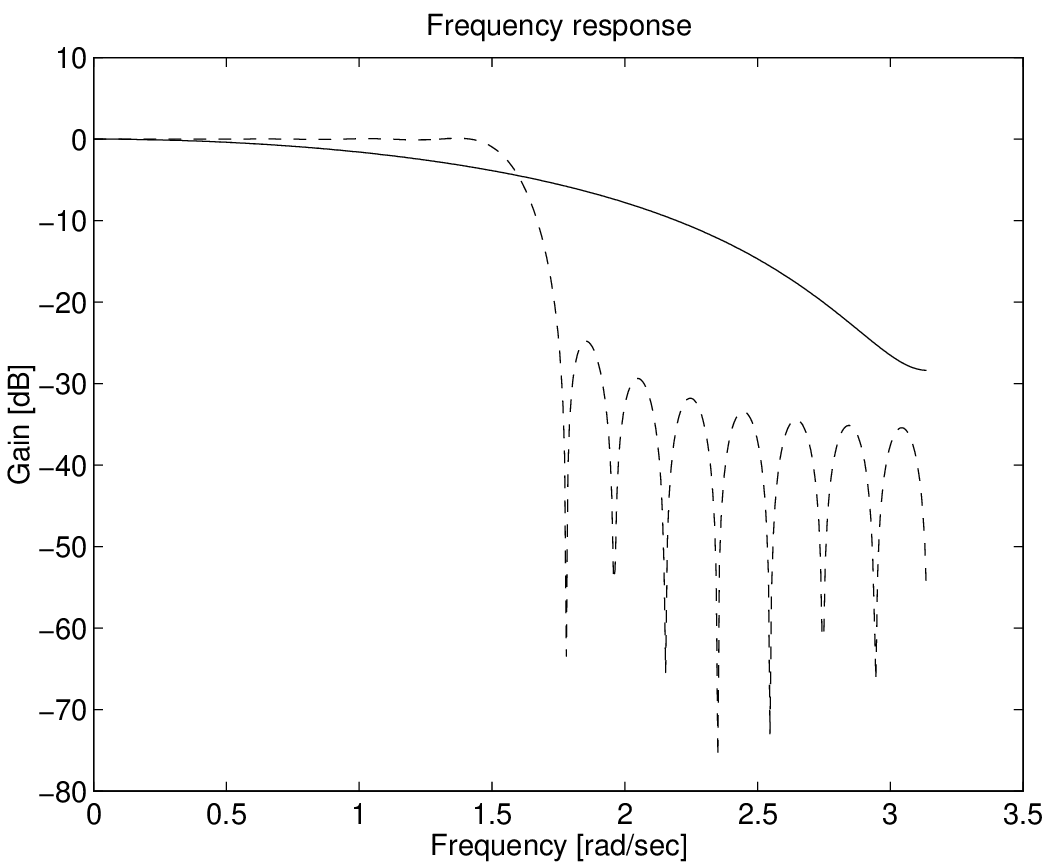}
    \end{center}
    \caption{Frequency responses of decimation filters: sampled-data design (solid) and Johnston filter (dash)}
    \label{fig:multirate_examples_decim_filter}
\end{figure}
The Johnston filter shows the sharper decay beyond
the cut-off frequency $\omega = \pi/2$,
while the filter by sampled-data design shows a rather slow decay.

Figure \ref{fig:multirate_examples_decim_Tew} shows the frequency response of the error system
(see Figure \ref{fig:multirate_design_decim_error}).
\begin{figure}[t]
    \begin{center}
	\includegraphics[width=0.7\linewidth]{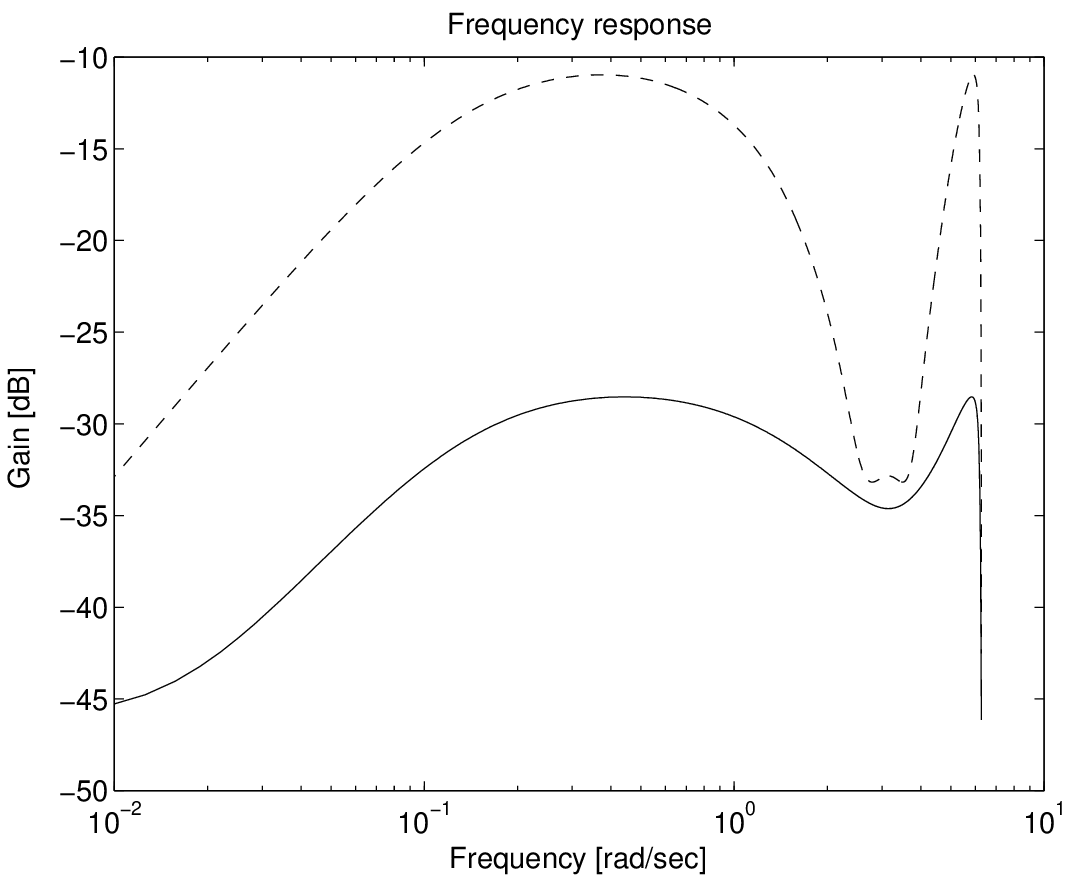}
    \end{center}
    \caption{Frequency responses of error system: sampled-data design (solid) and Johnston filter design (dash)}
    \label{fig:multirate_examples_decim_Tew}
\end{figure}
We can see from the frequency response of the error system that 
the decimator designed by the sampled-data method exhibits a clear advantage over all frequency range,
even though the sampled-data designed filter is of lower order than the Johnston filter.
In particular, around the frequency $\omega \approx$ 4 [rad/sec], the difference is about 20 dB.

Figure \ref{fig:multirate_examples_decim_t_y} and 
Figure \ref{fig:multirate_examples_decim_t_j} shows
the time responses against rectangular waves.
\begin{figure}[t]
    \begin{center}
	\includegraphics[width=0.7\linewidth]{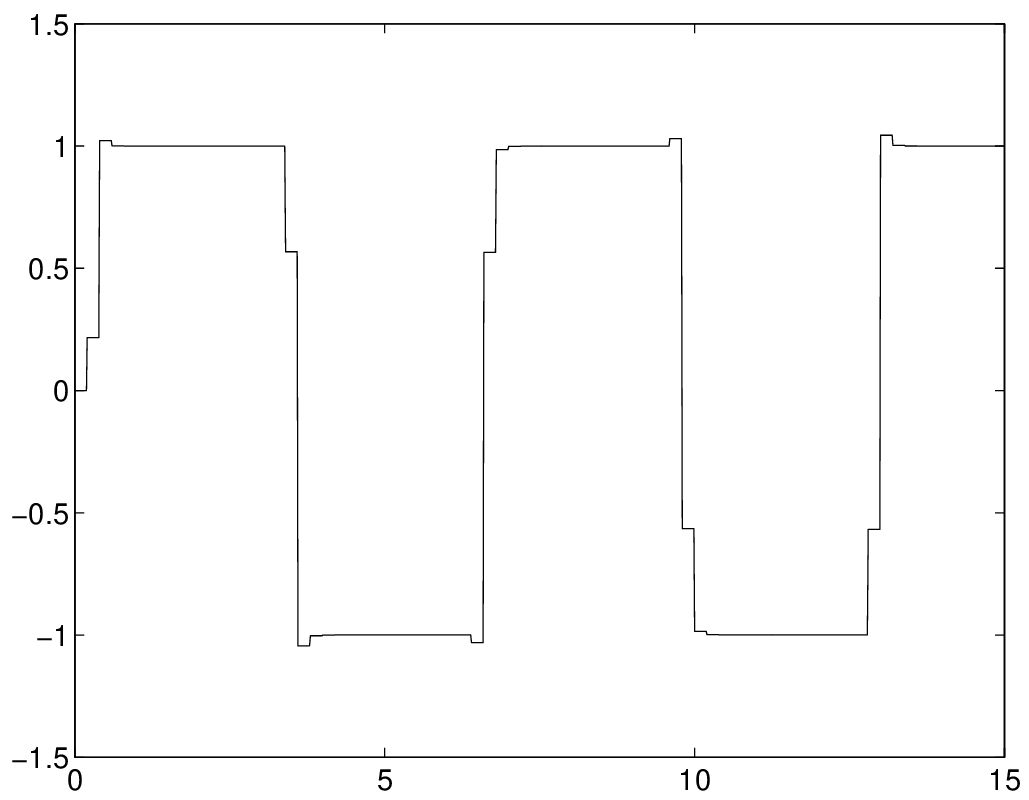}
    \end{center}
    \caption{Time response (sampled-data design)}
    \label{fig:multirate_examples_decim_t_y}
\end{figure}
\begin{figure}[t]
    \begin{center}
	\includegraphics[width=0.7\linewidth]{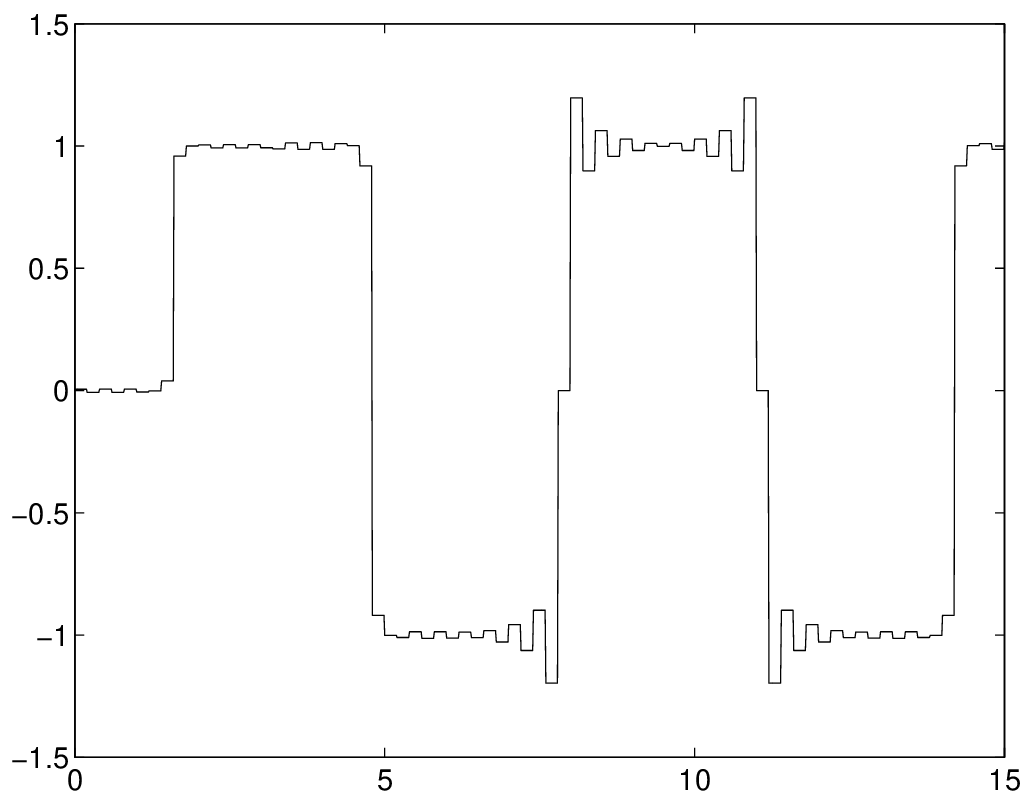}
    \end{center}
    \caption{Time response (Johnston filter)}
    \label{fig:multirate_examples_decim_t_j}
\end{figure}
The sampled-data designed decimator reconstructs the rectangular wave well,
while the decimator with the Johnston filter exhibits
a large amount of ringing due to the Gibbs phenomenon.

\subsection{Design of sampling rate converters}
In this section, we present a design example for the case of changing the sampling period
from $h_1=1$ to $h_2=4/3$. Then we have the sampling rate converter with 
$M_1=3$ and $M_2=4$ that are coprime (see Figure \ref{fig:multirate_design_src}).
Let the filter for the interpolator design 
(see Figure \ref{fig:multirate_design_src_interp}) be
\begin{gather*}
    F_1(s) = \frac{1}{(Ts+1)(0.1Ts+1)},\quad P_1(s) = 1,
\end{gather*}
and that for the decimator design 
(see Figure \ref{fig:multirate_design_src_decim}) be
\begin{gather*}
    F_2(s) = \frac{1}{(T_2s+1)(0.1T_2s+1)},\quad P_2(s) = 1,
\end{gather*}
where $T:=22.05/\pi,\quad T_2:=T/M_1$.
The filters simulate 
the frequency energy distribution of a typical orchestral music,
which are observed by FFT analysis of analog records
of some orchestral musics.

An approximate design is executed here for $N=M_1\times
4=12$ (interpolator) and $N=M_2\times 4=16$ (decimator).
For comparison, we compare it with the equiripple
filter obtained by Parks-McClellan method \cite{Vai, Zel} of order 31.
Parks-McClellan method is widely used for designing FIR filters \cite{Zel}.
The delay stemps $m_1$ and $m_2$ are 2.

The obtained (sub) optimal interpolation filter $K(z)$ is
of order 11 and the decimation filter $H(z)$ of order 15.
The sampling rate conversion filter $L(z)=H(z)K(z)$ 
is of order 22\footnote{The order of $L(z)$ is reduced by the minimal realization method.}.

Figure~\ref{fig:multirate_examples_src_filter} shows the gain characteristics of
these filters.  The equiripple filter shows the sharper decay
beyond the cut-off frequency ($\pi/4$ [rad/sec]) while the
sampled-data design shows a rather mild cut-off characteristic.

\begin{figure}[t]
    \begin{center}
	\includegraphics[width=0.7\linewidth]{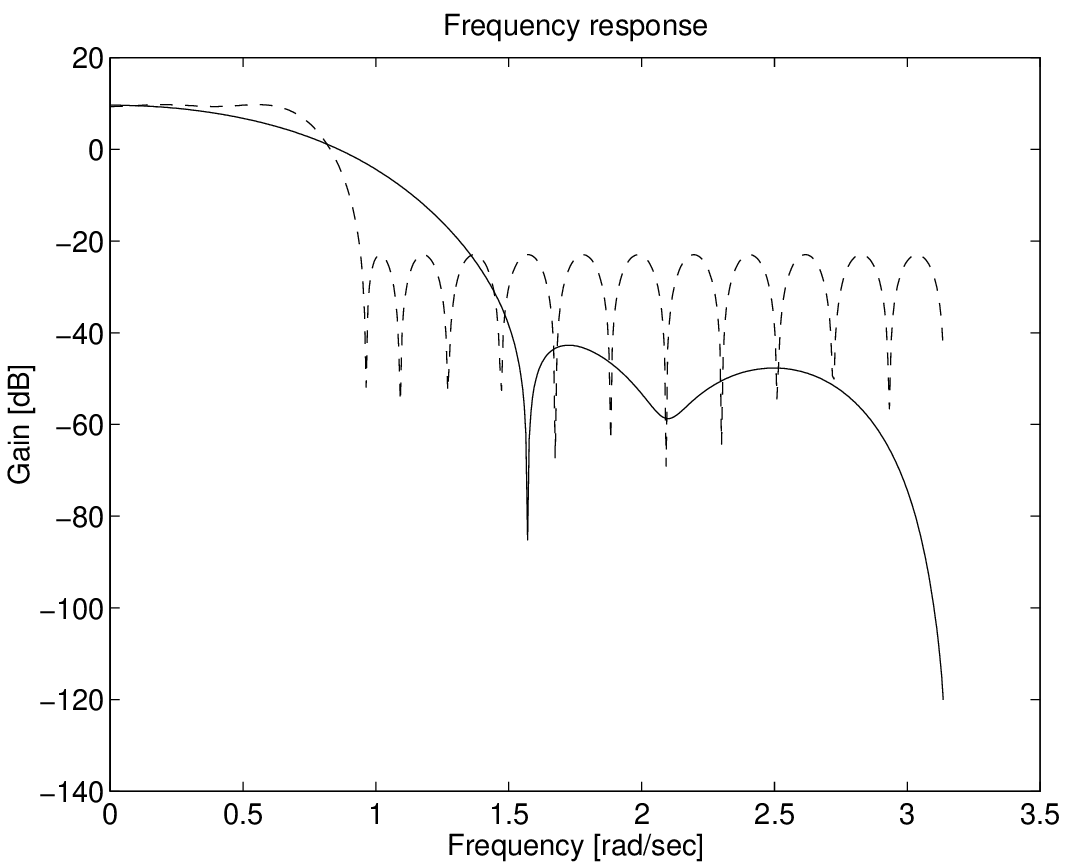}
    \end{center}
    \caption{Frequency responses of sampling rate conversion filters:
    sampled-data design $L(z)$ (solid) and quiripple filter $L_e(z)$ (dash)}
    \label{fig:multirate_examples_src_filter}
\end{figure}

In spite of these superficial differences, 
the frequency response of the error system 
(illustrated in Figure~\ref{fig:multirate_example_src_error}, where $m:=m_1+m_2$ and $P(s)=P_1(s)=P_2(s)=1$)
of sampling rate converter 
exhibits quite an admirable performance of the sampled-data design 
as shown in Figure~\ref{fig:multirate_examples_src_Tew}.  
\begin{figure}[t]
    \begin{center}
%	\input{multirate_example_src_error.tex}
%WinTpicVersion2.15
\unitlength 0.1in
\begin{picture}(44.00,16.50)(4.00,-18.00)
% VECTOR 2 0 3 0
% 2 400 800 800 800
% 
\special{pn 8}%
\special{pa 400 400}%
\special{pa 800 400}%
\special{fp}%
\special{sh 1}%
\special{pa 800 400}%
\special{pa 733 380}%
\special{pa 747 400}%
\special{pa 733 420}%
\special{pa 800 400}%
\special{fp}%
% BOX 2 0 3 0
% 2 800 600 1200 1000
% 
\special{pn 8}%
\special{pa 800 200}%
\special{pa 1200 200}%
\special{pa 1200 600}%
\special{pa 800 600}%
\special{pa 800 200}%
\special{fp}%
% VECTOR 2 0 3 0
% 2 1800 800 2800 800
% 
\special{pn 8}%
\special{pa 1800 400}%
\special{pa 2800 400}%
\special{fp}%
\special{sh 1}%
\special{pa 2800 400}%
\special{pa 2733 380}%
\special{pa 2747 400}%
\special{pa 2733 420}%
\special{pa 2800 400}%
\special{fp}%
% BOX 2 0 3 0
% 2 2800 600 3400 1000
% 
\special{pn 8}%
\special{pa 2800 200}%
\special{pa 3400 200}%
\special{pa 3400 600}%
\special{pa 2800 600}%
\special{pa 2800 200}%
\special{fp}%
% VECTOR 2 0 3 0
% 2 3400 800 3800 800
% 
\special{pn 8}%
\special{pa 3400 400}%
\special{pa 3800 400}%
\special{fp}%
\special{sh 1}%
\special{pa 3800 400}%
\special{pa 3733 380}%
\special{pa 3747 400}%
\special{pa 3733 420}%
\special{pa 3800 400}%
\special{fp}%
% BOX 2 0 3 0
% 2 3800 600 4200 1000
% 
\special{pn 8}%
\special{pa 3800 200}%
\special{pa 4200 200}%
\special{pa 4200 600}%
\special{pa 3800 600}%
\special{pa 3800 200}%
\special{fp}%
% LINE 2 0 3 0
% 2 4200 800 4600 800
% 
\special{pn 8}%
\special{pa 4200 400}%
\special{pa 4600 400}%
\special{fp}%
% VECTOR 2 0 3 0
% 2 4600 800 4600 1200
% 
\special{pn 8}%
\special{pa 4600 400}%
\special{pa 4600 800}%
\special{fp}%
\special{sh 1}%
\special{pa 4600 800}%
\special{pa 4620 733}%
\special{pa 4600 747}%
\special{pa 4580 733}%
\special{pa 4600 800}%
\special{fp}%
% BOX 2 0 3 0
% 2 4400 1200 4800 1600
% 
\special{pn 8}%
\special{pa 4400 800}%
\special{pa 4800 800}%
\special{pa 4800 1200}%
\special{pa 4400 1200}%
\special{pa 4400 800}%
\special{fp}%
% LINE 2 0 3 0
% 2 4600 1600 4600 2000
% 
\special{pn 8}%
\special{pa 4600 1200}%
\special{pa 4600 1600}%
\special{fp}%
% VECTOR 2 0 3 0
% 2 4600 2000 4200 2000
% 
\special{pn 8}%
\special{pa 4600 1600}%
\special{pa 4200 1600}%
\special{fp}%
\special{sh 1}%
\special{pa 4200 1600}%
\special{pa 4267 1620}%
\special{pa 4253 1600}%
\special{pa 4267 1580}%
\special{pa 4200 1600}%
\special{fp}%
% BOX 2 0 3 0
% 2 4200 1800 3800 2200
% 
\special{pn 8}%
\special{pa 4200 1400}%
\special{pa 3800 1400}%
\special{pa 3800 1800}%
\special{pa 4200 1800}%
\special{pa 4200 1400}%
\special{fp}%
% VECTOR 2 0 3 0
% 2 3800 2000 3400 2000
% 
\special{pn 8}%
\special{pa 3800 1600}%
\special{pa 3400 1600}%
\special{fp}%
\special{sh 1}%
\special{pa 3400 1600}%
\special{pa 3467 1620}%
\special{pa 3453 1600}%
\special{pa 3467 1580}%
\special{pa 3400 1600}%
\special{fp}%
% BOX 2 0 3 0
% 2 3400 1800 2800 2200
% 
\special{pn 8}%
\special{pa 3400 1400}%
\special{pa 2800 1400}%
\special{pa 2800 1800}%
\special{pa 3400 1800}%
\special{pa 3400 1400}%
\special{fp}%
% VECTOR 2 0 3 0
% 2 2800 2000 2400 2000
% 
\special{pn 8}%
\special{pa 2800 1600}%
\special{pa 2400 1600}%
\special{fp}%
\special{sh 1}%
\special{pa 2400 1600}%
\special{pa 2467 1620}%
\special{pa 2453 1600}%
\special{pa 2467 1580}%
\special{pa 2400 1600}%
\special{fp}%
% BOX 2 0 3 0
% 2 2400 1800 2000 2200
% 
\special{pn 8}%
\special{pa 2400 1400}%
\special{pa 2000 1400}%
\special{pa 2000 1800}%
\special{pa 2400 1800}%
\special{pa 2400 1400}%
\special{fp}%
% VECTOR 2 0 3 0
% 2 2000 2000 1600 2000
% 
\special{pn 8}%
\special{pa 2000 1600}%
\special{pa 1600 1600}%
\special{fp}%
\special{sh 1}%
\special{pa 1600 1600}%
\special{pa 1667 1620}%
\special{pa 1653 1600}%
\special{pa 1667 1580}%
\special{pa 1600 1600}%
\special{fp}%
% LINE 2 0 3 0
% 2 1200 800 1800 800
% 
\special{pn 8}%
\special{pa 1200 400}%
\special{pa 1800 400}%
\special{fp}%
% VECTOR 2 0 3 0
% 2 1500 800 1500 1200
% 
\special{pn 8}%
\special{pa 1500 400}%
\special{pa 1500 800}%
\special{fp}%
\special{sh 1}%
\special{pa 1500 800}%
\special{pa 1520 733}%
\special{pa 1500 747}%
\special{pa 1480 733}%
\special{pa 1500 800}%
\special{fp}%
% BOX 2 0 3 0
% 2 1200 1200 1800 1600
% 
\special{pn 8}%
\special{pa 1200 800}%
\special{pa 1800 800}%
\special{pa 1800 1200}%
\special{pa 1200 1200}%
\special{pa 1200 800}%
\special{fp}%
% VECTOR 2 0 3 0
% 2 1500 1600 1500 1900
% 
\special{pn 8}%
\special{pa 1500 1200}%
\special{pa 1500 1500}%
\special{fp}%
\special{sh 1}%
\special{pa 1500 1500}%
\special{pa 1520 1433}%
\special{pa 1500 1447}%
\special{pa 1480 1433}%
\special{pa 1500 1500}%
\special{fp}%
% CIRCLE 2 0 3 0
% 4 1500 2000 1400 2000 1400 2000 1400 2000
% 
\special{pn 8}%
\special{ar 1500 1600 100 100  0.0000000 6.2831853}%
% VECTOR 2 0 3 0
% 2 1400 2000 400 2000
% 
\special{pn 8}%
\special{pa 1400 1600}%
\special{pa 400 1600}%
\special{fp}%
\special{sh 1}%
\special{pa 400 1600}%
\special{pa 467 1620}%
\special{pa 453 1600}%
\special{pa 467 1580}%
\special{pa 400 1600}%
\special{fp}%
% STR 2 0 3 0
% 3 1000 700 1000 800 5 0
% $F_1(s)$
\put(10.0000,-4.0000){\makebox(0,0){$F_1(s)$}}%
% STR 2 0 3 0
% 3 3100 700 3100 800 5 0
% $\samp{h/M_2}$
\put(31.0000,-4.0000){\makebox(0,0){$\samp{h/M_2}$}}%
% STR 2 0 3 0
% 3 3100 1900 3100 2000 5 0
% $\hold{h/M_1}$
\put(31.0000,-16.0000){\makebox(0,0){$\hold{h/M_1}$}}%
% STR 2 0 3 0
% 3 1500 1300 1500 1400 5 0
% $e^{-mhs}$
\put(15.0000,-10.0000){\makebox(0,0){$e^{-mhs}$}}%
% STR 2 0 3 0
% 3 4000 700 4000 800 5 0
% $\us{M_1}$
\put(40.0000,-4.0000){\makebox(0,0){$\us{M_1}$}}%
% STR 2 0 3 0
% 3 4600 1300 4600 1400 5 0
% $L(z)$
\put(46.0000,-10.0000){\makebox(0,0){$L(z)$}}%
% STR 2 0 3 0
% 3 4000 1900 4000 2000 5 0
% $\ds{M_2}$
\put(40.0000,-16.0000){\makebox(0,0){$\ds{M_2}$}}%
% STR 2 0 3 0
% 3 2200 1900 2200 2000 5 0
% $P(s)$
\put(22.0000,-16.0000){\makebox(0,0){$P(s)$}}%
% STR 2 0 3 0
% 3 1570 1790 1570 1890 2 0
% $+$
\put(15.7000,-14.9000){\makebox(0,0)[lb]{$+$}}%
% STR 2 0 3 0
% 3 1680 1980 1680 2080 1 0
% $-$
\put(16.8000,-16.8000){\makebox(0,0)[lt]{$-$}}%
% STR 2 0 3 0
% 3 420 620 420 720 2 0
% $w_c$
\put(4.2000,-3.2000){\makebox(0,0)[lb]{$w_c$}}%
% STR 2 0 3 0
% 3 420 1820 420 1920 2 0
% $e_c$
\put(4.2000,-15.2000){\makebox(0,0)[lb]{$e_c$}}%
\end{picture}%
    \end{center}
    \caption{Error system of sampling rate converter}
    \label{fig:multirate_example_src_error}
\end{figure}
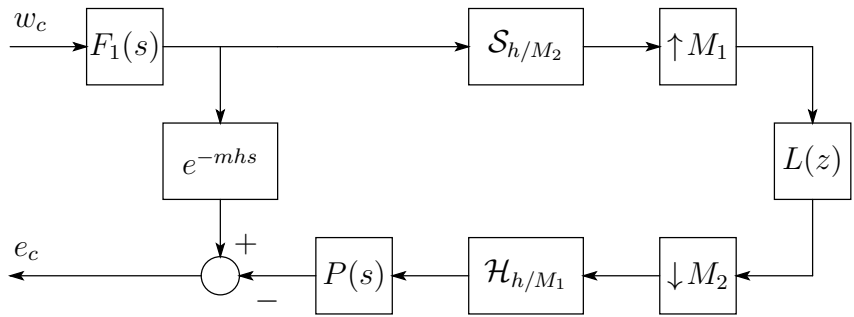
\begin{figure}[t]
    \begin{center}
	\includegraphics[width=0.7\linewidth]{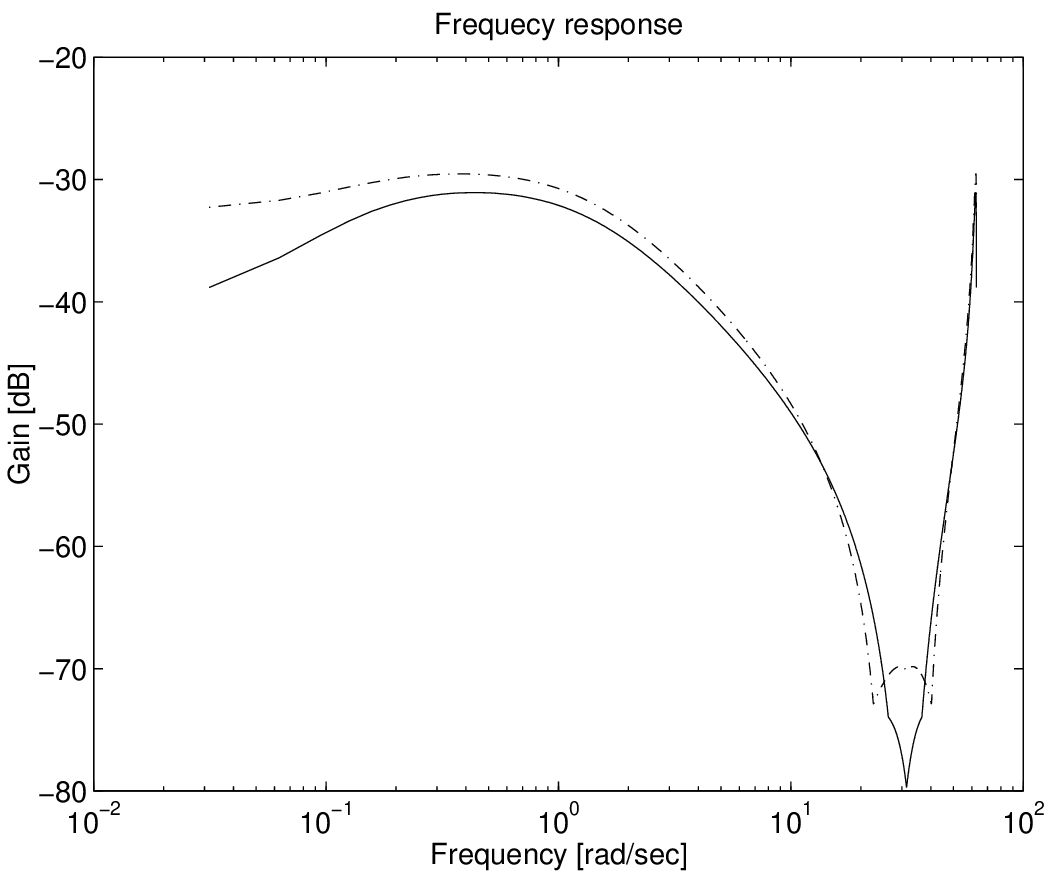}
    \end{center}
    \caption{Frequency responses of error system: sampled-data design (solid) and  equiripple design (dash)}
    \label{fig:multirate_examples_src_Tew}
\end{figure}

It is interesting to observe that the slow decay need not yield
an inferior design.  
%In fact, due to the underlying analog
%model (i.e., $F(s)$), there is an important information content
%beyond the Nyquist frequency, and such a slow decay is necessary to
%retain such information.  
To see this, let us see the time responses against a rectangular wave in 
Figure \ref{fig:multirate_examples_src_t_y} and Figure \ref{fig:multirate_examples_src_t_e}.
\begin{figure}[t]
    \begin{center}
	\includegraphics[width=0.7\linewidth]{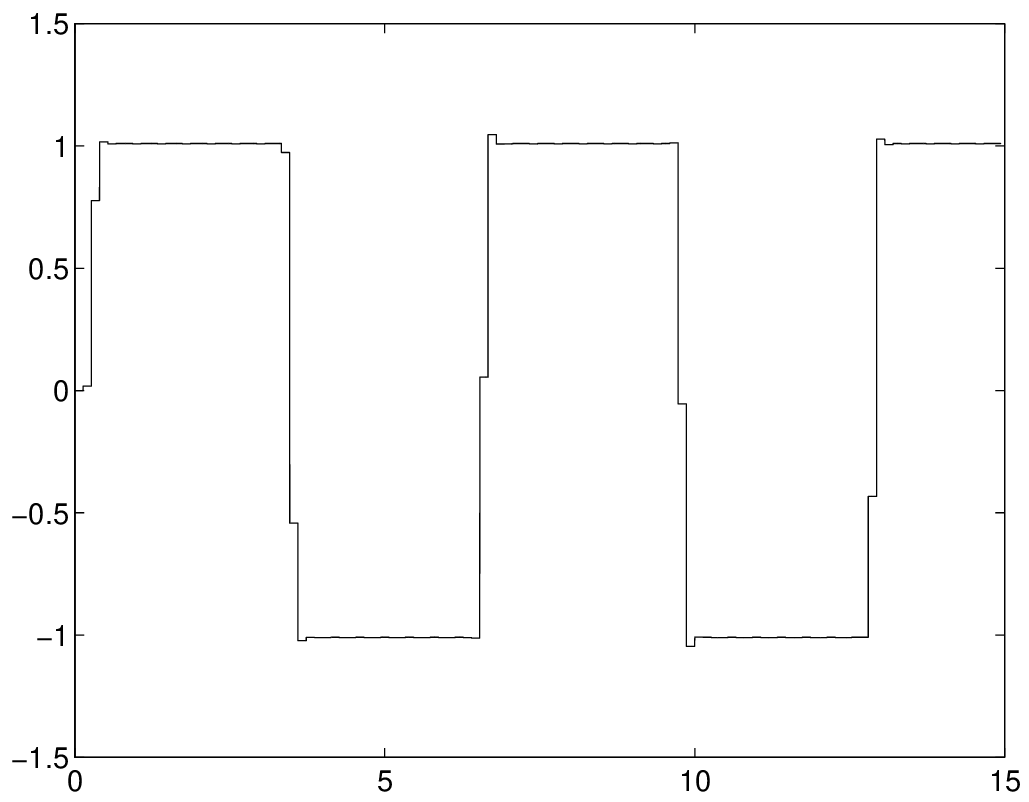}
    \end{center}
    \caption{Time response (sampled-data design)}
    \label{fig:multirate_examples_src_t_y}
\end{figure}
\begin{figure}[t]
    \begin{center}
	\includegraphics[width=0.7\linewidth]{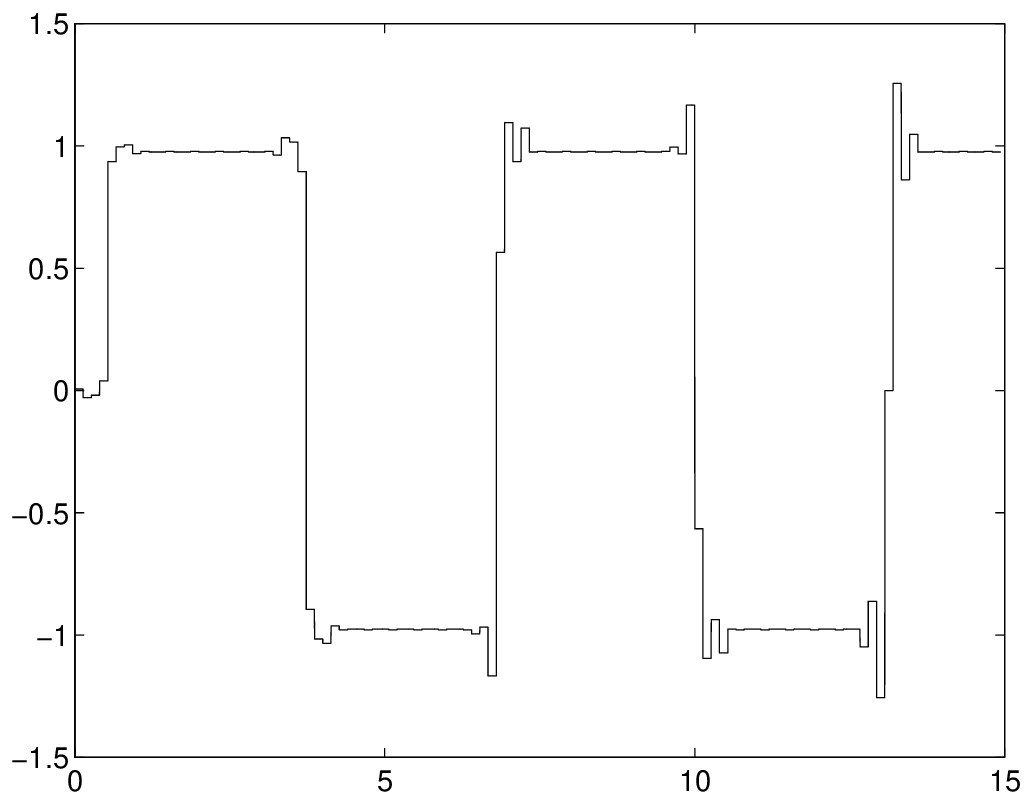}
    \end{center}
    \caption{Time response (equiripple design)}
    \label{fig:multirate_examples_src_t_e}
\end{figure}
The response with the equiripple filter shows a large amount of ringing, whereas
that with the filter by the sampled-data design has much less peak around the edge.  
Note also that $L(z)$ is nearly linear phase up to a certain frequency
as shown in Figure~\ref{fig:multirate_examples_src_phase_sd}.
\begin{figure}[t]
    \begin{center}
	\includegraphics[width=0.7\linewidth]{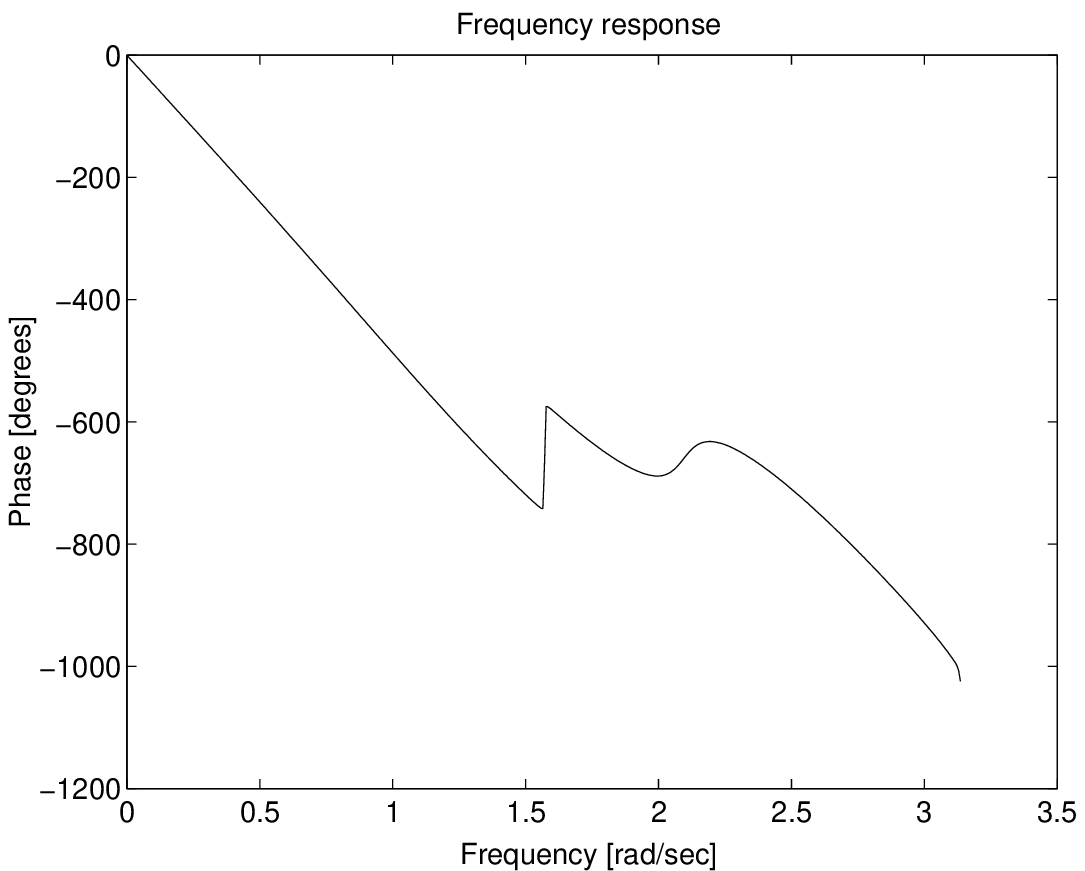}
    \end{center}
    \caption{Phase plot of $L(z)$}
    \label{fig:multirate_examples_src_phase_sd}
\end{figure}

\section{Conclusion}
Conventional theories of digital signal processing assert
that the ideal filter is the best for interpolation or decimation.
However, as we have shown above,
a sharp filter characteristic 
approximating the ideal filter does not necessarily behave well.
In particular, such a filter often exhibits a large amount of ringing
as illustrated in the previous section.
The ringing is due to the Gibbs phenomenon,
which is caused by the sharp characteristic of the filter.

On the other hand, our filter shows a slow decay.
The reason is that 
due to the underlying analog
characteristic (i.e., $F(s)$), there is an important information content
beyond the Nyquist frequency, and such a slow decay is necessary to
recover such information.  
%In particular, we can see this property in the decimator design.

Moreover, conventional design requires us to give a filter order in advance.
The higher the order is, the closer to the ideal characteristic the filter is,
and hence filters of a very high order are often used.

In contrast to the conventional design, our method is free from 
the choice of filter order.
Namely, the order of designed filter depends on the order of
filter $F(s)$ and $P(s)$, delay step $m$ and upsampling (or downsampling) ratio $M$.
As indicated in the previous section, in spite of the fact that the order of 
the obtained filter is not very high,
the response is better than high-order conventional filters.
This fact cannot be recognized without considering the analog performance.
\chapter{Application to Communication Systems}
\label{ch:comm}
\section{Introduction}
The importance of digital communication is ever increasing owing to
the rapid growth of the Internet, cellular phones, and so on \cite{Pro}.
In digital communication, especially in pulse amplitude modulation (PAM) or 
in pulse code modulation (PCM),
the analog signal to be transmitted is sampled and becomes a discrete-time signal.
In the conventional method, the analog characteristics of the signals are not considered,
and hence the total system is regarded as a discrete-time system.
Namely, one usually assumes that the original analog signal is band-limited
up to the Nyquist frequency. 

%This chapter proposes a new design methodology based on the sampled-data
%control theory \cite{CheFra} that takes account of intersample behavior or
%frequency components beyond the Nyquist frequency.
%Recently, sampled-data control theory is applied to
%some digital signal processing systems \cite{CheFra95,KhaYam96,YamFujKha97,NagYam00}.
%We propose a design for digital
%communication systems via sampled-data control.

In \cite{ErdHasKai2000}, a discrete-time $H^\infty$ design of 
receiving filters or equalizers is introduced.
This design is based on the assumption of full band-limitation,
but in reality no signals are fully band-limited.
Moreover, it is difficult to attenuate both the signal reconstruction error 
and the distortion caused by the channel only by equalizing after receiving.
Therefore an enhancer (or transmitting filter) that amplifies the signal before transmission
is often attached in order to increase signal-to-noise ratio \cite{Pro}.

In this chapter,  we propose a new design of receiving/transmitting filters
by using the sampled-data control theory.
Moreover, we introduce the $H^\infty$ method that takes account of a
tradeoff between the quality of signal reconstruction and the cost (i.e., 
the amount of energy of transmitting signals) with an appropriate weighting function.
Design examples are presented to illustrate the effectiveness of the proposed method.

\section{Digital communication systems}
Figure~\ref{fig:comm_digital_com} illustrates
a typical digital communication system.
\begin{figure}[t]
\begin{center}
%\input{comm_digital_com.tex}
%WinTpicVersion2.15
\unitlength 0.1in
\begin{picture}(48.51,13.86)(4.35,-17.86)
% BOX 2 0 3 0
% 2 435 800 1128 1146
% 
\special{pn 8}%
\special{pa 435 400}%
\special{pa 1128 400}%
\special{pa 1128 746}%
\special{pa 435 746}%
\special{pa 435 400}%
\special{fp}%
% BOX 2 0 3 0
% 2 1475 1146 2167 800
% 
\special{pn 8}%
\special{pa 1475 746}%
\special{pa 2167 746}%
\special{pa 2167 400}%
\special{pa 1475 400}%
\special{pa 1475 746}%
\special{fp}%
% BOX 2 0 3 0
% 2 2514 800 3207 1146
% 
\special{pn 8}%
\special{pa 2514 400}%
\special{pa 3207 400}%
\special{pa 3207 746}%
\special{pa 2514 746}%
\special{pa 2514 400}%
\special{fp}%
% BOX 2 0 3 0
% 2 3553 1146 4247 800
% 
\special{pn 8}%
\special{pa 3553 746}%
\special{pa 4247 746}%
\special{pa 4247 400}%
\special{pa 3553 400}%
\special{pa 3553 746}%
\special{fp}%
% BOX 2 0 3 0
% 2 4593 1320 5286 1666
% 
\special{pn 8}%
\special{pa 4593 920}%
\special{pa 5286 920}%
\special{pa 5286 1266}%
\special{pa 4593 1266}%
\special{pa 4593 920}%
\special{fp}%
% BOX 2 0 3 0
% 2 4247 1840 3553 2186
% 
\special{pn 8}%
\special{pa 4247 1440}%
\special{pa 3553 1440}%
\special{pa 3553 1786}%
\special{pa 4247 1786}%
\special{pa 4247 1440}%
\special{fp}%
% BOX 2 0 3 0
% 2 3207 1840 2514 2186
% 
\special{pn 8}%
\special{pa 3207 1440}%
\special{pa 2514 1440}%
\special{pa 2514 1786}%
\special{pa 3207 1786}%
\special{pa 3207 1440}%
\special{fp}%
% BOX 2 0 3 0
% 2 2167 1840 1475 2186
% 
\special{pn 8}%
\special{pa 2167 1440}%
\special{pa 1475 1440}%
\special{pa 1475 1786}%
\special{pa 2167 1786}%
\special{pa 2167 1440}%
\special{fp}%
% BOX 2 0 3 0
% 2 1128 1840 435 2186
% 
\special{pn 8}%
\special{pa 1128 1440}%
\special{pa 435 1440}%
\special{pa 435 1786}%
\special{pa 1128 1786}%
\special{pa 1128 1440}%
\special{fp}%
% VECTOR 2 0 3 0
% 2 1128 973 1475 973
% 
\special{pn 8}%
\special{pa 1128 573}%
\special{pa 1475 573}%
\special{fp}%
\special{sh 1}%
\special{pa 1475 573}%
\special{pa 1408 553}%
\special{pa 1422 573}%
\special{pa 1408 593}%
\special{pa 1475 573}%
\special{fp}%
% VECTOR 2 0 3 0
% 2 2167 973 2514 973
% 
\special{pn 8}%
\special{pa 2167 573}%
\special{pa 2514 573}%
\special{fp}%
\special{sh 1}%
\special{pa 2514 573}%
\special{pa 2447 553}%
\special{pa 2461 573}%
\special{pa 2447 593}%
\special{pa 2514 573}%
\special{fp}%
% VECTOR 2 0 3 0
% 2 3207 973 3553 973
% 
\special{pn 8}%
\special{pa 3207 573}%
\special{pa 3553 573}%
\special{fp}%
\special{sh 1}%
\special{pa 3553 573}%
\special{pa 3486 553}%
\special{pa 3500 573}%
\special{pa 3486 593}%
\special{pa 3553 573}%
\special{fp}%
% LINE 2 0 3 0
% 2 4247 973 4939 973
% 
\special{pn 8}%
\special{pa 4247 573}%
\special{pa 4939 573}%
\special{fp}%
% VECTOR 2 0 3 0
% 2 4939 973 4939 1320
% 
\special{pn 8}%
\special{pa 4939 573}%
\special{pa 4939 920}%
\special{fp}%
\special{sh 1}%
\special{pa 4939 920}%
\special{pa 4959 853}%
\special{pa 4939 867}%
\special{pa 4919 853}%
\special{pa 4939 920}%
\special{fp}%
% LINE 2 0 3 0
% 2 4939 1666 4939 2013
% 
\special{pn 8}%
\special{pa 4939 1266}%
\special{pa 4939 1613}%
\special{fp}%
% VECTOR 2 0 3 0
% 2 4939 2013 4247 2013
% 
\special{pn 8}%
\special{pa 4939 1613}%
\special{pa 4247 1613}%
\special{fp}%
\special{sh 1}%
\special{pa 4247 1613}%
\special{pa 4314 1633}%
\special{pa 4300 1613}%
\special{pa 4314 1593}%
\special{pa 4247 1613}%
\special{fp}%
% VECTOR 2 0 3 0
% 2 3553 2013 3207 2013
% 
\special{pn 8}%
\special{pa 3553 1613}%
\special{pa 3207 1613}%
\special{fp}%
\special{sh 1}%
\special{pa 3207 1613}%
\special{pa 3274 1633}%
\special{pa 3260 1613}%
\special{pa 3274 1593}%
\special{pa 3207 1613}%
\special{fp}%
% VECTOR 2 0 3 0
% 2 2514 2013 2167 2013
% 
\special{pn 8}%
\special{pa 2514 1613}%
\special{pa 2167 1613}%
\special{fp}%
\special{sh 1}%
\special{pa 2167 1613}%
\special{pa 2234 1633}%
\special{pa 2220 1613}%
\special{pa 2234 1593}%
\special{pa 2167 1613}%
\special{fp}%
% VECTOR 2 0 3 0
% 2 1475 2013 1128 2013
% 
\special{pn 8}%
\special{pa 1475 1613}%
\special{pa 1128 1613}%
\special{fp}%
\special{sh 1}%
\special{pa 1128 1613}%
\special{pa 1195 1633}%
\special{pa 1181 1613}%
\special{pa 1195 1593}%
\special{pa 1128 1613}%
\special{fp}%
% STR 2 0 3 0
% 3 781 800 781 973 5 0
% Source
\put(7.8100,-5.7300){\makebox(0,0){Source}}%
% STR 2 0 3 0
% 3 1821 800 1821 973 5 0
% Sampling
\put(18.2100,-5.7300){\makebox(0,0){Sampler}}%
% STR 2 0 3 0
% 3 2861 800 2861 973 5 0
% Quantizing
\put(28.6100,-5.7300){\makebox(0,0){Quantizer}}%
% STR 2 0 3 0
% 3 3900 800 3900 973 5 0
% Encoding
\put(39.0000,-5.7300){\makebox(0,0){Encoder}}%
% STR 2 0 3 0
% 3 3900 1840 3900 2013 5 0
% Decoding
\put(39.0000,-16.1300){\makebox(0,0){Decoder}}%
% STR 2 0 3 0
% 3 2861 1840 2861 2013 5 0
% Filter
\put(28.6100,-16.1300){\makebox(0,0){Filter}}%
% STR 2 0 3 0
% 3 1821 1840 1821 2013 5 0
% Hold
\put(18.2100,-16.1300){\makebox(0,0){Hold}}%
% STR 2 0 3 0
% 3 781 1840 781 2013 5 0
% Output
\put(7.8100,-16.1300){\makebox(0,0){Output}}%
% STR 2 0 3 0
% 3 4939 1320 4939 1493 5 0
% Channel
\put(49.3900,-10.9300){\makebox(0,0){Channel}}%
\end{picture}%
\end{center}
\caption{Digital communication system}
\label{fig:comm_digital_com}
\end{figure}
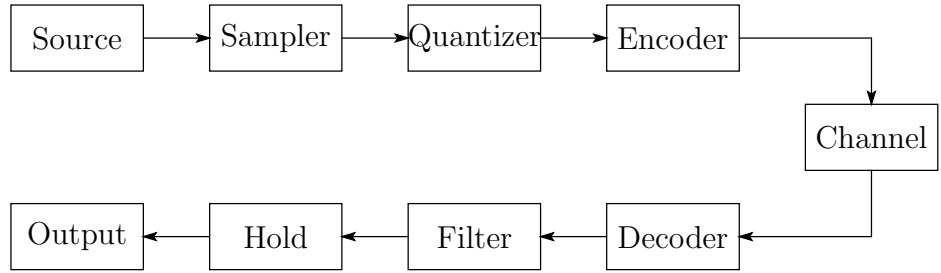
In this figure, the source is assumed to be an analog signal (e.g., audio, speech or image).
The analog signal will be discretized with an A/D converter,
which contains {\em sampler}, {\em quantizer} and {\em encoder} (or {\em coder}).

Sampling is a discretization in time, while quantization is that in amplitude.

Encoder converts the sampled and quantized signal to a binary valued signal.
It often contains {\em compressing} and {\em filtering} for
efficiency of signal transmission.
For example, {\em subband coding} is often used for efficient
communication.
Figure \ref{fig:comm_subband} illustrates a simple subband coding system.
\begin{figure}[t]
\begin{center}
%\input{comm_subband}
%WinTpicVersion2.15
\unitlength 0.1in
\begin{picture}(32.00,12.30)(6.00,-14.00)
% LINE 2 0 3 0
% 2 600 1200 1000 1200
% 
\special{pn 8}%
\special{pa 600 800}%
\special{pa 1000 800}%
\special{fp}%
% LINE 2 0 3 0
% 2 1000 800 1000 1600
% 
\special{pn 8}%
\special{pa 1000 400}%
\special{pa 1000 1200}%
\special{fp}%
% VECTOR 2 0 3 0
% 2 1000 800 1400 800
% 
\special{pn 8}%
\special{pa 1000 400}%
\special{pa 1400 400}%
\special{fp}%
\special{sh 1}%
\special{pa 1400 400}%
\special{pa 1333 380}%
\special{pa 1347 400}%
\special{pa 1333 420}%
\special{pa 1400 400}%
\special{fp}%
% BOX 2 0 3 0
% 2 1400 600 1800 1000
% 
\special{pn 8}%
\special{pa 1400 200}%
\special{pa 1800 200}%
\special{pa 1800 600}%
\special{pa 1400 600}%
\special{pa 1400 200}%
\special{fp}%
% VECTOR 2 0 3 0
% 2 1800 800 2200 800
% 
\special{pn 8}%
\special{pa 1800 400}%
\special{pa 2200 400}%
\special{fp}%
\special{sh 1}%
\special{pa 2200 400}%
\special{pa 2133 380}%
\special{pa 2147 400}%
\special{pa 2133 420}%
\special{pa 2200 400}%
\special{fp}%
% BOX 2 0 3 0
% 2 2200 600 2600 1000
% 
\special{pn 8}%
\special{pa 2200 200}%
\special{pa 2600 200}%
\special{pa 2600 600}%
\special{pa 2200 600}%
\special{pa 2200 200}%
\special{fp}%
% VECTOR 2 0 3 0
% 2 1000 1600 1400 1600
% 
\special{pn 8}%
\special{pa 1000 1200}%
\special{pa 1400 1200}%
\special{fp}%
\special{sh 1}%
\special{pa 1400 1200}%
\special{pa 1333 1180}%
\special{pa 1347 1200}%
\special{pa 1333 1220}%
\special{pa 1400 1200}%
\special{fp}%
% BOX 2 0 3 0
% 2 1400 1400 1800 1800
% 
\special{pn 8}%
\special{pa 1400 1000}%
\special{pa 1800 1000}%
\special{pa 1800 1400}%
\special{pa 1400 1400}%
\special{pa 1400 1000}%
\special{fp}%
% VECTOR 2 0 3 0
% 2 1800 1600 2200 1600
% 
\special{pn 8}%
\special{pa 1800 1200}%
\special{pa 2200 1200}%
\special{fp}%
\special{sh 1}%
\special{pa 2200 1200}%
\special{pa 2133 1180}%
\special{pa 2147 1200}%
\special{pa 2133 1220}%
\special{pa 2200 1200}%
\special{fp}%
% BOX 2 0 3 0
% 2 2200 1400 2600 1800
% 
\special{pn 8}%
\special{pa 2200 1000}%
\special{pa 2600 1000}%
\special{pa 2600 1400}%
\special{pa 2200 1400}%
\special{pa 2200 1000}%
\special{fp}%
% VECTOR 2 0 3 0
% 2 2600 800 3000 800
% 
\special{pn 8}%
\special{pa 2600 400}%
\special{pa 3000 400}%
\special{fp}%
\special{sh 1}%
\special{pa 3000 400}%
\special{pa 2933 380}%
\special{pa 2947 400}%
\special{pa 2933 420}%
\special{pa 3000 400}%
\special{fp}%
% VECTOR 2 0 3 0
% 2 2600 1600 3000 1600
% 
\special{pn 8}%
\special{pa 2600 1200}%
\special{pa 3000 1200}%
\special{fp}%
\special{sh 1}%
\special{pa 3000 1200}%
\special{pa 2933 1180}%
\special{pa 2947 1200}%
\special{pa 2933 1220}%
\special{pa 3000 1200}%
\special{fp}%
% BOX 2 0 3 0
% 2 3000 600 3400 1800
% 
\special{pn 8}%
\special{pa 3000 200}%
\special{pa 3400 200}%
\special{pa 3400 1400}%
\special{pa 3000 1400}%
\special{pa 3000 200}%
\special{fp}%
% VECTOR 2 0 3 0
% 2 3400 800 3800 800
% 
\special{pn 8}%
\special{pa 3400 400}%
\special{pa 3800 400}%
\special{fp}%
\special{sh 1}%
\special{pa 3800 400}%
\special{pa 3733 380}%
\special{pa 3747 400}%
\special{pa 3733 420}%
\special{pa 3800 400}%
\special{fp}%
% STR 2 0 3 0
% 3 3200 1100 3200 1200 5 0
% $Q$
\put(32.0000,-8.0000){\makebox(0,0){$Q$}}%
% STR 2 0 3 0
% 3 2400 1500 2400 1600 5 0
% $\ds{2}$
\put(24.0000,-12.0000){\makebox(0,0){$\ds{2}$}}%
% STR 2 0 3 0
% 3 2400 700 2400 800 5 0
% $\ds{2}$
\put(24.0000,-4.0000){\makebox(0,0){$\ds{2}$}}%
% STR 2 0 3 0
% 3 1600 700 1600 800 5 0
% $H_1(z)$
\put(16.0000,-4.0000){\makebox(0,0){$H_1(z)$}}%
% STR 2 0 3 0
% 3 1600 1500 1600 1600 5 0
% $H_2(z)$
\put(16.0000,-12.0000){\makebox(0,0){$H_2(z)$}}%
% STR 2 0 3 0
% 3 3600 640 3600 740 2 0
% $v_1$
\put(36.0000,-3.4000){\makebox(0,0)[lb]{$v_1$}}%
% STR 2 0 3 0
% 3 660 1040 660 1140 2 0
% $y$
\put(6.6000,-7.4000){\makebox(0,0)[lb]{$y$}}%
% STR 2 0 3 0
% 3 2660 640 2660 740 2 0
% $y_1$
\put(26.6000,-3.4000){\makebox(0,0)[lb]{$y_1$}}%
% STR 2 0 3 0
% 3 2660 1440 2660 1540 2 0
% $y_2$
\put(26.6000,-11.4000){\makebox(0,0)[lb]{$y_2$}}%
% VECTOR 2 0 3 0
% 2 3400 1600 3800 1600
% 
\special{pn 8}%
\special{pa 3400 1200}%
\special{pa 3800 1200}%
\special{fp}%
\special{sh 1}%
\special{pa 3800 1200}%
\special{pa 3733 1180}%
\special{pa 3747 1200}%
\special{pa 3733 1220}%
\special{pa 3800 1200}%
\special{fp}%
% STR 2 0 3 0
% 3 3600 1440 3600 1540 2 0
% $v_2$
\put(36.0000,-11.4000){\makebox(0,0)[lb]{$v_2$}}%
\end{picture}%
\end{center}
\caption{Subband coding}
\label{fig:comm_subband}
\end{figure}
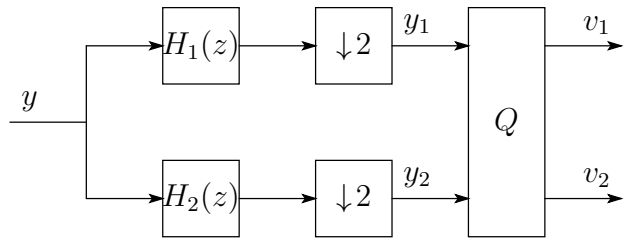
The input signal $y$ is divided into two subband signals\footnote{In real applications, the signal will be subdivided
into 16 or 32 signals.}.
By taking the filters $H_1(z)$ and $H_2(z)$ appropriately,
we can divide the frequency into two subbands,
for example, $y_1$ is a signal with low frequency and $y_2$
with high frequency.
If we need not reconstruct the information in the high frequency range precisely,
we can compress the signal by applying fewer bits to the high-frequency signal 
(i.e. the signal $y_2$ in Figure \ref{fig:comm_subband}).
For simplicity, we apply infinitely many bits to $y_1$ and 
no bit to $y_2$, that is, we assume $Q$ in Figure \ref{fig:comm_subband} as
\begin{equation}
\label{eq:subband_encode}
Q=\left[\begin{array}{cc}1&0\\0&0\end{array}\right].
\end{equation}
This means that we transmit only the signal $v_1$.

The signal then goes through the communication channel.
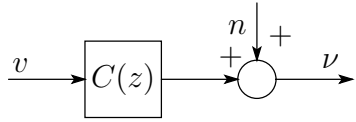
\begin{figure}[t]
\begin{center}
%\input{comm_channel}
%WinTpicVersion2.15
\unitlength 0.1in
\begin{picture}(18.00,6.10)(4.00,-6.00)
% VECTOR 2 0 3 0
% 2 400 800 800 800
% 
\special{pn 8}%
\special{pa 400 400}%
\special{pa 800 400}%
\special{fp}%
\special{sh 1}%
\special{pa 800 400}%
\special{pa 733 380}%
\special{pa 747 400}%
\special{pa 733 420}%
\special{pa 800 400}%
\special{fp}%
% BOX 2 0 3 0
% 2 800 600 1200 1000
% 
\special{pn 8}%
\special{pa 800 200}%
\special{pa 1200 200}%
\special{pa 1200 600}%
\special{pa 800 600}%
\special{pa 800 200}%
\special{fp}%
% VECTOR 2 0 3 0
% 2 1200 800 1600 800
% 
\special{pn 8}%
\special{pa 1200 400}%
\special{pa 1600 400}%
\special{fp}%
\special{sh 1}%
\special{pa 1600 400}%
\special{pa 1533 380}%
\special{pa 1547 400}%
\special{pa 1533 420}%
\special{pa 1600 400}%
\special{fp}%
% CIRCLE 2 0 3 0
% 4 1700 800 1600 800 1600 800 1600 800
% 
\special{pn 8}%
\special{ar 1700 400 100 100  0.0000000 6.2831853}%
% VECTOR 2 0 3 0
% 2 1800 800 2200 800
% 
\special{pn 8}%
\special{pa 1800 400}%
\special{pa 2200 400}%
\special{fp}%
\special{sh 1}%
\special{pa 2200 400}%
\special{pa 2133 380}%
\special{pa 2147 400}%
\special{pa 2133 420}%
\special{pa 2200 400}%
\special{fp}%
% VECTOR 2 0 3 0
% 2 1700 400 1700 700
% 
\special{pn 8}%
\special{pa 1700 0}%
\special{pa 1700 300}%
\special{fp}%
\special{sh 1}%
\special{pa 1700 300}%
\special{pa 1720 233}%
\special{pa 1700 247}%
\special{pa 1680 233}%
\special{pa 1700 300}%
\special{fp}%
% STR 2 0 3 0
% 3 1000 700 1000 800 5 0
% $C(z)$
\put(10.0000,-4.0000){\makebox(0,0){$C(z)$}}%
% STR 2 0 3 0
% 3 1500 640 1500 740 2 0
% $+$
\put(15.0000,-3.4000){\makebox(0,0)[lb]{$+$}}%
% STR 2 0 3 0
% 3 1760 530 1760 630 2 0
% $+$
\put(17.6000,-2.3000){\makebox(0,0)[lb]{$+$}}%
% STR 2 0 3 0
% 3 1550 460 1550 560 2 0
% $n$
\put(15.5000,-1.6000){\makebox(0,0)[lb]{$n$}}%
% STR 2 0 3 0
% 3 420 660 420 760 2 0
% $v$
\put(4.2000,-3.6000){\makebox(0,0)[lb]{$v$}}%
% STR 2 0 3 0
% 3 2040 640 2040 740 2 0
% $\nu$
\put(20.4000,-3.4000){\makebox(0,0)[lb]{$\nu$}}%
\end{picture}%
\end{center}
\caption{Communication channel model}
\label{fig:comm_channel}
\end{figure}
Since modeling for the channel is an important and difficult issue,
we assume a simple model: a linear time-invariant system and additive noise
as shown in Figure \ref{fig:comm_channel}.
Note that communication channels are generally time-varying and nonlinear,
in particular, in wireless communication, channel modeling is very complicated
\cite{Pro, ProSal}.

Then, the signal distorted by the channel enters the receiver.
The receiver decodes the binary signal, and then filters the decoded signal.
The filter reduces distortion of the received signal and
expands the compressed signal.
The process is shown in Figure \ref{fig:comm_subband_decode}.
Since we take the encoding (\ref{eq:subband_encode}),
we similarly assume $Q'$ to be the same as $Q$, that is, 
\begin{equation*}
Q'=\left[\begin{array}{cc}1&0\\0&0\end{array}\right].
\end{equation*}
\begin{figure}[t]
\begin{center}
%\input{comm_subband_decode}
%WinTpicVersion2.15
\unitlength 0.1in
\begin{picture}(33.00,12.30)(2.00,-14.00)
% VECTOR 2 0 3 0
% 2 1000 800 1400 800
% 
\special{pn 8}%
\special{pa 1000 400}%
\special{pa 1400 400}%
\special{fp}%
\special{sh 1}%
\special{pa 1400 400}%
\special{pa 1333 380}%
\special{pa 1347 400}%
\special{pa 1333 420}%
\special{pa 1400 400}%
\special{fp}%
% BOX 2 0 3 0
% 2 1400 600 1800 1000
% 
\special{pn 8}%
\special{pa 1400 200}%
\special{pa 1800 200}%
\special{pa 1800 600}%
\special{pa 1400 600}%
\special{pa 1400 200}%
\special{fp}%
% VECTOR 2 0 3 0
% 2 1800 800 2200 800
% 
\special{pn 8}%
\special{pa 1800 400}%
\special{pa 2200 400}%
\special{fp}%
\special{sh 1}%
\special{pa 2200 400}%
\special{pa 2133 380}%
\special{pa 2147 400}%
\special{pa 2133 420}%
\special{pa 2200 400}%
\special{fp}%
% BOX 2 0 3 0
% 2 2200 600 2600 1000
% 
\special{pn 8}%
\special{pa 2200 200}%
\special{pa 2600 200}%
\special{pa 2600 600}%
\special{pa 2200 600}%
\special{pa 2200 200}%
\special{fp}%
% VECTOR 2 0 3 0
% 2 1000 1600 1400 1600
% 
\special{pn 8}%
\special{pa 1000 1200}%
\special{pa 1400 1200}%
\special{fp}%
\special{sh 1}%
\special{pa 1400 1200}%
\special{pa 1333 1180}%
\special{pa 1347 1200}%
\special{pa 1333 1220}%
\special{pa 1400 1200}%
\special{fp}%
% BOX 2 0 3 0
% 2 1400 1400 1800 1800
% 
\special{pn 8}%
\special{pa 1400 1000}%
\special{pa 1800 1000}%
\special{pa 1800 1400}%
\special{pa 1400 1400}%
\special{pa 1400 1000}%
\special{fp}%
% VECTOR 2 0 3 0
% 2 1800 1600 2200 1600
% 
\special{pn 8}%
\special{pa 1800 1200}%
\special{pa 2200 1200}%
\special{fp}%
\special{sh 1}%
\special{pa 2200 1200}%
\special{pa 2133 1180}%
\special{pa 2147 1200}%
\special{pa 2133 1220}%
\special{pa 2200 1200}%
\special{fp}%
% BOX 2 0 3 0
% 2 2200 1400 2600 1800
% 
\special{pn 8}%
\special{pa 2200 1000}%
\special{pa 2600 1000}%
\special{pa 2600 1400}%
\special{pa 2200 1400}%
\special{pa 2200 1000}%
\special{fp}%
% STR 2 0 3 0
% 3 2400 1500 2400 1600 5 0
% $F_2(z)$
\put(24.0000,-12.0000){\makebox(0,0){$F_2(z)$}}%
% STR 2 0 3 0
% 3 2400 700 2400 800 5 0
% $F_1(z)$
\put(24.0000,-4.0000){\makebox(0,0){$F_1(z)$}}%
% STR 2 0 3 0
% 3 1600 700 1600 800 5 0
% $\us{2}$
\put(16.0000,-4.0000){\makebox(0,0){$\us{2}$}}%
% STR 2 0 3 0
% 3 1600 1500 1600 1600 5 0
% $\us{2}$
\put(16.0000,-12.0000){\makebox(0,0){$\us{2}$}}%
% STR 2 0 3 0
% 3 2660 640 2660 740 2 0
% $u_1$
\put(26.6000,-3.4000){\makebox(0,0)[lb]{$u_1$}}%
% STR 2 0 3 0
% 3 2660 1440 2660 1540 2 0
% $u_2$
\put(26.6000,-11.4000){\makebox(0,0)[lb]{$u_2$}}%
% LINE 2 0 3 0
% 2 2600 800 3000 800
% 
\special{pn 8}%
\special{pa 2600 400}%
\special{pa 3000 400}%
\special{fp}%
% CIRCLE 2 0 3 0
% 4 3000 1200 3000 1300 3000 1300 3000 1300
% 
\special{pn 8}%
\special{ar 3000 800 100 100  0.0000000 6.2831853}%
% VECTOR 2 0 3 0
% 2 3000 800 3000 1100
% 
\special{pn 8}%
\special{pa 3000 400}%
\special{pa 3000 700}%
\special{fp}%
\special{sh 1}%
\special{pa 3000 700}%
\special{pa 3020 633}%
\special{pa 3000 647}%
\special{pa 2980 633}%
\special{pa 3000 700}%
\special{fp}%
% VECTOR 2 0 3 0
% 2 3000 1600 3000 1300
% 
\special{pn 8}%
\special{pa 3000 1200}%
\special{pa 3000 900}%
\special{fp}%
\special{sh 1}%
\special{pa 3000 900}%
\special{pa 2980 967}%
\special{pa 3000 953}%
\special{pa 3020 967}%
\special{pa 3000 900}%
\special{fp}%
% LINE 2 0 3 0
% 2 2600 1600 3000 1600
% 
\special{pn 8}%
\special{pa 2600 1200}%
\special{pa 3000 1200}%
\special{fp}%
% VECTOR 2 0 3 0
% 2 3100 1200 3500 1200
% 
\special{pn 8}%
\special{pa 3100 800}%
\special{pa 3500 800}%
\special{fp}%
\special{sh 1}%
\special{pa 3500 800}%
\special{pa 3433 780}%
\special{pa 3447 800}%
\special{pa 3433 820}%
\special{pa 3500 800}%
\special{fp}%
% BOX 2 0 3 0
% 2 1000 600 600 1800
% 
\special{pn 8}%
\special{pa 1000 200}%
\special{pa 600 200}%
\special{pa 600 1400}%
\special{pa 1000 1400}%
\special{pa 1000 200}%
\special{fp}%
% VECTOR 2 0 3 0
% 2 200 800 600 800
% 
\special{pn 8}%
\special{pa 200 400}%
\special{pa 600 400}%
\special{fp}%
\special{sh 1}%
\special{pa 600 400}%
\special{pa 533 380}%
\special{pa 547 400}%
\special{pa 533 420}%
\special{pa 600 400}%
\special{fp}%
% VECTOR 2 0 3 0
% 2 200 1600 600 1600
% 
\special{pn 8}%
\special{pa 200 1200}%
\special{pa 600 1200}%
\special{fp}%
\special{sh 1}%
\special{pa 600 1200}%
\special{pa 533 1180}%
\special{pa 547 1200}%
\special{pa 533 1220}%
\special{pa 600 1200}%
\special{fp}%
% STR 2 0 3 0
% 3 800 1100 800 1200 5 0
% $Q'$
\put(8.0000,-8.0000){\makebox(0,0){$Q'$}}%
% STR 2 0 3 0
% 3 3260 1040 3260 1140 2 0
% $u$
\put(32.6000,-7.4000){\makebox(0,0)[lb]{$u$}}%
% STR 2 0 3 0
% 3 240 640 240 740 2 0
% $\nu_1$
\put(2.4000,-3.4000){\makebox(0,0)[lb]{$\nu_1$}}%
% STR 2 0 3 0
% 3 240 1440 240 1540 2 0
% $\nu_2$
\put(2.4000,-11.4000){\makebox(0,0)[lb]{$\nu_2$}}%
% STR 2 0 3 0
% 3 3040 960 3040 1060 2 0
% $+$
\put(30.4000,-6.6000){\makebox(0,0)[lb]{$+$}}%
% STR 2 0 3 0
% 3 3040 1240 3040 1340 1 0
% $+$
\put(30.4000,-9.4000){\makebox(0,0)[lt]{$+$}}%
\end{picture}%
\end{center}
\caption{Decoding}
\label{fig:comm_subband_decode}
\end{figure}

Then the decoded and filtered signal is converted to an analog signal
by a hold device, and finally we obtain a received signal.

The conventional design is executed in the discrete-time domain
by assuming that the original analog signal contains no frequency beyond the Nyquist frequency.
To satisfy this assumption, we often use a low-pass filter,
which deteriorates the quality of communication.

In contrast to this, we introduce the sampled-data $H^\infty$ control theory
to take the analog performance into account.
We will show that our design is superior to the conventional design
in the following sections.

\section{Design problem formulation}
We consider a digital communication system as shown in Figure \ref{fig:comm_sys}.
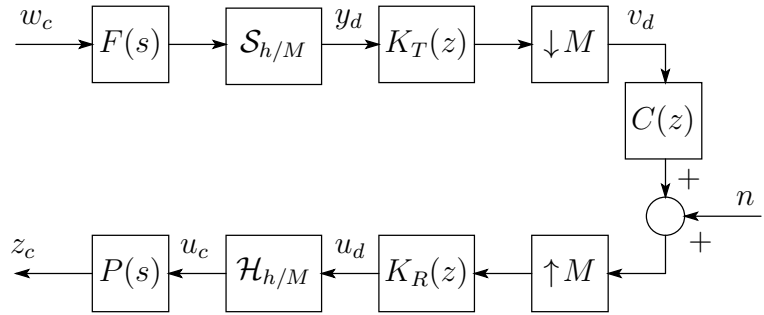
\begin{figure}[t]
    \begin{center}
%	\input{comm_design_comsys}
%WinTpicVersion2.15
\unitlength 0.1in
\begin{picture}(39.30,16.50)(0.70,-18.00)
% VECTOR 2 0 3 0
% 2 100 800 500 800
% 
\special{pn 8}%
\special{pa 100 400}%
\special{pa 500 400}%
\special{fp}%
\special{sh 1}%
\special{pa 500 400}%
\special{pa 433 380}%
\special{pa 447 400}%
\special{pa 433 420}%
\special{pa 500 400}%
\special{fp}%
% BOX 2 0 3 0
% 2 500 600 900 1000
% 
\special{pn 8}%
\special{pa 500 200}%
\special{pa 900 200}%
\special{pa 900 600}%
\special{pa 500 600}%
\special{pa 500 200}%
\special{fp}%
% BOX 2 0 3 0
% 2 1200 610 1700 1010
% 
\special{pn 8}%
\special{pa 1200 210}%
\special{pa 1700 210}%
\special{pa 1700 610}%
\special{pa 1200 610}%
\special{pa 1200 210}%
\special{fp}%
% BOX 2 0 3 0
% 2 2800 600 3200 1000
% 
\special{pn 8}%
\special{pa 2800 200}%
\special{pa 3200 200}%
\special{pa 3200 600}%
\special{pa 2800 600}%
\special{pa 2800 200}%
\special{fp}%
% BOX 2 0 3 0
% 2 3290 1000 3690 1400
% 
\special{pn 8}%
\special{pa 3290 600}%
\special{pa 3690 600}%
\special{pa 3690 1000}%
\special{pa 3290 1000}%
\special{pa 3290 600}%
\special{fp}%
% BOX 2 0 3 0
% 2 3200 1800 2800 2200
% 
\special{pn 8}%
\special{pa 3200 1400}%
\special{pa 2800 1400}%
\special{pa 2800 1800}%
\special{pa 3200 1800}%
\special{pa 3200 1400}%
\special{fp}%
% BOX 2 0 3 0
% 2 1700 1800 1200 2200
% 
\special{pn 8}%
\special{pa 1700 1400}%
\special{pa 1200 1400}%
\special{pa 1200 1800}%
\special{pa 1700 1800}%
\special{pa 1700 1400}%
\special{fp}%
% BOX 2 0 3 0
% 2 900 1800 500 2200
% 
\special{pn 8}%
\special{pa 900 1400}%
\special{pa 500 1400}%
\special{pa 500 1800}%
\special{pa 900 1800}%
\special{pa 900 1400}%
\special{fp}%
% STR 2 0 3 0
% 3 700 700 700 800 5 0
% $F(s)$
\put(7.0000,-4.0000){\makebox(0,0){$F(s)$}}%
% STR 2 0 3 0
% 3 1450 710 1450 810 5 0
% $\samp{h/M}$
\put(14.5000,-4.1000){\makebox(0,0){$\samp{h/M}$}}%
% STR 2 0 3 0
% 3 1450 1900 1450 2000 5 0
% $\hold{h/M}$
\put(14.5000,-16.0000){\makebox(0,0){$\hold{h/M}$}}%
% STR 2 0 3 0
% 3 3000 700 3000 800 5 0
% $\ds{M}$
\put(30.0000,-4.0000){\makebox(0,0){$\ds{M}$}}%
% STR 2 0 3 0
% 3 3490 1100 3490 1200 5 0
% $C(z)$
\put(34.9000,-8.0000){\makebox(0,0){$C(z)$}}%
% STR 2 0 3 0
% 3 3000 1900 3000 2000 5 0
% $\us{M}$
\put(30.0000,-16.0000){\makebox(0,0){$\us{M}$}}%
% STR 2 0 3 0
% 3 700 1900 700 2000 5 0
% $P(s)$
\put(7.0000,-16.0000){\makebox(0,0){$P(s)$}}%
% STR 2 0 3 0
% 3 120 620 120 720 2 0
% $w_c$
\put(1.2000,-3.2000){\makebox(0,0)[lb]{$w_c$}}%
% VECTOR 2 0 3 0
% 2 1700 800 2000 800
% 
\special{pn 8}%
\special{pa 1700 400}%
\special{pa 2000 400}%
\special{fp}%
\special{sh 1}%
\special{pa 2000 400}%
\special{pa 1933 380}%
\special{pa 1947 400}%
\special{pa 1933 420}%
\special{pa 2000 400}%
\special{fp}%
% BOX 2 0 3 0
% 2 2000 600 2500 1000
% 
\special{pn 8}%
\special{pa 2000 200}%
\special{pa 2500 200}%
\special{pa 2500 600}%
\special{pa 2000 600}%
\special{pa 2000 200}%
\special{fp}%
% VECTOR 2 0 3 0
% 2 2500 800 2800 800
% 
\special{pn 8}%
\special{pa 2500 400}%
\special{pa 2800 400}%
\special{fp}%
\special{sh 1}%
\special{pa 2800 400}%
\special{pa 2733 380}%
\special{pa 2747 400}%
\special{pa 2733 420}%
\special{pa 2800 400}%
\special{fp}%
% VECTOR 2 0 3 0
% 2 1200 2000 900 2000
% 
\special{pn 8}%
\special{pa 1200 1600}%
\special{pa 900 1600}%
\special{fp}%
\special{sh 1}%
\special{pa 900 1600}%
\special{pa 967 1620}%
\special{pa 953 1600}%
\special{pa 967 1580}%
\special{pa 900 1600}%
\special{fp}%
% VECTOR 2 0 3 0
% 2 2000 2000 1700 2000
% 
\special{pn 8}%
\special{pa 2000 1600}%
\special{pa 1700 1600}%
\special{fp}%
\special{sh 1}%
\special{pa 1700 1600}%
\special{pa 1767 1620}%
\special{pa 1753 1600}%
\special{pa 1767 1580}%
\special{pa 1700 1600}%
\special{fp}%
% BOX 2 0 3 0
% 2 2000 1800 2500 2200
% 
\special{pn 8}%
\special{pa 2000 1400}%
\special{pa 2500 1400}%
\special{pa 2500 1800}%
\special{pa 2000 1800}%
\special{pa 2000 1400}%
\special{fp}%
% VECTOR 2 0 3 0
% 2 2800 2000 2500 2000
% 
\special{pn 8}%
\special{pa 2800 1600}%
\special{pa 2500 1600}%
\special{fp}%
\special{sh 1}%
\special{pa 2500 1600}%
\special{pa 2567 1620}%
\special{pa 2553 1600}%
\special{pa 2567 1580}%
\special{pa 2500 1600}%
\special{fp}%
% VECTOR 2 0 3 0
% 2 3500 800 3500 1000
% 
\special{pn 8}%
\special{pa 3500 400}%
\special{pa 3500 600}%
\special{fp}%
\special{sh 1}%
\special{pa 3500 600}%
\special{pa 3520 533}%
\special{pa 3500 547}%
\special{pa 3480 533}%
\special{pa 3500 600}%
\special{fp}%
% VECTOR 2 0 3 0
% 2 3500 1400 3500 1600
% 
\special{pn 8}%
\special{pa 3500 1000}%
\special{pa 3500 1200}%
\special{fp}%
\special{sh 1}%
\special{pa 3500 1200}%
\special{pa 3520 1133}%
\special{pa 3500 1147}%
\special{pa 3480 1133}%
\special{pa 3500 1200}%
\special{fp}%
% CIRCLE 2 0 3 0
% 4 3500 1700 3500 1600 3500 1600 3500 1600
% 
\special{pn 8}%
\special{ar 3500 1300 100 100  0.0000000 6.2831853}%
% LINE 2 0 3 0
% 2 3500 1800 3500 2000
% 
\special{pn 8}%
\special{pa 3500 1400}%
\special{pa 3500 1600}%
\special{fp}%
% VECTOR 2 0 3 0
% 2 4000 1700 3600 1700
% 
\special{pn 8}%
\special{pa 4000 1300}%
\special{pa 3600 1300}%
\special{fp}%
\special{sh 1}%
\special{pa 3600 1300}%
\special{pa 3667 1320}%
\special{pa 3653 1300}%
\special{pa 3667 1280}%
\special{pa 3600 1300}%
\special{fp}%
% LINE 2 0 3 0
% 2 3200 800 3500 800
% 
\special{pn 8}%
\special{pa 3200 400}%
\special{pa 3500 400}%
\special{fp}%
% VECTOR 2 0 3 0
% 2 3500 2000 3200 2000
% 
\special{pn 8}%
\special{pa 3500 1600}%
\special{pa 3200 1600}%
\special{fp}%
\special{sh 1}%
\special{pa 3200 1600}%
\special{pa 3267 1620}%
\special{pa 3253 1600}%
\special{pa 3267 1580}%
\special{pa 3200 1600}%
\special{fp}%
% STR 2 0 3 0
% 3 2250 700 2250 800 5 0
% $K_T(z)$
\put(22.5000,-4.0000){\makebox(0,0){$K_T(z)$}}%
% STR 2 0 3 0
% 3 2250 1900 2250 2000 5 0
% $K_R(z)$
\put(22.5000,-16.0000){\makebox(0,0){$K_R(z)$}}%
% STR 2 0 3 0
% 3 3560 1460 3560 1560 2 0
% $+$
\put(35.6000,-11.6000){\makebox(0,0)[lb]{$+$}}%
% STR 2 0 3 0
% 3 3620 1680 3620 1780 1 0
% $+$
\put(36.2000,-13.8000){\makebox(0,0)[lt]{$+$}}%
% STR 2 0 3 0
% 3 3870 1540 3870 1640 2 0
% $n$
\put(38.7000,-12.4000){\makebox(0,0)[lb]{$n$}}%
% STR 2 0 3 0
% 3 1770 620 1770 720 2 0
% $y_d$
\put(17.7000,-3.2000){\makebox(0,0)[lb]{$y_d$}}%
% STR 2 0 3 0
% 3 3300 620 3300 720 2 0
% $v_d$
\put(33.0000,-3.2000){\makebox(0,0)[lb]{$v_d$}}%
% STR 2 0 3 0
% 3 1770 1820 1770 1920 2 0
% $u_d$
\put(17.7000,-15.2000){\makebox(0,0)[lb]{$u_d$}}%
% STR 2 0 3 0
% 3 960 1820 960 1920 2 0
% $u_c$
\put(9.6000,-15.2000){\makebox(0,0)[lb]{$u_c$}}%
% STR 2 0 3 0
% 3 70 1820 70 1920 2 0
% $z_c$
\put(0.7000,-15.2000){\makebox(0,0)[lb]{$z_c$}}%
% VECTOR 2 0 3 0
% 2 500 2000 100 2000
% 
\special{pn 8}%
\special{pa 500 1600}%
\special{pa 100 1600}%
\special{fp}%
\special{sh 1}%
\special{pa 100 1600}%
\special{pa 167 1620}%
\special{pa 153 1600}%
\special{pa 167 1580}%
\special{pa 100 1600}%
\special{fp}%
% VECTOR 2 0 3 0
% 2 900 800 1200 800
% 
\special{pn 8}%
\special{pa 900 400}%
\special{pa 1200 400}%
\special{fp}%
\special{sh 1}%
\special{pa 1200 400}%
\special{pa 1133 380}%
\special{pa 1147 400}%
\special{pa 1133 420}%
\special{pa 1200 400}%
\special{fp}%
\end{picture}%
    \end{center}
    \caption{Digital communication system}
    \label{fig:comm_sys}
\end{figure}
The incoming signal $w_c\in L^2[0,\infty)$ goes through an
analog low-pass filter $F_c$ and becomes $y_c$ that is nearly (but not
entirely) band limited.
The filter $F_c$ governs the frequency-domain characteristic of the
analog signal $y_c$\footnote{In the conventional design, $F_c$ is considered
to be an ideal filter that has the cut-off frequency up to the Nyquist frequency.}.
The signal $y_c$ is then sampled by the sampler $\samp{h/M}$ to become a discrete-time
signal (i.e., PAM signal) $y_d$ with sampling period $h/M$.
Then the signal is compressed by the downsampler with the factor $M$,
and becomes a discrete-time signal with sampling period $h$.
The downsampled signal is then shaped or enhanced by the transmitting digital filter $K_T$ to the signal
$v_d$ to be transmitted to a communication channel.

The transmitted signal $v_d$ is corrupted by the communication channel
$C$ and the additive noise $n_d$. In PCM communication, $n_d$ is also
considered as the noise generated by quantization and coding errors.
The received signal goes through the upsampler and
receiving digital filter $K_R$ that attempts to attenuate the imaging components caused by the upsampler
and the distortion by the channel.
Then the signal becomes an analog signal $u_c$ by the hold device $\hold{h/M}$ with sampling period
$h/M$ and this analog signal is smoothed by an analog low-pass filter $P_c$.
Finally we have the output signal $z_c$.

Our objective is to reconstruct the original analog signal $y_c$
by the transmitting filter $K_T$ and the receiving filter $K_R$ against 
data compression effect caused by the downsampler and
the distortion caused by the channel.

We thus consider the block diagram shown in Figure \ref{fig:comm_design_error} that is the
signal reconstruction error system for the design.
In the diagram the following points are taken into account:
\begin{itemize}
    \item The time delay $e^{-Ls}$ is introduced
	  because we allow a certain amount of time delay for signal
	  reconstruction.
    \item The transmitted signal $v_d$ is estimated with a weighting function
	  $W_z$ because the energy or the amplitude of the transmitted
	  signal $v_d$ is usually limited.
    \item The noise obeys a frequency characteristic $W_n$.
\end{itemize}
     
Our design problem is defined as follows:
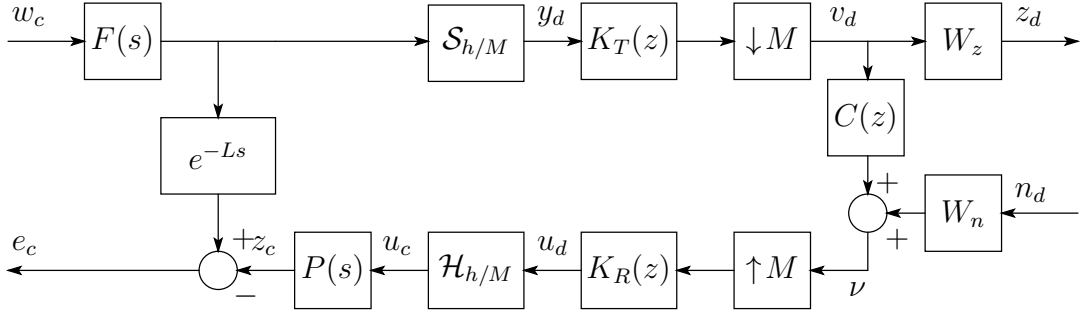
\begin{figure}[t]
    \begin{center}
%	\input{comm_design_error}
%WinTpicVersion2.15
\unitlength 0.1in
\begin{picture}(56.00,16.50)(4.00,-18.00)
% VECTOR 2 0 3 0
% 2 400 800 800 800
% 
\special{pn 8}%
\special{pa 400 400}%
\special{pa 800 400}%
\special{fp}%
\special{sh 1}%
\special{pa 800 400}%
\special{pa 733 380}%
\special{pa 747 400}%
\special{pa 733 420}%
\special{pa 800 400}%
\special{fp}%
% BOX 2 0 3 0
% 2 800 600 1200 1000
% 
\special{pn 8}%
\special{pa 800 200}%
\special{pa 1200 200}%
\special{pa 1200 600}%
\special{pa 800 600}%
\special{pa 800 200}%
\special{fp}%
% BOX 2 0 3 0
% 2 2600 610 3100 1010
% 
\special{pn 8}%
\special{pa 2600 210}%
\special{pa 3100 210}%
\special{pa 3100 610}%
\special{pa 2600 610}%
\special{pa 2600 210}%
\special{fp}%
% BOX 2 0 3 0
% 2 4200 600 4600 1000
% 
\special{pn 8}%
\special{pa 4200 200}%
\special{pa 4600 200}%
\special{pa 4600 600}%
\special{pa 4200 600}%
\special{pa 4200 200}%
\special{fp}%
% BOX 2 0 3 0
% 2 4690 1000 5090 1400
% 
\special{pn 8}%
\special{pa 4690 600}%
\special{pa 5090 600}%
\special{pa 5090 1000}%
\special{pa 4690 1000}%
\special{pa 4690 600}%
\special{fp}%
% BOX 2 0 3 0
% 2 4600 1800 4200 2200
% 
\special{pn 8}%
\special{pa 4600 1400}%
\special{pa 4200 1400}%
\special{pa 4200 1800}%
\special{pa 4600 1800}%
\special{pa 4600 1400}%
\special{fp}%
% BOX 2 0 3 0
% 2 3100 1800 2600 2200
% 
\special{pn 8}%
\special{pa 3100 1400}%
\special{pa 2600 1400}%
\special{pa 2600 1800}%
\special{pa 3100 1800}%
\special{pa 3100 1400}%
\special{fp}%
% BOX 2 0 3 0
% 2 2300 1800 1900 2200
% 
\special{pn 8}%
\special{pa 2300 1400}%
\special{pa 1900 1400}%
\special{pa 1900 1800}%
\special{pa 2300 1800}%
\special{pa 2300 1400}%
\special{fp}%
% LINE 2 0 3 0
% 2 1200 800 1800 800
% 
\special{pn 8}%
\special{pa 1200 400}%
\special{pa 1800 400}%
\special{fp}%
% VECTOR 2 0 3 0
% 2 1500 800 1500 1200
% 
\special{pn 8}%
\special{pa 1500 400}%
\special{pa 1500 800}%
\special{fp}%
\special{sh 1}%
\special{pa 1500 800}%
\special{pa 1520 733}%
\special{pa 1500 747}%
\special{pa 1480 733}%
\special{pa 1500 800}%
\special{fp}%
% BOX 2 0 3 0
% 2 1200 1200 1800 1600
% 
\special{pn 8}%
\special{pa 1200 800}%
\special{pa 1800 800}%
\special{pa 1800 1200}%
\special{pa 1200 1200}%
\special{pa 1200 800}%
\special{fp}%
% VECTOR 2 0 3 0
% 2 1500 1600 1500 1900
% 
\special{pn 8}%
\special{pa 1500 1200}%
\special{pa 1500 1500}%
\special{fp}%
\special{sh 1}%
\special{pa 1500 1500}%
\special{pa 1520 1433}%
\special{pa 1500 1447}%
\special{pa 1480 1433}%
\special{pa 1500 1500}%
\special{fp}%
% CIRCLE 2 0 3 0
% 4 1500 2000 1400 2000 1400 2000 1400 2000
% 
\special{pn 8}%
\special{ar 1500 1600 100 100  0.0000000 6.2831853}%
% VECTOR 2 0 3 0
% 2 1400 2000 400 2000
% 
\special{pn 8}%
\special{pa 1400 1600}%
\special{pa 400 1600}%
\special{fp}%
\special{sh 1}%
\special{pa 400 1600}%
\special{pa 467 1620}%
\special{pa 453 1600}%
\special{pa 467 1580}%
\special{pa 400 1600}%
\special{fp}%
% STR 2 0 3 0
% 3 1000 700 1000 800 5 0
% $F(s)$
\put(10.0000,-4.0000){\makebox(0,0){$F(s)$}}%
% STR 2 0 3 0
% 3 2850 710 2850 810 5 0
% $\samp{h/M}$
\put(28.5000,-4.1000){\makebox(0,0){$\samp{h/M}$}}%
% STR 2 0 3 0
% 3 2850 1900 2850 2000 5 0
% $\hold{h/M}$
\put(28.5000,-16.0000){\makebox(0,0){$\hold{h/M}$}}%
% STR 2 0 3 0
% 3 1500 1300 1500 1400 5 0
% $e^{-Ls}$
\put(15.0000,-10.0000){\makebox(0,0){$e^{-Ls}$}}%
% STR 2 0 3 0
% 3 4400 700 4400 800 5 0
% $\ds{M}$
\put(44.0000,-4.0000){\makebox(0,0){$\ds{M}$}}%
% STR 2 0 3 0
% 3 4890 1100 4890 1200 5 0
% $C(z)$
\put(48.9000,-8.0000){\makebox(0,0){$C(z)$}}%
% STR 2 0 3 0
% 3 4400 1900 4400 2000 5 0
% $\us{M}$
\put(44.0000,-16.0000){\makebox(0,0){$\us{M}$}}%
% STR 2 0 3 0
% 3 2100 1900 2100 2000 5 0
% $P(s)$
\put(21.0000,-16.0000){\makebox(0,0){$P(s)$}}%
% STR 2 0 3 0
% 3 1570 1790 1570 1890 2 0
% $+$
\put(15.7000,-14.9000){\makebox(0,0)[lb]{$+$}}%
% STR 2 0 3 0
% 3 1580 1980 1580 2080 1 0
% $-$
\put(15.8000,-16.8000){\makebox(0,0)[lt]{$-$}}%
% STR 2 0 3 0
% 3 420 620 420 720 2 0
% $w_c$
\put(4.2000,-3.2000){\makebox(0,0)[lb]{$w_c$}}%
% STR 2 0 3 0
% 3 420 1820 420 1920 2 0
% $e_c$
\put(4.2000,-15.2000){\makebox(0,0)[lb]{$e_c$}}%
% VECTOR 2 0 3 0
% 2 3100 800 3400 800
% 
\special{pn 8}%
\special{pa 3100 400}%
\special{pa 3400 400}%
\special{fp}%
\special{sh 1}%
\special{pa 3400 400}%
\special{pa 3333 380}%
\special{pa 3347 400}%
\special{pa 3333 420}%
\special{pa 3400 400}%
\special{fp}%
% BOX 2 0 3 0
% 2 3400 600 3900 1000
% 
\special{pn 8}%
\special{pa 3400 200}%
\special{pa 3900 200}%
\special{pa 3900 600}%
\special{pa 3400 600}%
\special{pa 3400 200}%
\special{fp}%
% VECTOR 2 0 3 0
% 2 3900 800 4200 800
% 
\special{pn 8}%
\special{pa 3900 400}%
\special{pa 4200 400}%
\special{fp}%
\special{sh 1}%
\special{pa 4200 400}%
\special{pa 4133 380}%
\special{pa 4147 400}%
\special{pa 4133 420}%
\special{pa 4200 400}%
\special{fp}%
% VECTOR 2 0 3 0
% 2 2600 2000 2300 2000
% 
\special{pn 8}%
\special{pa 2600 1600}%
\special{pa 2300 1600}%
\special{fp}%
\special{sh 1}%
\special{pa 2300 1600}%
\special{pa 2367 1620}%
\special{pa 2353 1600}%
\special{pa 2367 1580}%
\special{pa 2300 1600}%
\special{fp}%
% VECTOR 2 0 3 0
% 2 3400 2000 3100 2000
% 
\special{pn 8}%
\special{pa 3400 1600}%
\special{pa 3100 1600}%
\special{fp}%
\special{sh 1}%
\special{pa 3100 1600}%
\special{pa 3167 1620}%
\special{pa 3153 1600}%
\special{pa 3167 1580}%
\special{pa 3100 1600}%
\special{fp}%
% BOX 2 0 3 0
% 2 3400 1800 3900 2200
% 
\special{pn 8}%
\special{pa 3400 1400}%
\special{pa 3900 1400}%
\special{pa 3900 1800}%
\special{pa 3400 1800}%
\special{pa 3400 1400}%
\special{fp}%
% VECTOR 2 0 3 0
% 2 4200 2000 3900 2000
% 
\special{pn 8}%
\special{pa 4200 1600}%
\special{pa 3900 1600}%
\special{fp}%
\special{sh 1}%
\special{pa 3900 1600}%
\special{pa 3967 1620}%
\special{pa 3953 1600}%
\special{pa 3967 1580}%
\special{pa 3900 1600}%
\special{fp}%
% VECTOR 2 0 3 0
% 2 4900 800 4900 1000
% 
\special{pn 8}%
\special{pa 4900 400}%
\special{pa 4900 600}%
\special{fp}%
\special{sh 1}%
\special{pa 4900 600}%
\special{pa 4920 533}%
\special{pa 4900 547}%
\special{pa 4880 533}%
\special{pa 4900 600}%
\special{fp}%
% VECTOR 2 0 3 0
% 2 4900 1400 4900 1600
% 
\special{pn 8}%
\special{pa 4900 1000}%
\special{pa 4900 1200}%
\special{fp}%
\special{sh 1}%
\special{pa 4900 1200}%
\special{pa 4920 1133}%
\special{pa 4900 1147}%
\special{pa 4880 1133}%
\special{pa 4900 1200}%
\special{fp}%
% CIRCLE 2 0 3 0
% 4 4900 1700 4900 1600 4900 1600 4900 1600
% 
\special{pn 8}%
\special{ar 4900 1300 100 100  0.0000000 6.2831853}%
% LINE 2 0 3 0
% 2 4900 1800 4900 2000
% 
\special{pn 8}%
\special{pa 4900 1400}%
\special{pa 4900 1600}%
\special{fp}%
% LINE 2 0 3 0
% 2 4600 800 4900 800
% 
\special{pn 8}%
\special{pa 4600 400}%
\special{pa 4900 400}%
\special{fp}%
% VECTOR 2 0 3 0
% 2 4900 2000 4600 2000
% 
\special{pn 8}%
\special{pa 4900 1600}%
\special{pa 4600 1600}%
\special{fp}%
\special{sh 1}%
\special{pa 4600 1600}%
\special{pa 4667 1620}%
\special{pa 4653 1600}%
\special{pa 4667 1580}%
\special{pa 4600 1600}%
\special{fp}%
% VECTOR 2 0 3 0
% 2 1900 2000 1600 2000
% 
\special{pn 8}%
\special{pa 1900 1600}%
\special{pa 1600 1600}%
\special{fp}%
\special{sh 1}%
\special{pa 1600 1600}%
\special{pa 1667 1620}%
\special{pa 1653 1600}%
\special{pa 1667 1580}%
\special{pa 1600 1600}%
\special{fp}%
% VECTOR 2 0 3 0
% 2 1800 800 2600 800
% 
\special{pn 8}%
\special{pa 1800 400}%
\special{pa 2600 400}%
\special{fp}%
\special{sh 1}%
\special{pa 2600 400}%
\special{pa 2533 380}%
\special{pa 2547 400}%
\special{pa 2533 420}%
\special{pa 2600 400}%
\special{fp}%
% STR 2 0 3 0
% 3 3650 700 3650 800 5 0
% $K_T(z)$
\put(36.5000,-4.0000){\makebox(0,0){$K_T(z)$}}%
% STR 2 0 3 0
% 3 3650 1900 3650 2000 5 0
% $K_R(z)$
\put(36.5000,-16.0000){\makebox(0,0){$K_R(z)$}}%
% STR 2 0 3 0
% 3 4940 1500 4940 1600 2 0
% $+$
\put(49.4000,-12.0000){\makebox(0,0)[lb]{$+$}}%
% STR 2 0 3 0
% 3 4990 1680 4990 1780 1 0
% $+$
\put(49.9000,-13.8000){\makebox(0,0)[lt]{$+$}}%
% STR 2 0 3 0
% 3 5670 1540 5670 1640 2 0
% $n_d$
\put(56.7000,-12.4000){\makebox(0,0)[lb]{$n_d$}}%
% STR 2 0 3 0
% 3 3170 620 3170 720 2 0
% $y_d$
\put(31.7000,-3.2000){\makebox(0,0)[lb]{$y_d$}}%
% STR 2 0 3 0
% 3 4700 620 4700 720 2 0
% $v_d$
\put(47.0000,-3.2000){\makebox(0,0)[lb]{$v_d$}}%
% STR 2 0 3 0
% 3 3170 1820 3170 1920 2 0
% $u_d$
\put(31.7000,-15.2000){\makebox(0,0)[lb]{$u_d$}}%
% STR 2 0 3 0
% 3 2360 1820 2360 1920 2 0
% $u_c$
\put(23.6000,-15.2000){\makebox(0,0)[lb]{$u_c$}}%
% STR 2 0 3 0
% 3 1670 1820 1670 1920 2 0
% $z_c$
\put(16.7000,-15.2000){\makebox(0,0)[lb]{$z_c$}}%
% STR 2 0 3 0
% 3 4800 1960 4800 2060 1 0
% $\nu$
\put(48.0000,-16.6000){\makebox(0,0)[lt]{$\nu$}}%
% VECTOR 2 0 3 0
% 2 4900 800 5200 800
% 
\special{pn 8}%
\special{pa 4900 400}%
\special{pa 5200 400}%
\special{fp}%
\special{sh 1}%
\special{pa 5200 400}%
\special{pa 5133 380}%
\special{pa 5147 400}%
\special{pa 5133 420}%
\special{pa 5200 400}%
\special{fp}%
% BOX 2 0 3 0
% 2 5200 600 5600 1000
% 
\special{pn 8}%
\special{pa 5200 200}%
\special{pa 5600 200}%
\special{pa 5600 600}%
\special{pa 5200 600}%
\special{pa 5200 200}%
\special{fp}%
% VECTOR 2 0 3 0
% 2 5600 800 6000 800
% 
\special{pn 8}%
\special{pa 5600 400}%
\special{pa 6000 400}%
\special{fp}%
\special{sh 1}%
\special{pa 6000 400}%
\special{pa 5933 380}%
\special{pa 5947 400}%
\special{pa 5933 420}%
\special{pa 6000 400}%
\special{fp}%
% BOX 2 0 3 0
% 2 5200 1500 5600 1900
% 
\special{pn 8}%
\special{pa 5200 1100}%
\special{pa 5600 1100}%
\special{pa 5600 1500}%
\special{pa 5200 1500}%
\special{pa 5200 1100}%
\special{fp}%
% VECTOR 2 0 3 0
% 2 5200 1700 5000 1700
% 
\special{pn 8}%
\special{pa 5200 1300}%
\special{pa 5000 1300}%
\special{fp}%
\special{sh 1}%
\special{pa 5000 1300}%
\special{pa 5067 1320}%
\special{pa 5053 1300}%
\special{pa 5067 1280}%
\special{pa 5000 1300}%
\special{fp}%
% VECTOR 2 0 3 0
% 2 6000 1700 5600 1700
% 
\special{pn 8}%
\special{pa 6000 1300}%
\special{pa 5600 1300}%
\special{fp}%
\special{sh 1}%
\special{pa 5600 1300}%
\special{pa 5667 1320}%
\special{pa 5653 1300}%
\special{pa 5667 1280}%
\special{pa 5600 1300}%
\special{fp}%
% STR 2 0 3 0
% 3 5670 620 5670 720 2 0
% $z_d$
\put(56.7000,-3.2000){\makebox(0,0)[lb]{$z_d$}}%
% STR 2 0 3 0
% 3 5400 700 5400 800 5 0
% $W_z$
\put(54.0000,-4.0000){\makebox(0,0){$W_z$}}%
% STR 2 0 3 0
% 3 5400 1600 5400 1700 5 0
% $W_n$
\put(54.0000,-13.0000){\makebox(0,0){$W_n$}}%
\end{picture}%
    \end{center}
    \caption{Signal reconstruction error system}
    \label{fig:comm_design_error}
\end{figure}
\begin{problem}
\label{comm_prob_orig}
	Given a stable, strictly proper, continuous-time $F(s)$, stable, proper, continuous-time $P(s)$, 
	stable proper discrete-time weighting functions $W_n(z)$ and $W_z(z)$,
	stable proper channel model $C(z)$,
	downsampling factor $M$, delay time $L$ and a sampling period $h$,
	find digital filters $K_T(z)$ and $K_R(z)$ that minimizes
	\begin{equation}
		J(K_R,K_T) := \sup_{\substack{w_c\in L^2,n_d\in l^2\\ \|w_c\|_{L^2}+\|n_d\|_{l^2}\neq 0}}
		\frac{\|e_c\|_{L^2}^2+\|z_d\|_{l^2}^2}{\|w_c\|_{L^2}^2+\|n_d\|_{l^2}^2}.
		\label{eq:comm_prob_orig}
	\end{equation}
\end{problem}
\section{Design algorithm}
\subsection{Decomposing design problems}
Problem \ref{comm_prob_orig} is a simultaneous design problem of a
transmitting filter and a receiving filter, and it is difficult to solve this problem directly.
We thus decompose the design problem into two
steps, that is, the design of the receiving filter and that of the
transmitting filter.

Obviously the transmitting filter $K_T$ cannot attenuate the additive
noise $n_d$, hence the receiving filter $K_R$ must play this role.
Moreover $K_R$ must reconstruct the original signal from the corrupted signal
(if $K_R$ did not have to reconstruct, the optimal filter will be clearly $K_R=0$).
Therefore we first design the receiving filter $K_R$ in order to reconstruct
the original signal and to attenuate the noise by the block diagram
shown in Figure \ref{fig:comm_design_error} with $W_z=0$ and with $K_T=1$.
We then design the transmitting filter by the block diagram
shown in Figure \ref{fig:comm_design_error} with $W_n=0$ and with $K_R$ that is obtained 
in the previous design,
that is,
we consider the channel as $K_R(\us{M})C$.

Denote $\T_{R}$ the system from $[w_c,\; n_d]^T$ to $e_c$ (shown in Figure~\ref{fig:comm_err_Kr_syn}),
and $\T_{T}$ the system from $w_c$ to $[e_c,\; z_d]^T$ (shown in Figure~\ref{fig:comm_err_Kt_syn}).
The design procedure is then as follows:\\
{\bf \underline{Step 1} (Design of a receiving filter)}
Find a receiving filter $K_R$ that minimizes 
\begin{equation}
J_1(K_R):=\|\T_{R}\|^2:=
   \sup_{\substack{w_c\in L^2, n_d \in l^2\\ \|w_c\|_{L2}+\|n_d\|_{l^2}\neq 0}}
   \frac{\|e_c\|_{L^2}^2}{\|w_c\|_{L^2}^2+\|n_d\|_{l^2}^2},
\label{eq:prob_Kr}
\end{equation}
with fixed $K_T$.\\
{\bf \underline{Step 2} (Design of a transmitting filter)}
Find a transmitting filter $K_T$ that minimizes 
\begin{equation}
J_2(K_T):=\|\T_{T}\|^2:=
   \sup_{\substack{w_c\in L^2\\ w_c\neq 0}}
   \frac{\|e_c\|^2_{L^2}+\|z_d\|_{l^2}^2}{\|w_c\|_{L^2}^2},
\label{eq:prob_Kt}
\end{equation}
with $K_R$ obtained in the previous step.
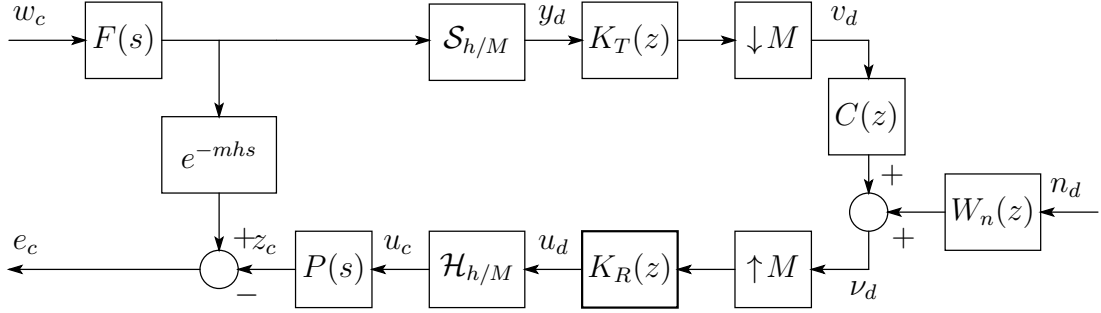
\begin{figure}[t]
    \begin{center}
%	\input{comm_design_error1}
%WinTpicVersion2.15
\unitlength 0.1in
\begin{picture}(57.00,16.50)(4.00,-18.00)
% VECTOR 2 0 3 0
% 2 400 800 800 800
% 
\special{pn 8}%
\special{pa 400 400}%
\special{pa 800 400}%
\special{fp}%
\special{sh 1}%
\special{pa 800 400}%
\special{pa 733 380}%
\special{pa 747 400}%
\special{pa 733 420}%
\special{pa 800 400}%
\special{fp}%
% BOX 2 0 3 0
% 2 800 600 1200 1000
% 
\special{pn 8}%
\special{pa 800 200}%
\special{pa 1200 200}%
\special{pa 1200 600}%
\special{pa 800 600}%
\special{pa 800 200}%
\special{fp}%
% BOX 2 0 3 0
% 2 2600 610 3100 1010
% 
\special{pn 8}%
\special{pa 2600 210}%
\special{pa 3100 210}%
\special{pa 3100 610}%
\special{pa 2600 610}%
\special{pa 2600 210}%
\special{fp}%
% BOX 2 0 3 0
% 2 4200 600 4600 1000
% 
\special{pn 8}%
\special{pa 4200 200}%
\special{pa 4600 200}%
\special{pa 4600 600}%
\special{pa 4200 600}%
\special{pa 4200 200}%
\special{fp}%
% BOX 2 0 3 0
% 2 4690 1000 5090 1400
% 
\special{pn 8}%
\special{pa 4690 600}%
\special{pa 5090 600}%
\special{pa 5090 1000}%
\special{pa 4690 1000}%
\special{pa 4690 600}%
\special{fp}%
% BOX 2 0 3 0
% 2 4600 1800 4200 2200
% 
\special{pn 8}%
\special{pa 4600 1400}%
\special{pa 4200 1400}%
\special{pa 4200 1800}%
\special{pa 4600 1800}%
\special{pa 4600 1400}%
\special{fp}%
% BOX 2 0 3 0
% 2 3100 1800 2600 2200
% 
\special{pn 8}%
\special{pa 3100 1400}%
\special{pa 2600 1400}%
\special{pa 2600 1800}%
\special{pa 3100 1800}%
\special{pa 3100 1400}%
\special{fp}%
% BOX 2 0 3 0
% 2 2300 1800 1900 2200
% 
\special{pn 8}%
\special{pa 2300 1400}%
\special{pa 1900 1400}%
\special{pa 1900 1800}%
\special{pa 2300 1800}%
\special{pa 2300 1400}%
\special{fp}%
% LINE 2 0 3 0
% 2 1200 800 1800 800
% 
\special{pn 8}%
\special{pa 1200 400}%
\special{pa 1800 400}%
\special{fp}%
% VECTOR 2 0 3 0
% 2 1500 800 1500 1200
% 
\special{pn 8}%
\special{pa 1500 400}%
\special{pa 1500 800}%
\special{fp}%
\special{sh 1}%
\special{pa 1500 800}%
\special{pa 1520 733}%
\special{pa 1500 747}%
\special{pa 1480 733}%
\special{pa 1500 800}%
\special{fp}%
% BOX 2 0 3 0
% 2 1200 1200 1800 1600
% 
\special{pn 8}%
\special{pa 1200 800}%
\special{pa 1800 800}%
\special{pa 1800 1200}%
\special{pa 1200 1200}%
\special{pa 1200 800}%
\special{fp}%
% VECTOR 2 0 3 0
% 2 1500 1600 1500 1900
% 
\special{pn 8}%
\special{pa 1500 1200}%
\special{pa 1500 1500}%
\special{fp}%
\special{sh 1}%
\special{pa 1500 1500}%
\special{pa 1520 1433}%
\special{pa 1500 1447}%
\special{pa 1480 1433}%
\special{pa 1500 1500}%
\special{fp}%
% CIRCLE 2 0 3 0
% 4 1500 2000 1400 2000 1400 2000 1400 2000
% 
\special{pn 8}%
\special{ar 1500 1600 100 100  0.0000000 6.2831853}%
% VECTOR 2 0 3 0
% 2 1400 2000 400 2000
% 
\special{pn 8}%
\special{pa 1400 1600}%
\special{pa 400 1600}%
\special{fp}%
\special{sh 1}%
\special{pa 400 1600}%
\special{pa 467 1620}%
\special{pa 453 1600}%
\special{pa 467 1580}%
\special{pa 400 1600}%
\special{fp}%
% STR 2 0 3 0
% 3 1000 700 1000 800 5 0
% $F(s)$
\put(10.0000,-4.0000){\makebox(0,0){$F(s)$}}%
% STR 2 0 3 0
% 3 2850 710 2850 810 5 0
% $\samp{h/M}$
\put(28.5000,-4.1000){\makebox(0,0){$\samp{h/M}$}}%
% STR 2 0 3 0
% 3 2850 1900 2850 2000 5 0
% $\hold{h/M}$
\put(28.5000,-16.0000){\makebox(0,0){$\hold{h/M}$}}%
% STR 2 0 3 0
% 3 1500 1300 1500 1400 5 0
% $e^{-mhs}$
\put(15.0000,-10.0000){\makebox(0,0){$e^{-mhs}$}}%
% STR 2 0 3 0
% 3 4400 700 4400 800 5 0
% $\ds{M}$
\put(44.0000,-4.0000){\makebox(0,0){$\ds{M}$}}%
% STR 2 0 3 0
% 3 4890 1100 4890 1200 5 0
% $C(z)$
\put(48.9000,-8.0000){\makebox(0,0){$C(z)$}}%
% STR 2 0 3 0
% 3 4400 1900 4400 2000 5 0
% $\us{M}$
\put(44.0000,-16.0000){\makebox(0,0){$\us{M}$}}%
% STR 2 0 3 0
% 3 2100 1900 2100 2000 5 0
% $P(s)$
\put(21.0000,-16.0000){\makebox(0,0){$P(s)$}}%
% STR 2 0 3 0
% 3 1570 1790 1570 1890 2 0
% $+$
\put(15.7000,-14.9000){\makebox(0,0)[lb]{$+$}}%
% STR 2 0 3 0
% 3 1580 1980 1580 2080 1 0
% $-$
\put(15.8000,-16.8000){\makebox(0,0)[lt]{$-$}}%
% STR 2 0 3 0
% 3 420 620 420 720 2 0
% $w_c$
\put(4.2000,-3.2000){\makebox(0,0)[lb]{$w_c$}}%
% STR 2 0 3 0
% 3 420 1820 420 1920 2 0
% $e_c$
\put(4.2000,-15.2000){\makebox(0,0)[lb]{$e_c$}}%
% VECTOR 2 0 3 0
% 2 3100 800 3400 800
% 
\special{pn 8}%
\special{pa 3100 400}%
\special{pa 3400 400}%
\special{fp}%
\special{sh 1}%
\special{pa 3400 400}%
\special{pa 3333 380}%
\special{pa 3347 400}%
\special{pa 3333 420}%
\special{pa 3400 400}%
\special{fp}%
% BOX 2 0 3 0
% 2 3400 600 3900 1000
% 
\special{pn 8}%
\special{pa 3400 200}%
\special{pa 3900 200}%
\special{pa 3900 600}%
\special{pa 3400 600}%
\special{pa 3400 200}%
\special{fp}%
% VECTOR 2 0 3 0
% 2 3900 800 4200 800
% 
\special{pn 8}%
\special{pa 3900 400}%
\special{pa 4200 400}%
\special{fp}%
\special{sh 1}%
\special{pa 4200 400}%
\special{pa 4133 380}%
\special{pa 4147 400}%
\special{pa 4133 420}%
\special{pa 4200 400}%
\special{fp}%
% VECTOR 2 0 3 0
% 2 2600 2000 2300 2000
% 
\special{pn 8}%
\special{pa 2600 1600}%
\special{pa 2300 1600}%
\special{fp}%
\special{sh 1}%
\special{pa 2300 1600}%
\special{pa 2367 1620}%
\special{pa 2353 1600}%
\special{pa 2367 1580}%
\special{pa 2300 1600}%
\special{fp}%
% VECTOR 2 0 3 0
% 2 3400 2000 3100 2000
% 
\special{pn 8}%
\special{pa 3400 1600}%
\special{pa 3100 1600}%
\special{fp}%
\special{sh 1}%
\special{pa 3100 1600}%
\special{pa 3167 1620}%
\special{pa 3153 1600}%
\special{pa 3167 1580}%
\special{pa 3100 1600}%
\special{fp}%
% BOX 1 0 3 0
% 2 3400 1800 3900 2200
% 
\special{pn 13}%
\special{pa 3400 1400}%
\special{pa 3900 1400}%
\special{pa 3900 1800}%
\special{pa 3400 1800}%
\special{pa 3400 1400}%
\special{fp}%
% VECTOR 2 0 3 0
% 2 4200 2000 3900 2000
% 
\special{pn 8}%
\special{pa 4200 1600}%
\special{pa 3900 1600}%
\special{fp}%
\special{sh 1}%
\special{pa 3900 1600}%
\special{pa 3967 1620}%
\special{pa 3953 1600}%
\special{pa 3967 1580}%
\special{pa 3900 1600}%
\special{fp}%
% VECTOR 2 0 3 0
% 2 4900 800 4900 1000
% 
\special{pn 8}%
\special{pa 4900 400}%
\special{pa 4900 600}%
\special{fp}%
\special{sh 1}%
\special{pa 4900 600}%
\special{pa 4920 533}%
\special{pa 4900 547}%
\special{pa 4880 533}%
\special{pa 4900 600}%
\special{fp}%
% VECTOR 2 0 3 0
% 2 4900 1400 4900 1600
% 
\special{pn 8}%
\special{pa 4900 1000}%
\special{pa 4900 1200}%
\special{fp}%
\special{sh 1}%
\special{pa 4900 1200}%
\special{pa 4920 1133}%
\special{pa 4900 1147}%
\special{pa 4880 1133}%
\special{pa 4900 1200}%
\special{fp}%
% CIRCLE 2 0 3 0
% 4 4900 1700 4900 1600 4900 1600 4900 1600
% 
\special{pn 8}%
\special{ar 4900 1300 100 100  0.0000000 6.2831853}%
% LINE 2 0 3 0
% 2 4900 1800 4900 2000
% 
\special{pn 8}%
\special{pa 4900 1400}%
\special{pa 4900 1600}%
\special{fp}%
% LINE 2 0 3 0
% 2 4600 800 4900 800
% 
\special{pn 8}%
\special{pa 4600 400}%
\special{pa 4900 400}%
\special{fp}%
% VECTOR 2 0 3 0
% 2 4900 2000 4600 2000
% 
\special{pn 8}%
\special{pa 4900 1600}%
\special{pa 4600 1600}%
\special{fp}%
\special{sh 1}%
\special{pa 4600 1600}%
\special{pa 4667 1620}%
\special{pa 4653 1600}%
\special{pa 4667 1580}%
\special{pa 4600 1600}%
\special{fp}%
% VECTOR 2 0 3 0
% 2 1900 2000 1600 2000
% 
\special{pn 8}%
\special{pa 1900 1600}%
\special{pa 1600 1600}%
\special{fp}%
\special{sh 1}%
\special{pa 1600 1600}%
\special{pa 1667 1620}%
\special{pa 1653 1600}%
\special{pa 1667 1580}%
\special{pa 1600 1600}%
\special{fp}%
% VECTOR 2 0 3 0
% 2 1800 800 2600 800
% 
\special{pn 8}%
\special{pa 1800 400}%
\special{pa 2600 400}%
\special{fp}%
\special{sh 1}%
\special{pa 2600 400}%
\special{pa 2533 380}%
\special{pa 2547 400}%
\special{pa 2533 420}%
\special{pa 2600 400}%
\special{fp}%
% STR 2 0 3 0
% 3 3650 700 3650 800 5 0
% $K_T(z)$
\put(36.5000,-4.0000){\makebox(0,0){$K_T(z)$}}%
% STR 2 0 3 0
% 3 3650 1900 3650 2000 5 0
% $K_R(z)$
\put(36.5000,-16.0000){\makebox(0,0){$K_R(z)$}}%
% STR 2 0 3 0
% 3 4960 1460 4960 1560 2 0
% $+$
\put(49.6000,-11.6000){\makebox(0,0)[lb]{$+$}}%
% STR 2 0 3 0
% 3 5020 1680 5020 1780 1 0
% $+$
\put(50.2000,-13.8000){\makebox(0,0)[lt]{$+$}}%
% STR 2 0 3 0
% 3 5850 1520 5850 1620 2 0
% $n_d$
\put(58.5000,-12.2000){\makebox(0,0)[lb]{$n_d$}}%
% STR 2 0 3 0
% 3 3170 620 3170 720 2 0
% $y_d$
\put(31.7000,-3.2000){\makebox(0,0)[lb]{$y_d$}}%
% STR 2 0 3 0
% 3 4700 620 4700 720 2 0
% $v_d$
\put(47.0000,-3.2000){\makebox(0,0)[lb]{$v_d$}}%
% STR 2 0 3 0
% 3 3170 1820 3170 1920 2 0
% $u_d$
\put(31.7000,-15.2000){\makebox(0,0)[lb]{$u_d$}}%
% STR 2 0 3 0
% 3 2360 1820 2360 1920 2 0
% $u_c$
\put(23.6000,-15.2000){\makebox(0,0)[lb]{$u_c$}}%
% STR 2 0 3 0
% 3 1670 1820 1670 1920 2 0
% $z_c$
\put(16.7000,-15.2000){\makebox(0,0)[lb]{$z_c$}}%
% BOX 2 0 3 0
% 2 5300 1500 5800 1900
% 
\special{pn 8}%
\special{pa 5300 1100}%
\special{pa 5800 1100}%
\special{pa 5800 1500}%
\special{pa 5300 1500}%
\special{pa 5300 1100}%
\special{fp}%
% VECTOR 2 0 3 0
% 2 6100 1700 5800 1700
% 
\special{pn 8}%
\special{pa 6100 1300}%
\special{pa 5800 1300}%
\special{fp}%
\special{sh 1}%
\special{pa 5800 1300}%
\special{pa 5867 1320}%
\special{pa 5853 1300}%
\special{pa 5867 1280}%
\special{pa 5800 1300}%
\special{fp}%
% STR 2 0 3 0
% 3 5550 1600 5550 1700 5 0
% $W_n(z)$
\put(55.5000,-13.0000){\makebox(0,0){$W_n(z)$}}%
% VECTOR 2 0 3 0
% 2 5300 1700 5000 1700
% 
\special{pn 8}%
\special{pa 5300 1300}%
\special{pa 5000 1300}%
\special{fp}%
\special{sh 1}%
\special{pa 5000 1300}%
\special{pa 5067 1320}%
\special{pa 5053 1300}%
\special{pa 5067 1280}%
\special{pa 5000 1300}%
\special{fp}%
% STR 2 0 3 0
% 3 4800 1960 4800 2060 1 0
% $\nu_d$
\put(48.0000,-16.6000){\makebox(0,0)[lt]{$\nu_d$}}%
\end{picture}%
    \end{center}
    \caption{Error system $\T_{R}$ for receiving filter design}
    \label{fig:comm_err_Kr_syn}
\end{figure}
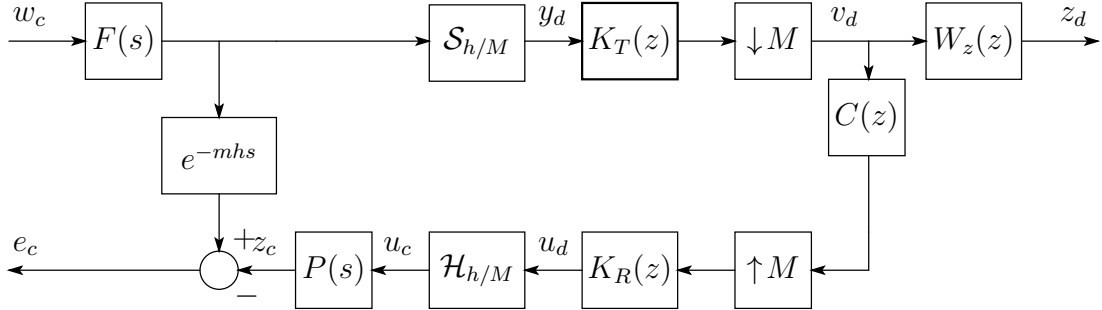
\begin{figure}[t]
    \begin{center}
%	\input{comm_design_error2}
%WinTpicVersion2.15
\unitlength 0.1in
\begin{picture}(57.00,16.50)(4.00,-24.00)
% VECTOR 2 0 3 0
% 2 400 1400 800 1400
% 
\special{pn 8}%
\special{pa 400 1000}%
\special{pa 800 1000}%
\special{fp}%
\special{sh 1}%
\special{pa 800 1000}%
\special{pa 733 980}%
\special{pa 747 1000}%
\special{pa 733 1020}%
\special{pa 800 1000}%
\special{fp}%
% BOX 2 0 3 0
% 2 800 1200 1200 1600
% 
\special{pn 8}%
\special{pa 800 800}%
\special{pa 1200 800}%
\special{pa 1200 1200}%
\special{pa 800 1200}%
\special{pa 800 800}%
\special{fp}%
% BOX 2 0 3 0
% 2 2600 1210 3100 1610
% 
\special{pn 8}%
\special{pa 2600 810}%
\special{pa 3100 810}%
\special{pa 3100 1210}%
\special{pa 2600 1210}%
\special{pa 2600 810}%
\special{fp}%
% BOX 2 0 3 0
% 2 4200 1200 4600 1600
% 
\special{pn 8}%
\special{pa 4200 800}%
\special{pa 4600 800}%
\special{pa 4600 1200}%
\special{pa 4200 1200}%
\special{pa 4200 800}%
\special{fp}%
% BOX 2 0 3 0
% 2 4690 1600 5090 2000
% 
\special{pn 8}%
\special{pa 4690 1200}%
\special{pa 5090 1200}%
\special{pa 5090 1600}%
\special{pa 4690 1600}%
\special{pa 4690 1200}%
\special{fp}%
% BOX 2 0 3 0
% 2 4600 2400 4200 2800
% 
\special{pn 8}%
\special{pa 4600 2000}%
\special{pa 4200 2000}%
\special{pa 4200 2400}%
\special{pa 4600 2400}%
\special{pa 4600 2000}%
\special{fp}%
% BOX 2 0 3 0
% 2 3100 2400 2600 2800
% 
\special{pn 8}%
\special{pa 3100 2000}%
\special{pa 2600 2000}%
\special{pa 2600 2400}%
\special{pa 3100 2400}%
\special{pa 3100 2000}%
\special{fp}%
% BOX 2 0 3 0
% 2 2300 2400 1900 2800
% 
\special{pn 8}%
\special{pa 2300 2000}%
\special{pa 1900 2000}%
\special{pa 1900 2400}%
\special{pa 2300 2400}%
\special{pa 2300 2000}%
\special{fp}%
% LINE 2 0 3 0
% 2 1200 1400 1800 1400
% 
\special{pn 8}%
\special{pa 1200 1000}%
\special{pa 1800 1000}%
\special{fp}%
% VECTOR 2 0 3 0
% 2 1500 1400 1500 1800
% 
\special{pn 8}%
\special{pa 1500 1000}%
\special{pa 1500 1400}%
\special{fp}%
\special{sh 1}%
\special{pa 1500 1400}%
\special{pa 1520 1333}%
\special{pa 1500 1347}%
\special{pa 1480 1333}%
\special{pa 1500 1400}%
\special{fp}%
% BOX 2 0 3 0
% 2 1200 1800 1800 2200
% 
\special{pn 8}%
\special{pa 1200 1400}%
\special{pa 1800 1400}%
\special{pa 1800 1800}%
\special{pa 1200 1800}%
\special{pa 1200 1400}%
\special{fp}%
% VECTOR 2 0 3 0
% 2 1500 2200 1500 2500
% 
\special{pn 8}%
\special{pa 1500 1800}%
\special{pa 1500 2100}%
\special{fp}%
\special{sh 1}%
\special{pa 1500 2100}%
\special{pa 1520 2033}%
\special{pa 1500 2047}%
\special{pa 1480 2033}%
\special{pa 1500 2100}%
\special{fp}%
% CIRCLE 2 0 3 0
% 4 1500 2600 1400 2600 1400 2600 1400 2600
% 
\special{pn 8}%
\special{ar 1500 2200 100 100  0.0000000 6.2831853}%
% VECTOR 2 0 3 0
% 2 1400 2600 400 2600
% 
\special{pn 8}%
\special{pa 1400 2200}%
\special{pa 400 2200}%
\special{fp}%
\special{sh 1}%
\special{pa 400 2200}%
\special{pa 467 2220}%
\special{pa 453 2200}%
\special{pa 467 2180}%
\special{pa 400 2200}%
\special{fp}%
% STR 2 0 3 0
% 3 1000 1300 1000 1400 5 0
% $F(s)$
\put(10.0000,-10.0000){\makebox(0,0){$F(s)$}}%
% STR 2 0 3 0
% 3 2850 1310 2850 1410 5 0
% $\samp{h/M}$
\put(28.5000,-10.1000){\makebox(0,0){$\samp{h/M}$}}%
% STR 2 0 3 0
% 3 2850 2500 2850 2600 5 0
% $\hold{h/M}$
\put(28.5000,-22.0000){\makebox(0,0){$\hold{h/M}$}}%
% STR 2 0 3 0
% 3 1500 1900 1500 2000 5 0
% $e^{-mhs}$
\put(15.0000,-16.0000){\makebox(0,0){$e^{-mhs}$}}%
% STR 2 0 3 0
% 3 4400 1300 4400 1400 5 0
% $\ds{M}$
\put(44.0000,-10.0000){\makebox(0,0){$\ds{M}$}}%
% STR 2 0 3 0
% 3 4890 1700 4890 1800 5 0
% $C(z)$
\put(48.9000,-14.0000){\makebox(0,0){$C(z)$}}%
% STR 2 0 3 0
% 3 4400 2500 4400 2600 5 0
% $\us{M}$
\put(44.0000,-22.0000){\makebox(0,0){$\us{M}$}}%
% STR 2 0 3 0
% 3 2100 2500 2100 2600 5 0
% $P(s)$
\put(21.0000,-22.0000){\makebox(0,0){$P(s)$}}%
% STR 2 0 3 0
% 3 1570 2390 1570 2490 2 0
% $+$
\put(15.7000,-20.9000){\makebox(0,0)[lb]{$+$}}%
% STR 2 0 3 0
% 3 1580 2580 1580 2680 1 0
% $-$
\put(15.8000,-22.8000){\makebox(0,0)[lt]{$-$}}%
% STR 2 0 3 0
% 3 420 1220 420 1320 2 0
% $w_c$
\put(4.2000,-9.2000){\makebox(0,0)[lb]{$w_c$}}%
% STR 2 0 3 0
% 3 420 2420 420 2520 2 0
% $e_c$
\put(4.2000,-21.2000){\makebox(0,0)[lb]{$e_c$}}%
% VECTOR 2 0 3 0
% 2 3100 1400 3400 1400
% 
\special{pn 8}%
\special{pa 3100 1000}%
\special{pa 3400 1000}%
\special{fp}%
\special{sh 1}%
\special{pa 3400 1000}%
\special{pa 3333 980}%
\special{pa 3347 1000}%
\special{pa 3333 1020}%
\special{pa 3400 1000}%
\special{fp}%
% BOX 1 0 3 0
% 2 3400 1200 3900 1600
% 
\special{pn 13}%
\special{pa 3400 800}%
\special{pa 3900 800}%
\special{pa 3900 1200}%
\special{pa 3400 1200}%
\special{pa 3400 800}%
\special{fp}%
% VECTOR 2 0 3 0
% 2 3900 1400 4200 1400
% 
\special{pn 8}%
\special{pa 3900 1000}%
\special{pa 4200 1000}%
\special{fp}%
\special{sh 1}%
\special{pa 4200 1000}%
\special{pa 4133 980}%
\special{pa 4147 1000}%
\special{pa 4133 1020}%
\special{pa 4200 1000}%
\special{fp}%
% VECTOR 2 0 3 0
% 2 2600 2600 2300 2600
% 
\special{pn 8}%
\special{pa 2600 2200}%
\special{pa 2300 2200}%
\special{fp}%
\special{sh 1}%
\special{pa 2300 2200}%
\special{pa 2367 2220}%
\special{pa 2353 2200}%
\special{pa 2367 2180}%
\special{pa 2300 2200}%
\special{fp}%
% VECTOR 2 0 3 0
% 2 3400 2600 3100 2600
% 
\special{pn 8}%
\special{pa 3400 2200}%
\special{pa 3100 2200}%
\special{fp}%
\special{sh 1}%
\special{pa 3100 2200}%
\special{pa 3167 2220}%
\special{pa 3153 2200}%
\special{pa 3167 2180}%
\special{pa 3100 2200}%
\special{fp}%
% BOX 2 0 3 0
% 2 3400 2400 3900 2800
% 
\special{pn 8}%
\special{pa 3400 2000}%
\special{pa 3900 2000}%
\special{pa 3900 2400}%
\special{pa 3400 2400}%
\special{pa 3400 2000}%
\special{fp}%
% VECTOR 2 0 3 0
% 2 4200 2600 3900 2600
% 
\special{pn 8}%
\special{pa 4200 2200}%
\special{pa 3900 2200}%
\special{fp}%
\special{sh 1}%
\special{pa 3900 2200}%
\special{pa 3967 2220}%
\special{pa 3953 2200}%
\special{pa 3967 2180}%
\special{pa 3900 2200}%
\special{fp}%
% VECTOR 2 0 3 0
% 2 4900 1400 4900 1600
% 
\special{pn 8}%
\special{pa 4900 1000}%
\special{pa 4900 1200}%
\special{fp}%
\special{sh 1}%
\special{pa 4900 1200}%
\special{pa 4920 1133}%
\special{pa 4900 1147}%
\special{pa 4880 1133}%
\special{pa 4900 1200}%
\special{fp}%
% LINE 2 0 3 0
% 2 4900 2400 4900 2600
% 
\special{pn 8}%
\special{pa 4900 2000}%
\special{pa 4900 2200}%
\special{fp}%
% LINE 2 0 3 0
% 2 4600 1400 4900 1400
% 
\special{pn 8}%
\special{pa 4600 1000}%
\special{pa 4900 1000}%
\special{fp}%
% VECTOR 2 0 3 0
% 2 4900 2600 4600 2600
% 
\special{pn 8}%
\special{pa 4900 2200}%
\special{pa 4600 2200}%
\special{fp}%
\special{sh 1}%
\special{pa 4600 2200}%
\special{pa 4667 2220}%
\special{pa 4653 2200}%
\special{pa 4667 2180}%
\special{pa 4600 2200}%
\special{fp}%
% VECTOR 2 0 3 0
% 2 1900 2600 1600 2600
% 
\special{pn 8}%
\special{pa 1900 2200}%
\special{pa 1600 2200}%
\special{fp}%
\special{sh 1}%
\special{pa 1600 2200}%
\special{pa 1667 2220}%
\special{pa 1653 2200}%
\special{pa 1667 2180}%
\special{pa 1600 2200}%
\special{fp}%
% VECTOR 2 0 3 0
% 2 1800 1400 2600 1400
% 
\special{pn 8}%
\special{pa 1800 1000}%
\special{pa 2600 1000}%
\special{fp}%
\special{sh 1}%
\special{pa 2600 1000}%
\special{pa 2533 980}%
\special{pa 2547 1000}%
\special{pa 2533 1020}%
\special{pa 2600 1000}%
\special{fp}%
% STR 2 0 3 0
% 3 3650 1300 3650 1400 5 0
% $K_T(z)$
\put(36.5000,-10.0000){\makebox(0,0){$K_T(z)$}}%
% STR 2 0 3 0
% 3 3650 2500 3650 2600 5 0
% $K_R(z)$
\put(36.5000,-22.0000){\makebox(0,0){$K_R(z)$}}%
% STR 2 0 3 0
% 3 3170 1220 3170 1320 2 0
% $y_d$
\put(31.7000,-9.2000){\makebox(0,0)[lb]{$y_d$}}%
% STR 2 0 3 0
% 3 4700 1220 4700 1320 2 0
% $v_d$
\put(47.0000,-9.2000){\makebox(0,0)[lb]{$v_d$}}%
% STR 2 0 3 0
% 3 3170 2420 3170 2520 2 0
% $u_d$
\put(31.7000,-21.2000){\makebox(0,0)[lb]{$u_d$}}%
% STR 2 0 3 0
% 3 2360 2420 2360 2520 2 0
% $u_c$
\put(23.6000,-21.2000){\makebox(0,0)[lb]{$u_c$}}%
% STR 2 0 3 0
% 3 1670 2420 1670 2520 2 0
% $z_c$
\put(16.7000,-21.2000){\makebox(0,0)[lb]{$z_c$}}%
% BOX 2 0 3 0
% 2 5200 1200 5700 1600
% 
\special{pn 8}%
\special{pa 5200 800}%
\special{pa 5700 800}%
\special{pa 5700 1200}%
\special{pa 5200 1200}%
\special{pa 5200 800}%
\special{fp}%
% STR 2 0 3 0
% 3 5450 1300 5450 1400 5 0
% $W_z(z)$
\put(54.5000,-10.0000){\makebox(0,0){$W_z(z)$}}%
% LINE 2 0 3 0
% 2 4900 2000 4900 2400
% 
\special{pn 8}%
\special{pa 4900 1600}%
\special{pa 4900 2000}%
\special{fp}%
% STR 2 0 3 0
% 3 5900 1220 5900 1320 2 0
% $z_d$
\put(59.0000,-9.2000){\makebox(0,0)[lb]{$z_d$}}%
% VECTOR 2 0 3 0
% 2 4900 1400 5200 1400
% 
\special{pn 8}%
\special{pa 4900 1000}%
\special{pa 5200 1000}%
\special{fp}%
\special{sh 1}%
\special{pa 5200 1000}%
\special{pa 5133 980}%
\special{pa 5147 1000}%
\special{pa 5133 1020}%
\special{pa 5200 1000}%
\special{fp}%
% VECTOR 2 0 3 0
% 2 5700 1400 6100 1400
% 
\special{pn 8}%
\special{pa 5700 1000}%
\special{pa 6100 1000}%
\special{fp}%
\special{sh 1}%
\special{pa 6100 1000}%
\special{pa 6033 980}%
\special{pa 6047 1000}%
\special{pa 6033 1020}%
\special{pa 6100 1000}%
\special{fp}%
\end{picture}%
    \end{center}
    \caption{Error system $\T_{T}$ for transmitting filter design}
    \label{fig:comm_err_Kt_syn}
\end{figure}
We iterate Step 1 and Step 2 alternately with initial condition $K_T=1$.

The filter $K_T$ designed in Step 2 cannot have any influence
on the performance of the system from $n_d$ to $e_c$.
Therefore, the objective function $J(K_R,K_T)$ in (\ref{eq:comm_prob_orig})
monotonically decreases by each step.
In fact, we have the following proposition:
\begin{prop}
For any integer $n\geq 1$, the following inequality holds:
\begin{equation*}
J(K_R^{(n-1)}, K_T^{(n-1)})\geq J(K_R^{(n)}, K_T^{(n)})\geq J_{\rm{opt}},
\end{equation*}
where $K_R^{(n)}$ and $K_T^{(n)}$ are filters obtained
at the $n$-th design step, and define
\begin{equation}
J_{\rm{opt}} := \min_{K_R, K_T} J(K_R, K_T).
\label{eq:jopt}
\end{equation}
\end{prop}
\begin{proof}
Let $\T_{1}(K_R,K_T)$ and $\T_{2}(K_R)$ be
the system from $w_c$ to $[e_c, z_d]^T$ and the system 
from $n_d$ to $[e_c,z_d]^T$, respectively,
in Figure \ref{fig:comm_design_error}.
At the $n$-th step of $K_R$ design (Step 1), we have
\begin{equation*}
K_R^{(n)} = \arg\min_{K_R} J_1(K_R)
          = \arg\min_{K_R} J(K_R,\; K_T^{(n-1)}),
\end{equation*}
and hence
\begin{equation}
J(K_R^{(n-1)},\; K_T^{(n-1)}) \geq J(K_R^{(n)},\; K_T^{(n-1)}),
\label{eq:proof_stepA}
\end{equation}
holds. Then at the $n$-th step of $K_T$ design (Step 2), we have
\begin{equation*}
K_T^{(n)} = \arg\min_{K_T} J_2(K_T)
          = \arg\min_{K_T} \|\T_{1}(K_R^{(n)},\; K_T)\|.
\end{equation*}
By using this equation, we have
\begin{equation*}
\|\T_{1}(K_R^{(n)},\; K_T^{(n-1)})\|\geq\|\T_{1}(K_R^{(n)},\; K_T^{(n)})\|.
\end{equation*}
It follows that
\begin{equation*}
\left\|\left[\begin{array}{cc}\T_{1}(K_R^{(n)},\; K_T^{(n-1)}),&\T_{2}(K_R^{(n)})\end{array}\right]\right\|
\geq
\left\|\left[\begin{array}{cc}\T_{1}(K_R^{(n)},\; K_T^{(n)}),&\T_{2}(K_R^{(n)})\end{array}\right]\right\|,
\end{equation*}
holds. This implies
\begin{equation}
J(K_R^{(n)},\; K_T^{(n-1)}) \geq J(K_R^{(n)},\; K_T^{(n)}).
\label{eq:proof_stepB}
\end{equation}
By (\ref{eq:proof_stepA}) and (\ref{eq:proof_stepB}), we have
\begin{equation*}
J(K_R^{(n-1)},\; K_T^{(n-1)})\geq J(K_R^{(n)},\; K_T^{(n)}).
\end{equation*}
Then, by the definition (\ref{eq:jopt}), 
\begin{equation*}
J(K_R^{(n)},\; K_T^{(n)})\geq J_{\rm{opt}},\nonumber
\end{equation*}
is obvious.
\end{proof}

\subsection{Fast-sampling/fast-hold approximation}
The design problems (\ref{eq:prob_Kr}) and (\ref{eq:prob_Kt}) involve a continuous-time delay
component $e^{-Ls}$, and hence they are infinite-dimensional sampled-data problems.
To avoid this difficulty, we employ the fast-sampling/fast-hold approximation method
\cite{KelAnd92,YamMadAnd99}.
%This method approximates continuous-time inputs and outputs via the ideal
%sampler and the zero-order hold that operate in the period $h/N$.
By the method, our design problems (\ref{eq:prob_Kr}) and (\ref{eq:prob_Kt})
are approximated by finite-dimensional discrete-time problems
assuming that the delay time $L$ is $mh$ ($m\in\Nset$).
\begin{theorem}\label{th:FSFH}
Assume that $L=mh$, $m\in\Nset$. Then, 
\begin{enumerate}
\item
for the error system $\T_R$ in Step 1, 
there exist finite-dimensional discrete-time systems $\{T_{R,N}: N=M,2M,\ldots\}$
such that 
     \begin{equation*}
	\lim_{N\to\infty}\|T_{R,N}\| = \|\T_{R}\|,
     \end{equation*}
\item
for the error system $\T_T$ in Step 2,
there exist finite-dimensional discrete-time systems $\{T_{T,N}: N=M,2M,\ldots\}$ 
such that 
\begin{equation*}
	\lim_{N\to\infty}\|T_{T,N}\| = \|\T_{T}\|.
\end{equation*}
\end{enumerate}
\end{theorem}
\begin{proof}
By the fast-sampling/fast-hold method, we approximate continuous-time inputs
and outputs to discrete-time ones
via the ideal  sampler and the zero-order hold that operate in the period $h/N$
(Figure \ref{fig:comm_fsfh_fig}). Assume $N=Ml, l\in\Nset$.
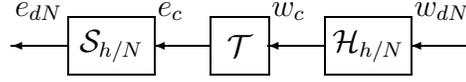
\begin{figure}[t]
\begin{center}
\setlength{\unitlength}{0.45mm}
\begin{picture}( 143,  25)
\put(  16.9,   0.0){\framebox(  25.3,  16.9){$\samp{h/N}$}}
\put( 134.9,   8.4){\vector(-1, 0){  16.9}}
\put(  92.7,   8.4){\vector(-1, 0){  16.9}}
\put(  59.0,   8.4){\vector(-1, 0){  16.9}}
\put(  16.9,   8.4){\vector(-1, 0){  16.9}}
\put(  92.7,   0.0){\framebox(  25.3,  16.9){$\hold{h/N}$}}
\put(  59.0,   0.0){\framebox(  16.9,  16.9){$\T$}}
\put(  43.2,  17.4){$e_c$}
\put(   1.0,  17.4){$e_{dN}$}
\put(  76.9,  17.4){$w_c$}
\put( 119.0,  17.4){$w_{dN}$}
\end{picture}
\end{center}
\caption{Fast-sampling/fast-hold discretization}
\label{fig:comm_fsfh_fig}
\end{figure}
Then apply the discrete-time lifting $\dlift{N}$ (see Section \ref{sec:SDC_fsfh})
to the discretized input/output signal $e_{dN}$ and $w_{dN}$,
we can obtain the lifted signals
\begin{equation*}
	\widetilde{e}_{dN} := \dlift{N}(e_{dN}) = \dlift{N}\samp{h/N}e_c,\quad 
	\widetilde{w}_{dN} := \dlift{N}(w_{dN}) = \dlift{N}\samp{h/N}w_c.
\end{equation*}
%These signals $\widetilde{e}_{dN}$ and $\widetilde{w}_{dN}$ are
%approximations of $e_c$ and $w_c$, respectively.
We next denote $T_{R,N}$ the system from $[\widetilde{w}_{dN}, n_d]^T$ to $\widetilde{e}_{dN}$
and the system $T_{T,N}$ from $\widetilde{w}_{dN}$ to $[\widetilde{e}_{dN}, z_d]^T$,
and define their norm
\begin{equation*}
\begin{split}
	\norm{T_{R,N}}^2&:= \sup_{\substack{\widetilde{w}_{dN},n_d\in l^2\\ \|\widetilde{w}_{dN}\|_{l^2}+\|n_d\|_{l^2}\neq 0}}
	\frac{\|\widetilde{e}_{dN}\|_{l^2}^2}
	     {\|\widetilde{w}_{dN}\|_{l^2}^2 + \frac{N}{h}\|n_d\|_{l^2}^2},\\
	\norm{T_{T,N}}^2&:= \sup_{\substack{\widetilde{w}_{dN}\in l^2\\ \widetilde{w}_{dN}\neq 0}}
	\frac{\|\widetilde{e}_{dN}\|_{l^2}^2 + \frac{N}{h}\|z_d\|_{l^2}^2}{\|\widetilde{w}_{dN}\|_{l^2}^2}.
\end{split}
\end{equation*}
Then we can show
$\|T_{R,N}\| \rightarrow \|\T_{R}\|$,
$\|T_{T,N}\| \rightarrow \|\T_{T}\|$,
as $N\rightarrow\infty$ by using the method as shown in \cite{YamMadAnd99}
under the assumption
$L=mh$.
\end{proof}

Once the problems have been reduced to discrete-time problems, they can be
solved by standard softwares such as MATLAB.
The resulting discrete-time approximant is given by the following:
\begin{theorem}
    The approximated discrete-time systems $T_{R,N}$ and $T_{T,N}$ are given as follows:
\begin{gather*}
\begin{split}
T_{R,N} &:= \lft\left(G_{R,N},\; \widetilde{K}_R\right),\quad
T_{T,N} := \lft\left(G_{T,N},\; \widetilde{K}_T\right),\\
G_{R,N}
&:= \left[\begin{array}{cc}
	\left[\begin{array}{cc}z^{-m}F_{dN},&0\end{array}\right],&-P_{dN}H\\
	\left[\begin{array}{cc}C \widetilde{K}_TSF_{dN},&W_n\end{array}\right],&0
\end{array}\right],\\
G_{T,N}
&:= \left[\begin{array}{cc}
	\left[\begin{array}{c}z^{-m}F_{dN}\\0\end{array}\right],
	&\left[\begin{array}{c}-P_{dN}H\widetilde{K}_RC\\W_z\end{array}\right]\\
	SF_{dN},&0
\end{array}\right],\\
S &:= \left.\left[\begin{array}{ccc}p& & \\ &\ddots & \\ & &p
	\end{array}\right]\right\}M,\quad 
p := [1,\underbrace{0,\ldots,0}_{k-1}],\\
H &:= \underbrace{\left[\begin{array}{ccc}q& & \\ &\ddots & \\ & &q
	\end{array}\right]}_{M},\quad
q := [\underbrace{1,\ldots,1}_{k}]^T,\\
F_{dN} &:= \dliftsys{N}{F},\quad P_{dN}(z):=\dliftsys{N}{P}.
\end{split}
\end{gather*}    
\end{theorem}
\begin{proof}
First, consider the receiving filter design in Figure \ref{fig:comm_err_Kr_syn}.
The relation between the input and the output is as follows:
\begin{align}
e_c &= e^{-mhs}Fw_c - P\hold{h/M}u_d,\label{eq:comm_thrm_r_e}\\
\nu_d &= C(\ds{M})K_T\samp{h/M}Fw_c + W_n n_d,\label{eq:comm_thrm_r_nu}\\
u_d &= K_R(\us{M})\nu_d.
\label{eq:comm_thrm_r_u}
\end{align}
Then applying the fast-sampling and the discrete-time lifting $\dlift{N}$
to the input $w_c$ and the output $w_c$, we have from (\ref{eq:comm_thrm_r_e})
\begin{equation*}
\begin{split}
\dlift{N}\samp{h/N}e_c &= \dlift{N}\samp{h/N}e^{-mhs}F\hold{h/N}\idlift{N}\widetilde{w}_{Nd}
	-\dlift{N}\samp{h/N}P\hold{h/M}u_d\\
&= z^{-m}\dliftsys{N}{F}\widetilde{w}_{Nd} - \dliftsys{N}{P}H\dlift{M}u_d,
\end{split}
\end{equation*}
and from (\ref{eq:comm_thrm_r_nu})
\begin{equation*}
\begin{split}
\nu_d &=C(\ds{M})K_T\samp{h/M}F\hold{h/N}\idlift{N}\widetilde{w}_{Nd} + W_n n_d\\
&= C(\ds{M})K_T\idlift{M}S\dliftsys{N}{F}\widetilde{w} + W_n n_d.
\end{split}
\end{equation*}
In these equations we have used Proposition \ref{prop:dlift}.
Then by using the relation (\ref{eq:dlift_prop2}),
we convert the multirate systems to single rate ones:
\begin{equation*}
\begin{split}
\dlift{M}u_d &= \dlift{M}K_R(\us{M})\nu_d = \dlift{M}K_R\idlift{M}[1, 0, \ldots, 0]^T\nu_d = \widetilde{K}_R\nu_d,\\
C(\ds{M})K_T\idlift{M} &= C[1, 0, \ldots, 0]\dlift{M}K_T\idlift{M} = C\widetilde{K}_T.
\end{split}
\end{equation*}
Consequently, we have
\begin{equation*}
\begin{split}
\widetilde{e}_{dN} &= 
\left[\begin{array}{cc}z^{-m}F_{dN}&0\end{array}\right]
\left[\begin{array}{c}\widetilde{w}_{Nd}\\n_d\end{array}\right] - P_{dN}H\widetilde{u}_d,\\
\nu_d &= 
\left[\begin{array}{cc}C \widetilde{K}_TSF_{dN}&Wn\end{array}\right]
\left[\begin{array}{c}\widetilde{w}_{Nd}\\n_d\end{array}\right],\\
\widetilde{u}_d &= \widetilde{K}_R\nu_d,
\end{split}
\end{equation*}
where $\widetilde{u}_d:=\dlift{M}u_d$.
It implies that $\widetilde{e}_{dN} = \lft(G_{R,N},\widetilde{K}_R)\widetilde{w}_{dN}$.

The proof for $G_{T,N}$ is almost similar to that for $G_{R,N}$ and hence 
we omit the proof.

\end{proof}
Then our design problems (\ref{eq:prob_Kr}) and (\ref{eq:prob_Kt}) are reduced to 
finite-dimensional discrete-time $H^\infty$ design problems, which are
shown in Figure \ref{fig:comm_GRN} and Figure \ref{fig:comm_GTN}.
\begin{figure}[t]
\begin{center}
%\input{comm_GRN}
%\setlength{\unitlength}{0.28mm}
%WinTpicVersion2.15
\unitlength 0.1in
\begin{picture}(16.00,11.20)(2.00,-13.00)
% VECTOR 2 0 3 0
% 2 600 800 200 800
% 
\special{pn 8}%
\special{pa 600 400}%
\special{pa 200 400}%
\special{fp}%
\special{sh 1}%
\special{pa 200 400}%
\special{pa 267 420}%
\special{pa 253 400}%
\special{pa 267 380}%
\special{pa 200 400}%
\special{fp}%
% LINE 2 0 3 0
% 6 1200 1600 1600 1600 1600 1600 1600 1200 1600 1200 1600 1200
% 
\special{pn 8}%
\special{pa 1200 1200}%
\special{pa 1600 1200}%
\special{fp}%
\special{pa 1600 1200}%
\special{pa 1600 800}%
\special{fp}%
\special{pa 1600 800}%
\special{pa 1600 800}%
\special{fp}%
% VECTOR 2 0 3 0
% 2 1600 1200 1400 1200
% 
\special{pn 8}%
\special{pa 1600 800}%
\special{pa 1400 800}%
\special{fp}%
\special{sh 1}%
\special{pa 1400 800}%
\special{pa 1467 820}%
\special{pa 1453 800}%
\special{pa 1467 780}%
\special{pa 1400 800}%
\special{fp}%
% LINE 2 0 3 0
% 2 400 1600 400 1600
% 
\special{pn 8}%
\special{pa 400 1200}%
\special{pa 400 1200}%
\special{fp}%
% VECTOR 2 0 3 0
% 2 400 1600 800 1600
% 
\special{pn 8}%
\special{pa 400 1200}%
\special{pa 800 1200}%
\special{fp}%
\special{sh 1}%
\special{pa 800 1200}%
\special{pa 733 1180}%
\special{pa 747 1200}%
\special{pa 733 1220}%
\special{pa 800 1200}%
\special{fp}%
% LINE 2 0 3 0
% 4 600 1200 400 1200 400 1200 400 1600
% 
\special{pn 8}%
\special{pa 600 800}%
\special{pa 400 800}%
\special{fp}%
\special{pa 400 800}%
\special{pa 400 1200}%
\special{fp}%
% VECTOR 2 0 3 0
% 2 1800 800 1400 800
% 
\special{pn 8}%
\special{pa 1800 400}%
\special{pa 1400 400}%
\special{fp}%
\special{sh 1}%
\special{pa 1400 400}%
\special{pa 1467 420}%
\special{pa 1453 400}%
\special{pa 1467 380}%
\special{pa 1400 400}%
\special{fp}%
% BOX 2 0 3 0
% 2 1400 700 600 1300
% 
\special{pn 8}%
\special{pa 1400 300}%
\special{pa 600 300}%
\special{pa 600 900}%
\special{pa 1400 900}%
\special{pa 1400 300}%
\special{fp}%
% BOX 2 0 3 0
% 2 800 1500 1200 1700
% 
\special{pn 8}%
\special{pa 800 1100}%
\special{pa 1200 1100}%
\special{pa 1200 1300}%
\special{pa 800 1300}%
\special{pa 800 1100}%
\special{fp}%
% STR 2 0 3 0
% 3 1000 900 1000 1000 5 0
% $G_{R,N}$
\put(10.0000,-6.0000){\makebox(0,0){$G_{R,N}$}}%
% STR 2 0 3 0
% 3 1000 1500 1000 1600 5 0
% $K_R$
\put(10.0000,-12.0000){\makebox(0,0){$K_R$}}%
% STR 2 0 3 0
% 3 260 650 260 750 2 0
% $\widetilde{e}_{dN}$
\put(2.6000,-3.5000){\makebox(0,0)[lb]{$\widetilde{e}_{dN}$}}%
% STR 2 0 3 0
% 3 1660 650 1660 750 2 0
% $\widetilde{w}_{dN}$
\put(16.6000,-3.5000){\makebox(0,0)[lb]{$\widetilde{w}_{dN}$}}%
% STR 2 0 3 0
% 3 1460 850 1460 950 2 0
% $n_d$
\put(14.6000,-5.5000){\makebox(0,0)[lb]{$n_d$}}%
% VECTOR 2 0 3 0
% 2 1800 1000 1400 1000
% 
\special{pn 8}%
\special{pa 1800 600}%
\special{pa 1400 600}%
\special{fp}%
\special{sh 1}%
\special{pa 1400 600}%
\special{pa 1467 620}%
\special{pa 1453 600}%
\special{pa 1467 580}%
\special{pa 1400 600}%
\special{fp}%
\end{picture}%
\end{center}
\caption{Discrete-time system for designing receiving filter $K_R$}
\label{fig:comm_GRN}
\end{figure}
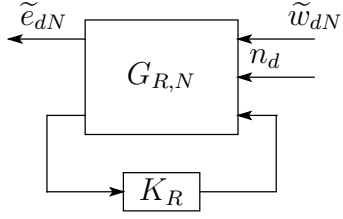
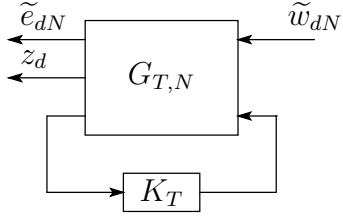
\begin{figure}[t]
\begin{center}
%\input{comm_GTN}
%WinTpicVersion2.15
\unitlength 0.1in
\begin{picture}(16.00,11.20)(2.00,-13.00)
% VECTOR 2 0 3 0
% 2 600 800 200 800
% 
\special{pn 8}%
\special{pa 600 400}%
\special{pa 200 400}%
\special{fp}%
\special{sh 1}%
\special{pa 200 400}%
\special{pa 267 420}%
\special{pa 253 400}%
\special{pa 267 380}%
\special{pa 200 400}%
\special{fp}%
% LINE 2 0 3 0
% 6 1200 1600 1600 1600 1600 1600 1600 1200 1600 1200 1600 1200
% 
\special{pn 8}%
\special{pa 1200 1200}%
\special{pa 1600 1200}%
\special{fp}%
\special{pa 1600 1200}%
\special{pa 1600 800}%
\special{fp}%
\special{pa 1600 800}%
\special{pa 1600 800}%
\special{fp}%
% VECTOR 2 0 3 0
% 2 1600 1200 1400 1200
% 
\special{pn 8}%
\special{pa 1600 800}%
\special{pa 1400 800}%
\special{fp}%
\special{sh 1}%
\special{pa 1400 800}%
\special{pa 1467 820}%
\special{pa 1453 800}%
\special{pa 1467 780}%
\special{pa 1400 800}%
\special{fp}%
% LINE 2 0 3 0
% 2 400 1600 400 1600
% 
\special{pn 8}%
\special{pa 400 1200}%
\special{pa 400 1200}%
\special{fp}%
% VECTOR 2 0 3 0
% 2 400 1600 800 1600
% 
\special{pn 8}%
\special{pa 400 1200}%
\special{pa 800 1200}%
\special{fp}%
\special{sh 1}%
\special{pa 800 1200}%
\special{pa 733 1180}%
\special{pa 747 1200}%
\special{pa 733 1220}%
\special{pa 800 1200}%
\special{fp}%
% LINE 2 0 3 0
% 4 600 1200 400 1200 400 1200 400 1600
% 
\special{pn 8}%
\special{pa 600 800}%
\special{pa 400 800}%
\special{fp}%
\special{pa 400 800}%
\special{pa 400 1200}%
\special{fp}%
% VECTOR 2 0 3 0
% 2 1800 800 1400 800
% 
\special{pn 8}%
\special{pa 1800 400}%
\special{pa 1400 400}%
\special{fp}%
\special{sh 1}%
\special{pa 1400 400}%
\special{pa 1467 420}%
\special{pa 1453 400}%
\special{pa 1467 380}%
\special{pa 1400 400}%
\special{fp}%
% BOX 2 0 3 0
% 2 1400 700 600 1300
% 
\special{pn 8}%
\special{pa 1400 300}%
\special{pa 600 300}%
\special{pa 600 900}%
\special{pa 1400 900}%
\special{pa 1400 300}%
\special{fp}%
% BOX 2 0 3 0
% 2 800 1500 1200 1700
% 
\special{pn 8}%
\special{pa 800 1100}%
\special{pa 1200 1100}%
\special{pa 1200 1300}%
\special{pa 800 1300}%
\special{pa 800 1100}%
\special{fp}%
% STR 2 0 3 0
% 3 1000 900 1000 1000 5 0
% $G_{T,N}$
\put(10.0000,-6.0000){\makebox(0,0){$G_{T,N}$}}%
% STR 2 0 3 0
% 3 1000 1500 1000 1600 5 0
% $K_T$
\put(10.0000,-12.0000){\makebox(0,0){$K_T$}}%
% STR 2 0 3 0
% 3 260 650 260 750 2 0
% $\widetilde{e}_{dN}$
\put(2.6000,-3.5000){\makebox(0,0)[lb]{$\widetilde{e}_{dN}$}}%
% STR 2 0 3 0
% 3 1660 650 1660 750 2 0
% $\widetilde{w}_{dN}$
\put(16.6000,-3.5000){\makebox(0,0)[lb]{$\widetilde{w}_{dN}$}}%
% STR 2 0 3 0
% 3 260 850 260 950 2 0
% $z_d$
\put(2.6000,-5.5000){\makebox(0,0)[lb]{$z_d$}}%
% VECTOR 2 0 3 0
% 2 600 1000 200 1000
% 
\special{pn 8}%
\special{pa 600 600}%
\special{pa 200 600}%
\special{fp}%
\special{sh 1}%
\special{pa 200 600}%
\special{pa 267 620}%
\special{pa 253 600}%
\special{pa 267 580}%
\special{pa 200 600}%
\special{fp}%
\end{picture}%
\end{center}
\caption{Discrete-time system for designing transmitting filter $K_T$}
\label{fig:comm_GTN}
\end{figure}

\section{Design examples}
\subsection{The case of no compression ($M=1$)}
\subsubsection{Design for $W_z=0$}
In this section, we present a design example for 
\begin{gather*}
    F(s) := \frac{1}{10s+1},\quad P(s) := 1,\quad W_n(z):=1,\\
    C(z):=1+0.65z^{-1}-0.52z^{-2}-0.2975z^{-3},\\ 
\end{gather*}
with sampling period $h=1$ and time delay $L=mh=2$.
We here consider the case of no compression (i.e., $M=1$) and
no limitation on transmission (i.e., $W_z=0$).
An approximate design is executed here for $N=8$.
For comparison, the discrete-time $H^\infty$ design \cite{ErdHasKai2000}
is also executed.

\begin{figure}[t]
    \begin{center}
	\includegraphics[width=0.7\linewidth]{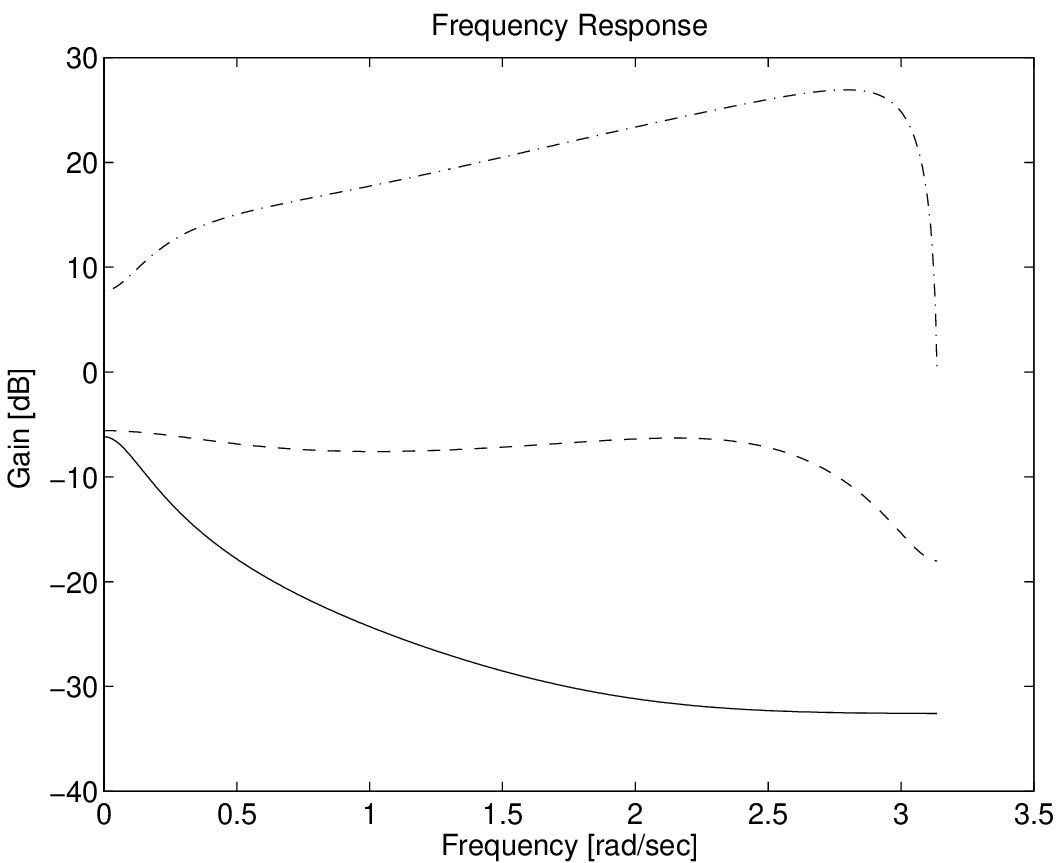}
    \end{center}
    \caption{Gain responses of filters: sampled-data $H^\infty$
    design (transmitting filter: solid, receiving filter: dots) and discrete-time $H^\infty$ design (dash)}
    \label{fig:equalizers}
\end{figure}
\begin{figure}[t]
    \begin{center}
	\includegraphics[width=0.7\linewidth]{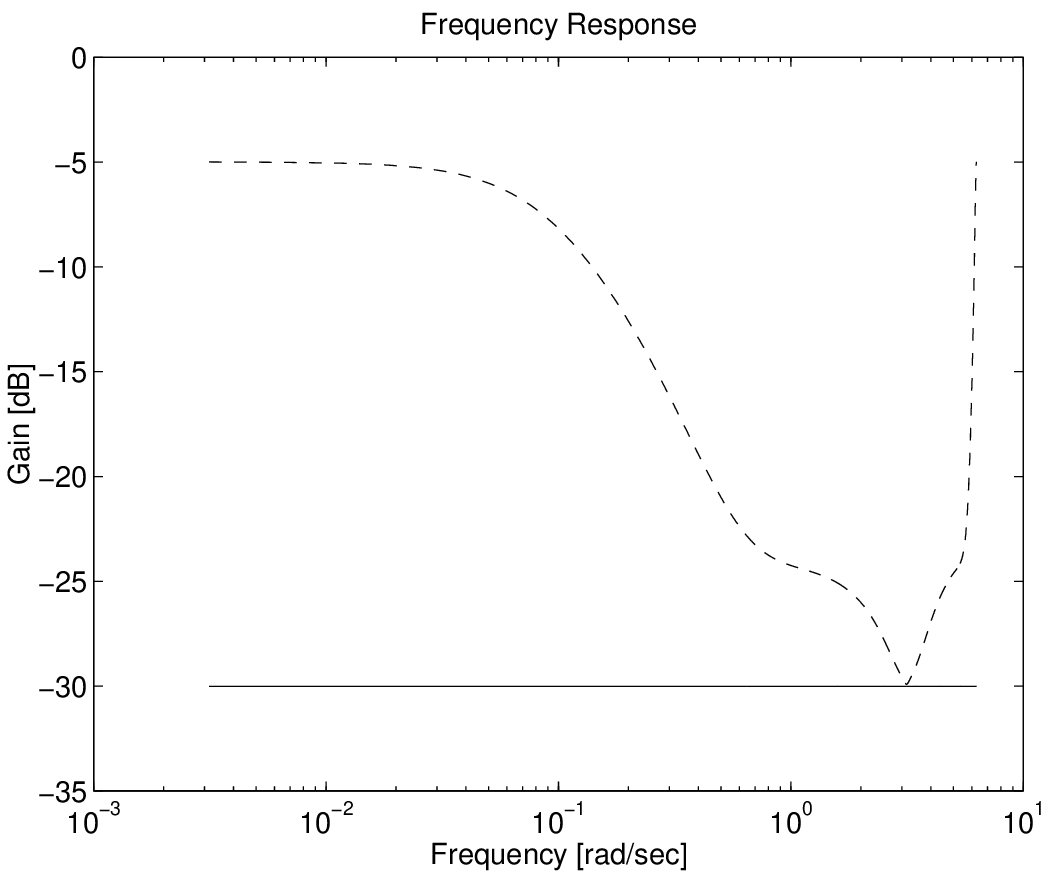}
    \end{center}
    \caption{Frequency responses of $\T_{ew}$: sampled-data $H^\infty$
    design (solid) and discrete-time $H^\infty$ design (dash)}
    \label{fig:Tew}
\end{figure}
\begin{figure}[t]
    \begin{center}
	\includegraphics[width=0.7\linewidth]{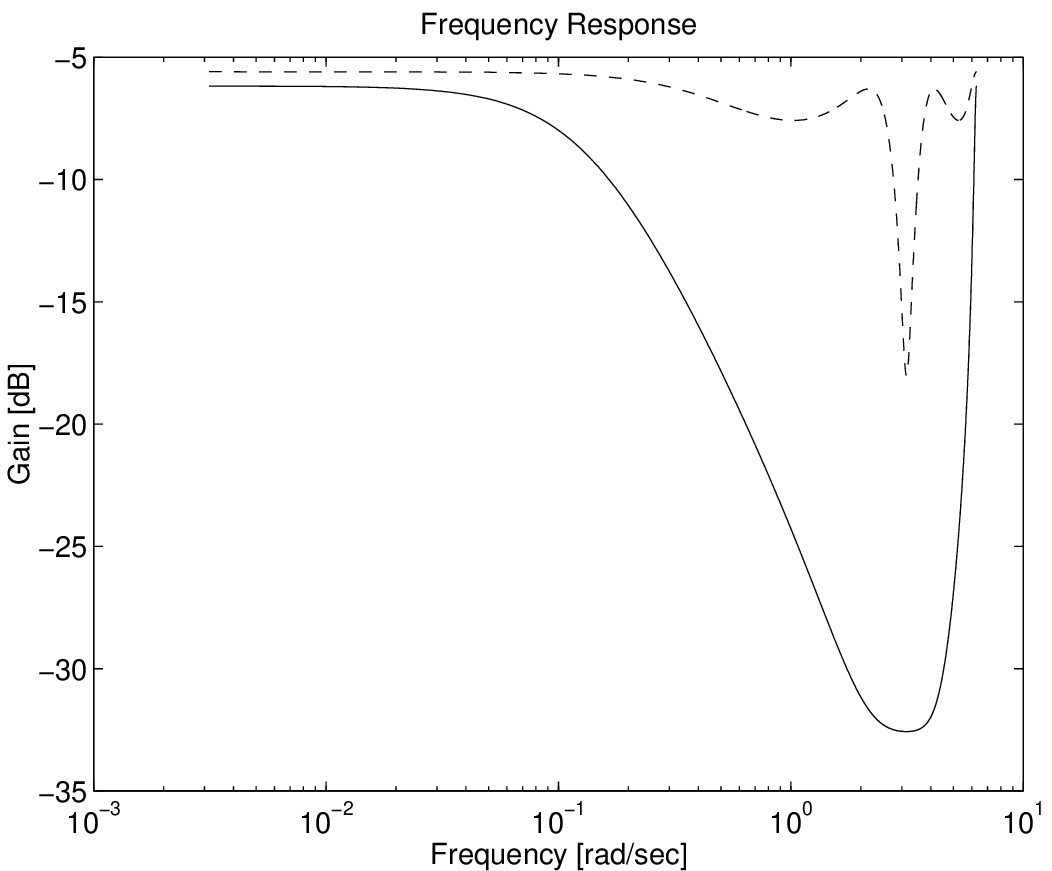}
    \end{center}
    \caption{Frequency responses of $\T_{zn}$: sampled-data $H^\infty$
    design (solid) and discrete-time $H^\infty$ design (dash)}
    \label{fig:Ten}
\end{figure}

Figure \ref{fig:equalizers} shows the gain responses of the filters, and
Figure \ref{fig:Tew} shows the frequency responses of $\T_{ew}$ 
(the system from the input $w_c$ to the error $e_c$), and Figure~\ref{fig:Ten} 
shows those of $\T_{zn}$ (the system from the additive noise $n_d$ to the output $z_c$).
Compared with the discrete-time design, the sampled-data one shows better frequency response
both in $\T_{ew}$ and in $\T_{zn}$.
Moreover, we can say that an equalizer alone cannot attenuate the
corruption caused by the channel and the additive noise,
that is,
we need an appropriate transmitter for transmission.

To explain this fact, we show a simulation of these communication systems.
The input signal $y_c$ is a rectangular wave with amplitude 1,
and the noise disturbance $n_d$ is a discrete-time sinusoid: $n_d[k] = \sin(2k)$.
Figure~\ref{fig:simsd}
shows the output $z_c$ with the receiving filter and the
transmitting filter designed via sampled-data method, and
Figure \ref{fig:sim_dt} shows that with the receiving filter
designed in discrete-time (and without any transmitting filter).
We see that the former shows much better
reconstruction against the noise than the latter.
\begin{figure}[t]
    \begin{center}
	\includegraphics[width=0.7\linewidth]{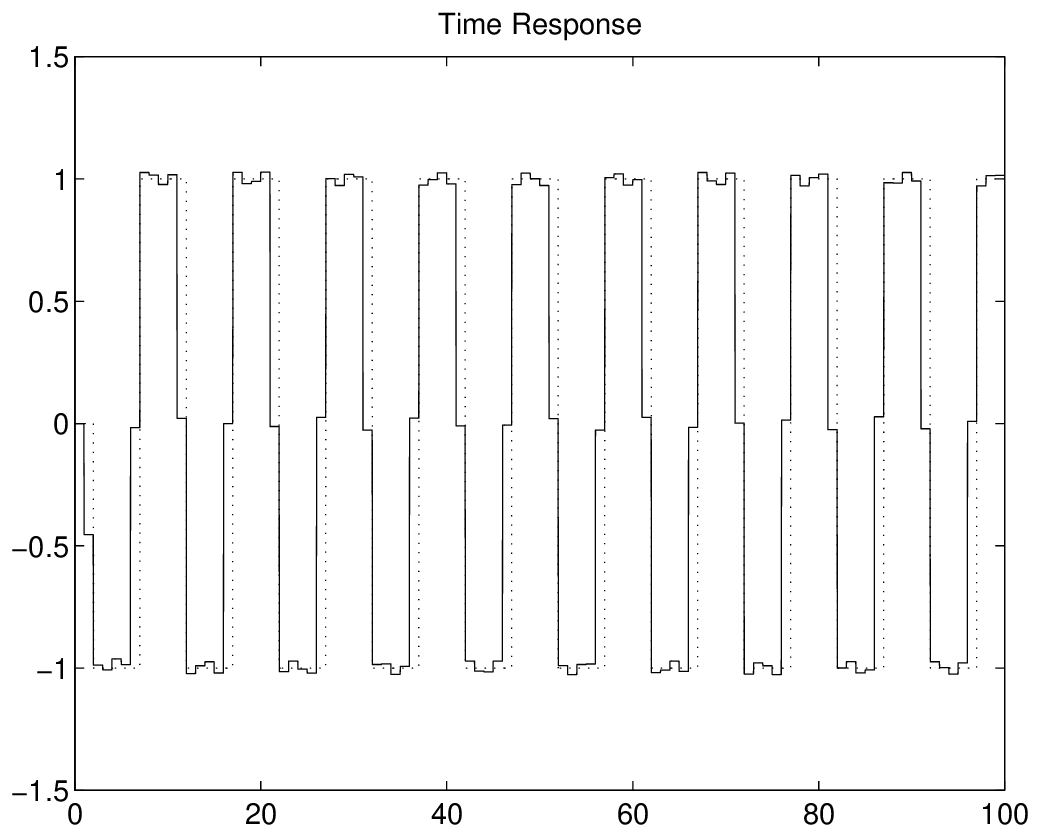}
    \end{center}
    \caption{Time response with sampled-data design}
    \label{fig:simsd}
\end{figure}
\begin{figure}[t]
    \begin{center}
    \includegraphics[width=0.7\linewidth]{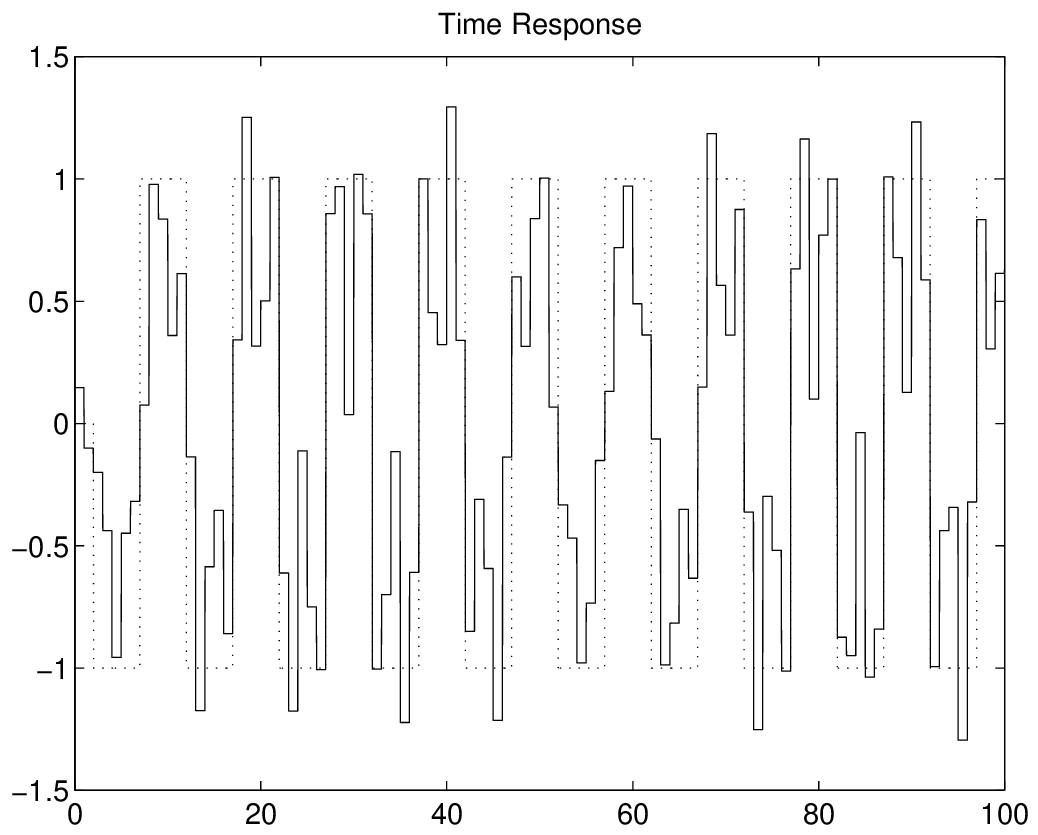}
    \end{center}
    \caption{Time response with discrete-time design}
    \label{fig:sim_dt}
\end{figure}

\subsubsection{Design for $W_z(z)\neq 0$}
We now consider the design with the estimation of the
transmitting signal $v_d$, that is $W_z(z)\neq 0$.

We observe from Figure \ref{fig:equalizers} that the transmitting filter
shows high gains around the Nyquist frequency (i.e. $\omega\approx\pi$), and hence
we take 
\begin{equation*}
W_z(z) = r\cdot\frac{z-1}{z+0.5},
\end{equation*}
as the weighting function of the transmitting signal,
where the parameter $r=0.21$.
The other design parameters are the same as the example above.

Figure \ref{fig:rz_inf_norm} shows the $H^\infty$-norm
of $\T_{ew}$ and $\T_{vw}$ (the system from $w_c$ to
$v_d$ in Figure \ref{fig:comm_design_error}) that varies with 
$r \in [0,5]$.
We can take account of a trade-off between the error attenuation level and 
the amount of the transmitting signal with Figure \ref{fig:rz_inf_norm}.
\begin{figure}[t]
    \begin{center}
	\includegraphics[width=0.7\linewidth]{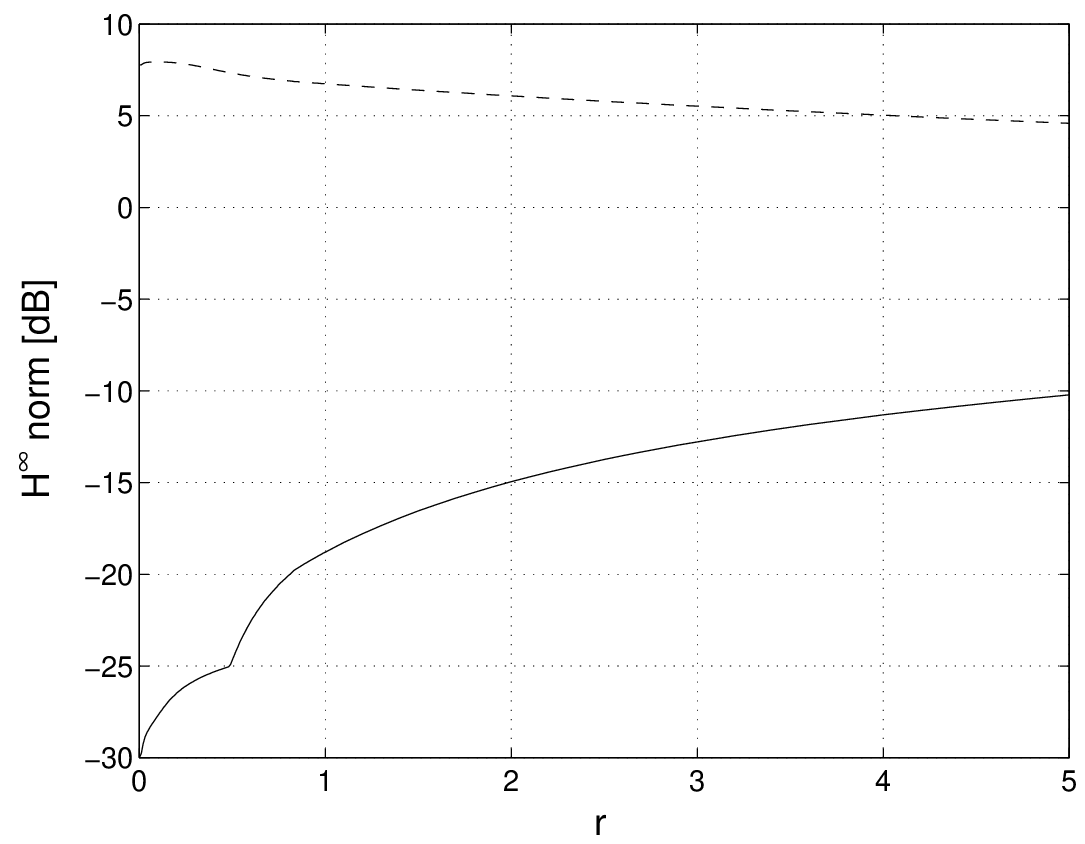}
    \end{center}
    \caption{Relation between $r$ and $\|T_{ew}\|$ (solid),
    $\|T_{vw}\|$ (dash)}
    \label{fig:rz_inf_norm}
\end{figure}
For example, we choose $r=0.21$ in order to attenuate the error less than $-26$dB.
\begin{figure}[t]
    \begin{center}
	\includegraphics[width=0.7\linewidth]{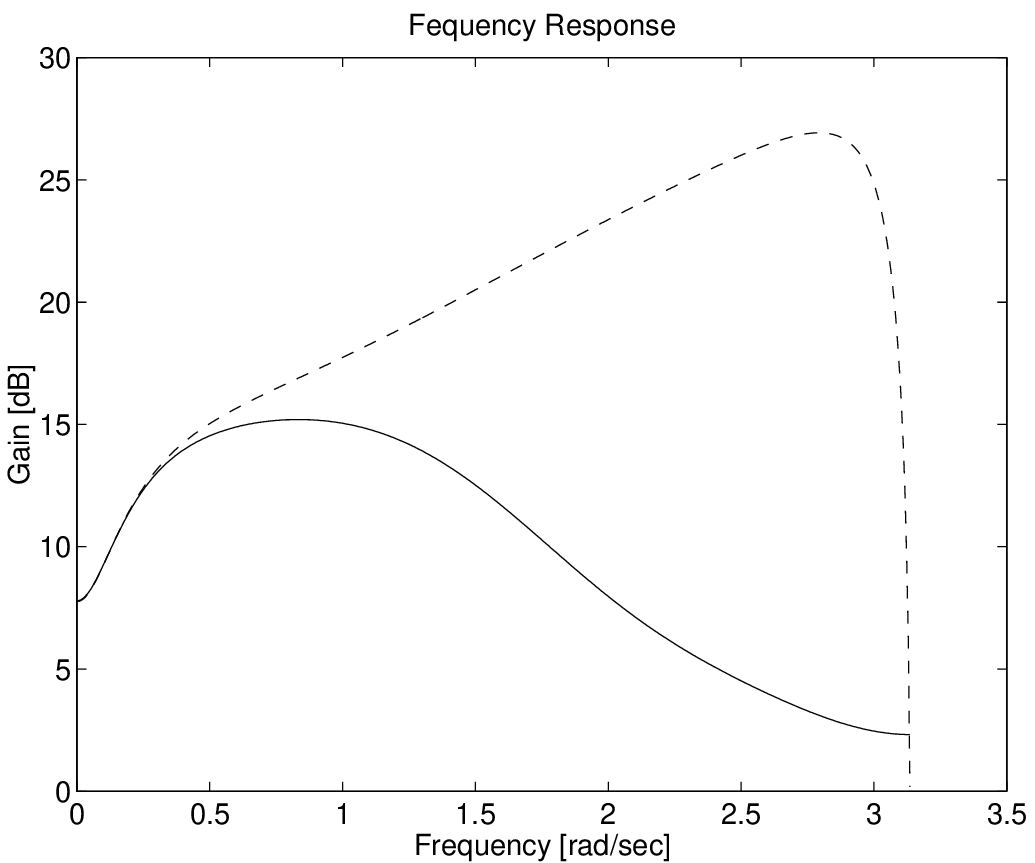}
    \end{center}
    \caption{Gain responses of transmitting filters designed for $r=0.21$ (solid) 
      and $r=0$ (dash)}
    \label{fig:equalizers2}
\end{figure}

Figure \ref{fig:equalizers2} shows the gain responses of transmitting filters designed for 
$r=0$ and $r=0.21$. We can see that the new filter shows better attenuation 
than the filter designed for $r=0$ at high frequency.

Figure \ref{fig:Tew2} shows the frequency responses of the error system $\T_{ew}$.
We see that the attenuation level of $\T_{ew}$ designed for $r=0.21$ is less than $-26$dB.
Figure \ref{fig:Tvw} shows the frequency responses of $\T_{vw}$.
We can see that the amount of the transmitting signal is attenuated at high frequency.

\begin{figure}[t]
    \begin{center}
	\includegraphics[width=0.7\linewidth]{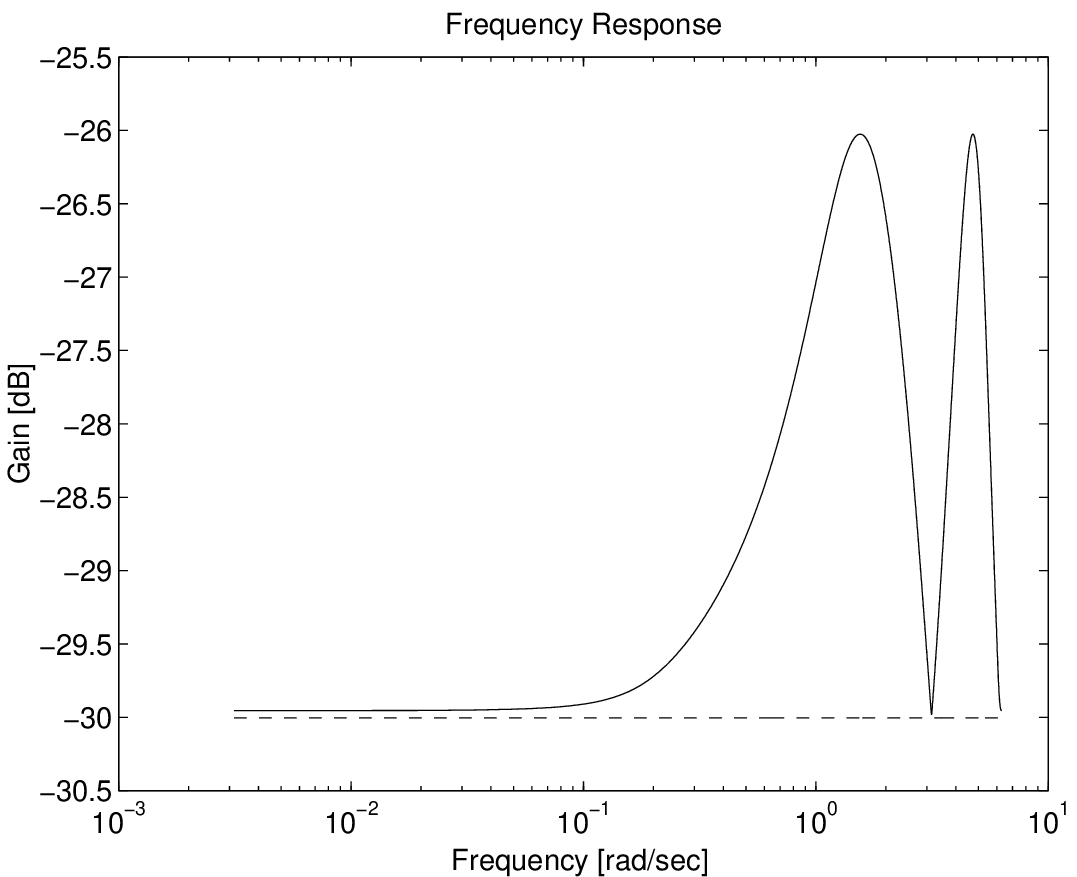}
    \end{center}
    \caption{Frequency responses of $\T_{ew}$ designed for $r=0.21$
    (solid) and $r=0$ (dash)}
    \label{fig:Tew2}
\end{figure}
\begin{figure}[t]
    \begin{center}
	\includegraphics[width=0.7\linewidth]{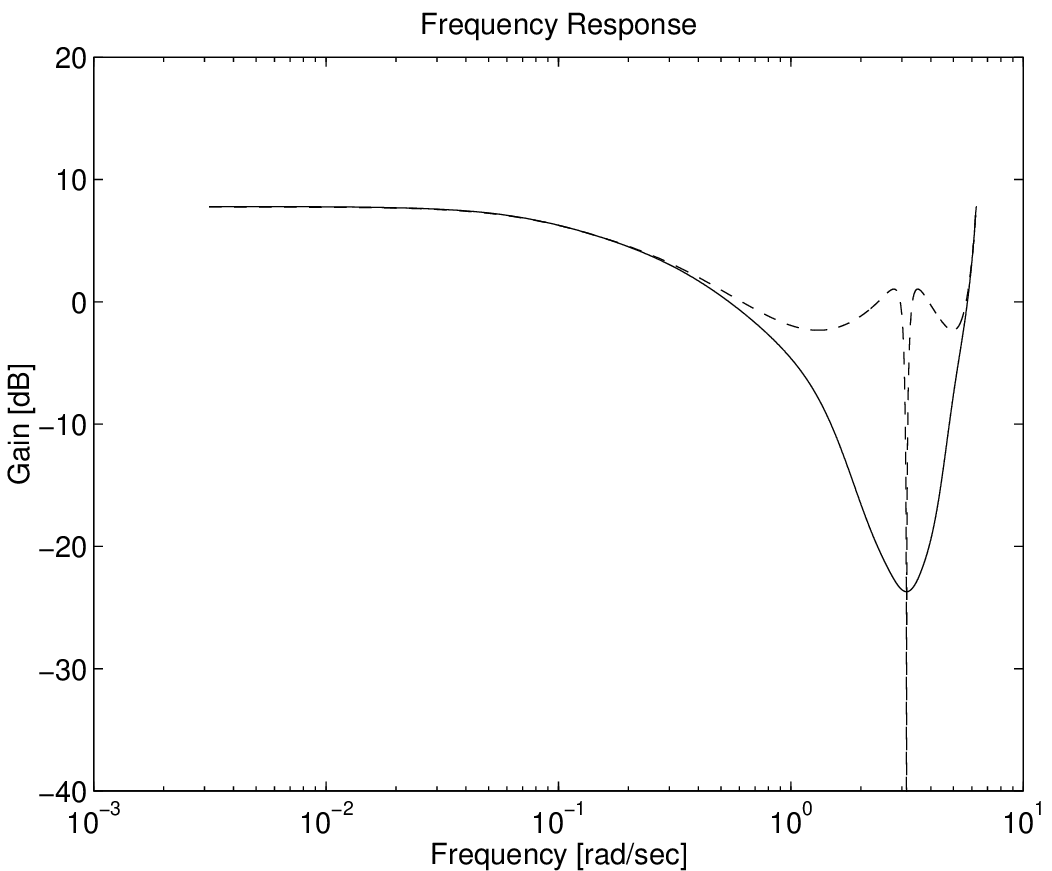}
    \end{center}
    \caption{Frequency responses of $\T_{vw}$ designed for $r=0.21$
    (solid) and $r=0$ (dash)}
    \label{fig:Tvw}
\end{figure}

\subsection{Compression effects}
In this section, we consider compression effects.
First, we take compression with the downsampling factor  $M=1,2,4,8$.
It means that the size of data transmitted is compressed
to $1, 1/2, 1/4, 1/8$, respectively.
For simplicity, we put $W_z=0$.
The other parameters are the same as the previous example.
Denote $\T_{ew}$ the system from $w_c$ to $e_c$ and
$\T_{zn}$ the system from $n_d$ to $e_c$.

We show the frequency responses of $\T_{ew}$ in Figure \ref{fig:Tew_M} and
those of $\T_{zn}$ in Figure \ref{fig:Ten_M} for $M=1,2,4,8$.
\begin{figure}[t]
     \begin{center}
    @\includegraphics[width=.7\linewidth]{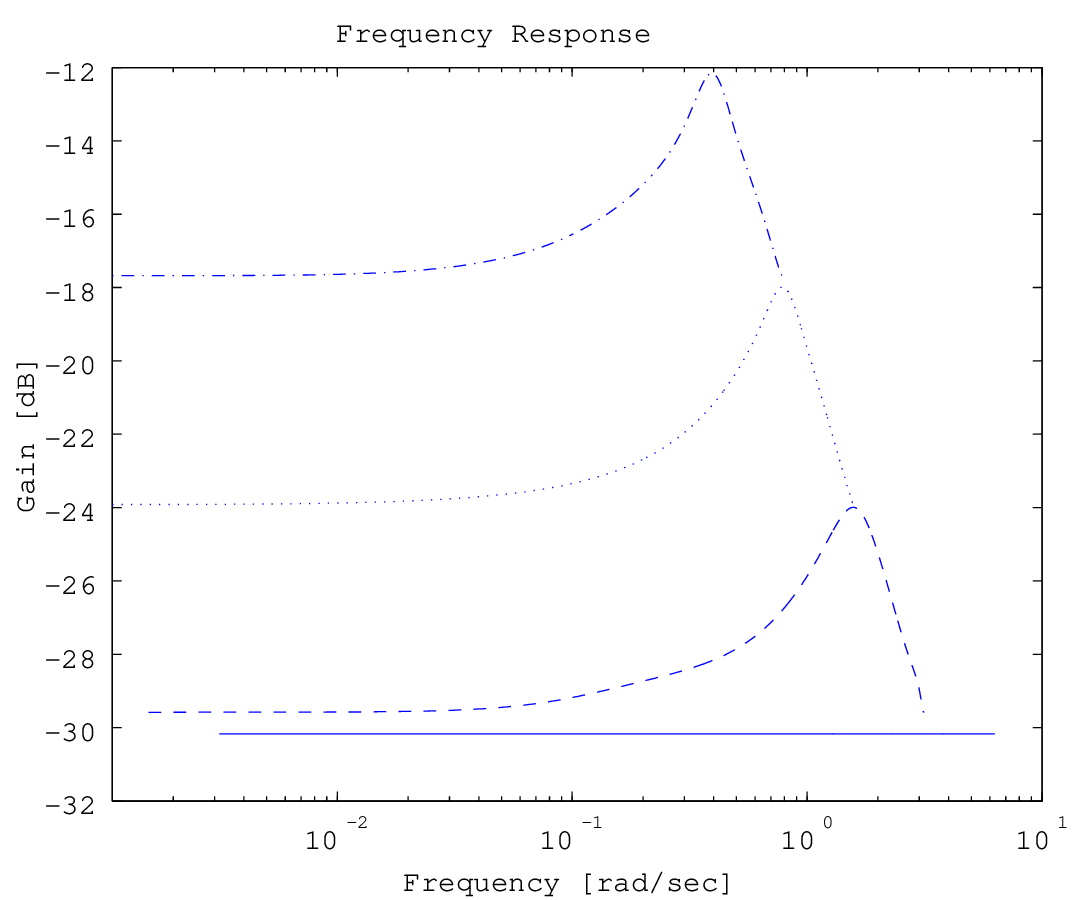}
     \end{center}
     \caption{Frequency responses of $T_{ew}$: compression ratio $M=1$ (solid), $M=2$ (dash), $M=4$ (dot) and $M=8$ (dash-dot)}
     \label{fig:Tew_M} 
\end{figure} 
\begin{figure}[t]
     \begin{center}
    \includegraphics[width=.7\linewidth]{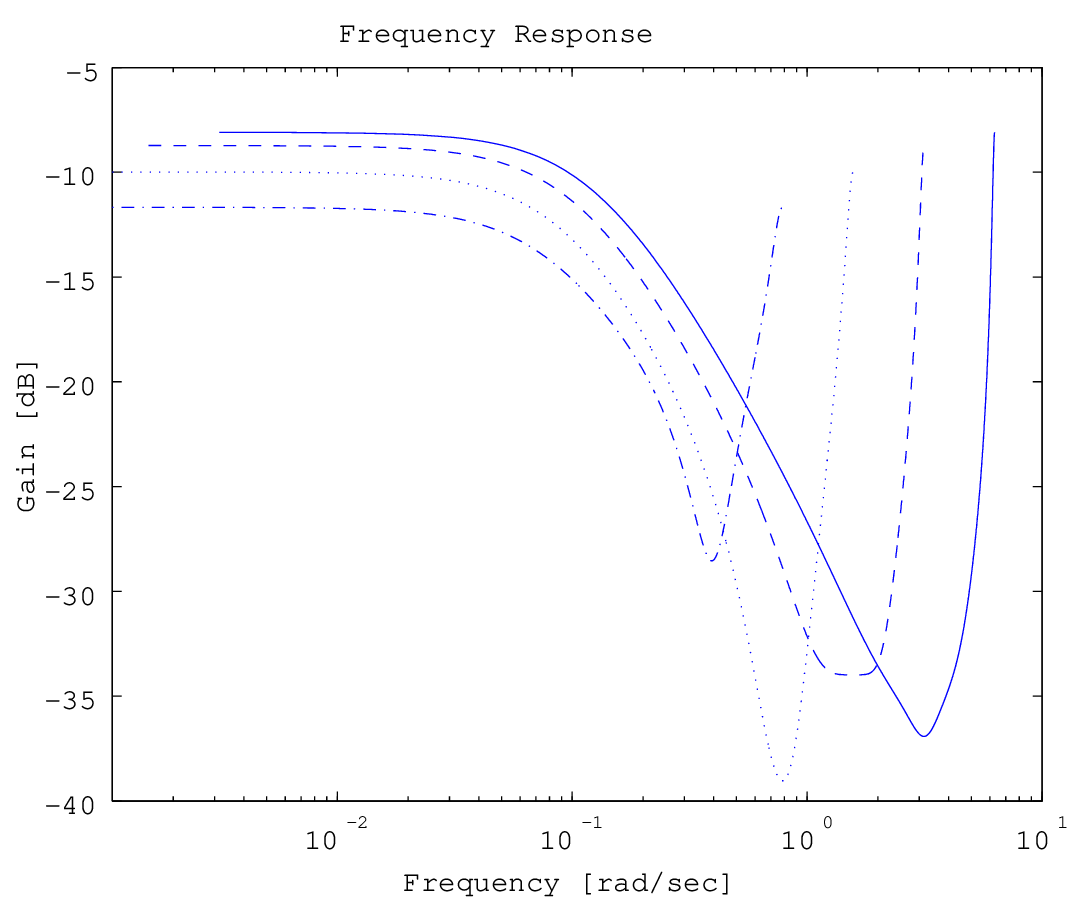}
     \end{center}
     \caption{Frequency responses of $T_{zn}$: compression ratio  $M=1$ (solid), $M=2$ (dash), $M=4$ (dot) and $M=8$ (dash-dot)}
     \label{fig:Ten_M}
\end{figure} 
From Figure \ref{fig:Tew_M}, we can see that when we compress the signal to $1/M$,
the gain of $\|\T_{ew}\|$ (i.e., the maximum gain) increases about $6\times M$ dB.
In particular, the compression of $M=2$ is almost comparative 
with the case of no compression in the low frequency range; the difference is about 0.5dB.
Figure \ref{fig:Ten_M} indicates that the larger the compression ratio $M$, the smaller
$\|T_{zn}\|$.

Next we simulate a transmission with compression ratio $M=4$.
For simplicity, we put $C_d=1$.
The source is a speech signal with the sampling frequency 22.05kHz.
The frequency of the downsampled signal transmitted is 5.5125kHz.

Figure \ref{fig:fft_orig} shows the FFT of the source speech signal.
\begin{figure}[t]
     \begin{center}
    @\includegraphics[width=.7\linewidth]{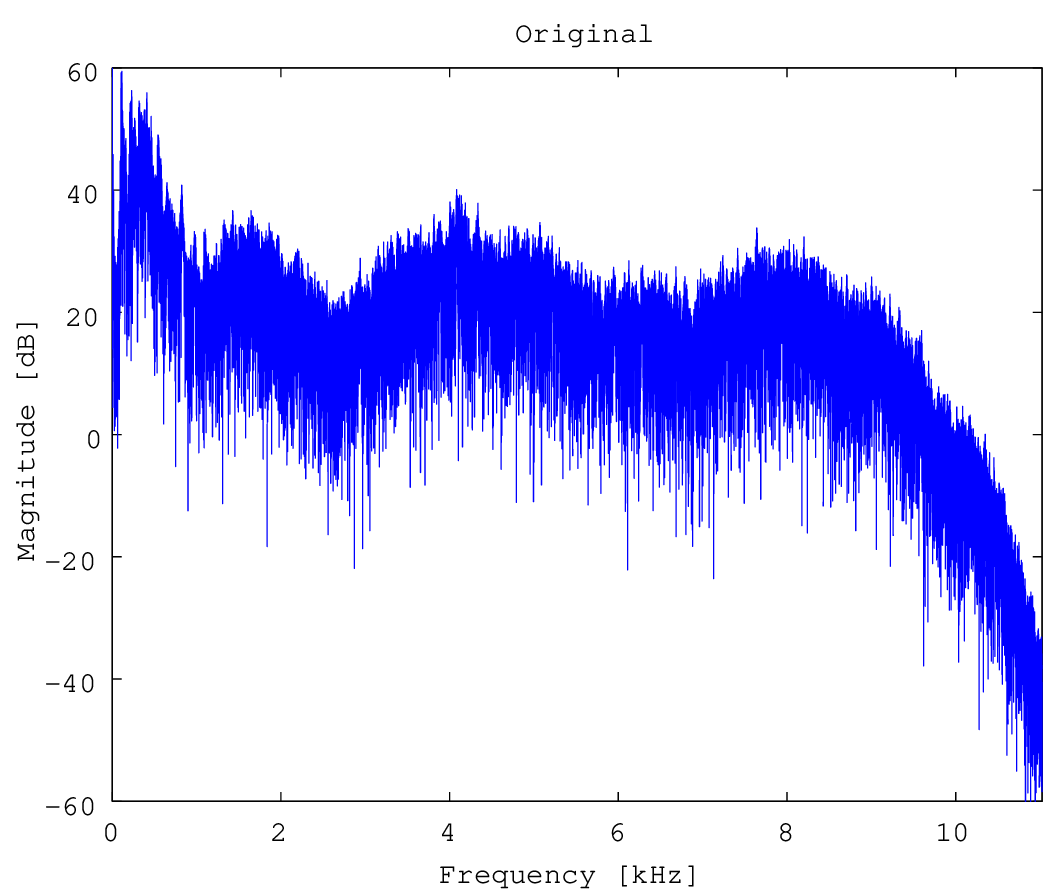}
     \end{center}
     \caption{FFT of the source}
     \label{fig:fft_orig}
\end{figure} 
Then we show the FFT of the reconstructed signal $u_d$ 
(see Figure~\ref{fig:comm_sys}) that is designed by the sampled-data $H^\infty$ optimization.
\begin{figure}[t]
     \begin{center}
    \includegraphics[width=.7\linewidth]{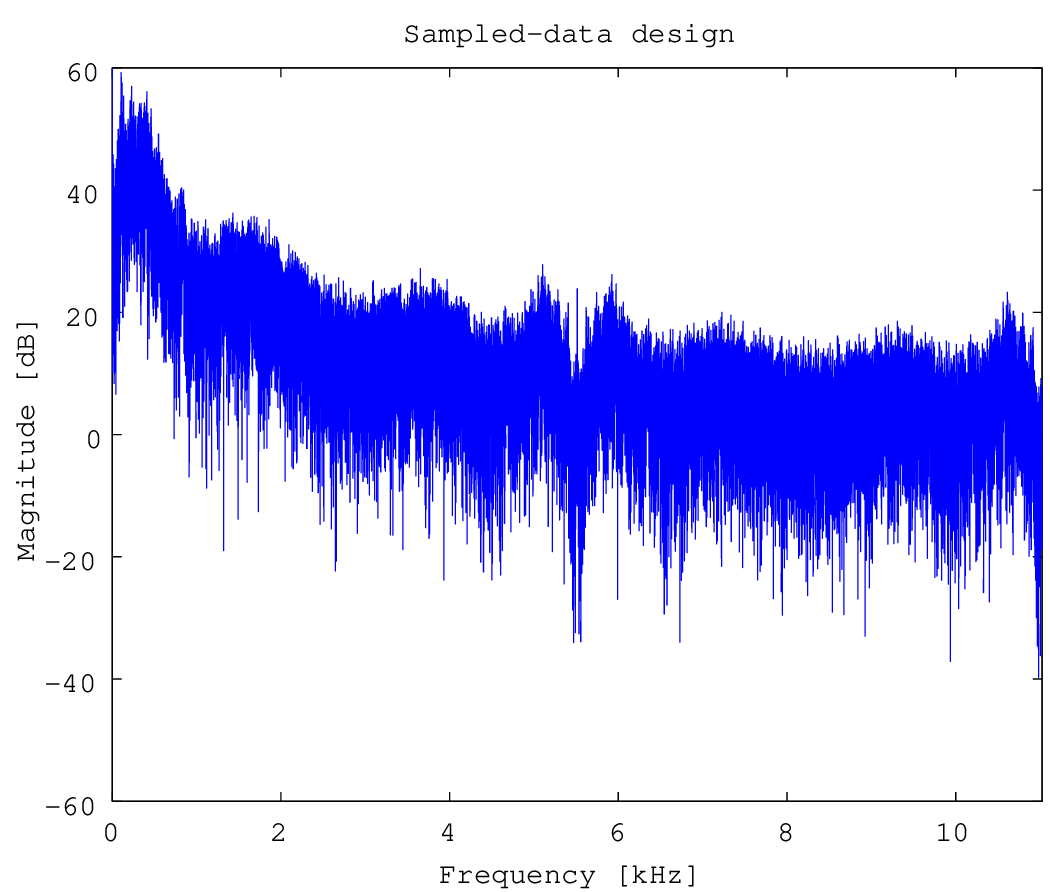}
     \end{center}
     \caption{FFT of the reconstructed signal (sampled-data design)}
     \label{fig:fft_sd}
\end{figure} 
For comparison,  we take the equiripple filter \cite{Fli,Vai,Zel} with
the cut-off frequency $\omega_c=22.05/8\approx 2.76$kHz.
This filter is often used in multirate signal processing.
The FFT of the reconstructed signal is shown in Figure \ref{fig:fft_d}.
\begin{figure}[t]
     \begin{center}
    \includegraphics[width=.7\linewidth]{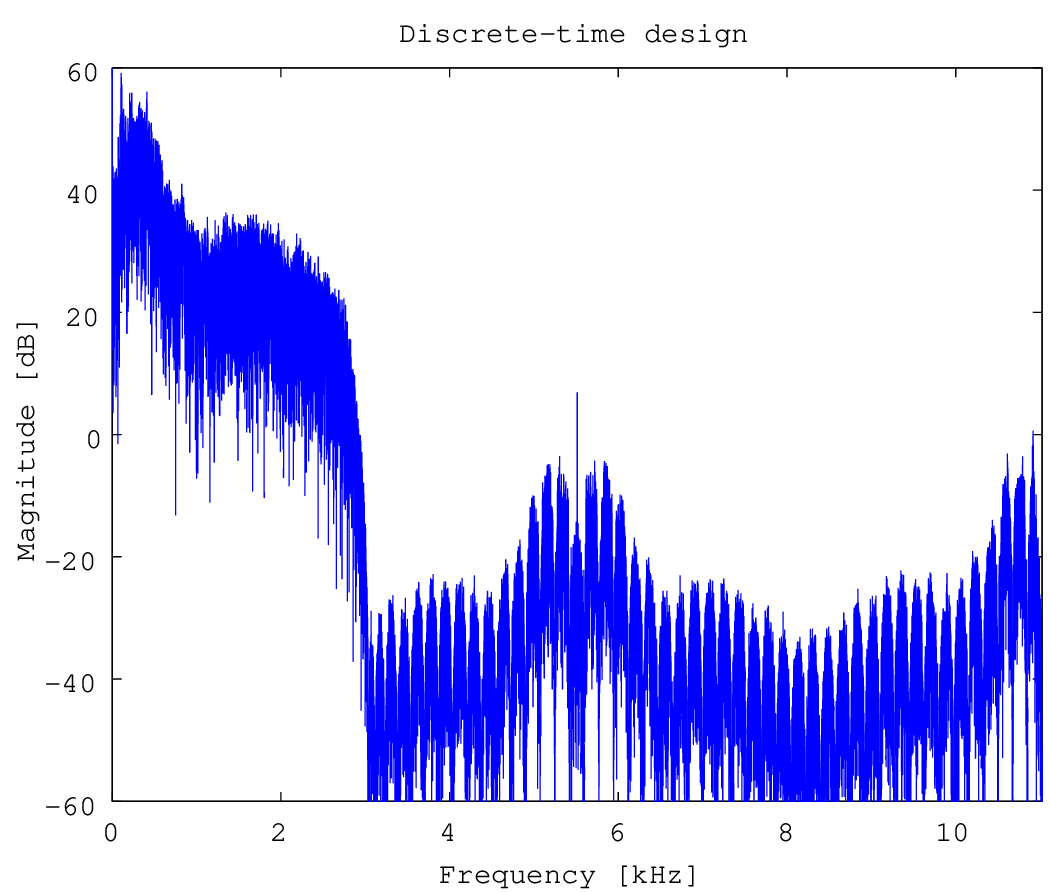}
     \end{center}
     \caption{FFT of the reconstructed signal (equiripple  design)}
     \label{fig:fft_d}
\end{figure} 
We can see from Figure \ref{fig:fft_d} that the FFT of the signal 
processed by the equiripple filter
is sharply cut at the frequency about 2.76kHz,
while that of sampled-data design shows slow decay.
The FFT plot of the source signal shown in Figure \ref{fig:fft_orig}
indicates that the frequency of the original signal does not very decay up to 8kHz,
and we can say that the response by the sampled-data design is better
than that by the conventional design.
In fact, the reconstructed voice by equiripple filter sounds blur because of lack
of high frequency, 
while that by the sampled-data design sounds clearer.
However, the sound by the sampled-data design has high frequency noise.
To reduce the noise, we have to introduce a low-pass filter,
whose design needs further discussion.

\section{Conclusion}
In this chapter, we have treated communication systems
which contains signal compression.
Under distortion by a channel, we have presented a design method
of transmitting/receiving filters
by using sampled-data $H^\infty$ optimization.
By iterating a transmitting filter design and a receiving filter design,
we can obtain sub optimal filters.
We have shown that the objective function monotonically decreases
by the iteration.

In this design, the channel is assumed to be time-invariant.
However, the real channel often contains time-varying systems, in particular, in the case of wireless communication.
Moreover, the real channel is very complicated and we should notice that
the model of the channel always contains a modeling error.
To overcome this, we have to choose an adaptive filter.
Design for adaptive filters by using sampled-data theory is an important subject for the future.
\chapter{Minimization of Quantization Errors}
\label{ch:qqq}
\section{Introduction}
In digital signal processing, digital communications and digital control,
analog signals have to be discretized by an A/D converters to become digital signals.
In discretization, we have two operations; sampling and quantization.
Sampling is discretization in time, and the model is described as
a linear one, which is relatively easy to analyze mathematically
(e.g., via lifting discussed in Chapter \ref{ch:SDC}).
On the other hand, quantization, discretization in amplitude, 
is a nonlinear operation, and its analysis is much more difficult.

For analyzing such a nonlinearity, an additive noise model has been widely used
\cite{DCDS,Del90}. The model is then linear and hence we can easily
analyze or design a system with quantization,
but there has not been much established knowledge on the characteristic of 
the nonlinear system designed by the linearization method.

In this chapter, we first discuss the stability of quantized sampled-data control systems.
Then we investigate how much quantization influences  the performance of a sampled-data system.
We will show that
\begin{itemize}
\item if the linear model is stable, the states of the quantized system are bounded, 
\item if the linear model has a small $L^2$ gain (i.e., the $H^\infty$-norm of the system), 
the quantized system has a small power gain.
\end{itemize}
It follows that the linearization method is valid for analyzing and designing 
systems with quantizations.

Next we apply the above results to the design of a quantizer 
that is called differential pulse code modulation (DPCM) \cite{Pro}.
Although a large number of studies have been made on DPCM systems,
little is known on the stability and the performance,
which we will discuss by means of linearization.

In this chapter, we measure the performance by
the power norm defined as follows:
\begin{equation*}
\begin{split}
\pow(z)^2&:=\lim_{T\rightarrow\infty}\frac{1}{2T}\int_{-T}^{T}|z(t)|^2dt \quad \text{(continuous-time)},\\
\pow[z_d]^2&:=\lim_{N\rightarrow\infty}\frac{1}{2N}\sum_{k=-N}^{N}|z_d[k]|^2 \quad \text{(discrete-time)},
\end{split}
\end{equation*}
and the supremum norm
\begin{equation*}
\|z_d\|_\infty := \sup_{k\in\Zp} |z_d[k]|.
\end{equation*}
\section{Sampled-data control systems with quantization}
\subsection{Additive noise model for quantizer}
We will begin with a sampled-data system with quantization shown in Figure \ref{fig:qsd}.
\begin{figure}[t]
\begin{center}
% \input{qsd_block.tex}
%WinTpicVersion2.15
\unitlength 0.1in
\begin{picture}(21.00,18.55)(2.17,-21.33)
% BOX 2 0 3 0
% 2 817 733 1717 1633
% 
\special{pn 8}%
\special{pa 817 333}%
\special{pa 1717 333}%
\special{pa 1717 1233}%
\special{pa 817 1233}%
\special{pa 817 333}%
\special{fp}%
% VECTOR 2 0 3 0
% 2 817 883 217 883
% 
\special{pn 8}%
\special{pa 817 483}%
\special{pa 217 483}%
\special{fp}%
\special{sh 1}%
\special{pa 217 483}%
\special{pa 284 503}%
\special{pa 270 483}%
\special{pa 284 463}%
\special{pa 217 483}%
\special{fp}%
% BOX 2 0 3 0
% 2 367 1783 667 2083
% 
\special{pn 8}%
\special{pa 367 1383}%
\special{pa 667 1383}%
\special{pa 667 1683}%
\special{pa 367 1683}%
\special{pa 367 1383}%
\special{fp}%
% BOX 2 0 3 0
% 2 817 2233 1117 2533
% 
\special{pn 8}%
\special{pa 817 1833}%
\special{pa 1117 1833}%
\special{pa 1117 2133}%
\special{pa 817 2133}%
\special{pa 817 1833}%
\special{fp}%
% VECTOR 2 0 3 0
% 2 1117 2383 1417 2383
% 
\special{pn 8}%
\special{pa 1117 1983}%
\special{pa 1417 1983}%
\special{fp}%
\special{sh 1}%
\special{pa 1417 1983}%
\special{pa 1350 1963}%
\special{pa 1364 1983}%
\special{pa 1350 2003}%
\special{pa 1417 1983}%
\special{fp}%
% LINE 2 0 3 0
% 2 517 2083 517 2383
% 
\special{pn 8}%
\special{pa 517 1683}%
\special{pa 517 1983}%
\special{fp}%
% VECTOR 2 0 3 0
% 2 517 2383 817 2383
% 
\special{pn 8}%
\special{pa 517 1983}%
\special{pa 817 1983}%
\special{fp}%
\special{sh 1}%
\special{pa 817 1983}%
\special{pa 750 1963}%
\special{pa 764 1983}%
\special{pa 750 2003}%
\special{pa 817 1983}%
\special{fp}%
% LINE 2 0 3 0
% 4 1717 2383 2017 2383 2017 2383 2017 2383
% 
\special{pn 8}%
\special{pa 1717 1983}%
\special{pa 2017 1983}%
\special{fp}%
\special{pa 2017 1983}%
\special{pa 2017 1983}%
\special{fp}%
% VECTOR 2 0 3 0
% 2 2017 2383 2017 2083
% 
\special{pn 8}%
\special{pa 2017 1983}%
\special{pa 2017 1683}%
\special{fp}%
\special{sh 1}%
\special{pa 2017 1683}%
\special{pa 1997 1750}%
\special{pa 2017 1736}%
\special{pa 2037 1750}%
\special{pa 2017 1683}%
\special{fp}%
% LINE 2 0 3 0
% 2 2017 1783 2017 1483
% 
\special{pn 8}%
\special{pa 2017 1383}%
\special{pa 2017 1083}%
\special{fp}%
% VECTOR 2 0 3 0
% 2 2017 1483 1717 1483
% 
\special{pn 8}%
\special{pa 2017 1083}%
\special{pa 1717 1083}%
\special{fp}%
\special{sh 1}%
\special{pa 1717 1083}%
\special{pa 1784 1103}%
\special{pa 1770 1083}%
\special{pa 1784 1063}%
\special{pa 1717 1083}%
\special{fp}%
% LINE 2 0 3 0
% 2 817 1483 517 1483
% 
\special{pn 8}%
\special{pa 817 1083}%
\special{pa 517 1083}%
\special{fp}%
% VECTOR 2 0 3 0
% 2 517 1483 517 1783
% 
\special{pn 8}%
\special{pa 517 1083}%
\special{pa 517 1383}%
\special{fp}%
\special{sh 1}%
\special{pa 517 1383}%
\special{pa 537 1316}%
\special{pa 517 1330}%
\special{pa 497 1316}%
\special{pa 517 1383}%
\special{fp}%
% VECTOR 2 0 3 0
% 2 2317 883 1717 883
% 
\special{pn 8}%
\special{pa 2317 483}%
\special{pa 1717 483}%
\special{fp}%
\special{sh 1}%
\special{pa 1717 483}%
\special{pa 1784 503}%
\special{pa 1770 483}%
\special{pa 1784 463}%
\special{pa 1717 483}%
\special{fp}%
% BOX 2 0 3 0
% 2 1417 2233 1717 2533
% 
\special{pn 8}%
\special{pa 1417 1833}%
\special{pa 1717 1833}%
\special{pa 1717 2133}%
\special{pa 1417 2133}%
\special{pa 1417 1833}%
\special{fp}%
% BOX 2 0 3 0
% 2 1867 1783 2167 2083
% 
\special{pn 8}%
\special{pa 1867 1383}%
\special{pa 2167 1383}%
\special{pa 2167 1683}%
\special{pa 1867 1683}%
\special{pa 1867 1383}%
\special{fp}%
% STR 2 0 3 0
% 3 967 2233 967 2383 5 0
% $Q$
\put(9.6700,-19.8300){\makebox(0,0){$Q$}}%
% STR 2 0 3 0
% 3 1567 2233 1567 2383 5 0
% $K$
\put(15.6700,-19.8300){\makebox(0,0){$K$}}%
% STR 2 0 3 0
% 3 2017 1783 2017 1933 5 0
% $\hold{h}$
\put(20.1700,-15.3300){\makebox(0,0){$\hold{h}$}}%
% STR 2 0 3 0
% 3 517 1783 517 1933 5 0
% $\samp{h}$
\put(5.1700,-15.3300){\makebox(0,0){$\samp{h}$}}%
% STR 2 0 3 0
% 3 1267 1033 1267 1183 5 0
% $G$
\put(12.6700,-7.8300){\makebox(0,0){$G$}}%
% STR 2 0 3 0
% 3 2167 613 2167 763 5 0
% $w$
\put(21.6700,-3.6300){\makebox(0,0){$w$}}%
% STR 2 0 3 0
% 3 367 613 367 763 5 0
% $z$
\put(3.6700,-3.6300){\makebox(0,0){$z$}}%
% STR 2 0 3 0
% 3 367 2113 367 2263 5 0
% $y$
\put(3.6700,-18.6300){\makebox(0,0){$y$}}%
% STR 2 0 3 0
% 3 2167 2113 2167 2263 5 0
% $u$
\put(21.6700,-18.6300){\makebox(0,0){$u$}}%
\end{picture}%
 \caption{Sampled-data control system with quantization}
 \label{fig:qsd}
\end{center}
\end{figure}
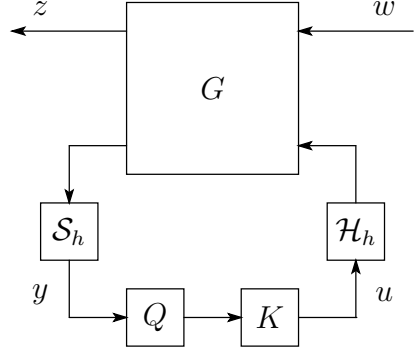
In the figure, $Q$ is a uniform quantizer with a quantization level $\Delta$, and
$K(z)$ is a discrete-time controller. 
Let $\{A_K, B_K, C_K, D_K\}$ be a realization of the controller $K(z)$, and
set
\begin{equation*}
G(s):=\begin{bmatrix}G_{11}(s)&G_{12}(s)\\G_{21}(s)&G_{22}(s)\end{bmatrix}
=:\gss{A}{B_1}{B_2}{C_1}{C_2}{0}{D_{12}}{0}{0}(s).
\end{equation*}

We use the additive noise model for the uniform quantizer;
the quantizer $Q$ is modeled by
\begin{equation*}
Qy = y + d,\quad \|d\|_\infty\leq \Delta/2.
\end{equation*}
Since we have $\|d\|_\infty=\|y-Qy\|_\infty\leq \Delta/2$,
the additive noise model covers the input/output relation of the uniform quantization
with a belt-like region as shown in Figure \ref{fig:quantization}.
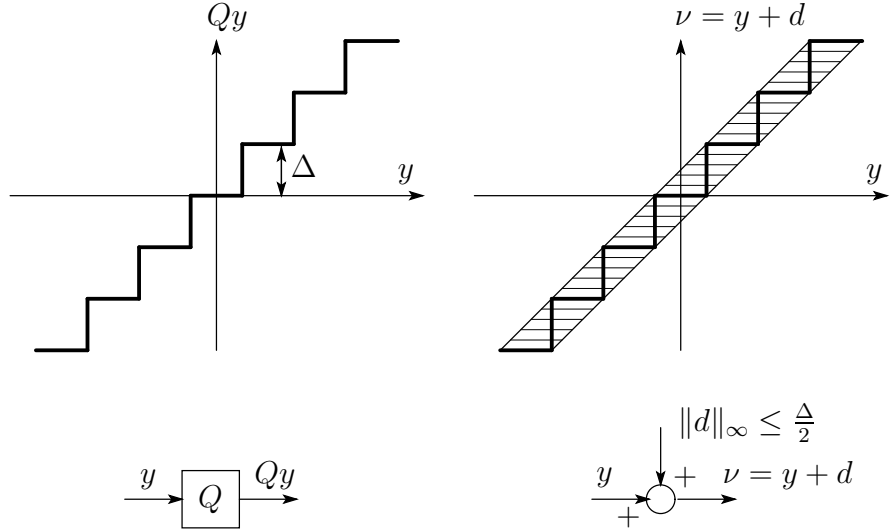
\begin{figure}[t]
\begin{center}
% \input{quantization.tex}
%WinTpicVersion2.15
\unitlength 0.1in
\begin{picture}(45.90,27.70)(4.00,-27.50)
% VECTOR 2 0 3 0
% 2 400 1410 2560 1410
% 
\special{pn 8}%
\special{pa 400 1010}%
\special{pa 2560 1010}%
\special{fp}%
\special{sh 1}%
\special{pa 2560 1010}%
\special{pa 2493 990}%
\special{pa 2507 1010}%
\special{pa 2493 1030}%
\special{pa 2560 1010}%
\special{fp}%
% VECTOR 2 0 3 0
% 2 1480 2220 1480 600
% 
\special{pn 8}%
\special{pa 1480 1820}%
\special{pa 1480 200}%
\special{fp}%
\special{sh 1}%
\special{pa 1480 200}%
\special{pa 1460 267}%
\special{pa 1480 253}%
\special{pa 1500 267}%
\special{pa 1480 200}%
\special{fp}%
% LINE 0 0 3 0
% 28 535 2220 805 2220 805 2220 805 1950 805 1950 1075 1950 1075 1950 1075 1680 1075 1680 1345 1680 1345 1680 1345 1410 1345 1410 1615 1410 1615 1410 1615 1140 1615 1140 1885 1140 1885 1140 1885 870 1885 870 2155 870 2155 870 2155 600 2155 600 2425 600 2425 600 2425 600
% 
\special{pn 20}%
\special{pa 535 1820}%
\special{pa 805 1820}%
\special{fp}%
\special{pa 805 1820}%
\special{pa 805 1550}%
\special{fp}%
\special{pa 805 1550}%
\special{pa 1075 1550}%
\special{fp}%
\special{pa 1075 1550}%
\special{pa 1075 1280}%
\special{fp}%
\special{pa 1075 1280}%
\special{pa 1345 1280}%
\special{fp}%
\special{pa 1345 1280}%
\special{pa 1345 1010}%
\special{fp}%
\special{pa 1345 1010}%
\special{pa 1615 1010}%
\special{fp}%
\special{pa 1615 1010}%
\special{pa 1615 740}%
\special{fp}%
\special{pa 1615 740}%
\special{pa 1885 740}%
\special{fp}%
\special{pa 1885 740}%
\special{pa 1885 470}%
\special{fp}%
\special{pa 1885 470}%
\special{pa 2155 470}%
\special{fp}%
\special{pa 2155 470}%
\special{pa 2155 200}%
\special{fp}%
\special{pa 2155 200}%
\special{pa 2425 200}%
\special{fp}%
\special{pa 2425 200}%
\special{pa 2425 200}%
\special{fp}%
% VECTOR 2 0 3 0
% 2 2830 1410 4990 1410
% 
\special{pn 8}%
\special{pa 2830 1010}%
\special{pa 4990 1010}%
\special{fp}%
\special{sh 1}%
\special{pa 4990 1010}%
\special{pa 4923 990}%
\special{pa 4937 1010}%
\special{pa 4923 1030}%
\special{pa 4990 1010}%
\special{fp}%
% VECTOR 2 0 3 0
% 2 3910 2220 3910 600
% 
\special{pn 8}%
\special{pa 3910 1820}%
\special{pa 3910 200}%
\special{fp}%
\special{sh 1}%
\special{pa 3910 200}%
\special{pa 3890 267}%
\special{pa 3910 253}%
\special{pa 3930 267}%
\special{pa 3910 200}%
\special{fp}%
% LINE 0 0 3 0
% 28 2965 2220 3235 2220 3235 2220 3235 1950 3235 1950 3505 1950 3505 1950 3505 1680 3505 1680 3775 1680 3775 1680 3775 1410 3775 1410 4045 1410 4045 1410 4045 1140 4045 1140 4315 1140 4315 1140 4315 870 4315 870 4585 870 4585 870 4585 600 4585 600 4855 600 4855 600 4855 600
% 
\special{pn 20}%
\special{pa 2965 1820}%
\special{pa 3235 1820}%
\special{fp}%
\special{pa 3235 1820}%
\special{pa 3235 1550}%
\special{fp}%
\special{pa 3235 1550}%
\special{pa 3505 1550}%
\special{fp}%
\special{pa 3505 1550}%
\special{pa 3505 1280}%
\special{fp}%
\special{pa 3505 1280}%
\special{pa 3775 1280}%
\special{fp}%
\special{pa 3775 1280}%
\special{pa 3775 1010}%
\special{fp}%
\special{pa 3775 1010}%
\special{pa 4045 1010}%
\special{fp}%
\special{pa 4045 1010}%
\special{pa 4045 740}%
\special{fp}%
\special{pa 4045 740}%
\special{pa 4315 740}%
\special{fp}%
\special{pa 4315 740}%
\special{pa 4315 470}%
\special{fp}%
\special{pa 4315 470}%
\special{pa 4585 470}%
\special{fp}%
\special{pa 4585 470}%
\special{pa 4585 200}%
\special{fp}%
\special{pa 4585 200}%
\special{pa 4855 200}%
\special{fp}%
\special{pa 4855 200}%
\special{pa 4855 200}%
\special{fp}%
% POLYGON 2 0 3 0
% 6 2965 2220 4585 600 4855 600 3235 2220 3235 2220 2965 2220
% 
\special{pn 8}%
\special{pa 2965 1820}%
\special{pa 4585 200}%
\special{pa 4855 200}%
\special{pa 3235 1820}%
\special{pa 3235 1820}%
\special{pa 2965 1820}%
\special{fp}%
% LINE 3 0 3 0
% 60 4585 600 4855 600 4531 654 4801 654 4477 708 4747 708 4423 762 4693 762 4369 816 4639 816 4315 870 4585 870 4261 924 4531 924 4207 978 4477 978 4153 1032 4423 1032 4099 1086 4369 1086 4045 1140 4315 1140 3991 1194 4261 1194 3937 1248 4207 1248 3883 1302 4153 1302 3829 1356 4099 1356 3775 1410 4045 1410 3721 1464 3991 1464 3667 1518 3937 1518 3613 1572 3883 1572 3559 1626 3829 1626 3505 1680 3775 1680 3451 1734 3721 1734 3397 1788 3667 1788 3343 1842 3613 1842 3289 1896 3559 1896 3235 1950 3505 1950 3181 2004 3451 2004 3127 2058 3397 2058 3073 2112 3343 2112 3019 2166 3289 2166
% 
\special{pn 4}%
\special{pa 4585 200}%
\special{pa 4855 200}%
\special{fp}%
\special{pa 4531 254}%
\special{pa 4801 254}%
\special{fp}%
\special{pa 4477 308}%
\special{pa 4747 308}%
\special{fp}%
\special{pa 4423 362}%
\special{pa 4693 362}%
\special{fp}%
\special{pa 4369 416}%
\special{pa 4639 416}%
\special{fp}%
\special{pa 4315 470}%
\special{pa 4585 470}%
\special{fp}%
\special{pa 4261 524}%
\special{pa 4531 524}%
\special{fp}%
\special{pa 4207 578}%
\special{pa 4477 578}%
\special{fp}%
\special{pa 4153 632}%
\special{pa 4423 632}%
\special{fp}%
\special{pa 4099 686}%
\special{pa 4369 686}%
\special{fp}%
\special{pa 4045 740}%
\special{pa 4315 740}%
\special{fp}%
\special{pa 3991 794}%
\special{pa 4261 794}%
\special{fp}%
\special{pa 3937 848}%
\special{pa 4207 848}%
\special{fp}%
\special{pa 3883 902}%
\special{pa 4153 902}%
\special{fp}%
\special{pa 3829 956}%
\special{pa 4099 956}%
\special{fp}%
\special{pa 3775 1010}%
\special{pa 4045 1010}%
\special{fp}%
\special{pa 3721 1064}%
\special{pa 3991 1064}%
\special{fp}%
\special{pa 3667 1118}%
\special{pa 3937 1118}%
\special{fp}%
\special{pa 3613 1172}%
\special{pa 3883 1172}%
\special{fp}%
\special{pa 3559 1226}%
\special{pa 3829 1226}%
\special{fp}%
\special{pa 3505 1280}%
\special{pa 3775 1280}%
\special{fp}%
\special{pa 3451 1334}%
\special{pa 3721 1334}%
\special{fp}%
\special{pa 3397 1388}%
\special{pa 3667 1388}%
\special{fp}%
\special{pa 3343 1442}%
\special{pa 3613 1442}%
\special{fp}%
\special{pa 3289 1496}%
\special{pa 3559 1496}%
\special{fp}%
\special{pa 3235 1550}%
\special{pa 3505 1550}%
\special{fp}%
\special{pa 3181 1604}%
\special{pa 3451 1604}%
\special{fp}%
\special{pa 3127 1658}%
\special{pa 3397 1658}%
\special{fp}%
\special{pa 3073 1712}%
\special{pa 3343 1712}%
\special{fp}%
\special{pa 3019 1766}%
\special{pa 3289 1766}%
\special{fp}%
% VECTOR 2 0 3 0
% 2 1000 3000 1300 3000
% 
\special{pn 8}%
\special{pa 1000 2600}%
\special{pa 1300 2600}%
\special{fp}%
\special{sh 1}%
\special{pa 1300 2600}%
\special{pa 1233 2580}%
\special{pa 1247 2600}%
\special{pa 1233 2620}%
\special{pa 1300 2600}%
\special{fp}%
% BOX 2 0 3 0
% 2 1300 2850 1600 3150
% 
\special{pn 8}%
\special{pa 1300 2450}%
\special{pa 1600 2450}%
\special{pa 1600 2750}%
\special{pa 1300 2750}%
\special{pa 1300 2450}%
\special{fp}%
% VECTOR 2 0 3 0
% 2 1600 3000 1900 3000
% 
\special{pn 8}%
\special{pa 1600 2600}%
\special{pa 1900 2600}%
\special{fp}%
\special{sh 1}%
\special{pa 1900 2600}%
\special{pa 1833 2580}%
\special{pa 1847 2600}%
\special{pa 1833 2620}%
\special{pa 1900 2600}%
\special{fp}%
% VECTOR 2 0 3 0
% 2 3445 3000 3745 3000
% 
\special{pn 8}%
\special{pa 3445 2600}%
\special{pa 3745 2600}%
\special{fp}%
\special{sh 1}%
\special{pa 3745 2600}%
\special{pa 3678 2580}%
\special{pa 3692 2600}%
\special{pa 3678 2620}%
\special{pa 3745 2600}%
\special{fp}%
% VECTOR 2 0 3 0
% 2 3895 3000 4195 3000
% 
\special{pn 8}%
\special{pa 3895 2600}%
\special{pa 4195 2600}%
\special{fp}%
\special{sh 1}%
\special{pa 4195 2600}%
\special{pa 4128 2580}%
\special{pa 4142 2600}%
\special{pa 4128 2620}%
\special{pa 4195 2600}%
\special{fp}%
% CIRCLE 2 0 3 0
% 4 3805 3000 3730 3000 3730 3000 3730 3000
% 
\special{pn 8}%
\special{ar 3805 2600 75 75  0.0000000 6.2831853}%
% VECTOR 2 0 3 0
% 2 3805 2625 3805 2925
% 
\special{pn 8}%
\special{pa 3805 2225}%
\special{pa 3805 2525}%
\special{fp}%
\special{sh 1}%
\special{pa 3805 2525}%
\special{pa 3825 2458}%
\special{pa 3805 2472}%
\special{pa 3785 2458}%
\special{pa 3805 2525}%
\special{fp}%
% STR 2 0 3 0
% 3 3895 2550 3895 2700 2 0
% $\dd$
\put(38.9500,-23.0000){\makebox(0,0)[lb]{$\dd$}}%
% STR 2 0 3 0
% 3 1450 2850 1450 3000 5 0
% $Q$
\put(14.5000,-26.0000){\makebox(0,0){$Q$}}%
% STR 2 0 3 0
% 3 1080 2850 1080 2950 2 0
% $y$
\put(10.8000,-25.5000){\makebox(0,0)[lb]{$y$}}%
% STR 2 0 3 0
% 3 1680 2850 1680 2950 2 0
% $Qy$
\put(16.8000,-25.5000){\makebox(0,0)[lb]{$Qy$}}%
% STR 2 0 3 0
% 3 2430 1250 2430 1350 2 0
% $y$
\put(24.3000,-9.5000){\makebox(0,0)[lb]{$y$}}%
% STR 2 0 3 0
% 3 1430 450 1430 550 2 0
% $Qy$
\put(14.3000,-1.5000){\makebox(0,0)[lb]{$Qy$}}%
% VECTOR 2 0 3 0
% 2 1820 1150 1820 1410
% 
\special{pn 8}%
\special{pa 1820 750}%
\special{pa 1820 1010}%
\special{fp}%
\special{sh 1}%
\special{pa 1820 1010}%
\special{pa 1840 943}%
\special{pa 1820 957}%
\special{pa 1800 943}%
\special{pa 1820 1010}%
\special{fp}%
% VECTOR 2 0 3 0
% 2 1820 1410 1820 1160
% 
\special{pn 8}%
\special{pa 1820 1010}%
\special{pa 1820 760}%
\special{fp}%
\special{sh 1}%
\special{pa 1820 760}%
\special{pa 1800 827}%
\special{pa 1820 813}%
\special{pa 1840 827}%
\special{pa 1820 760}%
\special{fp}%
% STR 2 0 3 0
% 3 1870 1210 1870 1310 2 0
% $\Delta$
\put(18.7000,-9.1000){\makebox(0,0)[lb]{$\Delta$}}%
% STR 2 0 3 0
% 3 3470 2840 3470 2940 2 0
% $y$
\put(34.7000,-25.4000){\makebox(0,0)[lb]{$y$}}%
% STR 2 0 3 0
% 3 3870 2840 3870 2940 2 0
% $+$
\put(38.7000,-25.4000){\makebox(0,0)[lb]{$+$}}%
% STR 2 0 3 0
% 3 3700 2940 3700 3040 4 0
% $+$
\put(37.0000,-26.4000){\makebox(0,0)[rt]{$+$}}%
% STR 2 0 3 0
% 3 4130 2850 4130 2950 2 0
% $\nu=y+d$
\put(41.3000,-25.5000){\makebox(0,0)[lb]{$\nu=y+d$}}%
% STR 2 0 3 0
% 3 4880 1250 4880 1350 2 0
% $y$
\put(48.8000,-9.5000){\makebox(0,0)[lb]{$y$}}%
% STR 2 0 3 0
% 3 3880 450 3880 550 2 0
% $\nu=y+d$
\put(38.8000,-1.5000){\makebox(0,0)[lb]{$\nu=y+d$}}%
\end{picture}%
 \caption{Uniform quantization (left) and its additive noise model (right)}
 \label{fig:quantization}
\end{center}
\end{figure}
\subsection{Stability of sampled-data systems with quantization}
Let us consider the block diagram shown in Figure \ref{fig:qsd_st_block} 
to check the stability of the closed loop system.
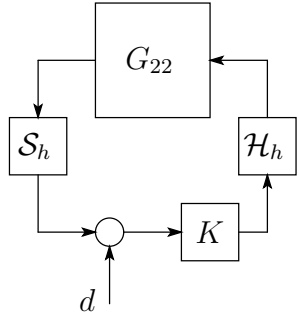
\begin{figure}[t]
\begin{center}
% \input{qsd_st_block.tex}
%WinTpicVersion2.15
\unitlength 0.1in
\begin{picture}(15.01,15.75)(5.22,-17.75)
% BOX 2 0 3 0
% 2 973 600 1573 1200
% 
\special{pn 8}%
\special{pa 973 200}%
\special{pa 1573 200}%
\special{pa 1573 800}%
\special{pa 973 800}%
\special{pa 973 200}%
\special{fp}%
% LINE 2 0 3 0
% 2 973 900 673 900
% 
\special{pn 8}%
\special{pa 973 500}%
\special{pa 673 500}%
\special{fp}%
% VECTOR 2 0 3 0
% 2 673 900 673 1200
% 
\special{pn 8}%
\special{pa 673 500}%
\special{pa 673 800}%
\special{fp}%
\special{sh 1}%
\special{pa 673 800}%
\special{pa 693 733}%
\special{pa 673 747}%
\special{pa 653 733}%
\special{pa 673 800}%
\special{fp}%
% BOX 2 0 3 0
% 2 523 1200 823 1500
% 
\special{pn 8}%
\special{pa 523 800}%
\special{pa 823 800}%
\special{pa 823 1100}%
\special{pa 523 1100}%
\special{pa 523 800}%
\special{fp}%
% LINE 2 0 3 0
% 2 673 1500 673 1800
% 
\special{pn 8}%
\special{pa 673 1100}%
\special{pa 673 1400}%
\special{fp}%
% VECTOR 2 0 3 0
% 2 673 1800 973 1800
% 
\special{pn 8}%
\special{pa 673 1400}%
\special{pa 973 1400}%
\special{fp}%
\special{sh 1}%
\special{pa 973 1400}%
\special{pa 906 1380}%
\special{pa 920 1400}%
\special{pa 906 1420}%
\special{pa 973 1400}%
\special{fp}%
% CIRCLE 2 0 3 0
% 4 1047 1800 973 1800 973 1800 973 1800
% 
\special{pn 8}%
\special{ar 1047 1400 74 74  0.0000000 6.2831853}%
% VECTOR 2 0 3 0
% 2 1123 1800 1423 1800
% 
\special{pn 8}%
\special{pa 1123 1400}%
\special{pa 1423 1400}%
\special{fp}%
\special{sh 1}%
\special{pa 1423 1400}%
\special{pa 1356 1380}%
\special{pa 1370 1400}%
\special{pa 1356 1420}%
\special{pa 1423 1400}%
\special{fp}%
% BOX 2 0 3 0
% 2 1423 1650 1723 1950
% 
\special{pn 8}%
\special{pa 1423 1250}%
\special{pa 1723 1250}%
\special{pa 1723 1550}%
\special{pa 1423 1550}%
\special{pa 1423 1250}%
\special{fp}%
% LINE 2 0 3 0
% 2 1723 1800 1873 1800
% 
\special{pn 8}%
\special{pa 1723 1400}%
\special{pa 1873 1400}%
\special{fp}%
% VECTOR 2 0 3 0
% 2 1873 1800 1873 1500
% 
\special{pn 8}%
\special{pa 1873 1400}%
\special{pa 1873 1100}%
\special{fp}%
\special{sh 1}%
\special{pa 1873 1100}%
\special{pa 1853 1167}%
\special{pa 1873 1153}%
\special{pa 1893 1167}%
\special{pa 1873 1100}%
\special{fp}%
% LINE 2 0 3 0
% 2 1873 1200 1873 900
% 
\special{pn 8}%
\special{pa 1873 800}%
\special{pa 1873 500}%
\special{fp}%
% VECTOR 2 0 3 0
% 2 1873 900 1573 900
% 
\special{pn 8}%
\special{pa 1873 500}%
\special{pa 1573 500}%
\special{fp}%
\special{sh 1}%
\special{pa 1573 500}%
\special{pa 1640 520}%
\special{pa 1626 500}%
\special{pa 1640 480}%
\special{pa 1573 500}%
\special{fp}%
% BOX 2 0 3 0
% 2 1723 1200 2023 1500
% 
\special{pn 8}%
\special{pa 1723 800}%
\special{pa 2023 800}%
\special{pa 2023 1100}%
\special{pa 1723 1100}%
\special{pa 1723 800}%
\special{fp}%
% VECTOR 2 0 3 0
% 2 1047 2175 1047 1875
% 
\special{pn 8}%
\special{pa 1047 1775}%
\special{pa 1047 1475}%
\special{fp}%
\special{sh 1}%
\special{pa 1047 1475}%
\special{pa 1027 1542}%
\special{pa 1047 1528}%
\special{pa 1067 1542}%
\special{pa 1047 1475}%
\special{fp}%
% STR 2 0 3 0
% 3 1257 750 1257 900 5 0
% $G_{22}$
\put(12.5700,-5.0000){\makebox(0,0){$G_{22}$}}%
% STR 2 0 3 0
% 3 657 1200 657 1350 5 0
% $\samp{h}$
\put(6.5700,-9.5000){\makebox(0,0){$\samp{h}$}}%
% STR 2 0 3 0
% 3 1857 1200 1857 1350 5 0
% $\hold{h}$
\put(18.5700,-9.5000){\makebox(0,0){$\hold{h}$}}%
% STR 2 0 3 0
% 3 1573 1650 1573 1800 5 0
% $K$
\put(15.7300,-14.0000){\makebox(0,0){$K$}}%
% STR 2 0 3 0
% 3 973 1950 973 2100 4 0
% $d$
\put(9.7300,-17.0000){\makebox(0,0)[rt]{$d$}}%
\end{picture}%
 \caption{Closed loop system}
 \label{fig:qsd_st_block}
\end{center}
\end{figure}
In the figure,  $\samp{h}G_{22}\hold{h}K$ is a discrete-time system, 
of which let $\{\widetilde{A},\widetilde{B},\widetilde{C},\widetilde{D}\}$ be a realization.
Note that $\widetilde{A},\widetilde{B},\widetilde{C},\widetilde{D}$ are given by
\begin{gather*}
\widetilde{A}:=\begin{bmatrix}A_K&0\\ B_{2d}C_K&A_d\end{bmatrix},\quad
\widetilde{B}:=\begin{bmatrix}B_K\\ B_{2d}D_K\end{bmatrix},\quad
\widetilde{C}:=\begin{bmatrix}0&C_2\end{bmatrix},\quad
\widetilde{D}:=0,
\end{gather*}
where
\begin{equation*}
A_d:=e^{Ah},\quad B_{2d}:=\int_0^h e^{At}B_{2}dt.
\end{equation*}

Then, we have the state-space representation of the system with the additive noise model 
shown in Figure \ref{fig:qsd_st_block} as follows:
\begin{equation}
x[n+1] = (\widetilde{A}+\widetilde{B}\widetilde{C})x[n] + \widetilde{B}d[n],\quad \|d\|_\infty\leq\Delta/2.
\label{eq:st_ss}
\end{equation}
We should notice that without quantization (i.e., $d=0$) the stability of the feedback system 
is equal to the stability of the matrix $\widetilde{A}+\widetilde{B}\widetilde{C}$.
We have the following theorem of the stability of the system in Figure \ref{fig:qsd}.
\begin{theorem}
Assume that $\widetilde{A}+\widetilde{B}\widetilde{C}=:F\in\Real^{n\times n}$ is stable,
and let $\gamma>0$ be a real number such that $r(F)<\gamma<1$, where $r(F)$ denotes the maximum absolute value of
the eigenvalues of $F$.
Then there exists $c>0$ such that
\begin{equation}
\label{eq:theorem_st1}
|x[k] - F^kx_0| \leq c \frac{1-\gamma^k}{1-\gamma}\|\widetilde{B}\|\frac{\Delta}{2} =: r_k,
\end{equation}
for all $x_0:=x[0]\in \Real^n$ and $k \geq 0$,
and 
\begin{equation}
\label{eq:theorem_st2}
\lim_{k\rightarrow\infty}|x[k]| \leq c \frac{1}{1-\gamma}\|\widetilde{B}\|\frac{\Delta}{2}=:r_\infty,
\end{equation}
for all $x_0\in\Real^n$.
\end{theorem}

\begin{proof}
First of all, let $J$ be the Jordan canonical form of $F$. That is, for a nonsingular matrix $P$,
we have
\begin{equation*}
J=P^{-1}FP=\Lambda + U,
\end{equation*}
where
\begin{equation*}
\Lambda := \left[\begin{array}{cccc}
		\lambda_1 & 0        & \ldots & 0\\
		0         &\ddots    & \ddots & \vdots\\
		\vdots    &\ddots    & \ddots & 0\\
		0         &\ldots    & 0      & \lambda_{n}
	\end{array}\right],\quad
U := \left[\begin{array}{ccccc}
		0         & *        & 0      & \ldots  &0\\
		\vdots    &\ddots    & \ddots & \ddots  &\vdots\\
		\vdots    &          & \ddots & \ddots  &0\\
		\vdots    &          &        & \ddots  &*\\
		0         &\ldots    & \ldots & \ldots  &0
	\end{array}\right],\quad *=0 \text{ or } 1,
\end{equation*}
and $\lambda_i$ ($i=1,\ldots,n$) denotes the $i$-th eigenvalue of $F$.
Take $\delta>0$ such that $r(F) + \delta =\gamma < 1$, and define
\begin{equation*}
D := \left[\begin{array}{cccc}
		1         & 0          & \ldots & 0\\
		0         &\delta^{-1} & \ddots & \vdots\\
		\vdots    &\ddots      & \ddots & 0\\
		0         &\ldots      & 0      & \delta^{-1+n}
	\end{array}\right].
\end{equation*}
Then we have
\begin{equation*}
D^{-1}JD = \Lambda + \delta U.
\end{equation*}
Put $T:=DP$. Then for any $x\in\Real^n$,
\begin{equation*}
|T^{-1}FTx| = |(\Lambda + \delta U)x|
\leq (r(F)+\delta)|x| = \gamma|x|.
\end{equation*}
It follows that
\begin{equation*}
\begin{split}
|x[k]-F^kx_0| &\leq \sum_{l=1}^{k} |F^{k-l}\widetilde{B}d[l]|\\
&=\sum_{l=1}^{k} |T(T^{-1}FT)^{k-l}T^{-1}\widetilde{B}d[l]|\\
&\leq \sum_{l=1}^{k} \|T\|\cdot\gamma^{k-l}\cdot\|T^{-1}\|\|\widetilde{B}\|\frac{\Delta}{2}\\
&=\|T\|\|T^{-1}\|\frac{1-\gamma^{k}}{1-\gamma}\|\widetilde{B}\|\frac{\Delta}{2}.
\end{split}
\end{equation*}
Then put $c:=\|T\|\|T^{-1}\|$, we have (\ref{eq:theorem_st1}).
Since $F$ is stable, for any $x_0\in\Real^n$, we have
\begin{equation*}
\lim_{k\rightarrow\infty} F^kx_0=0.
\end{equation*}
Hence by taking $n\rightarrow\infty$ for the inequality (\ref{eq:theorem_st1}),
we obtain the second inequality (\ref{eq:theorem_st2}).
\end{proof}
From this theorem, we conclude that if we stabilize the linearized model, the state of the 
sampled-data system with quantization shown in Figure \ref{fig:qsd} are bounded, that is,
the system is bounded-input bounded-output (BIBO) stable.

Moreover, the set
\begin{equation*}
\mathcal{D}:=\left\{ x\in\Real^n: |x|\leq r_\infty\right\}
\end{equation*}
is a positive invariant set of the system (\ref{eq:st_ss}).
In fact, if $x[0]\in \mathcal{D}\in\Real^n$, that is, $|x[0]|\leq r_\infty$, we have
for any $k\geq 1$
\begin{equation*}
\begin{split}
|x[k]|&\leq c\gamma^k |x[0]| + r_k\\
&\leq c\gamma^k r_\infty + r_k\\
&=c \left\{ \gamma^k \frac{1}{1-\gamma} + \frac{1-\gamma^k}{1-\gamma}\right\}\|\widetilde{B}\|\frac{\Delta}{2}\\
&=c \frac{1}{1-\gamma} \|\widetilde{B}\|\frac{\Delta}{2}\\
&=r_\infty,
\end{split}
\end{equation*}
that is, $x[k]\in \mathcal{D}$.

\subsection{Performance analysis of sampled-data systems with quantization}
Generally, quantization deteriorates the performance of control systems, and
it is important to know how much quantization affects on the systems.
The aim of this section is to analyze the quantization effect by using the
additive noise model discussed above.
Consider the block diagram shown in Figure~\ref{fig:qsd_block_ad}.
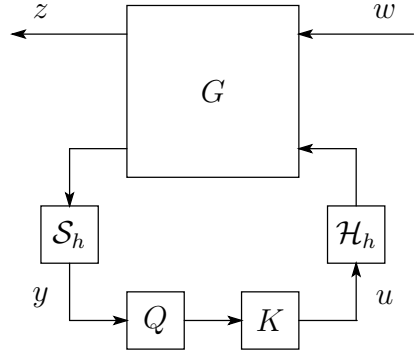
\begin{figure}[t]
\begin{center}
% \input{qsd_block_ad.tex}
%WinTpicVersion2.15
\unitlength 0.1in
\begin{picture}(22.50,18.55)(0.67,-21.33)
% BOX 2 0 3 0
% 2 817 733 1717 1633
% 
\special{pn 8}%
\special{pa 817 333}%
\special{pa 1717 333}%
\special{pa 1717 1233}%
\special{pa 817 1233}%
\special{pa 817 333}%
\special{fp}%
% VECTOR 2 0 3 0
% 2 817 883 217 883
% 
\special{pn 8}%
\special{pa 817 483}%
\special{pa 217 483}%
\special{fp}%
\special{sh 1}%
\special{pa 217 483}%
\special{pa 284 503}%
\special{pa 270 483}%
\special{pa 284 463}%
\special{pa 217 483}%
\special{fp}%
% BOX 2 0 3 0
% 2 367 1783 667 2083
% 
\special{pn 8}%
\special{pa 367 1383}%
\special{pa 667 1383}%
\special{pa 667 1683}%
\special{pa 367 1683}%
\special{pa 367 1383}%
\special{fp}%
% BOX 2 0 3 0
% 2 817 2233 1117 2533
% 
\special{pn 8}%
\special{pa 817 1833}%
\special{pa 1117 1833}%
\special{pa 1117 2133}%
\special{pa 817 2133}%
\special{pa 817 1833}%
\special{fp}%
% VECTOR 2 0 3 0
% 2 1117 2383 1417 2383
% 
\special{pn 8}%
\special{pa 1117 1983}%
\special{pa 1417 1983}%
\special{fp}%
\special{sh 1}%
\special{pa 1417 1983}%
\special{pa 1350 1963}%
\special{pa 1364 1983}%
\special{pa 1350 2003}%
\special{pa 1417 1983}%
\special{fp}%
% LINE 2 0 3 0
% 2 517 2083 517 2383
% 
\special{pn 8}%
\special{pa 517 1683}%
\special{pa 517 1983}%
\special{fp}%
% VECTOR 2 0 3 0
% 2 517 2383 817 2383
% 
\special{pn 8}%
\special{pa 517 1983}%
\special{pa 817 1983}%
\special{fp}%
\special{sh 1}%
\special{pa 817 1983}%
\special{pa 750 1963}%
\special{pa 764 1983}%
\special{pa 750 2003}%
\special{pa 817 1983}%
\special{fp}%
% LINE 2 0 3 0
% 4 1717 2383 2017 2383 2017 2383 2017 2383
% 
\special{pn 8}%
\special{pa 1717 1983}%
\special{pa 2017 1983}%
\special{fp}%
\special{pa 2017 1983}%
\special{pa 2017 1983}%
\special{fp}%
% VECTOR 2 0 3 0
% 2 2017 2383 2017 2083
% 
\special{pn 8}%
\special{pa 2017 1983}%
\special{pa 2017 1683}%
\special{fp}%
\special{sh 1}%
\special{pa 2017 1683}%
\special{pa 1997 1750}%
\special{pa 2017 1736}%
\special{pa 2037 1750}%
\special{pa 2017 1683}%
\special{fp}%
% LINE 2 0 3 0
% 2 2017 1783 2017 1483
% 
\special{pn 8}%
\special{pa 2017 1383}%
\special{pa 2017 1083}%
\special{fp}%
% VECTOR 2 0 3 0
% 2 2017 1483 1717 1483
% 
\special{pn 8}%
\special{pa 2017 1083}%
\special{pa 1717 1083}%
\special{fp}%
\special{sh 1}%
\special{pa 1717 1083}%
\special{pa 1784 1103}%
\special{pa 1770 1083}%
\special{pa 1784 1063}%
\special{pa 1717 1083}%
\special{fp}%
% LINE 2 0 3 0
% 2 817 1483 517 1483
% 
\special{pn 8}%
\special{pa 817 1083}%
\special{pa 517 1083}%
\special{fp}%
% VECTOR 2 0 3 0
% 2 517 1483 517 1783
% 
\special{pn 8}%
\special{pa 517 1083}%
\special{pa 517 1383}%
\special{fp}%
\special{sh 1}%
\special{pa 517 1383}%
\special{pa 537 1316}%
\special{pa 517 1330}%
\special{pa 497 1316}%
\special{pa 517 1383}%
\special{fp}%
% VECTOR 2 0 3 0
% 2 2317 883 1717 883
% 
\special{pn 8}%
\special{pa 2317 483}%
\special{pa 1717 483}%
\special{fp}%
\special{sh 1}%
\special{pa 1717 483}%
\special{pa 1784 503}%
\special{pa 1770 483}%
\special{pa 1784 463}%
\special{pa 1717 483}%
\special{fp}%
% BOX 2 0 3 0
% 2 1417 2233 1717 2533
% 
\special{pn 8}%
\special{pa 1417 1833}%
\special{pa 1717 1833}%
\special{pa 1717 2133}%
\special{pa 1417 2133}%
\special{pa 1417 1833}%
\special{fp}%
% BOX 2 0 3 0
% 2 1867 1783 2167 2083
% 
\special{pn 8}%
\special{pa 1867 1383}%
\special{pa 2167 1383}%
\special{pa 2167 1683}%
\special{pa 1867 1683}%
\special{pa 1867 1383}%
\special{fp}%
% STR 2 0 3 0
% 3 967 2233 967 2383 5 0
% $Q$
\put(9.6700,-19.8300){\makebox(0,0){$Q$}}%
% STR 2 0 3 0
% 3 1567 2233 1567 2383 5 0
% $K$
\put(15.6700,-19.8300){\makebox(0,0){$K$}}%
% STR 2 0 3 0
% 3 2017 1783 2017 1933 5 0
% $\hold{h}$
\put(20.1700,-15.3300){\makebox(0,0){$\hold{h}$}}%
% STR 2 0 3 0
% 3 517 1783 517 1933 5 0
% $\samp{h}$
\put(5.1700,-15.3300){\makebox(0,0){$\samp{h}$}}%
% STR 2 0 3 0
% 3 1267 1033 1267 1183 5 0
% $G$
\put(12.6700,-7.8300){\makebox(0,0){$G$}}%
% STR 2 0 3 0
% 3 2167 613 2167 763 5 0
% $w$
\put(21.6700,-3.6300){\makebox(0,0){$w$}}%
% STR 2 0 3 0
% 3 367 613 367 763 5 0
% $z$
\put(3.6700,-3.6300){\makebox(0,0){$z$}}%
% STR 2 0 3 0
% 3 367 2113 367 2263 5 0
% $y$
\put(3.6700,-18.6300){\makebox(0,0){$y$}}%
% STR 2 0 3 0
% 3 2167 2113 2167 2263 5 0
% $u$
\put(21.6700,-18.6300){\makebox(0,0){$u$}}%
\end{picture}%
 \caption{Sampled-data control system with additive noise}
 \label{fig:qsd_block_ad}
\end{center}
\end{figure}
Although the quantization noise $d$ is usually taken as the white noise, 
we assume the noise $d$ in $l^2$.
The reason is as follows:
\begin{itemize}
\item generally the quantization noise is not white; the noise will depend on the input signal of the quantizer,
\item we can use the $H^\infty$-norm that has a connection with the worst case analysis.
\end{itemize}
We denote by $\T_{zw}$ and $\T_{zd}$ the system shown in Figure \ref{fig:qsd_block_ad}
from $w$ to $z$ and from $d$ to $z$ respectively.
Assume $\T_{zw}$ and $\T_{zd}$ are stable and
\begin{equation*}
\|\T_{zw}\| := \sup_{\substack{w\in L^2\\w\neq 0}}\frac{\|\T_{zw}w\|_{L^2}}{\|w\|_{L^2}}=\gamma_1,\quad 
\|\T_{zd}\| = \sup_{\substack{d\in l^2\\d\neq 0}}\frac{\|\T_{zd}d\|_{L^2}}{\|d\|_{l^2}}=\gamma_2.
\end{equation*}
We will now discuss the performance of the system in Figure~\ref{fig:qsd}.
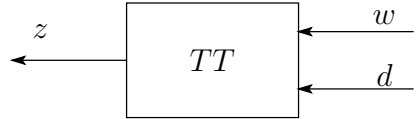
\begin{figure}[t]
\begin{center}
% \input{qsd_per_block.tex}
%WinTpicVersion2.15
\unitlength 0.1in
\begin{picture}(21.00,6.09)(2.17,-10.17)
% VECTOR 2 0 3 0
% 2 817 1117 217 1117
% 
\special{pn 8}%
\special{pa 817 717}%
\special{pa 217 717}%
\special{fp}%
\special{sh 1}%
\special{pa 217 717}%
\special{pa 284 737}%
\special{pa 270 717}%
\special{pa 284 697}%
\special{pa 217 717}%
\special{fp}%
% BOX 2 0 3 0
% 2 817 817 1717 1417
% 
\special{pn 8}%
\special{pa 817 417}%
\special{pa 1717 417}%
\special{pa 1717 1017}%
\special{pa 817 1017}%
\special{pa 817 417}%
\special{fp}%
% VECTOR 2 0 3 0
% 2 2317 967 1717 967
% 
\special{pn 8}%
\special{pa 2317 567}%
\special{pa 1717 567}%
\special{fp}%
\special{sh 1}%
\special{pa 1717 567}%
\special{pa 1784 587}%
\special{pa 1770 567}%
\special{pa 1784 547}%
\special{pa 1717 567}%
\special{fp}%
% VECTOR 2 0 3 0
% 2 2317 1267 1717 1267
% 
\special{pn 8}%
\special{pa 2317 867}%
\special{pa 1717 867}%
\special{fp}%
\special{sh 1}%
\special{pa 1717 867}%
\special{pa 1784 887}%
\special{pa 1770 867}%
\special{pa 1784 847}%
\special{pa 1717 867}%
\special{fp}%
% STR 2 0 3 0
% 3 367 817 367 967 5 0
% $z$
\put(3.6700,-5.6700){\makebox(0,0){$z$}}%
% STR 2 0 3 0
% 3 2167 743 2167 893 5 0
% $w$
\put(21.6700,-4.9300){\makebox(0,0){$w$}}%
% STR 2 0 3 0
% 3 2167 1043 2167 1193 5 0
% $d$
\put(21.6700,-7.9300){\makebox(0,0){$d$}}%
% STR 2 0 3 0
% 3 1267 967 1267 1117 5 0
% $TT$
\put(12.6700,-7.1700){\makebox(0,0){$TT$}}%
\end{picture}%
 \caption{Additive noise model for sampled-data system with quantization}
 \label{fig:qsd_per}
\end{center}
\end{figure}
\begin{lemma}
\label{lemma:qsd}
Let $\T$ be a stable sampled-data system with discrete-time inputs and continuous-time outputs.
We have the following inequality for any discrete-time, power signal $u$:
\begin{equation*}
\pow(\T u) \leq \|\T\| \pow[u].
\end{equation*}
\end{lemma}
\begin{proof}
First of all, define the power norm of lifted signal $\widetilde{y}=\lift y$ as follows:
\begin{equation*}
\pow\{\widetilde{y}\}^2:=\lim_{N\rightarrow\infty}\frac{1}{2N}\sum_{k=-N}^{N}\|\widetilde{y}[k]\|_{L^2[0,h)}^2,\quad
\widetilde{y}[k]\in L^2[0,h).
\end{equation*}
Note that for $y$ and $\widetilde{y}=\lift y$, we have
\begin{equation}
\pow(y) = \pow\{\widetilde{y}\}.
\label{eq:pow_eq}
\end{equation}
For a lifted power signal $\widetilde{y}$ (i.e., $\pow\{\widetilde{y}\}<\infty$), 
define the autocorrelation function
\begin{equation*}
R_{\widetilde{y}}[k] := \lim_{N\rightarrow\infty}
	\frac{1}{2N}\sum_{n=-N}^N \langle \widetilde{y}[n], \widetilde{y}[n+k]\rangle,
\end{equation*}
where $\langle \cdot, \cdot \rangle$  is the inner product on $L^2_{[0,h)}$,
that is,
$\langle x, y \rangle:=\int_0^h y(t)^Tx(t)dt$.
Let $S_{\widetilde{y}}$ denote the (discrete) Fourier transform of $R_{\widetilde{y}}$:
\begin{equation*}
S_{\widetilde{y}}(e^{j\omega h}) := \sum_{k=-\infty}^{\infty} R_{\widetilde{y}}[k]e^{-j\omega kh}.
\end{equation*}
Similarly, for a discrete-time $u$, define $R_u$ and $S_u$ as
\begin{equation*}
\begin{split}
R_u[k] &:= \lim_{N\rightarrow\infty}\frac{1}{2N}\sum_{n=-N}^{N} u[n+k]^Tu[n],\\
S_u(e^{j\omega h}) &:= \sum_{k=-\infty}^{\infty} R_u[k]e^{-j\omega kh}.
\end{split}
\end{equation*}
Then we have
\begin{equation}
\begin{split}
\pow\{\widetilde{y}\}^2 &= R_{\widetilde{y}}[0] = \frac{h}{2\pi}\int_{0}^{\frac{2\pi}{h}}S_{\widetilde{y}}(e^{j\omega h})d\omega,\\
\pow[u]^2 &= R_u[0] = \frac{h}{2\pi}\int_{0}^{\frac{2\pi}{h}}S_u(e^{j\omega h})d\omega.
\end{split}
\label{eq:app1}
\end{equation}
For the sampled-data system $\T$, we denote by $\widetilde{\T}(e^{j\omega h})$ the frequency response \cite{YamKha96}
of the system $\T$. Let $u$ and $y$ be the input and output of $\T$ respectively, that is, $y=\T u$.
Then we obtain
\begin{equation}
S_{\widetilde{y}}(e^{j\omega h}) = \widetilde{\T}(e^{j\omega h})\widetilde{\T^*}(e^{j\omega h})S_u(e^{j\omega h})
	= \|\widetilde{\T}(e^{j\omega h})\|^2S_u(e^{j\omega h}),
\label{eq:app2}
\end{equation}
where $\T^*$ is the dual system \cite{YamKha96} of $\T$.
The equation (\ref{eq:app2}) can be proven by the same method as the continuous-time version \cite{DFT}.
By using (\ref{eq:pow_eq}), (\ref{eq:app1}) and (\ref{eq:app2}), we have
\begin{equation*}
\begin{split}
\pow(y)^2 = \pow\{\widetilde{y}\}^2
  &=\frac{h}{2\pi}\int_{0}^{\frac{2\pi}{h}}S_{\widetilde{y}}(e^{j\omega h})d\omega\\
  &=\frac{h}{2\pi}\int_{0}^{\frac{2\pi}{h}}\|\widetilde{T}(e^{j\omega h})\|^2S_u(e^{j\omega h})d\omega\\
  &\leq \sup_{\omega\in (0,2\pi/h)} \|\widetilde{T}(e^{j\omega h})\|^2 
	\frac{h}{2\pi}\int_{0}^{\frac{2\pi}{h}}S_u(e^{j\omega h})d\omega\\
  &= \|\widetilde{\T}\|^2_\infty \pow[u]^2
  = \|\T\|^2\pow[u]^2.
\end{split}
\end{equation*}
\end{proof}

Using Lemma \ref{lemma:qsd}, we have the following theorem.
\begin{theorem}
For any input $w\in L^2$ and $d\in l^\infty$ 
of the sampled-data system $[\T_{zw}, \T_{zd}]$ defined above, 
the output $z$ satisfies
\begin{equation*}
\pow(z)\leq\frac{\gamma_2\Delta}{2}.
\end{equation*}
\end{theorem}
\begin{proof}
Since $\T_{zw}w \in L^2$, we have $\pow(\T_{zw}w)=0$.
Generally, if $\|d\|_\infty\leq 1$ then $\pow[d]\leq 1$, and hence
\begin{equation*}
\sup_{\|d\|_\infty\leq 1}\pow(\T_{zd}d) \leq \sup_{\pow[d]\leq 1}\pow(\T_{zd}d).
\end{equation*}
From Lemma \ref{lemma:qsd} we have
\begin{equation*}
\sup_{\pow[d]\leq 1}\pow(\T_{zd}d) \leq \|\T_{zd}\|.
\end{equation*}
Therefore
\begin{equation*}
\begin{split}
\pow(z) &= \pow(\T_{zw}w+\T_{zd}d)\\
&\leq\pow(\T_{zw}w)+\pow(\T_{zd}d)\\
&=\pow(T_{zd}d)\leq\|T_{zd}\|\|d\|_\infty\leq\frac{\gamma_2\Delta}{2}.
\end{split}
\end{equation*}
\end{proof}

The theorem leads us to the conclusion that if we take the $H^\infty$ design to
attenuate $\|T_{zd}\|_\infty$, the quantization has a small effect on the output $z$ with respect
to the power, and hence the $H^\infty$ design will be valid.

\section{Differential pulse code modulation}
\subsection{Differential pulse code modulation}
Differential pulse code modulation (DPCM) is a quantization system,
which is used, for example, in the telephone communication systems.
Figure \ref{fig:blockofdpcm} shows a DPCM system.
The encoder quantizes the error $e=r-u$, 
where $u$ is a prediction of the input signal $r$.
If the filter $K_1(z)$ well predicts $r$ from the quantizer output $\hat{e}:=Qe$,
the error $e$ will be smaller than the input $r$, 
and hence fewer bits are required to quantize the signal.
Assume that the quantizer $Q$ is a uniform quantizer with a quantization level $\Delta$.
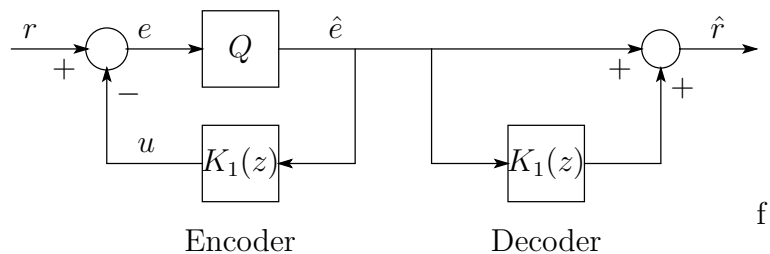
\begin{figure}[t]
\begin{center}
%\input{qsd_dpcm1}
%WinTpicVersion2.15
\unitlength 0.1in
\begin{picture}(39.00,11.45)(4.00,-13.15)
% VECTOR 2 0 3 0
% 2 400 800 800 800
% 
\special{pn 8}%
\special{pa 400 400}%
\special{pa 800 400}%
\special{fp}%
\special{sh 1}%
\special{pa 800 400}%
\special{pa 733 380}%
\special{pa 747 400}%
\special{pa 733 420}%
\special{pa 800 400}%
\special{fp}%
% CIRCLE 2 0 3 0
% 4 900 800 790 800 790 800 790 800
% 
\special{pn 8}%
\special{ar 900 400 110 110  0.0000000 6.2831853}%
% VECTOR 2 0 3 0
% 2 1000 800 1400 800
% 
\special{pn 8}%
\special{pa 1000 400}%
\special{pa 1400 400}%
\special{fp}%
\special{sh 1}%
\special{pa 1400 400}%
\special{pa 1333 380}%
\special{pa 1347 400}%
\special{pa 1333 420}%
\special{pa 1400 400}%
\special{fp}%
% BOX 2 0 3 0
% 2 1400 600 1800 1000
% 
\special{pn 8}%
\special{pa 1400 200}%
\special{pa 1800 200}%
\special{pa 1800 600}%
\special{pa 1400 600}%
\special{pa 1400 200}%
\special{fp}%
% LINE 2 0 3 0
% 2 2200 800 2200 1400
% 
\special{pn 8}%
\special{pa 2200 400}%
\special{pa 2200 1000}%
\special{fp}%
% VECTOR 2 0 3 0
% 2 2200 1400 1800 1400
% 
\special{pn 8}%
\special{pa 2200 1000}%
\special{pa 1800 1000}%
\special{fp}%
\special{sh 1}%
\special{pa 1800 1000}%
\special{pa 1867 1020}%
\special{pa 1853 1000}%
\special{pa 1867 980}%
\special{pa 1800 1000}%
\special{fp}%
% BOX 2 0 3 0
% 2 1800 1200 1400 1600
% 
\special{pn 8}%
\special{pa 1800 800}%
\special{pa 1400 800}%
\special{pa 1400 1200}%
\special{pa 1800 1200}%
\special{pa 1800 800}%
\special{fp}%
% LINE 2 0 3 0
% 2 1400 1400 900 1400
% 
\special{pn 8}%
\special{pa 1400 1000}%
\special{pa 900 1000}%
\special{fp}%
% VECTOR 2 0 3 0
% 2 900 1400 900 900
% 
\special{pn 8}%
\special{pa 900 1000}%
\special{pa 900 500}%
\special{fp}%
\special{sh 1}%
\special{pa 900 500}%
\special{pa 880 567}%
\special{pa 900 553}%
\special{pa 920 567}%
\special{pa 900 500}%
\special{fp}%
% LINE 2 0 3 0
% 2 1800 800 2600 800
% 
\special{pn 8}%
\special{pa 1800 400}%
\special{pa 2600 400}%
\special{fp}%
% LINE 2 0 3 0
% 2 2600 800 2600 1400
% 
\special{pn 8}%
\special{pa 2600 400}%
\special{pa 2600 1000}%
\special{fp}%
% VECTOR 2 0 3 0
% 2 2600 1400 3000 1400
% 
\special{pn 8}%
\special{pa 2600 1000}%
\special{pa 3000 1000}%
\special{fp}%
\special{sh 1}%
\special{pa 3000 1000}%
\special{pa 2933 980}%
\special{pa 2947 1000}%
\special{pa 2933 1020}%
\special{pa 3000 1000}%
\special{fp}%
% BOX 2 0 3 0
% 2 3000 1200 3400 1600
% 
\special{pn 8}%
\special{pa 3000 800}%
\special{pa 3400 800}%
\special{pa 3400 1200}%
\special{pa 3000 1200}%
\special{pa 3000 800}%
\special{fp}%
% LINE 2 0 3 0
% 2 3400 1400 3800 1400
% 
\special{pn 8}%
\special{pa 3400 1000}%
\special{pa 3800 1000}%
\special{fp}%
% VECTOR 2 0 3 0
% 2 3800 1400 3800 900
% 
\special{pn 8}%
\special{pa 3800 1000}%
\special{pa 3800 500}%
\special{fp}%
\special{sh 1}%
\special{pa 3800 500}%
\special{pa 3780 567}%
\special{pa 3800 553}%
\special{pa 3820 567}%
\special{pa 3800 500}%
\special{fp}%
% CIRCLE 2 0 3 0
% 4 3800 800 3800 900 3800 900 3800 900
% 
\special{pn 8}%
\special{ar 3800 400 100 100  0.0000000 6.2831853}%
% VECTOR 2 0 3 0
% 2 2600 800 3700 800
% 
\special{pn 8}%
\special{pa 2600 400}%
\special{pa 3700 400}%
\special{fp}%
\special{sh 1}%
\special{pa 3700 400}%
\special{pa 3633 380}%
\special{pa 3647 400}%
\special{pa 3633 420}%
\special{pa 3700 400}%
\special{fp}%
% VECTOR 2 0 3 0
% 2 3900 800 4300 800
% 
\special{pn 8}%
\special{pa 3900 400}%
\special{pa 4300 400}%
\special{fp}%
\special{sh 1}%
\special{pa 4300 400}%
\special{pa 4233 380}%
\special{pa 4247 400}%
\special{pa 4233 420}%
\special{pa 4300 400}%
\special{fp}%
% STR 2 0 3 0
% 3 1600 1700 1600 1800 5 0
% Encoder
\put(16.0000,-14.0000){\makebox(0,0){Encoder}}%
% STR 2 0 3 0
% 3 3200 1700 3200 1800 5 0
% Decoder
\put(32.0000,-14.0000){\makebox(0,0){Decoder}}%
% STR 2 0 3 0
% 3 1600 1300 1600 1400 5 0
% $K_1(z)$
\put(16.0000,-10.0000){\makebox(0,0){$K_1(z)$}}%
% STR 2 0 3 0
% 3 1600 700 1600 800 5 0
% $Q$
\put(16.0000,-4.0000){\makebox(0,0){$Q$}}%
% STR 2 0 3 0
% 3 3200 1300 3200 1400 5 0
% $K_1(z)$
\put(32.0000,-10.0000){\makebox(0,0){$K_1(z)$}}%
% STR 2 0 3 0
% 3 460 640 460 740 2 0
% $r$
\put(4.6000,-3.4000){\makebox(0,0)[lb]{$r$}}%
% STR 2 0 3 0
% 3 1060 640 1060 740 2 0
% $e$
\put(10.6000,-3.4000){\makebox(0,0)[lb]{$e$}}%
% STR 2 0 3 0
% 3 2060 640 2060 740 2 0
% $\hat{e}$
\put(20.6000,-3.4000){\makebox(0,0)[lb]{$\hat{e}$}}%
% STR 2 0 3 0
% 3 4060 640 4060 740 2 0
% $\hat{r}$
\put(40.6000,-3.4000){\makebox(0,0)[lb]{$\hat{r}$}}%
% STR 2 0 3 0
% 3 1060 1240 1060 1340 2 0
% $u$
\put(10.6000,-9.4000){\makebox(0,0)[lb]{$u$}}%
% STR 2 0 3 0
% 3 950 880 950 980 1 0
% $-$
\put(9.5000,-5.8000){\makebox(0,0)[lt]{$-$}}%
% STR 2 0 3 0
% 3 740 750 740 850 4 0
% $+$
\put(7.4000,-4.5000){\makebox(0,0)[rt]{$+$}}%
% STR 2 0 3 0
% 3 3650 750 3650 850 4 0
% $+$
\put(36.5000,-4.5000){\makebox(0,0)[rt]{$+$}}%
% STR 2 0 3 0
% 3 3850 850 3850 950 1 0
% $+$
\put(38.5000,-5.5000){\makebox(0,0)[lt]{$+$}}%
\end{picture}%
f\end{center}
 \caption{DPCM system}
 \label{fig:blockofdpcm}
\end{figure}

The signal $\hat{e}$ is transmitted into a communication channel or stored in digital media, 
and then the decoder reconstructs the original signal $r$ and makes the output $\hat{r}$.
Note that the filter in the decoder is the same as that in the encoder,
and we have the following error estimate
\begin{equation*}
\| \hat{r} - r \|_\infty = \| \hat{e} - e \|_\infty \leq \frac{\Delta}{2}.
\end{equation*}
That is to say, the reconstruction error is less than $\Delta/2$.
Therefore, DPCM can transmit data with fewer bits.

However, the model does not take account of the channel noise that adds the transmitted signal $\hat{e}$,
and hence the estimate is not valid.
For example, in the $\Delta$ modulation, the filter $K_1$ is the adder
$K_1(z) = 1/(z-1)$.
Since this filter is unstable, the channel noise will be amplified in the decoder.

Therefore, we use a model taking the channel noise into account
for designing the encoder and the decoder.

\subsection{Problem formulation}
The block diagram of our DPCM system to be designed is illustrated in Figure \ref{fig:blockofdpcm2}.
In the figure, $n$ is the channel noise.
Our purpose is to attenuate the prediction error $e$ to be quantized and 
the reconstruction error.
For this purpose, we design the filter $K_1(z)$ and $K_2(z)$ in Figure \ref{fig:blockofdpcm2}
by using the sampled-data $H^\infty$ optimization method.
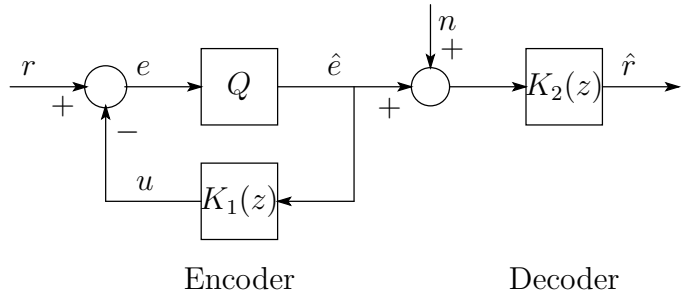
\begin{figure}[t]
\begin{center}
%\input{qsd_dpcm2}
%WinTpicVersion2.15
\unitlength 0.1in
\begin{picture}(35.00,13.85)(4.00,-15.15)
% VECTOR 2 0 3 0
% 2 400 1000 800 1000
% 
\special{pn 8}%
\special{pa 400 600}%
\special{pa 800 600}%
\special{fp}%
\special{sh 1}%
\special{pa 800 600}%
\special{pa 733 580}%
\special{pa 747 600}%
\special{pa 733 620}%
\special{pa 800 600}%
\special{fp}%
% CIRCLE 2 0 3 0
% 4 900 1000 790 1000 790 1000 790 1000
% 
\special{pn 8}%
\special{ar 900 600 110 110  0.0000000 6.2831853}%
% VECTOR 2 0 3 0
% 2 1000 1000 1400 1000
% 
\special{pn 8}%
\special{pa 1000 600}%
\special{pa 1400 600}%
\special{fp}%
\special{sh 1}%
\special{pa 1400 600}%
\special{pa 1333 580}%
\special{pa 1347 600}%
\special{pa 1333 620}%
\special{pa 1400 600}%
\special{fp}%
% BOX 2 0 3 0
% 2 1400 800 1800 1200
% 
\special{pn 8}%
\special{pa 1400 400}%
\special{pa 1800 400}%
\special{pa 1800 800}%
\special{pa 1400 800}%
\special{pa 1400 400}%
\special{fp}%
% LINE 2 0 3 0
% 2 2200 1000 2200 1600
% 
\special{pn 8}%
\special{pa 2200 600}%
\special{pa 2200 1200}%
\special{fp}%
% VECTOR 2 0 3 0
% 2 2200 1600 1800 1600
% 
\special{pn 8}%
\special{pa 2200 1200}%
\special{pa 1800 1200}%
\special{fp}%
\special{sh 1}%
\special{pa 1800 1200}%
\special{pa 1867 1220}%
\special{pa 1853 1200}%
\special{pa 1867 1180}%
\special{pa 1800 1200}%
\special{fp}%
% BOX 2 0 3 0
% 2 1800 1400 1400 1800
% 
\special{pn 8}%
\special{pa 1800 1000}%
\special{pa 1400 1000}%
\special{pa 1400 1400}%
\special{pa 1800 1400}%
\special{pa 1800 1000}%
\special{fp}%
% LINE 2 0 3 0
% 2 1400 1600 900 1600
% 
\special{pn 8}%
\special{pa 1400 1200}%
\special{pa 900 1200}%
\special{fp}%
% VECTOR 2 0 3 0
% 2 900 1600 900 1100
% 
\special{pn 8}%
\special{pa 900 1200}%
\special{pa 900 700}%
\special{fp}%
\special{sh 1}%
\special{pa 900 700}%
\special{pa 880 767}%
\special{pa 900 753}%
\special{pa 920 767}%
\special{pa 900 700}%
\special{fp}%
% BOX 2 0 3 0
% 2 3100 800 3500 1200
% 
\special{pn 8}%
\special{pa 3100 400}%
\special{pa 3500 400}%
\special{pa 3500 800}%
\special{pa 3100 800}%
\special{pa 3100 400}%
\special{fp}%
% STR 2 0 3 0
% 3 1600 1900 1600 2000 5 0
% Encoder
\put(16.0000,-16.0000){\makebox(0,0){Encoder}}%
% STR 2 0 3 0
% 3 3300 1900 3300 2000 5 0
% Decoder
\put(33.0000,-16.0000){\makebox(0,0){Decoder}}%
% STR 2 0 3 0
% 3 1600 1500 1600 1600 5 0
% $K_1(z)$
\put(16.0000,-12.0000){\makebox(0,0){$K_1(z)$}}%
% STR 2 0 3 0
% 3 1600 900 1600 1000 5 0
% $Q$
\put(16.0000,-6.0000){\makebox(0,0){$Q$}}%
% STR 2 0 3 0
% 3 3300 900 3300 1000 5 0
% $K_2(z)$
\put(33.0000,-6.0000){\makebox(0,0){$K_2(z)$}}%
% STR 2 0 3 0
% 3 460 840 460 940 2 0
% $r$
\put(4.6000,-5.4000){\makebox(0,0)[lb]{$r$}}%
% STR 2 0 3 0
% 3 1060 840 1060 940 2 0
% $e$
\put(10.6000,-5.4000){\makebox(0,0)[lb]{$e$}}%
% STR 2 0 3 0
% 3 2060 840 2060 940 2 0
% $\hat{e}$
\put(20.6000,-5.4000){\makebox(0,0)[lb]{$\hat{e}$}}%
% STR 2 0 3 0
% 3 3600 840 3600 940 2 0
% $\hat{r}$
\put(36.0000,-5.4000){\makebox(0,0)[lb]{$\hat{r}$}}%
% STR 2 0 3 0
% 3 1060 1440 1060 1540 2 0
% $u$
\put(10.6000,-11.4000){\makebox(0,0)[lb]{$u$}}%
% STR 2 0 3 0
% 3 950 1080 950 1180 1 0
% $-$
\put(9.5000,-7.8000){\makebox(0,0)[lt]{$-$}}%
% STR 2 0 3 0
% 3 740 950 740 1050 4 0
% $+$
\put(7.4000,-6.5000){\makebox(0,0)[rt]{$+$}}%
% VECTOR 2 0 3 0
% 2 1800 1000 2500 1000
% 
\special{pn 8}%
\special{pa 1800 600}%
\special{pa 2500 600}%
\special{fp}%
\special{sh 1}%
\special{pa 2500 600}%
\special{pa 2433 580}%
\special{pa 2447 600}%
\special{pa 2433 620}%
\special{pa 2500 600}%
\special{fp}%
% CIRCLE 2 0 3 0
% 4 2600 1000 2500 1000 2500 1000 2500 1000
% 
\special{pn 8}%
\special{ar 2600 600 100 100  0.0000000 6.2831853}%
% VECTOR 2 0 3 0
% 2 2700 1000 3100 1000
% 
\special{pn 8}%
\special{pa 2700 600}%
\special{pa 3100 600}%
\special{fp}%
\special{sh 1}%
\special{pa 3100 600}%
\special{pa 3033 580}%
\special{pa 3047 600}%
\special{pa 3033 620}%
\special{pa 3100 600}%
\special{fp}%
% VECTOR 2 0 3 0
% 2 2600 570 2600 900
% 
\special{pn 8}%
\special{pa 2600 170}%
\special{pa 2600 500}%
\special{fp}%
\special{sh 1}%
\special{pa 2600 500}%
\special{pa 2620 433}%
\special{pa 2600 447}%
\special{pa 2580 433}%
\special{pa 2600 500}%
\special{fp}%
% STR 2 0 3 0
% 3 2640 600 2640 700 2 0
% $n$
\put(26.4000,-3.0000){\makebox(0,0)[lb]{$n$}}%
% STR 2 0 3 0
% 3 2450 960 2450 1060 4 0
% $+$
\put(24.5000,-6.6000){\makebox(0,0)[rt]{$+$}}%
% STR 2 0 3 0
% 3 2650 760 2650 860 2 0
% $+$
\put(26.5000,-4.6000){\makebox(0,0)[lb]{$+$}}%
% VECTOR 2 0 3 0
% 2 3500 1000 3900 1000
% 
\special{pn 8}%
\special{pa 3500 600}%
\special{pa 3900 600}%
\special{fp}%
\special{sh 1}%
\special{pa 3900 600}%
\special{pa 3833 580}%
\special{pa 3847 600}%
\special{pa 3833 620}%
\special{pa 3900 600}%
\special{fp}%
\end{picture}%
\end{center}
 \caption{DPCM system with channel noise}
 \label{fig:blockofdpcm2}
\end{figure}

The block diagram of the error system for designing these filters is shown is Figure \ref{fig:qsd_errorsystem}.
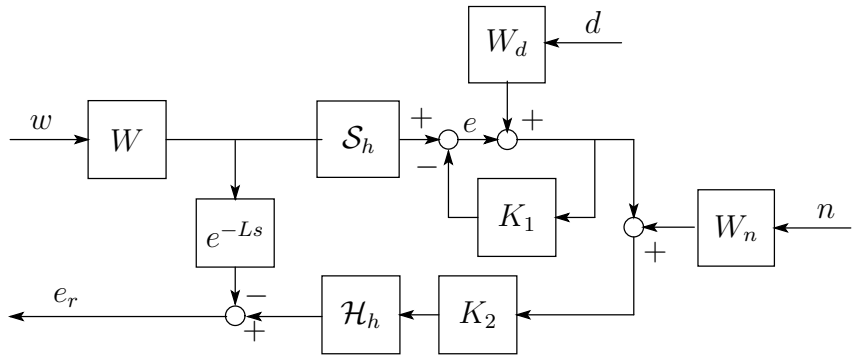
\begin{figure}[t]
\begin{center}
% \input{qsd_error_system}
%WinTpicVersion2.15
\unitlength 0.1in
\begin{picture}(44.00,18.24)(2.00,-19.29)
% BOX 2 0 3 0
% 2 608 1003 1016 1411
% 
\special{pn 8}%
\special{pa 608 603}%
\special{pa 1016 603}%
\special{pa 1016 1011}%
\special{pa 608 1011}%
\special{pa 608 603}%
\special{fp}%
% VECTOR 2 0 3 0
% 2 2240 1207 2444 1207
% 
\special{pn 8}%
\special{pa 2240 807}%
\special{pa 2444 807}%
\special{fp}%
\special{sh 1}%
\special{pa 2444 807}%
\special{pa 2377 787}%
\special{pa 2391 807}%
\special{pa 2377 827}%
\special{pa 2444 807}%
\special{fp}%
% CIRCLE 2 0 3 0
% 4 2494 1207 2444 1207 2444 1207 2444 1207
% 
\special{pn 8}%
\special{ar 2494 807 50 50  0.0000000 6.2831853}%
% VECTOR 2 0 3 0
% 2 2546 1207 2750 1207
% 
\special{pn 8}%
\special{pa 2546 807}%
\special{pa 2750 807}%
\special{fp}%
\special{sh 1}%
\special{pa 2750 807}%
\special{pa 2683 787}%
\special{pa 2697 807}%
\special{pa 2683 827}%
\special{pa 2750 807}%
\special{fp}%
% CIRCLE 2 0 3 0
% 4 2802 1207 2750 1207 2750 1207 2750 1207
% 
\special{pn 8}%
\special{ar 2802 807 52 52  0.0000000 6.2831853}%
% LINE 2 0 3 0
% 2 2852 1207 3464 1207
% 
\special{pn 8}%
\special{pa 2852 807}%
\special{pa 3464 807}%
\special{fp}%
% VECTOR 2 0 3 0
% 2 3464 1207 3464 1615
% 
\special{pn 8}%
\special{pa 3464 807}%
\special{pa 3464 1215}%
\special{fp}%
\special{sh 1}%
\special{pa 3464 1215}%
\special{pa 3484 1148}%
\special{pa 3464 1162}%
\special{pa 3444 1148}%
\special{pa 3464 1215}%
\special{fp}%
% CIRCLE 2 0 3 0
% 4 3464 1666 3464 1615 3464 1615 3464 1615
% 
\special{pn 8}%
\special{ar 3464 1266 51 51  0.0000000 6.2831853}%
% LINE 2 0 3 0
% 2 3464 1717 3464 2125
% 
\special{pn 8}%
\special{pa 3464 1317}%
\special{pa 3464 1725}%
\special{fp}%
% VECTOR 2 0 3 0
% 2 3464 2125 2852 2125
% 
\special{pn 8}%
\special{pa 3464 1725}%
\special{pa 2852 1725}%
\special{fp}%
\special{sh 1}%
\special{pa 2852 1725}%
\special{pa 2919 1745}%
\special{pa 2905 1725}%
\special{pa 2919 1705}%
\special{pa 2852 1725}%
\special{fp}%
% BOX 2 0 3 0
% 2 2852 1921 2444 2329
% 
\special{pn 8}%
\special{pa 2852 1521}%
\special{pa 2444 1521}%
\special{pa 2444 1929}%
\special{pa 2852 1929}%
\special{pa 2852 1521}%
\special{fp}%
% VECTOR 2 0 3 0
% 2 2444 2125 2240 2125
% 
\special{pn 8}%
\special{pa 2444 1725}%
\special{pa 2240 1725}%
\special{fp}%
\special{sh 1}%
\special{pa 2240 1725}%
\special{pa 2307 1745}%
\special{pa 2293 1725}%
\special{pa 2307 1705}%
\special{pa 2240 1725}%
\special{fp}%
% BOX 2 0 3 0
% 2 2240 1921 1832 2329
% 
\special{pn 8}%
\special{pa 2240 1521}%
\special{pa 1832 1521}%
\special{pa 1832 1929}%
\special{pa 2240 1929}%
\special{pa 2240 1521}%
\special{fp}%
% CIRCLE 2 0 3 0
% 4 1378 2130 1429 2130 1429 2130 1429 2130
% 
\special{pn 8}%
\special{ar 1378 1730 51 51  0.0000000 6.2831853}%
% VECTOR 2 0 3 0
% 4 1378 1212 1378 1518 1378 1518 1378 1518
% 
\special{pn 8}%
\special{pa 1378 812}%
\special{pa 1378 1118}%
\special{fp}%
\special{sh 1}%
\special{pa 1378 1118}%
\special{pa 1398 1051}%
\special{pa 1378 1065}%
\special{pa 1358 1051}%
\special{pa 1378 1118}%
\special{fp}%
\special{pa 1378 1118}%
\special{pa 1378 1118}%
\special{fp}%
% BOX 2 0 3 0
% 2 1174 1518 1582 1875
% 
\special{pn 8}%
\special{pa 1174 1118}%
\special{pa 1582 1118}%
\special{pa 1582 1475}%
\special{pa 1174 1475}%
\special{pa 1174 1118}%
\special{fp}%
% VECTOR 2 0 3 0
% 2 1378 1875 1378 2068
% 
\special{pn 8}%
\special{pa 1378 1475}%
\special{pa 1378 1668}%
\special{fp}%
\special{sh 1}%
\special{pa 1378 1668}%
\special{pa 1398 1601}%
\special{pa 1378 1615}%
\special{pa 1358 1601}%
\special{pa 1378 1668}%
\special{fp}%
% VECTOR 2 0 3 0
% 2 2802 901 2802 1166
% 
\special{pn 8}%
\special{pa 2802 501}%
\special{pa 2802 766}%
\special{fp}%
\special{sh 1}%
\special{pa 2802 766}%
\special{pa 2822 699}%
\special{pa 2802 713}%
\special{pa 2782 699}%
\special{pa 2802 766}%
\special{fp}%
% LINE 2 0 3 0
% 2 3260 1207 3260 1615
% 
\special{pn 8}%
\special{pa 3260 807}%
\special{pa 3260 1215}%
\special{fp}%
% VECTOR 2 0 3 0
% 2 3260 1615 3056 1615
% 
\special{pn 8}%
\special{pa 3260 1215}%
\special{pa 3056 1215}%
\special{fp}%
\special{sh 1}%
\special{pa 3056 1215}%
\special{pa 3123 1235}%
\special{pa 3109 1215}%
\special{pa 3123 1195}%
\special{pa 3056 1215}%
\special{fp}%
% BOX 2 0 3 0
% 2 3056 1411 2648 1819
% 
\special{pn 8}%
\special{pa 3056 1011}%
\special{pa 2648 1011}%
\special{pa 2648 1419}%
\special{pa 3056 1419}%
\special{pa 3056 1011}%
\special{fp}%
% LINE 2 0 3 0
% 2 2648 1615 2494 1615
% 
\special{pn 8}%
\special{pa 2648 1215}%
\special{pa 2494 1215}%
\special{fp}%
% VECTOR 2 0 3 0
% 2 2494 1615 2494 1268
% 
\special{pn 8}%
\special{pa 2494 1215}%
\special{pa 2494 868}%
\special{fp}%
\special{sh 1}%
\special{pa 2494 868}%
\special{pa 2474 935}%
\special{pa 2494 921}%
\special{pa 2514 935}%
\special{pa 2494 868}%
\special{fp}%
% VECTOR 2 0 3 0
% 2 3781 1666 3514 1666
% 
\special{pn 8}%
\special{pa 3781 1266}%
\special{pa 3514 1266}%
\special{fp}%
\special{sh 1}%
\special{pa 3514 1266}%
\special{pa 3581 1286}%
\special{pa 3567 1266}%
\special{pa 3581 1246}%
\special{pa 3514 1266}%
\special{fp}%
% STR 2 0 3 0
% 3 812 1105 812 1207 5 0
% $W$
\put(8.1200,-8.0700){\makebox(0,0){$W$}}%
% STR 2 0 3 0
% 3 3250 489 3250 590 5 0
% $d$
\put(32.5000,-1.9000){\makebox(0,0){$d$}}%
% STR 2 0 3 0
% 3 2852 1513 2852 1615 5 0
% $K_1$
\put(28.5200,-12.1500){\makebox(0,0){$K_1$}}%
% STR 2 0 3 0
% 3 4470 1478 4470 1580 5 0
% $n$
\put(44.7000,-11.8000){\makebox(0,0){$n$}}%
% STR 2 0 3 0
% 3 2648 2023 2648 2125 5 0
% $K_2$
\put(26.4800,-17.2500){\makebox(0,0){$K_2$}}%
% STR 2 0 3 0
% 3 2036 2023 2036 2125 5 0
% $\hold{h}$
\put(20.3600,-17.2500){\makebox(0,0){$\hold{h}$}}%
% STR 2 0 3 0
% 3 1378 1599 1378 1701 5 0
% $e^{-Ls}$
\put(13.7800,-13.0100){\makebox(0,0){$e^{-Ls}$}}%
% STR 2 0 3 0
% 3 363 1014 363 1114 5 0
% $w$
\put(3.6300,-7.1400){\makebox(0,0){$w$}}%
% STR 2 0 3 0
% 3 495 1930 495 2034 5 0
% $e_r$
\put(4.9500,-16.3400){\makebox(0,0){$e_r$}}%
% STR 2 0 3 0
% 3 1490 1926 1490 2028 5 0
% $-$
\put(14.9000,-16.2800){\makebox(0,0){$-$}}%
% STR 2 0 3 0
% 3 1480 2109 1480 2211 5 0
% $+$
\put(14.8000,-18.1100){\makebox(0,0){$+$}}%
% STR 2 0 3 0
% 3 2353 992 2353 1095 5 0
% $+$
\put(23.5300,-6.9500){\makebox(0,0){$+$}}%
% STR 2 0 3 0
% 3 2383 1248 2383 1350 5 0
% $-$
\put(23.8300,-9.5000){\makebox(0,0){$-$}}%
% STR 2 0 3 0
% 3 2934 1003 2934 1105 5 0
% $+$
\put(29.3400,-7.0500){\makebox(0,0){$+$}}%
% STR 2 0 3 0
% 3 3577 1696 3577 1798 5 0
% $+$
\put(35.7700,-13.9800){\makebox(0,0){$+$}}%
% LINE 2 0 3 0
% 2 1016 1207 1832 1207
% 
\special{pn 8}%
\special{pa 1016 807}%
\special{pa 1832 807}%
\special{fp}%
% VECTOR 2 0 3 0
% 2 200 1207 608 1207
% 
\special{pn 8}%
\special{pa 200 807}%
\special{pa 608 807}%
\special{fp}%
\special{sh 1}%
\special{pa 608 807}%
\special{pa 541 787}%
\special{pa 555 807}%
\special{pa 541 827}%
\special{pa 608 807}%
\special{fp}%
% VECTOR 2 0 3 0
% 2 1832 2134 1444 2134
% 
\special{pn 8}%
\special{pa 1832 1734}%
\special{pa 1444 1734}%
\special{fp}%
\special{sh 1}%
\special{pa 1444 1734}%
\special{pa 1511 1754}%
\special{pa 1497 1734}%
\special{pa 1511 1714}%
\special{pa 1444 1734}%
\special{fp}%
% VECTOR 2 0 3 0
% 2 1333 2134 200 2134
% 
\special{pn 8}%
\special{pa 1333 1734}%
\special{pa 200 1734}%
\special{fp}%
\special{sh 1}%
\special{pa 200 1734}%
\special{pa 267 1754}%
\special{pa 253 1734}%
\special{pa 267 1714}%
\special{pa 200 1734}%
\special{fp}%
% BOX 2 0 3 0
% 2 2600 510 3000 910
% 
\special{pn 8}%
\special{pa 2600 110}%
\special{pa 3000 110}%
\special{pa 3000 510}%
\special{pa 2600 510}%
\special{pa 2600 110}%
\special{fp}%
% BOX 2 0 3 0
% 2 3800 1470 4200 1870
% 
\special{pn 8}%
\special{pa 3800 1070}%
\special{pa 4200 1070}%
\special{pa 4200 1470}%
\special{pa 3800 1470}%
\special{pa 3800 1070}%
\special{fp}%
% VECTOR 2 0 3 0
% 2 3400 700 3000 700
% 
\special{pn 8}%
\special{pa 3400 300}%
\special{pa 3000 300}%
\special{fp}%
\special{sh 1}%
\special{pa 3000 300}%
\special{pa 3067 320}%
\special{pa 3053 300}%
\special{pa 3067 280}%
\special{pa 3000 300}%
\special{fp}%
% VECTOR 2 0 3 0
% 2 4600 1670 4200 1670
% 
\special{pn 8}%
\special{pa 4600 1270}%
\special{pa 4200 1270}%
\special{fp}%
\special{sh 1}%
\special{pa 4200 1270}%
\special{pa 4267 1290}%
\special{pa 4253 1270}%
\special{pa 4267 1250}%
\special{pa 4200 1270}%
\special{fp}%
% STR 2 0 3 0
% 3 4000 1570 4000 1670 5 0
% $W_n$
\put(40.0000,-12.7000){\makebox(0,0){$W_n$}}%
% STR 2 0 3 0
% 3 2800 600 2800 700 5 0
% $W_d$
\put(28.0000,-3.0000){\makebox(0,0){$W_d$}}%
% BOX 2 0 3 0
% 2 1806 1005 2236 1405
% 
\special{pn 8}%
\special{pa 1806 605}%
\special{pa 2236 605}%
\special{pa 2236 1005}%
\special{pa 1806 1005}%
\special{pa 1806 605}%
\special{fp}%
% STR 2 0 3 0
% 3 2026 1105 2026 1205 5 0
% $\samp{h}$
\put(20.2600,-8.0500){\makebox(0,0){$\samp{h}$}}%
% STR 2 0 3 0
% 3 2570 1070 2570 1170 2 0
% $e$
\put(25.7000,-7.7000){\makebox(0,0)[lb]{$e$}}%
\end{picture}%
 \caption{Error system for designing filters}
 \label{fig:qsd_errorsystem}
\end{center}
\end{figure}
We first take the quantization error caused by $Q$ for the additive noise $d$.
Then, assume that the analog signal to be quantized has frequency characteristic $W(s)$, 
and introduce a time delay $e^{-Ls}$.
This delay time $L$ is the time that $K_1(z)$ and $K_2(z)$ will take to process signals.
In the error system, let $\T_1$ be the system from $[w,d]^T$ to $e$,
and $\T_2$ be the system from $[w,d,n]^T$ to $e_r$.
Then the design problem is formulated as follows:
\begin{problem}
Given an analog filter $W(s)$, time delay $L$ and a sampling period $h$, 
find filters $K_1(z)$ and $K_2(z)$ that minimize 
\begin{equation*}
\begin{split}
\|\T_1\|^2&:=\sup_{\substack{w\in L^2, d\in l^2\\ \|w\|_{L^2}+\|d\|_{l^2}\neq 0}} \frac{\|e\|^2_{l^2}}{\|w\|^2_{L^2}+\|d\|^2_{l^2}},\\
\|\T_2\|^2&:=\sup_{\substack{w\in L^2, d,n\in l^2\\ \|w\|_{L^2}+\|d\|_{l^2}+\|n\|_{l^2}\neq 0}} 
	\frac{\|e_r\|^2_{l^2}}{\|w\|^2_{L^2}+\|d\|^2_{l^2}+\|n\|^2_{l^2}},
\end{split}
\end{equation*}
respectively.
%\begin{itemize}
% \item[\underline{Step1}] Find the filter $K_1(z)$ which minimize $\| T_1 \|_\infty$.
% \item[\underline{Step2}] Find the filter $K_2(z)$ which minimize $\| T_2 \|_\infty$.
%\end{itemize}  
\end{problem}

This is a sampled-data $H^\infty$ optimization problem and assuming $L=mh$ ( $m\in\Nset$ ),
the solution can be obtained by using the fast-sampling/fast-hold method discussed in 
Chapter \ref{ch:SDC}.
\begin{theorem}\label{th:qqq_FSFH}
Assume that $L=mh$, $m\in\Nset$. Then, 
for the sampled-data systems $\T_1$ and $\T_2$,
there exist finite-dimensional discrete-time systems $\{T_{1,N}: N=1,2,\ldots\}$
and $\{T_{2,N}: N=1,2,\ldots\}$ such that
\begin{equation*}
\begin{split}
\lim_{N\to\infty}\|T_{1,N}\|&= \|\T_{1}\|,\\
\lim_{N\to\infty}\|T_{2,N}\|&= \|\T_{2}\|.
\end{split}
\end{equation*}
\end{theorem}
The proof of this theorem is almost the same as the proof of Theorem 5 in Chapter 4.
The approximated discrete-time systems $T_{1,N}$ and $T_{2,N}$ are as follows:
\begin{equation*}
\begin{split}
T_{1,N}&=\lft(G_{1,N}, K_1),\quad
G_{1,N}:=\left[\begin{array}{cc}
		\left[\widetilde{W}_{N}, 0\right] & -1\\
		\left[\widetilde{W}_{N}, W_d\right] & -1
	\end{array}\right],\\
T_{2,N}&=\lft(G_{2,N}, K_2),\quad
G_{2,N}:=\left[\begin{array}{cc}
		\left[z^{-m}, 0, 0\right] & -1\\
		\left[S_d\widetilde{W}_{N}, S_dW_d, W_n\right] & 0
	\end{array}\right],\\
\widetilde{W}_{N}&:=[1,\underbrace{0,\ldots,0}_{N-1}]\dliftsys{N}{W},\quad
S_d:=(1+K_1)^{-1}.
\end{split}
\end{equation*}

\section{Design example}
In this section, we present a design example of DPCM.
The design parameters are as follows:
the sampling period $h=1$, the reconstruction delay $L=2$, the weighting functions are
\begin{equation*}
W_d = 0.5, \quad W_n(z) = 0.1\times\frac{0.01753z^2-0.03506z+0.01753}{z^2+0.572z+0.3147},
\end{equation*}
where $W_n(z)$ is a Chebyshev type I high-pass filter \cite{Fli,Vai,Zel},
and the analog filter $W(s)$ is
\begin{equation*}
W(s) = \frac{1}{(10s+1)^2}.
\end{equation*}
For comparison, we take the $\Delta$ modulation, that is, $K_1(z)=1/(z-1)$ and $K_2(z)=1+K_1(z)=z/(z-1)$.

Figure \ref{fig:qsd_freqT1} shows the frequency responses of the system $\T_1$. 
We can see that the system $\T_1$ designed by $H^\infty$ optimization has lower
gain than that of the $\Delta$ modulation over the whole frequency,
in particular, in the high frequency range, our system shows better attenuation.

Figure \ref{fig:qsd_freqT2} shows the frequency responses of the system $\T_2$.
Since the $\Delta$ modulation is not stable (i.e., the system has a pole at $z=1$), 
the frequency response of $\T_2$ is indefinite.
Therefore in Figure \ref{fig:qsd_freqT2},
we take a decoder with the decoding filter $K_2 = 1+K_1$ where $K_1$ is the $H^\infty$ (sub) 
optimal filter
of the encoder, and compare it with the optimal decoding filter designed by sampled-data
$H^\infty$ optimization.
We can see that the proposed system attenuates the gain over the whole frequency.
It follows that the channel noise on the received signal will be attenuated more than
the conventional system.
\begin{figure}[t]
\begin{center}
 \includegraphics[width=0.7\linewidth]{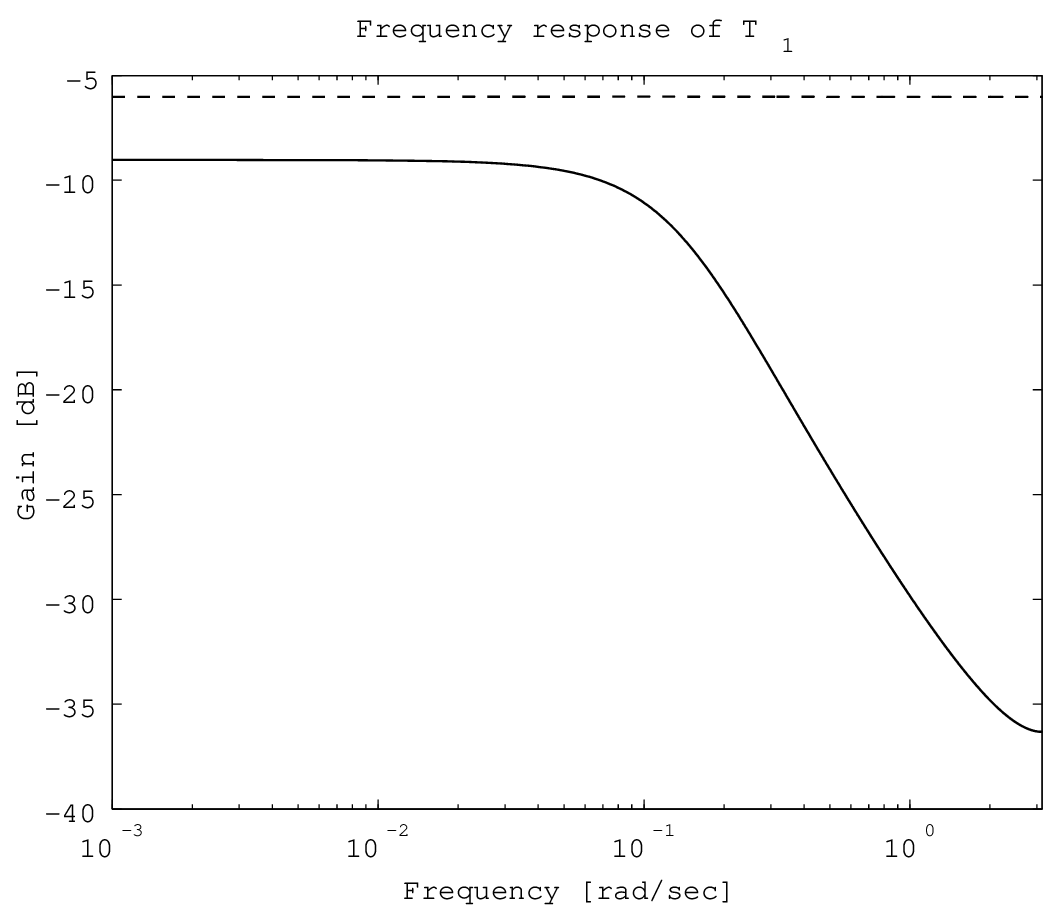}
 \caption{Frequency responses of $\T_1$: proposed (solid) and conventional (dash)}
 \label{fig:qsd_freqT1}
\end{center}
\end{figure}
\begin{figure}[t]
\begin{center}
 \includegraphics[width=0.7\linewidth]{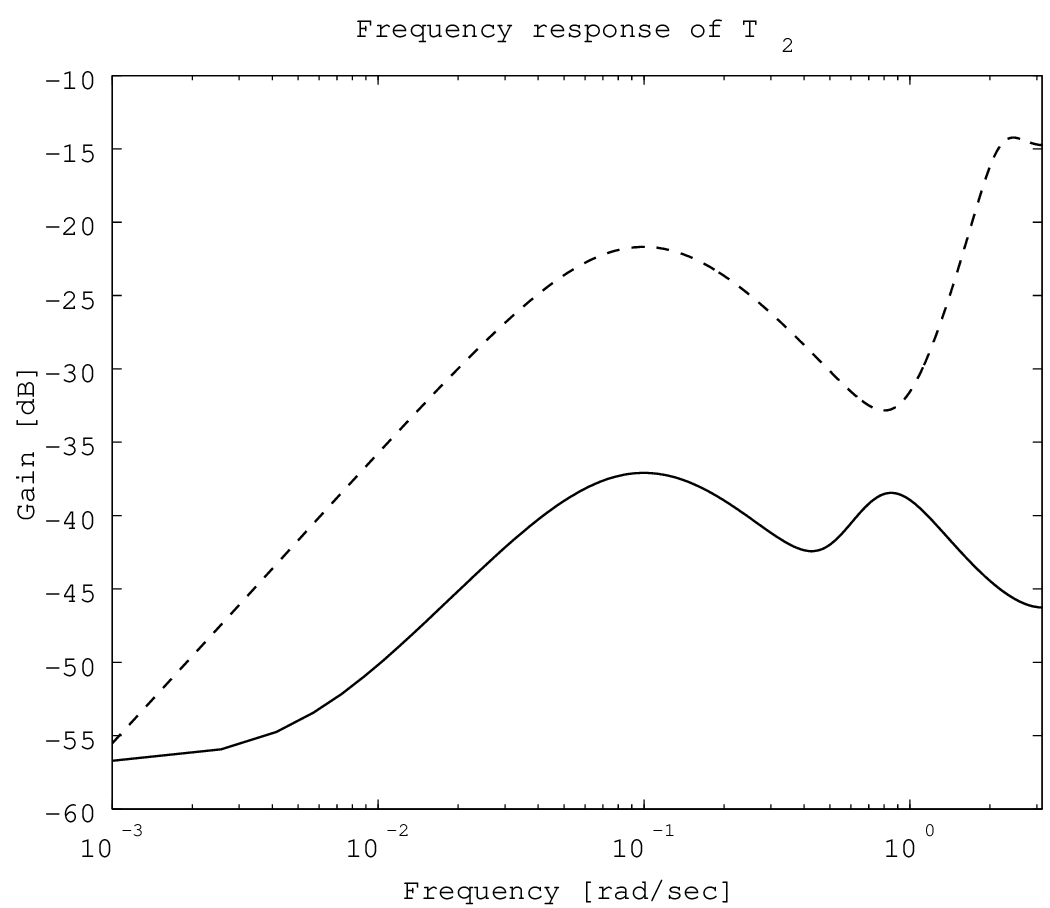}
 \caption{Frequency responses of $\T_2$: proposed (solid) and conventional (dash)}
 \label{fig:qsd_freqT2}
\end{center}
\end{figure}

Then we show a simulation for the obtained DPCM systems.
The parameters are as follows:
the quantization level $\Delta = 0.125$, the input $r = \sin(\frac{\pi}{10}t)$, the channel noise $n = 0.1\sin(2t)$
and the sampling period $h=1$.
Figure \ref{fig:qsd_time_sd} shows the time response of the sampled-data designed DPCM system,
and Figure \ref{fig:qsd_time_dt} shows that of the conventional $\Delta$ modulation system.
We can see that the proposed system attenuates the channel noise 
considerably better than the conventional system.
\begin{figure}[t]
\begin{center}
 \includegraphics[width=0.7\linewidth]{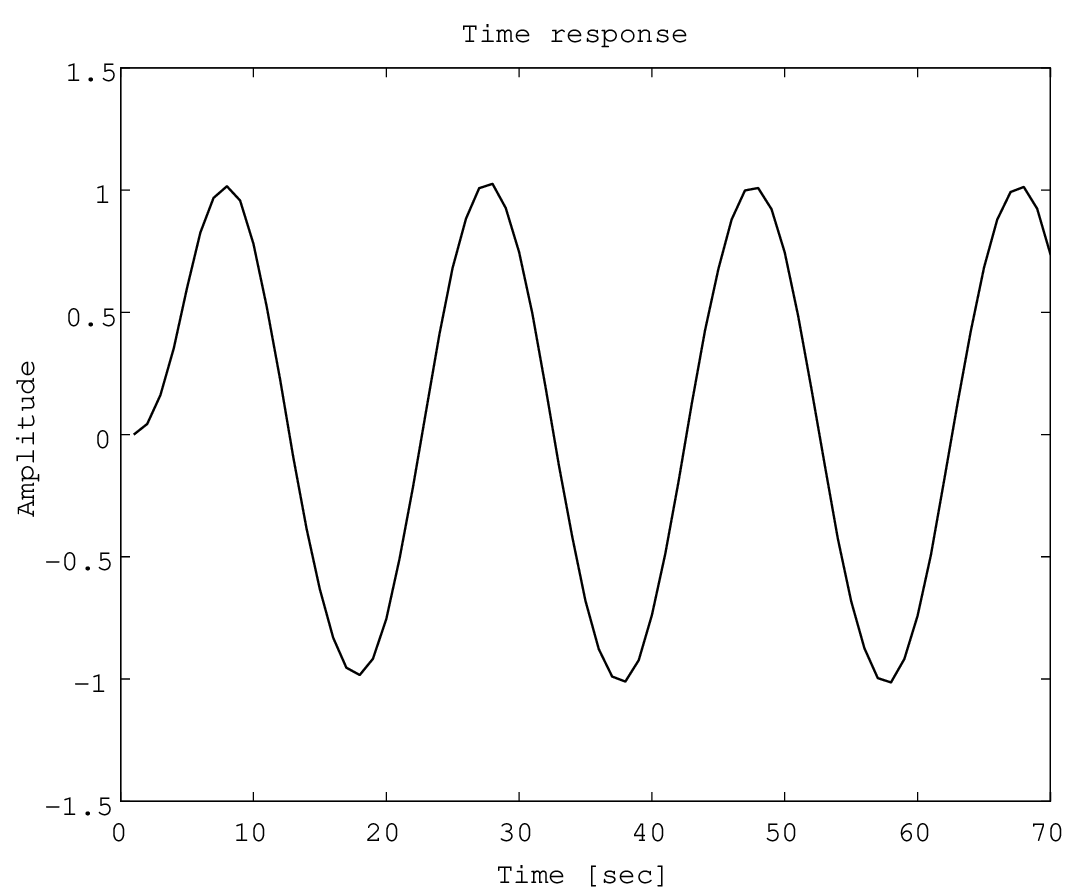}
 \caption{Time response of sampled-data designed system}
 \label{fig:qsd_time_sd}
\end{center}
\end{figure}
\begin{figure}[t]
\begin{center}
 \includegraphics[width=0.7\linewidth]{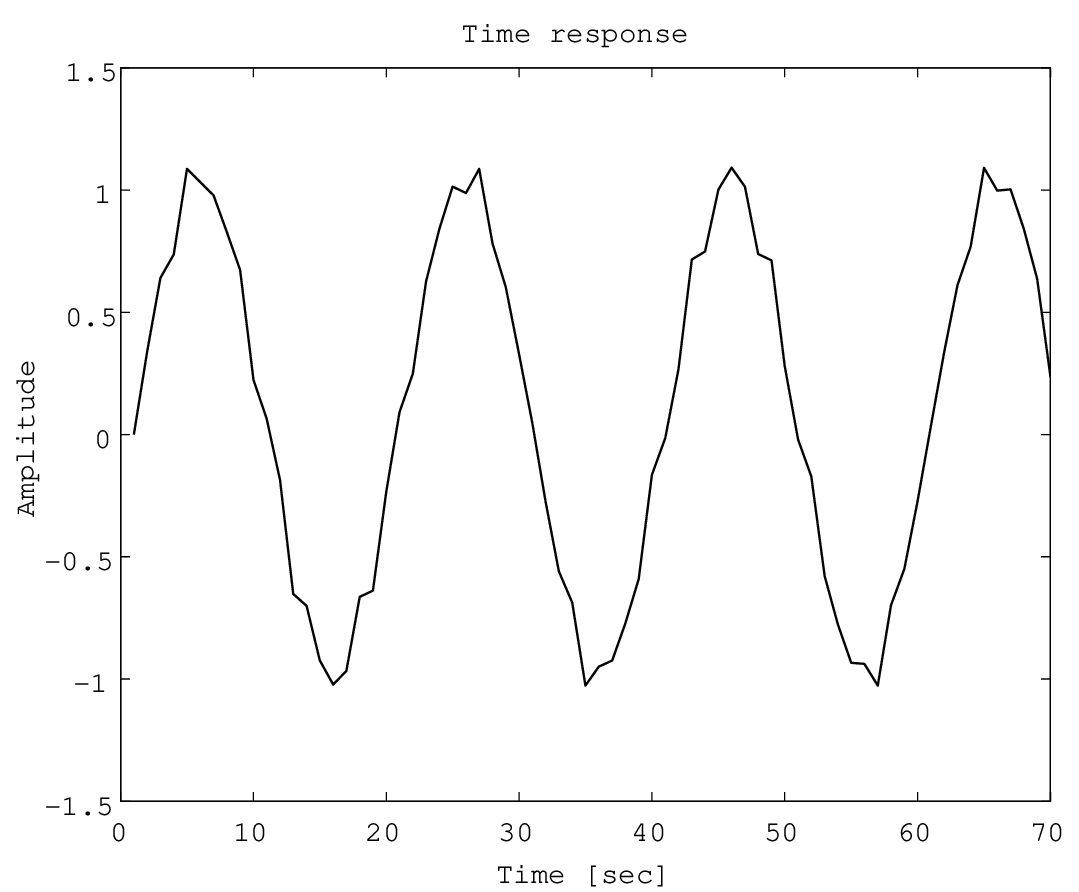}
 \caption{Time response of conventional $\Delta$ modulation system}
 \label{fig:qsd_time_dt}
\end{center}
\end{figure}

\section{Conclusion}
In this chapter, we have discussed the stability and the performance of quantized sampled-data control systems.
We have adopted the additive noise model, and by using it we have shown
the BIBO stability and the performance of quantized systems.
Moreover, we have proposed a new method for designing DPCM systems.
Since the conventional $\Delta$ modulation system is not stable,
a channel noise can be amplified at the encoder,
while our system will attenuate the channel noise.

However, there remains an issue. We have not dealt with
saturation in a quantizer.
The additive noise model can be applied effectively in the case that
the noise is small, but with saturation, the noise can be enormous.
In this case, we have to treat the quantizer as a nonlinear system.
We consider that hybrid system theory,
in particular, switching system theory may be applicable to that case.
%\input{firapp}
%\newcommand{\norm}[1]{\left\| #1 \right\|}
%\newcommand{\Hinfnorm}[1]{\norm{#1}_{\infty}}
%%%%%%%%%%%%%%%%%%%%%%%%%%%%%%%%%%%%%%%%%%%%%%%%%%%%%%%%%%%%%%%%%%%%%%%%%%%%%

\chapter{Optimal FIR Approximation}
\label{ch:firapp}

\section{Introduction}
According to  the method we have discussed in the previous chapters,
the filter we obtain is an IIR (Infinite Impulse Response) filter:
\begin{equation}
  F(z) = 
   \frac{\sum^{M}_{k=0} a_{k} z^{-k}}{1+\sum^{N}_{k=1}b_{k}z^{-k}}.
\label{eq:fir_IIR}
\end{equation}
The structure of an IIR filter is shown in Figure~\ref{fig:fir_IIR}.
\begin{figure}[t]
\begin{center}
\includegraphics[width=.4\linewidth]{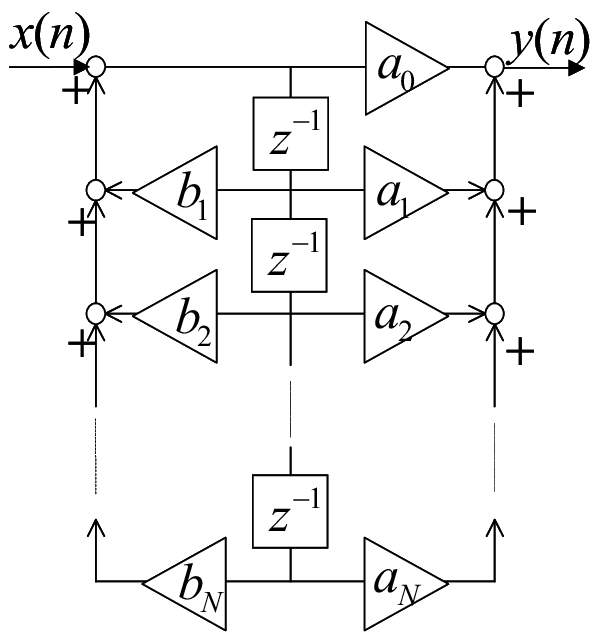}
\end{center}
\caption{IIR filter}
\label{fig:fir_IIR}
\end{figure}

In practice, FIR (Finite Impulse Response) filters are often preferred to IIR ones.
They have finitely many nonzero Markov parameters:
\begin{equation} \label{eq:fir_FIR}
  F(z)=\sum^{M}_{k=0} a_{k} z^{-k}.
\end{equation}
The structure of an FIR filter is shown in Figure~\ref{fig:fir_FIR}.
\begin{figure}[t]
\begin{center}
\includegraphics[width=.4\linewidth]{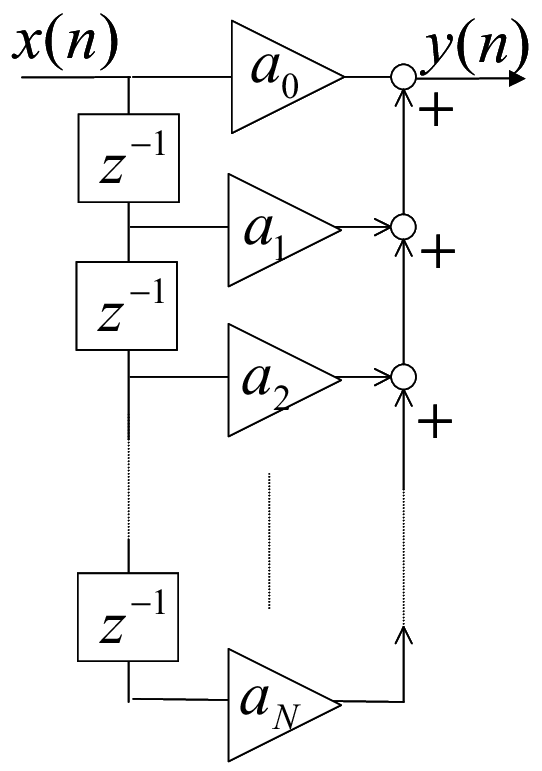}
\end{center}
\caption{FIR filter}
\label{fig:fir_FIR}
\end{figure}

The reasons why FIR filters are often preferred to 
IIR filters are as follows \cite{Zel}:

\begin{itemize}
\item FIR filters are intrinsically stable; the stability issue
is a non-issue.

\item They can easily realize various features that are not possible or are
difficult to achieve with IIR filters, e.g., linear phase property.

\item They can be free from certain problems in implementation,
for example, limit cycles, attributed to quantization and 
the existence of a feedback loop in IIR filters.
\end{itemize}

On the other hand, 
a design process may have to start with an IIR filter
for variety of reasons.  For example, we have a large number
of continuous-time filters available, and a digital filter
may be obtained by discretizing one of them.  It is then desired that
such an IIR filter be approximated by an FIR filter.
An easy way of doing this is to just truncate the Markov parameters
of the IIR filter at a desired number of steps.  This may however lead
to either a very high-dimensional filter with good approximation, or a lower
dimensional filter with an unsatisfactory approximation, depending on the
truncation point.

The following problem is thus very natural and of importance:

\begin{problem}
Given an IIR filter $K(z)$ and a positive integer $N$, find
an optimal FIR approximant $K_{f}(z)$ that has order $N$
and approximates $K(z)$ with respect to a certain performance measure.
\end{problem}

There is a very elegant method called the {\em Nehari shuffle},
proposed by Kootsookos et al.\ \cite{Kootsookos92,Kootsookosbkchap}.
Its basic idea may be described as follows: 
For a given IIR filter $G$ and a desired degree $r-1$ that
an approximating FIR filter should assume, one first
truncate the impulse response of $G$ to the first $r$ steps.
This is a mere truncation, and it may induce a large error.
One then takes its residual $G_{1}$, and suitably shifting and 
taking the mirror image, one can reduce this to the
situation of the Nehari extension (approximation) problem.
This will induce a truncation in the second step.  By taking
the residual further, this process can be continued, and
the approximation can be improved in each step.  (Details
may be found in \cite{ObinataAndersonbk}.  
An advantage here is that this procedure 
gives rise to certain a priori and a posteriori error bounds.
On the other hand, it does not necessarily give an 
optimal approximation 
with respect to the $H^{\infty}$-norm.

In contrast to the Nehari shuffle, we here propose a method
that directly deals with (sub)optimal approximants with respect to
the $H^{\infty}$ error norm.  It is shown that
\begin{itemize}
\item the design problem is reducible to a Linear Matrix Inequality(LMI) \cite{Boydetal}; and

\item the obtained filter can be made close to be optimal by
an iterative procedure.
\end{itemize}
A comparison with the Nehari shuffle is made for the Chebyshev
filter of order 8, which has been studied in detail in
\cite{Kootsookosbkchap}.

\section{FIR approximation problem}
Consider the block diagram Figure~\ref{fig:fir_syn}.
\begin{figure}[t]
\begin{center}
%WinTpicVersion2.15
\unitlength 0.1in
\begin{picture}(29.00,12.00)(4.00,-14.00)
% VECTOR 2 0 3 0
% 4 400 1200 800 1200 800 1000 800 1000
%
\special{pn 8}%
\special{pa 400 800}%
\special{pa 800 800}%
\special{fp}%
\special{sh 1}%
\special{pa 800 800}%
\special{pa 733 780}%
\special{pa 747 800}%
\special{pa 733 820}%
\special{pa 800 800}%
\special{fp}%
\special{pa 800 600}%
\special{pa 800 600}%
\special{fp}%
% BOX 2 0 3 0
% 2 800 1000 1200 1400
%
\special{pn 8}%
\special{pa 800 600}%
\special{pa 1200 600}%
\special{pa 1200 1000}%
\special{pa 800 1000}%
\special{pa 800 600}%
\special{fp}%
% LINE 2 0 3 0
% 4 1200 1200 1600 1200 1200 1200 1600 1200
%
\special{pn 8}%
\special{pa 1200 800}%
\special{pa 1600 800}%
\special{fp}%
\special{pa 1200 800}%
\special{pa 1600 800}%
\special{fp}%
% LINE 2 0 3 0
% 2 1600 1600 1600 800
%
\special{pn 8}%
\special{pa 1600 1200}%
\special{pa 1600 400}%
\special{fp}%
% VECTOR 2 0 3 0
% 2 1600 800 2000 800
%
\special{pn 8}%
\special{pa 1600 400}%
\special{pa 2000 400}%
\special{fp}%
\special{sh 1}%
\special{pa 2000 400}%
\special{pa 1933 380}%
\special{pa 1947 400}%
\special{pa 1933 420}%
\special{pa 2000 400}%
\special{fp}%
% BOX 2 0 3 0
% 2 2000 600 2400 1000
%
\special{pn 8}%
\special{pa 2000 200}%
\special{pa 2400 200}%
\special{pa 2400 600}%
\special{pa 2000 600}%
\special{pa 2000 200}%
\special{fp}%
% VECTOR 2 0 3 0
% 2 1600 1600 2000 1600
%
\special{pn 8}%
\special{pa 1600 1200}%
\special{pa 2000 1200}%
\special{fp}%
\special{sh 1}%
\special{pa 2000 1200}%
\special{pa 1933 1180}%
\special{pa 1947 1200}%
\special{pa 1933 1220}%
\special{pa 2000 1200}%
\special{fp}%
% BOX 2 0 3 0
% 2 2000 1400 2400 1800
%
\special{pn 8}%
\special{pa 2000 1000}%
\special{pa 2400 1000}%
\special{pa 2400 1400}%
\special{pa 2000 1400}%
\special{pa 2000 1000}%
\special{fp}%
% LINE 2 0 3 0
% 2 2400 800 2800 800
%
\special{pn 8}%
\special{pa 2400 400}%
\special{pa 2800 400}%
\special{fp}%
% VECTOR 2 0 3 0
% 2 2800 800 2800 1100
%
\special{pn 8}%
\special{pa 2800 400}%
\special{pa 2800 700}%
\special{fp}%
\special{sh 1}%
\special{pa 2800 700}%
\special{pa 2820 633}%
\special{pa 2800 647}%
\special{pa 2780 633}%
\special{pa 2800 700}%
\special{fp}%
% VECTOR 2 0 3 0
% 2 2800 1600 2800 1300
%
\special{pn 8}%
\special{pa 2800 1200}%
\special{pa 2800 900}%
\special{fp}%
\special{sh 1}%
\special{pa 2800 900}%
\special{pa 2780 967}%
\special{pa 2800 953}%
\special{pa 2820 967}%
\special{pa 2800 900}%
\special{fp}%
% LINE 2 0 3 0
% 2 2800 1600 2400 1600
%
\special{pn 8}%
\special{pa 2800 1200}%
\special{pa 2400 1200}%
\special{fp}%
% CIRCLE 2 0 3 0
% 4 2800 1200 2900 1200 2900 1200 2900 1200
%
\special{pn 8}%
\special{ar 2800 800 100 100  0.0000000 6.2831853}%
% VECTOR 2 0 3 0
% 2 2900 1200 3300 1200
%
\special{pn 8}%
\special{pa 2900 800}%
\special{pa 3300 800}%
\special{fp}%
\special{sh 1}%
\special{pa 3300 800}%
\special{pa 3233 780}%
\special{pa 3247 800}%
\special{pa 3233 820}%
\special{pa 3300 800}%
\special{fp}%
% STR 2 0 3 0
% 3 2900 900 2900 1000 2 0
% $-$
\put(29.0000,-6.0000){\makebox(0,0)[lb]{$-$}}%
% STR 2 0 3 0
% 3 2900 1300 2900 1400 1 0
% $+$
\put(29.0000,-10.0000){\makebox(0,0)[lt]{$+$}}%
% STR 2 0 3 0
% 3 1000 1100 1000 1200 5 0
% $W$
\put(10.0000,-8.0000){\makebox(0,0){$W(z)$}}%
% STR 2 0 3 0
% 3 2200 700 2200 800 5 0
% $K(z)$
\put(22.0000,-4.0000){\makebox(0,0){$K(z)$}}%
% STR 2 0 3 0
% 3 2200 1500 2200 1600 5 0
% $K_f(z)$
\put(22.0000,-12.0000){\makebox(0,0){$K_f(z)$}}%
% STR 2 0 3 0
% 3 450 1050 450 1150 2 0
% $w$
\put(4.5000,-7.5000){\makebox(0,0)[lb]{$w$}}%
% STR 2 0 3 0
% 3 3250 1050 3250 1150 2 0
% $e$
\put(32.5000,-7.5000){\makebox(0,0)[lb]{$e$}}%
\end{picture}%
\end{center}
\caption{Error system}
\label{fig:fir_syn}
\end{figure}
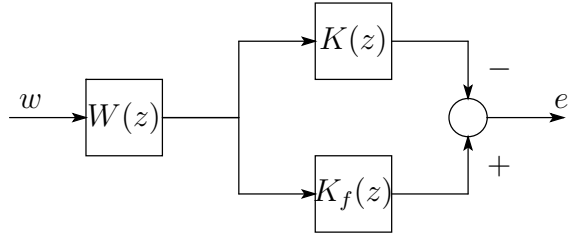
$K(z)$ is a given (rational and stable) IIR filter,
$W(z)$ is a proper and rational weighting function, and
$K_{f}(z)$ is an FIR filter of order $N$.  
Denote by $T_{ew}(z)$ the transfer function from $w$ to $e$ in
Figure~\ref{fig:fir_syn}.  
The objective here is
to find $K_{f}(z)$ that makes the $H^{\infty}$ error norm less than a
prespecified bound $\gamma>0$, that is,
\[
\|T_{ew}\|_\infty := \sup_{\substack{w\in l^2\\w\neq 0}}\frac{\|T_{ew}w\|_2}{\|w\|_2} < \gamma.
\]

Introduce state space realizations
\begin{equation*}
	\begin{split}
	W(z) &:= C_W(zI-A_W)^{-1}B_W+D_W,\\
	K(z) &:= C_K(zI-A_K)^{-1}B_K+D_K,
	\end{split}
\end{equation*}
and
\begin{equation*}
	\begin{split}
	K_f(z)&=\sum_{k=0}^{N} a_kz^{-k}=C_f(\alpha)(zI-A_f)^{-1}B_f+D_f(\alpha),\\
	%C_f(\alpha)&=&\begin{bmatrix}a_N&a_{N-1}&\ldots&a_1\end{bmatrix},\quad
	C_f(\alpha)&=\left[
                       \begin{array}{cccc}
                        a_N & a_{N-1} & \ldots & a_1
                       \end{array}
                      \right],\quad
        D_f(\alpha)=a_0,\\
	%\alpha=\begin{bmatrix}a_N&a_{N-1}&\ldots&a_1\end{bmatrix}
        \alpha&=\left[
                \begin{array}{cccc}
                 a_N & a_{N-1} & \ldots & a_0
                \end{array}
               \right],
\label{eq:fir_fir}
\end{split}
\end{equation*}
where $\alpha=\left[ 
                \begin{array}{cccc}
                 a_N & a_{N-1} & \ldots & a_0
                 \end{array} 
                \right]$ denote
the Markov parameters of the filter $K_{f}(z)$ to be designed. 
The matrices $A_f$ and $B_f$ are defined as follows:
\begin{equation*}
A_{f}= \left[
          \begin{array}{ccccc}
           0 & 1 & 0 & \cdots & 0 \\
           \vdots & \ddots & \ddots & \ddots & \vdots \\
           \vdots &  & \ddots & \ddots & 0 \\
           \vdots &  & & \ddots & 1 \\
           0 & \cdots &\cdots & \cdots & 0
          \end{array}
         \right],\quad
B_{f} = \left[
          \begin{array}{c}
          0 \\ \vdots \\ 0 \\1
          \end{array}
          \right],
\end{equation*}
and they contain just zeros and ones.

A realization of $T_{ew}$ is given as follows:
\begin{equation*}
\begin{split}
%\hspace*{1cm}
T_{ew}(z) &=: C(\alpha)(zI-A)^{-1}B+D(\alpha),\\
A&=\left[
              \begin{array}{ccc}
               A_W & 0 & 0 \\
               B_{K}C_W & A_{K} & 0   \\
               B_{f}C_W & 0 & A_{f}  \\ 
              \end{array}
        \right], \quad
B=\left[
       \begin{array}{c}
       B_W\\
       B_{K}D_W\\
       B_{f}D_W
       \end{array}
       \right], \\
C(\alpha)&=\left[
       \begin{array}{ccc}
        (D_{f}(\alpha) - D_{K})C_W & -C_{K} & C_{f}(\alpha) 
       \end{array}
      \right],\\
D(\alpha)&= \left[
          \begin{array}{c}
               (D_{K}+D_{f}(\alpha))D_W 
              \end{array}
             \right].
\end{split}        	
\end{equation*}

The important asset here is that the design parameter $\alpha$ appears
only in the $C$ and $D$ matrices linearly, and the underlying structure is of
the one-block type.  Hence the overall transfer operator is linear in
$\alpha$, and  the design problem of choosing $\alpha$ to minimize the
$H_{\infty}$-norm can be expected to become a linear matrix 
inequality.  In fact, the bounded real lemma
\cite{Boydetal} readily yields the following:
\begin{theorem} \label{thm:fir_main}
$\|T_{ew}\|_\infty<\gamma$ if and only if there exists
$P>0$ such that
\begin{eqnarray}
	\left[\begin{array}{ccc}
		A^TPA-P&A^TPB&C(\alpha)^T\\
		B^TPA&-\gamma I+B^TPB&D(\alpha)^T\\
		C(\alpha)&D(\alpha)&-\gamma I
	\end{array}\right]<0.
\label{eq:fir_TH1}
\end{eqnarray}
\end{theorem}
\begin{proof}
By the bounded real lemma \cite{Boydetal}, $\|T_{ew}\|_\infty<\gamma$ is equivalent to the condition 
that there exists a matrix $\tilde{P}>0$ such that
\begin{eqnarray}
        Q^{T}
         \left[
          \begin{array}{cc}  
           \tilde{P} & 0 \\
           0 & I
          \end{array}
         \right]
         Q 
        < 
        \left[\begin{array}{cc}
	\tilde{P}&0\\0&\gamma^2I
	\end{array}\right], 
\label{eq:fir_LMI}
\end{eqnarray}
where
\begin{eqnarray*}
        Q :=
        \left[
         \begin{array}{cc}
          A & B \\
          C(\alpha) & D(\alpha)
         \end{array}
        \right].  
%\end{split}
\end{eqnarray*}
Although the inequality (\ref{eq:fir_LMI}) is not affine in $\alpha$,
it can be converted to an affine one by the Schur complement \cite{Boydetal}:
\begin{equation}
	\left[
	  \begin{array}{cc}
	    \Phi_{11} & \Phi_{12}\\
	    \Phi_{12}^T & \Phi_{22}
	  \end{array}
	\right]
	<0,    \nonumber
\end{equation}
is equivalent to $\Phi_{22}<0$ and $\Phi_{11}<\Phi_{12}\Phi_{22}^{-1}\Phi_{12}^T$.
By dividing the inequality (\ref{eq:fir_LMI}) by $\gamma>0$, we get
\begin{equation*}
	\left[
	  \begin{array}{cc}
	    A^T\gamma^{-1}\tilde{P}A-\gamma^{-1}\tilde{P} & A^T\gamma^{-1}\tilde{P}B\\
	    B^T\gamma^{-1}\tilde{P}A & B^T\gamma^{-1}\tilde{P}B-\gamma I
	  \end{array}
	\right]
	< 
	\left[
	  \begin{array}{c}
	    C^T\\ 
	    D^T
	  \end{array}
	\right]
	(-\gamma^{-1}I)
	\left[
	  \begin{array}{cc}
	    C & D
	  \end{array}
	\right].
\end{equation*}
Then by using the Sure complement for
\begin{eqnarray*}
	\Phi_{11}&:=&
	\left[
	  \begin{array}{cc}
	    A^TPA-P & A^TPB\\
	    B^TPA & B^TPB-\gamma I
	  \end{array}
	\right]^T,\\
	\Phi_{22}&:=&-\gamma I,\\
	\Phi_{12}&:=&
	\left[
	  \begin{array}{cc}
	    C & D
	  \end{array}
	\right]^T,
\end{eqnarray*}
where $P:=\gamma^{-1}\tilde{P}>0$,
we get the inequality (\ref{eq:fir_TH1}).
\end{proof}	

The obtained condition is an LMI in $\alpha$, and can be
effectively solved by standard MATLAB routines \cite{MATLAB}.

\section{Numerical example}
\subsection{Comparison of $H^\infty$ design via LMI and the Nehari shuffle}
Take the following Chebyshev filter of order $8$
{\footnotesize
\begin{equation*}
\begin{split}
K(z) &= 10^{-3}\\
&\times\frac{0.04705z^{8}+0.3764z^{7}+1.317z^{6}+2.635z^{5}+3.294z^{4}+2.635z^{3}+1.317z^{2}+0.3764z+0.04705}
{z^{8}-4.953z^{7}+11.71z^{6}-16.95z^{5}+16.29z^{4}-10.58z^{3}+4.552z^{2}-1.161z+0.1369}
\end{split}
\end{equation*}
}as a target filter to be approximated.  This has been studied thoroughly
by Kootsookos and Bitmead \cite{Kootsookosbkchap} for the
Nehari shuffle, and is suitable for comparison with the
present method.  For simplicity, 
we confine ourselves to approximations by FIR filters with $32$ tap
coefficients (of order $31$).

\begin{figure}[t]
\begin{center}
\includegraphics[width=.7\linewidth]{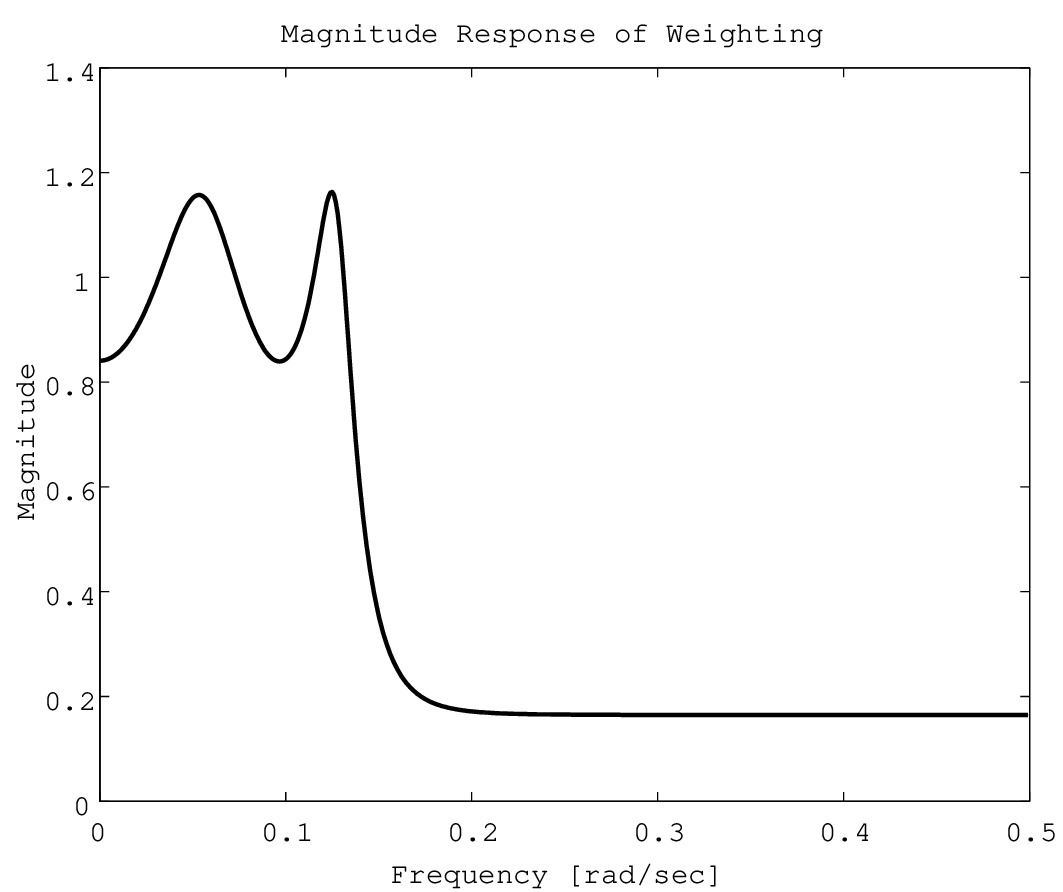}
\end{center}
\caption{Inverse of the weighting function}
\label{fig:fir_weight}
\end{figure}

The design depends crucially on the choice of the weight $W(z)$.
One natural choice (\cite{ObinataAndersonbk})
would be to take $W(z)$ to be equal to $K^{-1}(z)$
(or some variant of it having the same gain on the imaginary axis,
since $K$ is not minimum phase).
This is relative error approximation, where (approximately) dB and phase
errors are weighted uniformly with frequency. Since the optimal 
overall error in
Figure~\ref{fig:fir_syn} will become all-pass, this will have the effect of
attenuating the stop-band error with the weight of
$K^{-1}(z)$ (which is very large) while maintaining reasonable pass-band
characteristic.  Unfortunately, however, due to the very small gain 
of $K(z)$ in
the stop-band, this will make the solution of the approximation problem
Figure~\ref{fig:fir_syn}  numerically hard.  Neither the Nehari shuffle nor the
LMI method gave a satisfactory result in this case.  Hence one should sacrifice
the stop-band attenuation to obtain a reasonable $W(z)$.  There is also a
trade-off, empirically observed, between the stop-band attenuation
and the pass-band ripples.

Kootsookos and Bitmead \cite{Kootsookosbkchap} thus employed the
weight as depicted in Figure~\ref{fig:fir_weight}.
To be precise, the frequency response shown here is the
inverse of the para-Hermitian conjugate of the weight
function.  The reason for taking the para-Hermitian conjugate
is that the Nehari shuffle makes use of causal approximation
for anti-causal transfer function, so that we must 
reciprocate the poles and zeros.  Then by taking the inverse,
the weight attenuates the stop-band by the inverse of
its gain and approximately shapes the pass-band as
it is in the pass band.  On the other hand, for the
FIR approximation as in Figure~\ref{fig:fir_syn}, we 
simply take the inverse of this weight, since we do
not need to make the weight anti-stable.

The gain responses of obtained FIR filters
based on the Nehari shuffle and 
Theorem \ref{thm:fir_main} are given in Figure~\ref{fig:fir_filter_w}.  
Figure \ref{fig:fir_filter_w_p} shows their phase plots.

\begin{figure}[t]
\begin{center}
\includegraphics[width=.7\linewidth]{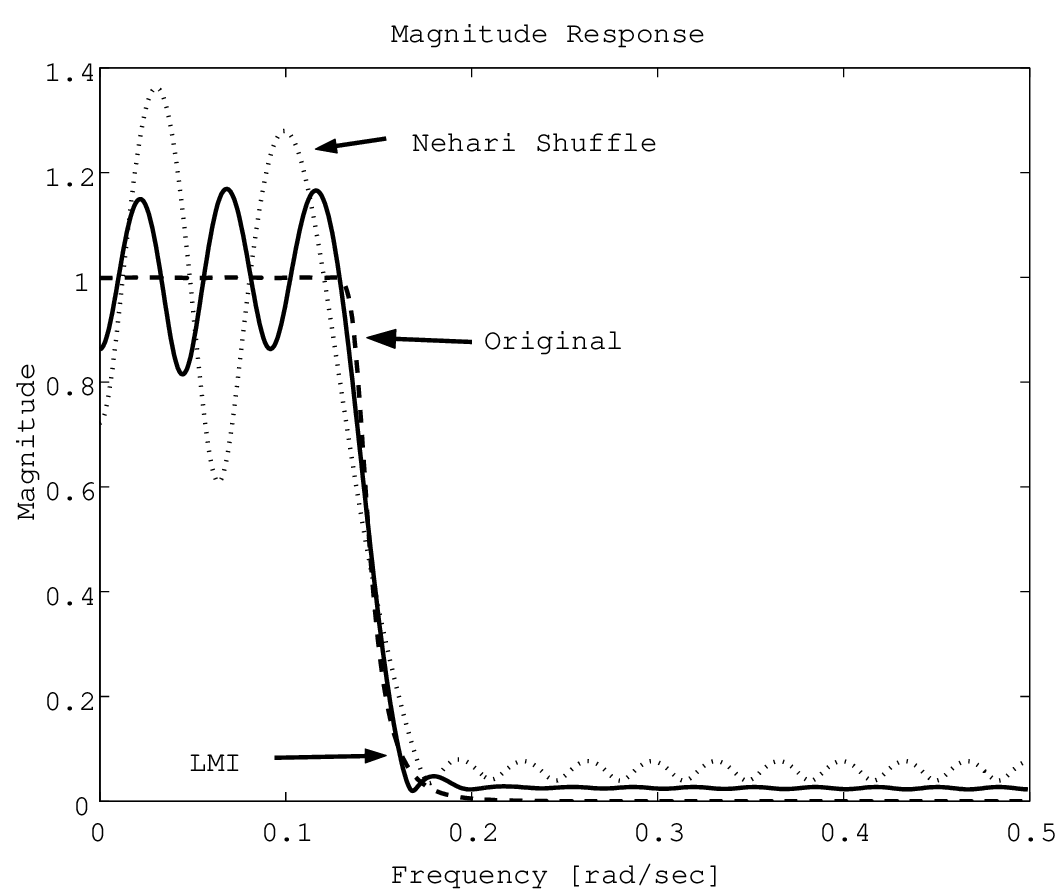}
\end{center}
\caption{Gain responses of FIR approximants
with weight function in Figure~\ref{fig:fir_weight}: $H^{\infty}$ via
LMI (solid), Nehari shuffle (dots) and
original IIR Chebyshev filter (dash)}
\label{fig:fir_filter_w}
\end{figure}

\begin{figure}[t]
\begin{center}
\includegraphics[width=.7\linewidth]{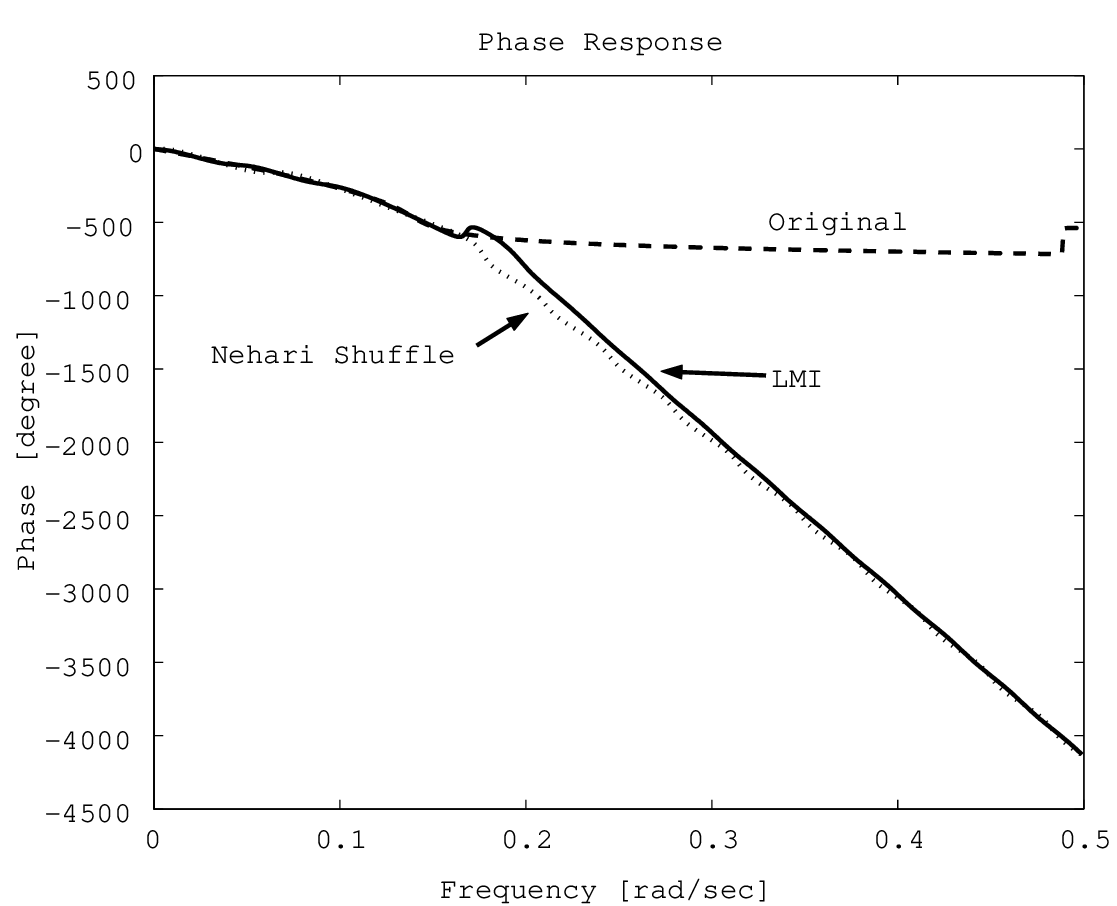}
\end{center}
\caption{Phase plots of FIR approximants
with weight function in Figure~\ref{fig:fir_weight}: $H^{\infty}$ via
LMI (solid), Nehari shuffle (dots) and
Chebyshev (dash)}
\label{fig:fir_filter_w_p}
\end{figure}

We see that the gain of the $H^{\infty}$ 
approximant shows smaller pass-band ripples and 
better stop-band attenuation than those 
by the Nehari shuffle.  
The phase characteristics of these are about the same up to the
edge of the transition band.  

Figure~\ref{fig:fir_error_w}
shows the error magnitude responses.
The design by the LMI method has the advantage of 5--7 dB
over the one by the Nehari shuffle.

\begin{figure}[t]
\begin{center}
\includegraphics[width=.7\linewidth]{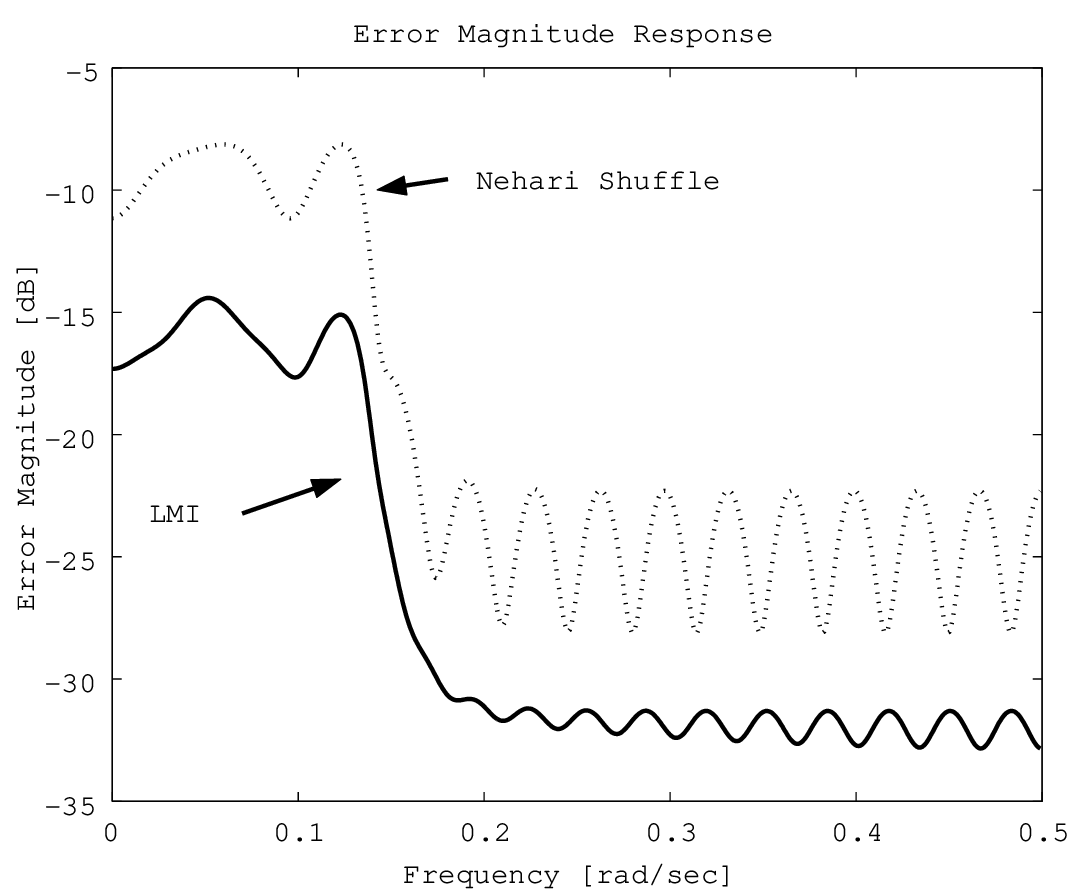}
\end{center}
\caption{Gain of the error $K-K_f$: $H^\infty$ design
via LMI (solid); Nehari shuffle (dots)}
\label{fig:fir_error_w}
\end{figure}

%\begin{figure}[htbh]
%\begin{center}
%\includegraphics[width=.7\linewidth]{error_w_p}
%\end{center}
%\caption{Phase characteristics of the error $K-K_f$: $H^\infty$ design via LMI (solid); Nehari shuffle (dots)}
%\label{fig:error_w_p}
%\end{figure}

\subsection{Trade-off between pass-band and stop-band characteristics}
The design in the previous subsection depends crucially
on the weighting function.  It is desirable to 
obtain smaller pass-band ripples while maintaining 
reasonable stop-band attenuation.  
In this section we attempt to see how the choice of a
weighting function affects the overall approximation.

We consider the following three weighting functions:
{\small
\begin{equation*}
\begin{split}
W_{1} &= \frac{0.7661 z^{2} - 1.305 z + 0.675}{z^{2} - 1.735 z + 0.9289},\\
\\
W_{2} &= \frac{0.2831 z^{4} - 0.5515 z^{3} + 0.5416 z^{2}- 0.2708 z +0.05882}
	{z^{4} - 2.865 z^{3} + 3.6 z^{2}- 2.268 z + 0.6056},\\
\\
W_{3} &= 10^{-3}\times\frac{14.44 z^{7} - 7.838 z^{6} + 19.02 z^{5}-4.448+ 6.697 z^{3} - 0.1857 z^{2}+ 0.5287 z +0.01134}
{z^{7} - 4.229 z^{6} + 8.561 z^{5}-10.43z^{4}+ 8.172 z^{3} - 4.089z^{2}+ 1.206 z - 0.1613}.
\end{split}
\end{equation*}
}
These functions $W_{1}$, $W_{2}$, $W_{3}$ are,
respectively, obtained as
the 2nd, 4th, 7th-order Hankel norm approximations
\cite{ObinataAndersonbk} 
of the IIR Chebyshev filter to be approximated.
The weight $W_{2}$ is the same as that used in the previous section.
Their magnitude frequency responses are shown in
Figure~\ref{fig:fir_pooh}.  % and Figure \ref{fig:pooh2}.

\begin{figure}[t]
  \begin{center}
    \includegraphics[width=.7\linewidth]{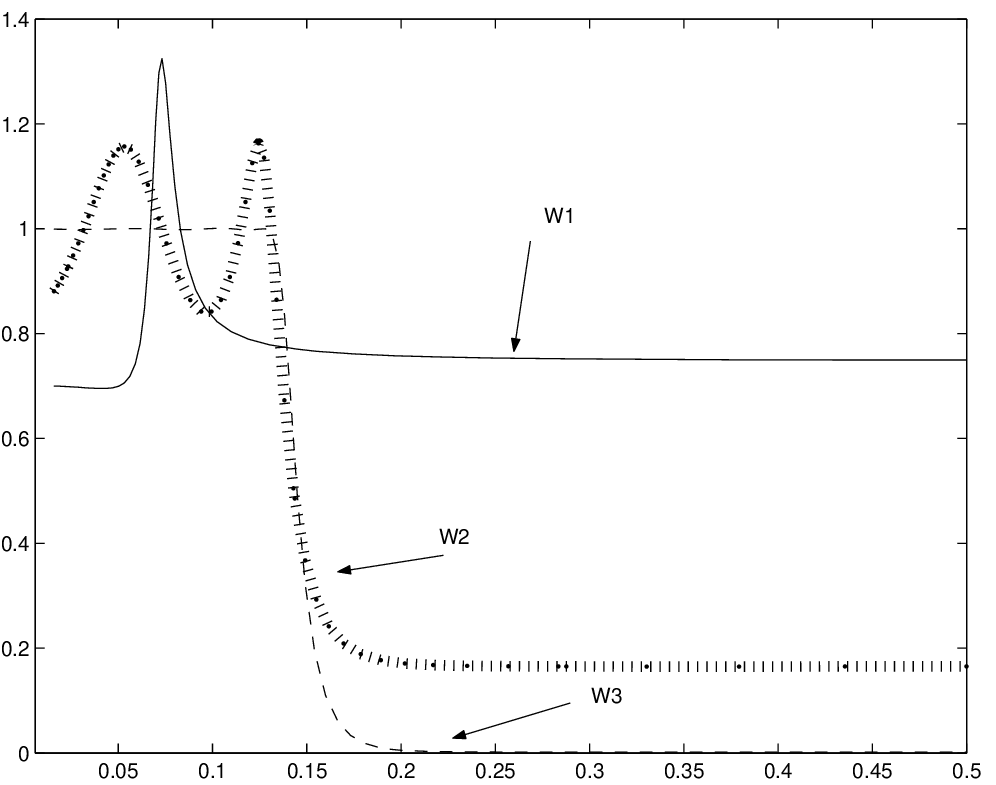}
  \end{center}
  \caption{Gain of inverse of weighting functions}
%: $W_1$(solid); $W_2$(dots); $W_3$(broken)}
  \label{fig:fir_pooh}
\end{figure}

%\begin{figure}[htbp]
%  \begin{center}
%    \includegraphics[width=.7\linewidth]{weightpha.eps}
%  \end{center}
%  \caption{Phase plots of the weighting functions}
% : $W1$(---):$W2$($\cdots$); $W3$(- - -)}
%  \label{fig:pooh2}
%\end{figure}

Figure~\ref{fir_chip} and Figure~\ref{fir_chip2} show the resulting
gain and phase responses of the FIR filters
designed with respective weighting functions.
Figure~\ref{fir_chip} in particular shows that there is a
clear trade-off between the magnitude of the pass-band ripples and 
the stop-band attenuation.  That is, if we attempt to
decrease the stop-band error, we must sacrifice the
pass-band characteristic (i.e., larger ripples), and
vice versa.  

Figure~\ref{fir_dale} shows the error magnitude responses.  
Table 1 also shows the $H^{\infty}$ and $H^{2}$ error norms.
Interestingly, the design via $W_{1}$ exhibits the
best overall approximation in both performance measures, although
its stop-band attenuation is not as good as those
by $W_{2}$ and $W_{3}$.  Note also that the one by
$W_{3}$ approximates the phase characteristic of the
original filter up to the edge of the stop-band as
Figure~\ref{fir_chip2} shows.

\begin{figure}[t]
  \begin{center}
    \includegraphics[width=.7\linewidth]{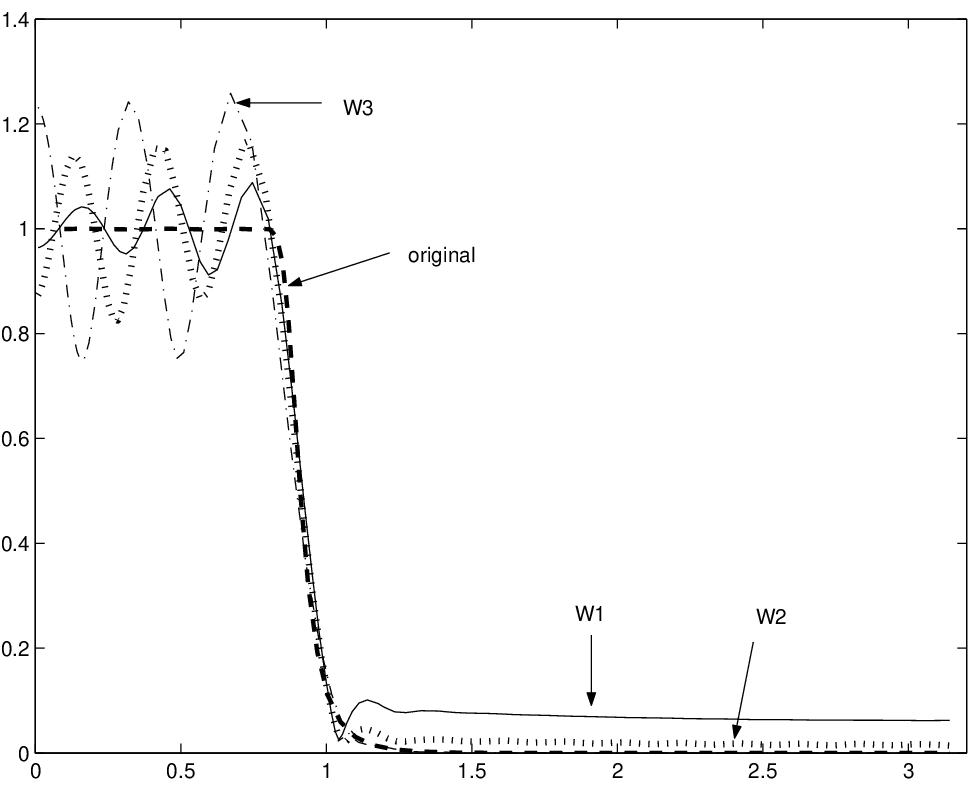}
  \end{center}
  \caption{Gain responses of FIR filters via LMI}
%: $W1$(---):$W2$($\cdots$); $W3$(- $\cdot$ -); Chebyshev(- - -)}
  \label{fir_chip}
\end{figure}

\begin{figure}[t]
  \begin{center}
    \includegraphics[width=.7\linewidth]{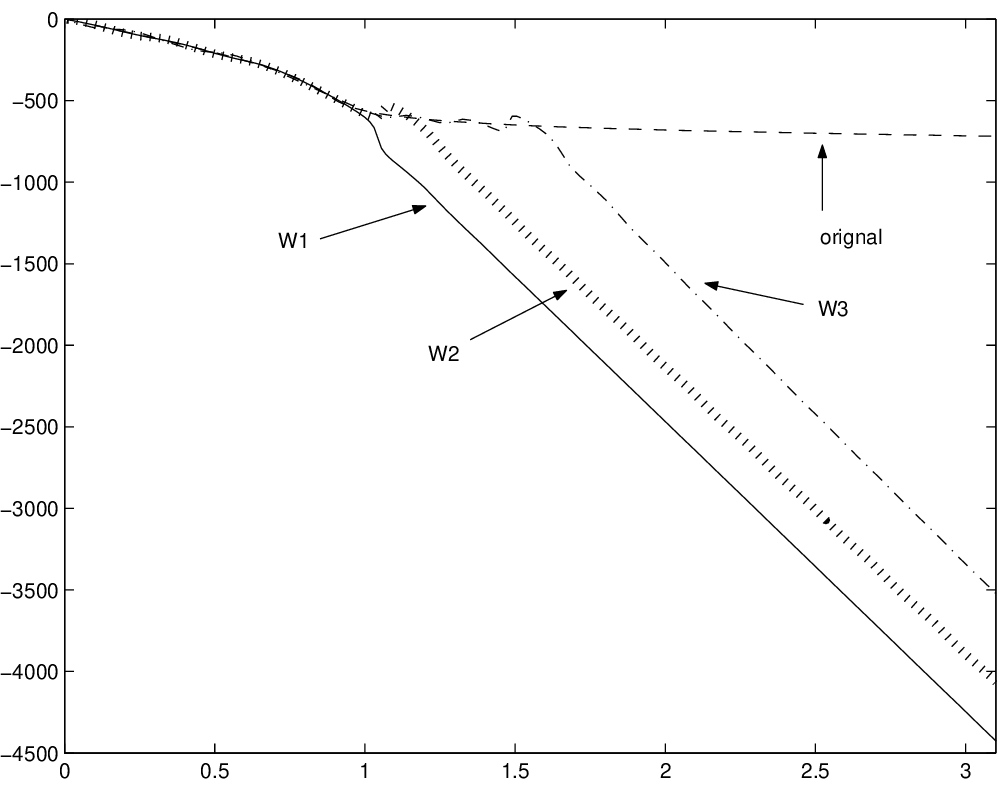}
  \end{center}
  \caption{Phase responses of FIR filters via LMI}
% : $W1$(---):$W2$($\cdots$); $W3$(- $\cdot$ -); Chebyshev(- - -)}
  \label{fir_chip2}
\end{figure}

\begin{figure}[t]
  \begin{center}
    \includegraphics[width=.7\linewidth]{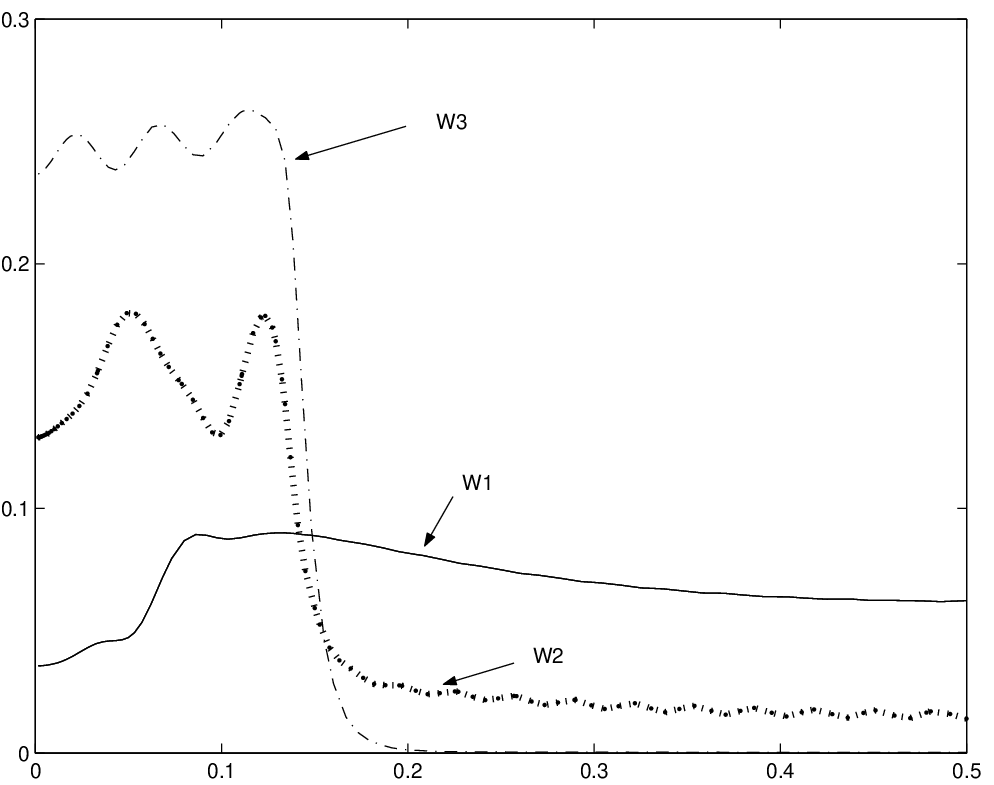}
  \end{center}
  \caption{Gain responses of the error $K_{f}-K$}
%$: $W1$(---):$W2$($\cdots$); $W3$(- $\cdot$ -)}
  \label{fir_dale}
\end{figure}

%\begin{figure}[htbh]
%  \begin{center}
%    \includegraphics[width=.7\linewidth]{errorpha.eps}
%  \end{center}
%  \caption{Phase responses of the error $K_{f}-K$}
%: $W1$(---):$W2$($\cdots$); $W3$(- $\cdot$ -)}
%  \label{dale2}
%\end{figure}

\begin{table}
\caption{$H^{\infty}$ and $H^{2}$ error norm}
\vspace*{0.5cm}
\begin{center}
\begin{tabular}{|c|c|c|c|} \hline
{}& $W_{1}$ & $W_{2}$ & $W_{3}$ \\ \hline
$H^{\infty}$ error norm & 0.0954 & 0.1838 & 0.2627 \\ \hline
$H^{2}$ error norm &{0.0713} &{0.0840} &{0.1333} \\ \hline
\end{tabular}
\end{center}
\label{fir_hyou}
\end{table} 

\section{Conclusion}
We have given an LMI solution to the optimal $H^{\infty}$ 
approximation of IIR filters via FIR filters.  
A comparison with the Nehari shuffle is made with
a numerical example, and it is observed that 
the LMI solution generally performs better.  
Another numerical study also indicates that 
there is a trade-off between the pass-band and
stop-band approximation characteristics.

%\bigskip
%{\bf Acknowledgment}  The authors wish to thank Dr.\ P. Kootsookos and
%Professor R. Bitmead for providing us with detailed data of the
%weighting function in Figure \ref{fig:weight}, and particularly the
%MATLAB routines generating it.
%\begin{figure}[tbh]
%\begin{center}
%\includegraphics[width=.7\linewidth]{error_w_p}
%\end{center}
%\caption{Phase characteristics of the error $K-K_f$:
%$H^\infty$ design via LMI (solid); Nehari shuffle (dots)}
%\label{fig:error_w_p}
%\end{figure}
%\input{conclusion}
\chapter{Conclusion}
\label{ch:conclusion}
In this thesis, we proposed a new design method for multirate
digital signal processing and digital communication systems.
Conventionally, they are designed under the assumption that the original analog signal 
is fully band-limited,
while our method takes the analog characteristic into account
that goes often beyond the Nyquist frequency,
and optimizes the analog performance via
the sampled-data $H^\infty$ optimization.

In Chapter \ref{ch:multirate}, we have presented a sampled-data design
for multirate signal processors, in particular, interpolation, decimation and
sampling rate conversion.
Conventionally, a filter used in these systems is designed 
to have a sharp characteristic which approximates the ideal filter.
However, as shown by design examples, such a sharp filter is not 
necessarily optimal for reconstruction.
This fact will not be recognized without taking the analog signal into account.

In Chapter \ref{ch:comm}, we have treated communication systems
which contains signal compression.
Under distortions by a channel, we have presented a design method
of a transmitting filter and a receiving filter 
by using sampled-data $H^\infty$ optimization.
By iterating a transmitting filter design and a receiving filter design,
we can obtain sub-optimal filters.
We have shown that the objective function monotonically decreases
by the iteration.

In Chapter \ref{ch:qqq}, we have investigated the stability and the performance
of quantized sampled-data control systems.
By using an additive noise model (linearized model) for the quantization, 
we have shown that if the linearized model is stable, the states of the quantized system 
are bounded, and that if the linearized model has small $L^2$ gain, the quantized system
has small power gain.
Then we have applied the results to DPCM design.
We have pointed out that the conventional $\Delta$ modulation is not stable and 
a channel noise will be amplified at the decoder.
Therefore we have proposed a design of the decoder and the encoder;
the decoder reduces the quantization noise, while the encoder reduces 
the channel noise.

In Chapter \ref{ch:firapp}, 
We have given an LMI solution to the optimal $H^{\infty}$ 
approximation of IIR filters via FIR filters.  
A comparison with the Nehari shuffle is made with
a numerical example, and it is observed that 
the LMI solution generally performs better.  
Another numerical study also indicates that 
there is a trade-off between the pass-band and
stop-band approximation characteristics.

We conclude by giving future directions of research as follows:

\begin{itemize}
\item We have treated one-dimensional signals (e.g., audio/speech)
and we believe that the present method can be effectively extended to image processing,
which is a future direction of our research.
In particular, the popular format JPEG or MPEG is a multirate filter bank system,
which can be designed by using the method discussed in Chapter \ref{ch:multirate}.

\item We have discussed a design of communication systems, where the channel is time-invariant.
However, the real channel often contains time-varying systems, in particular, in the case of wireless communication.
Moreover, the real channel is very complicated and we should notice that
the model of the channel always contains a modeling error.
To overcome this, we have to choose an adaptive filter.
Design for adaptive filters by using sampled-data theory is an important subject for the future.
\end{itemize}
%\input{bibliograph}

%%%
%%%	setting for Emacs
%%%
%%% Local Variables:
%%% tab-width:4
%%% tab-stop-list:(4 8 12 16 20 24 28 32 36 40 44 48 52 56 60 64 68 72 76 80 84 88 92 96 100 104 108 112 116 120)
%%% YaTeX-environment-indent:4
%%% End:

%----------------------------------------¼Õ¼­
%\include{thanks}
%----------------------------------------ÉÕÏ¿
%\include{appendix}
%----------------------------------------½¤ÏÀ¤Î½ª¤ï¤ê
\end{document}